%% file: dl_focm.tex
\documentclass[letter, 11pt]{article}
\usepackage[T1]{fontenc}
\usepackage{tgpagella}
\usepackage{graphicx,amsmath,amsfonts,amscd,amssymb,bm,url,color,latexsym}
\usepackage{fullpage}
\usepackage[small,bf]{caption}
\usepackage{subcaption}
\usepackage{bbm}
\usepackage{microtype}
\usepackage{algorithm}
\usepackage{algorithmicx}
\usepackage{algpseudocode} 
\usepackage{hyperref}
\usepackage{verbatim}
\usepackage{framed}

\allowdisplaybreaks  

\hypersetup{
    colorlinks=true,%
    citecolor=blue,%
    filecolor=blue,%
    linkcolor=blue,%
    urlcolor=blue
}
\usepackage[toc, page]{appendix}

\newtheorem{theorem}{Theorem}[section]
\newtheorem{lemma}[theorem]{Lemma}
\newtheorem{corollary}[theorem]{Corollary}
\newtheorem{proposition}[theorem]{Proposition}

\usepackage{tikz}
\usepackage{fge}

\renewcommand{\mathbf}{\boldsymbol}

\newcommand{\mb}{\mathbf}
\newcommand{\mc}{\mathcal}

\newcommand{\bb}{\mathbb}
\newcommand{\magnitude}[1]{ \left| #1 \right| } 
\newcommand{\set}[1]{\left\{ #1 \right\}}
\newcommand{\condset}[2]{ \left\{ #1 \;\middle|\; #2 \right\} }

\newcommand{\reals}{\bb R}

\newcommand{\eps}{\varepsilon}
\newcommand{\R}{\reals}

\newcommand{\N}{\bb N}

\newcommand{\indicator}[1]{\mathbbm 1_{#1}}

\newcommand{ \Brac }[1]{\left\lbrace #1 \right\rbrace}
\newcommand{ \brac }[1]{\left[ #1 \right]}
\newcommand{ \paren }[1]{ \left( #1 \right) }

\DeclareMathOperator{\trace}{tr}
\DeclareMathOperator{\supp}{supp}

\DeclareMathOperator{\diag}{diag}
\DeclareMathOperator{\sign}{sign}
\DeclareMathOperator{\grad}{grad}
\DeclareMathOperator{\Hess}{Hess}
\DeclareMathOperator{\mini}{minimize}

\DeclareMathOperator{\st}{subject\; to}

\newcommand{\event}{\mc E}

\newcommand{\rconcave}{r_\fgecap}
\newcommand{\Lconcave}{L_\fgecap}
\newcommand{\rconvex}{r_\fgecup}

\newcommand{\Lconvex}{L_\fgecup}

\newcommand{\wh}{\widehat}
\newcommand{\wt}{\widetilde}
\newcommand{\ol}{\overline}

\newcommand{\betaconcave}{\beta_\fgecap}
\newcommand{\rI}{R_{\mathtt{I}}}
\newcommand{\rII}{R_{\mathtt{II}}}
\newcommand{\rIII}{R_{\mathtt{III}}}
\newcommand{\dI}{d_{\mathtt{I}}}
\newcommand{\dII}{d_{\mathtt{II}}}
\newcommand{\dIII}{d_{\mathtt{III}}}
\newcommand{\betagrad}{\beta_{\mathrm{grad}}}

\newcommand{\norm}[2]{\left\| #1 \right\|_{#2}}
\newcommand{\abs}[1]{\left| #1 \right|}
\newcommand{\innerprod}[2]{\left\langle #1,  #2 \right\rangle}
\newcommand{\prob}[1]{\bb P\left[ #1 \right]}
\newcommand{\expect}[1]{\bb E\left[ #1 \right]}

\numberwithin{equation}{section}

\def \endprf{\hfill {\vrule height6pt width6pt depth0pt}\medskip}
\newenvironment{proof}{\noindent {\bf Proof} }{\endprf\par}

\title{Complete Dictionary Recovery over the Sphere}
\author{Ju Sun, Qing Qu, and John Wright \\
\texttt{\{js4038, qq2105, jw2966\}@columbia.edu} \\
Department of Electrical Engineering, Columbia University, New York, USA
}

\date{April 25, 2015  \quad Revised: \today}
\begin{document}
\maketitle

\vspace{-0.3in}
\begin{abstract}
We consider the problem of recovering a complete (i.e., square and invertible) matrix $\mb A_0$, from $\mb Y \in \R^{n \times p}$ with $\mb Y = \mb A_0 \mb X_0$, provided $\mb X_0$ is sufficiently sparse. This recovery problem is central to the theoretical understanding of dictionary learning, which seeks a sparse representation for a collection of input signals, and finds numerous applications in modern signal processing and machine learning. We give the first efficient algorithm that provably recovers $\mb A_0$ when $\mb X_0$ has $O\paren{n}$ nonzeros per column, under suitable probability model for $\mb X_0$. In contrast, prior results based on efficient algorithms provide recovery guarantees when $\mb X_0$ has only $O\paren{n^{1-\delta}}$ nonzeros per column for any constant $\delta \in (0, 1)$.

Our algorithmic pipeline centers around solving a certain nonconvex optimization problem with a spherical constraint, and hence is naturally phrased in the language of manifold optimization. To show this apparently hard problem is tractable, we first provide a geometric characterization of the high-dimensional objective landscape, which shows that with high probability there are no ``spurious'' local minima. This particular geometric structure allows us to design a Riemannian trust region algorithm over the sphere that provably converges to one local minimizer with an arbitrary initialization, despite the presence of saddle points. The geometric approach we develop here may also shed light on other problems arising from nonconvex recovery of structured signals. 
\end{abstract}



\indent {\bf Acknowledgement.} We thank Dr. Boaz Barak for pointing out an inaccurate comment made on overcomplete dictionary learning using SOS. We thank Cun Mu and Henry Kuo of Columbia University for discussions related to this project. JS thanks the Wei Family Private Foundation for their generous support. This work was partially supported by grants ONR N00014-13-1-0492, NSF 1343282, and funding from the Moore and Sloan Foundations. 

\indent {\bf Note.} This technical report has subsequently been divided into two papers~\cite{sun2015complete_a} and~\cite{sun2015complete_b}. All future updates will be made only to the separate papers. 

\newpage
\tableofcontents 

\input{sec/intro}
\input{sec/geometry}

\input{sec/algorithm}

\input{sec/main_result}

\input{sec/exp}

\input{sec/discuss}
\input{sec/proof_geometry}

\input{sec/proof_algorithm}

\input{sec/proof_main}

\input{sec/appendix}

{\small
\bibliographystyle{amsalpha}
\bibliography{dl_focm,ncvx}
}

\end{document}

%% file: sec/intro.tex
\section{Introduction}
Given $p$ signal samples from $\R^n$, i.e., $\mb Y \doteq \brac{\mb y_1, \dots, \mb y_p}$, is it possible to construct a dictionary $\mb A \doteq \brac{\mb a_1, \dots, \mb a_m}$ with $m$ much smaller than $p$, such that $\mb Y \approx \mb A \mb X$ and the coefficient matrix $\mb X$ has as few nonzeros as possible? In other words, this model \emph{dictionary learning} (DL) problem seeks a concise representation for a collection of input signals. Concise signal representations play a central role in compression, and also prove useful for many other important tasks, such as signal acquisition, denoising, and classification. 

Traditionally, concise signal representations have relied heavily on explicit analytic bases constructed in nonlinear approximation and harmonic analysis. This constructive approach has proved highly successfully; the numerous theoretical advances in these fields (see, e.g., ~\cite{devore1998nonlinear, temlyakov2003nonlinear, devore2009nonlinear, candes2002new, ma2010review} for summary of relevant results) provide ever more powerful representations, ranging from the classic Fourier to modern multidimensional, multidirectional, multiresolution bases, including wavelets, curvelets, ridgelets, and so on. However, two challenges confront practitioners in adapting these results to new domains: which function class best describes signals at hand, and consequently which representation is most appropriate. These challenges are coupled, as function classes with known ``good'' analytic bases are rare. \footnote{As Donoho et al~\cite{donoho1998data} put it, ``...in effect, uncovering the optimal codebook structure of naturally occurring data involves more challenging empirical questions than any that have ever been solved in empirical work in the mathematical sciences.''}

Around 1996, neuroscientists Olshausen and Field discovered that sparse coding, the principle of encoding a signal with few atoms from a learned dictionary, reproduces important properties of the receptive fields of the simple cells that perform early visual processing~\cite{olshausen1996emergence, olshausen1997sparse}. The discovery has spurred a flurry of algorithmic developments and successful applications for DL in the past two decades, spanning classical image processing, visual recognition, compressive signal acquisition, and also recent deep architectures for signal classification (see, e.g., \cite{elad2010sparse, mairal2014sparse} for review this development). 

The learning approach is particularly relevant to modern signal processing and machine learning, which deal with data of huge volume and great variety (e.g., images, audios, graphs, texts, genome sequences, time series, etc). The proliferation of problems and data seems to preclude analytically deriving optimal representations for each new class of data in a timely manner. On the other hand, as datasets grow, learning dictionaries directly from data looks increasingly attractive and promising. When armed with sufficiently many data samples of one signal class, by solving the model DL problem, one would expect to obtain a dictionary that allows sparse representation for the whole class. This hope has been borne out in a number of successful examples~\cite{elad2010sparse, mairal2014sparse} and theories~\cite{maurer2010dimensional, Vainsencher:2011, Mehta13, gribonval2013sample}. 

\subsection{Theoretical and Algorithmic Challenges}

In contrast to the above empirical successes, the theoretical study of dictionary learning is still developing. For applications in which dictionary learning is to be applied in a ``hands-free'' manner, it is desirable to have efficient algorithms which are guaranteed to perform correctly, when the input data admit a sparse model. There have been several important recent results in this direction, which we will review in Section \ref{sec:lit_review}, after our sketching main results. Nevertheless, obtaining algorithms that provably succeed under broad and realistic conditions remains an important research challenge. 

To understand where the difficulties arise, we can consider a model formulation, in which we attempt to obtain the dictionary $\mb A$ and coefficients $\mb X$ which best trade-off sparsity and fidelity to the observed data:
\begin{align} \label{eq:dl_concept}
\mini_{\mb A \in \R^{n \times m}, \mb X \in \R^{m \times p}}\; \lambda \norm{\mb X}{1} + \frac{1}{2} \norm{\mb A \mb X - \mb Y}{F}^2, \; \st \;  \mb A \in \mc A.  
\end{align}
Here, $\norm{\mb X}{1} \doteq \sum_{i, j} \abs{X_{ij}}$ promotes sparsity of the coefficients, $\lambda \ge 0$ trades off the level of coefficient sparsity and quality of approximation, and $\mc A$ imposes desired structures on the dictionary.

This formulation is nonconvex: the admissible set $\mc A$ is typically nonconvex (e.g., orthogonal group, matrices with normalized columns)\footnote{For example, in nonlinear approximation and harmonic analysis, orthonormal basis or (tight-)frames are preferred; to fix the scale ambiguity discussed in the text, a common practice is to require that $\mb A$ to be column-normalized. There is no obvious reason to believe that convexifying these constraint sets would leave the optima unchanged. For example, the convex hull of the orthogonal group $O_n$ is the operator norm ball $\Brac{\mb X \in \R^{n \times n}: \norm{\mb X}{} \le 1}$. If there are no effective symmetry breaking constraints, any convex objective function tends to have minimizers inside the ball, which obviously will not be orthogonal matrices. Other ideas such as lifting may not play together with the objective function, nor yield tight relaxations (see, e.g.,~\cite{bandeira2013approximating, briet2014tight}).}, while the most daunting nonconvexity comes from the bilinear mapping: $\paren{\mb A, \mb X} \mapsto \mb A \mb X$. Because $\paren{\mb A, \mb X}$ and $\paren{\mb A \mb \Pi \mb \Sigma, \mb \Sigma^{-1} \mb \Pi^* \mb X}$ result in the same objective value for the conceptual formulation~\eqref{eq:dl_concept}, where $\mb \Pi$ is any permutation matrix, and $\mb \Sigma$ any diagonal matrix with diagonal entries in $\{ \pm 1 \}$, and $\paren{\cdot}^*$ denotes matrix transpose. Thus, we should expect the problem to have combinatorially many global minima. Because there are multiple isolated global minima, the problem does not appear to be amenable to convex relaxation (see similar discussions in, e.g.,~\cite{gribonval2010dictionary} and~\cite{geng2011local}).\footnote{Semidefinite programming (SDP) lifting may be one useful general strategy to convexify bilinear inverse problems, see, e.g., \cite{ahmed2014blind, choudhary2014identifiability}. However, for problems with general nonlinear constraints, it is unclear whether the lifting always yield tight relaxation, consider, e.g.,~\cite{bandeira2013approximating, briet2014tight} again.} This contrasts sharply with problems in sparse recovery and compressed sensing, in which simple convex relaxations are often provably effective
\cite{donoho2009observed, oymak2010new, candes2011robust, donoho2013phase, mccoy2014sharp, mu2013square, chandrasekaran2012convex, candes2013phaselift, amelunxen2014living, candes2014mathematics}. Is there any hope to obtain global solutions to the DL problem? 

\subsection{An Intriguing Numerical Experiment with Real Images} \label{sec:intro_exp}
\begin{figure}[t]
\centering  
\begin{subfigure}[t]{0.3\textwidth}
\centering
\includegraphics[width = 0.9\linewidth]{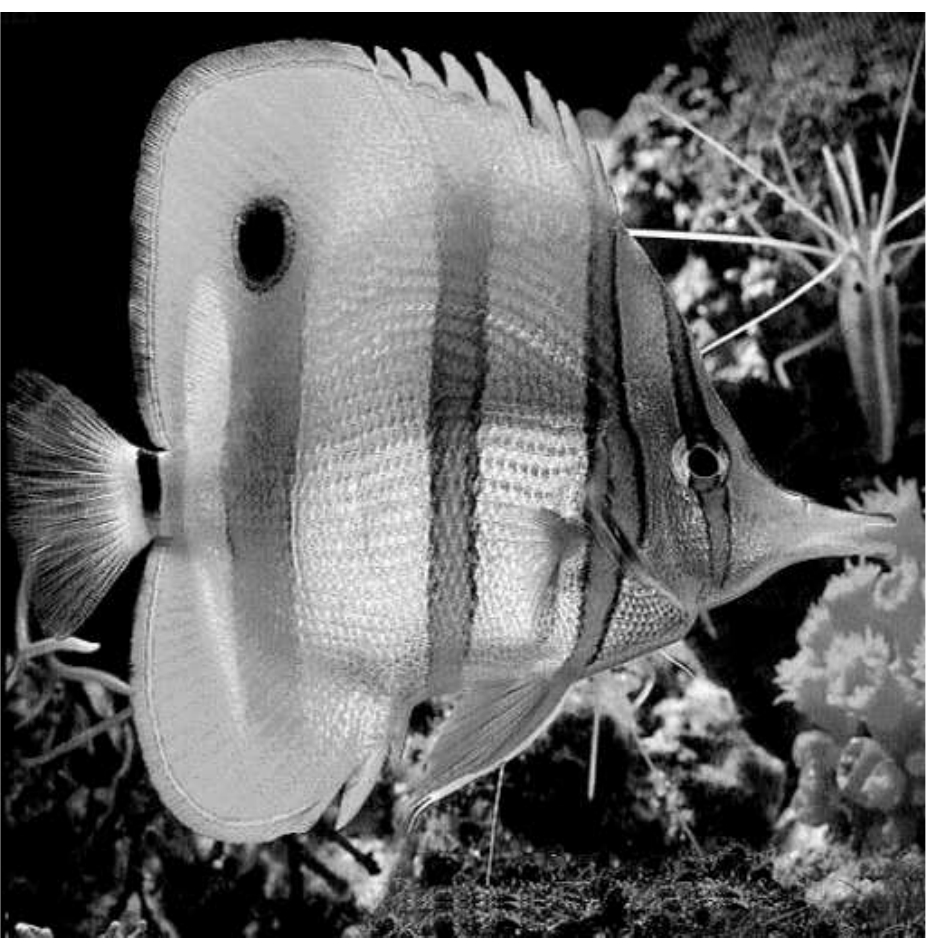} 
\end{subfigure}
\begin{subfigure}[t]{0.3\textwidth}
\centering
\includegraphics[width = 0.9\linewidth]{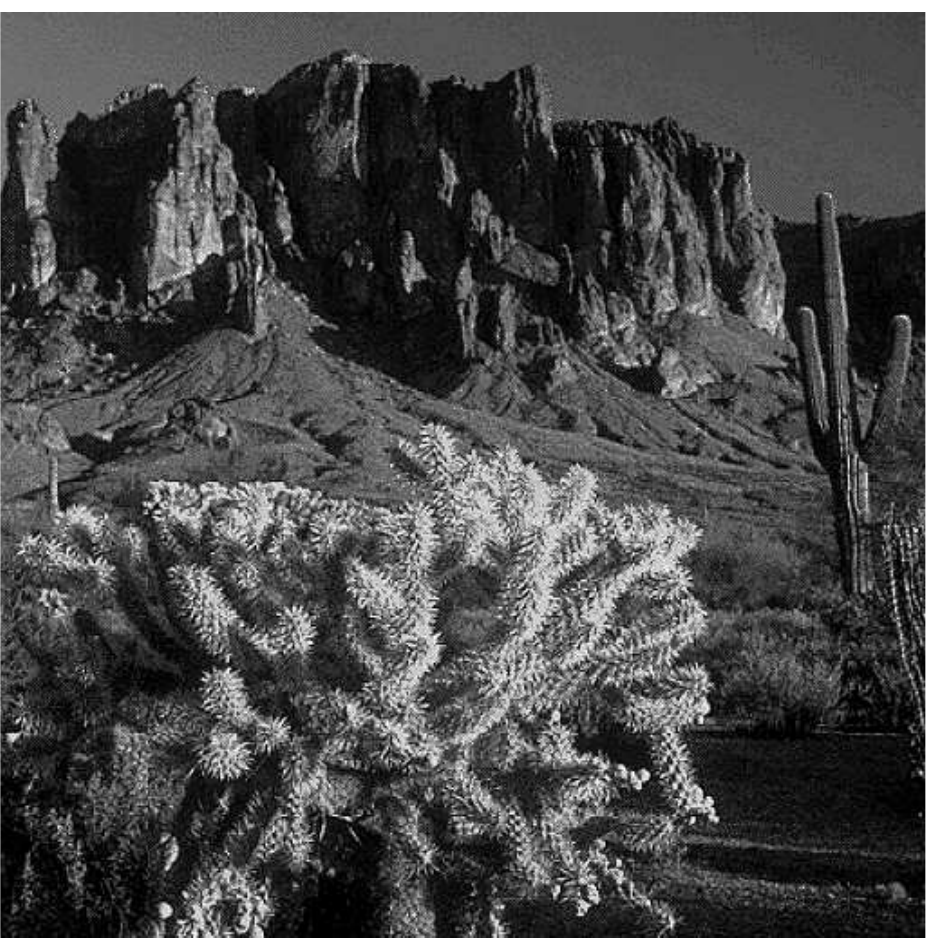} 
\end{subfigure}
\begin{subfigure}[t]{0.3\textwidth}
\centering
\includegraphics[width = 0.9\linewidth]{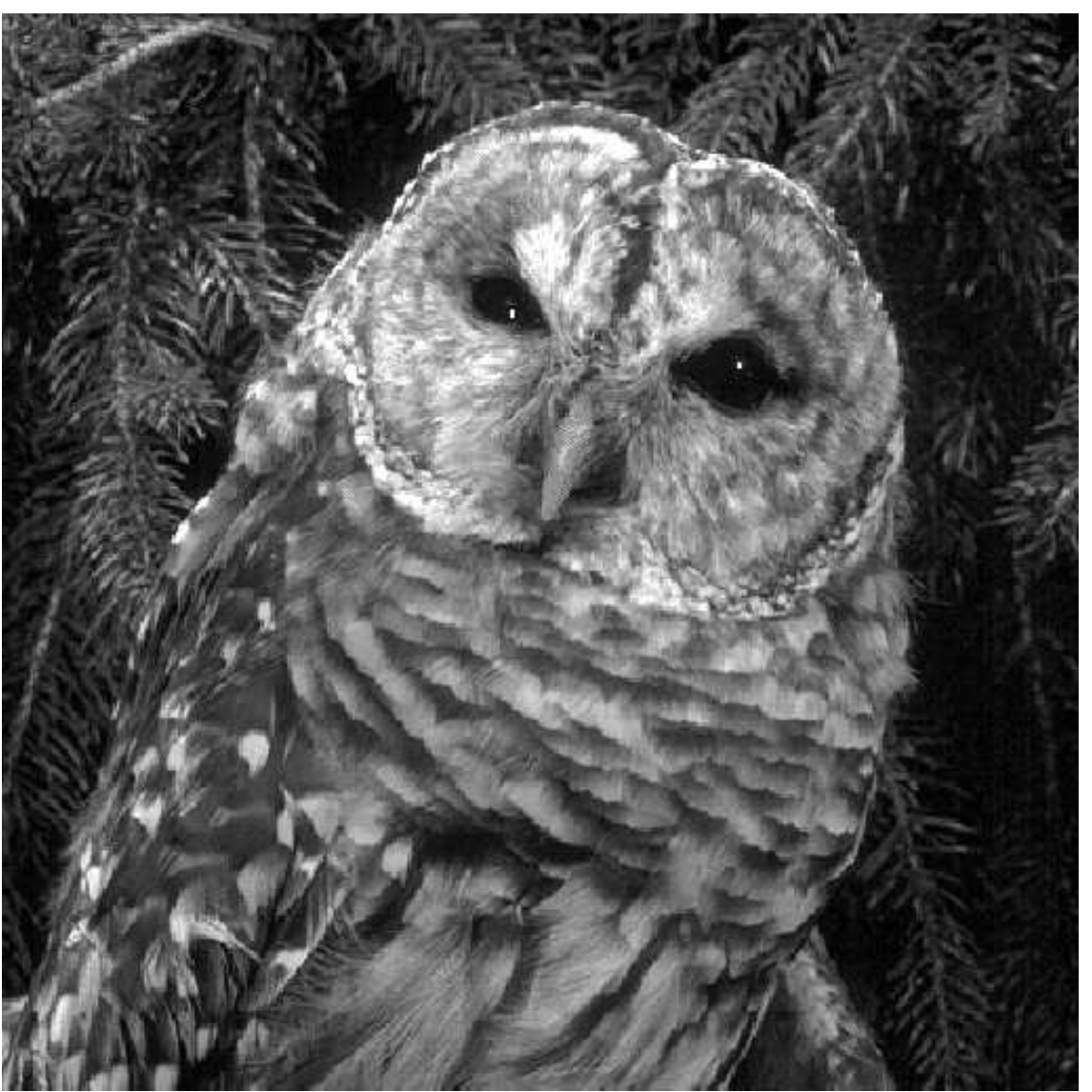} 
\end{subfigure} \\
\begin{subfigure}[t]{0.3\textwidth}
\centering
\includegraphics[width = 0.95\linewidth]{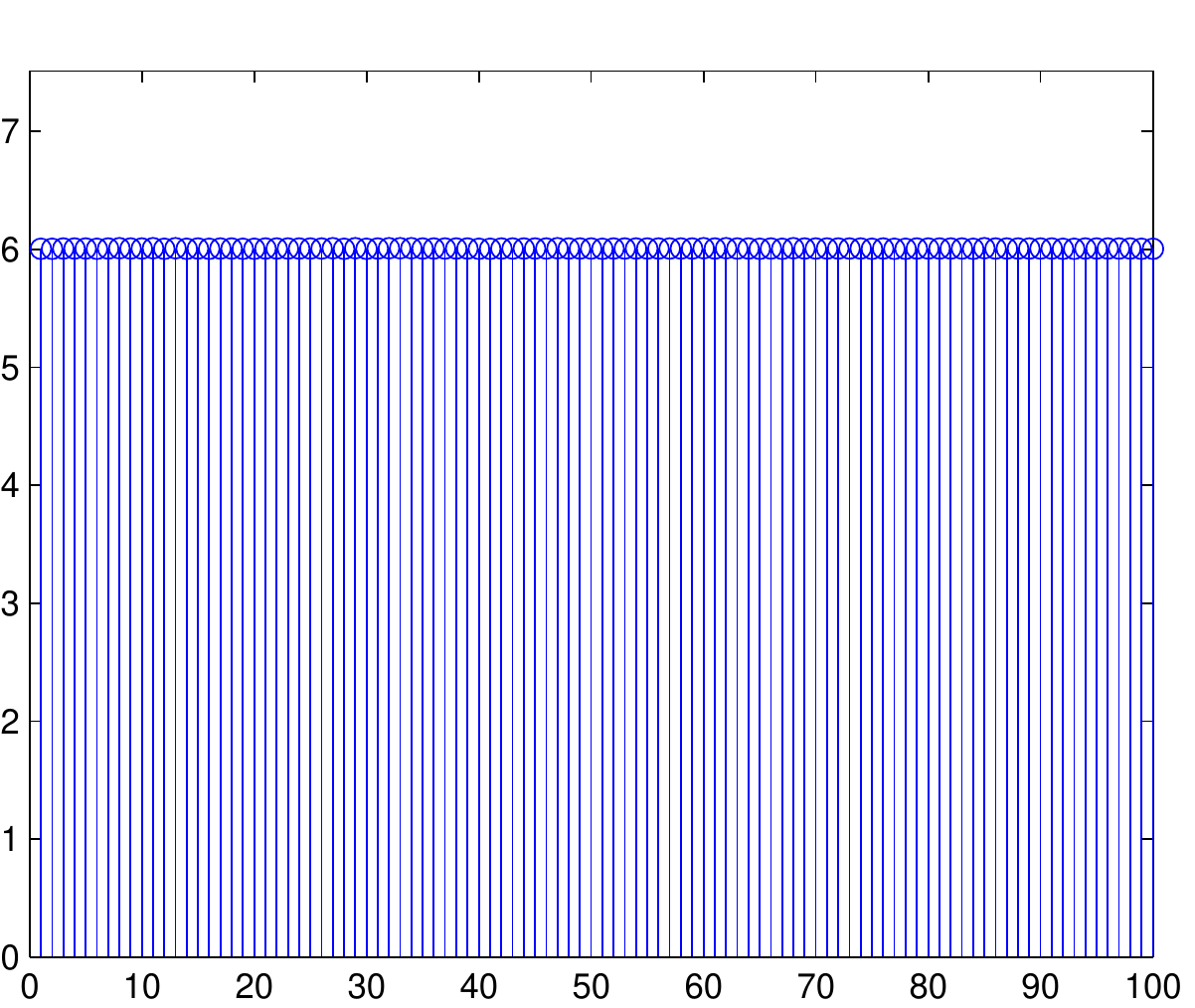} 
\end{subfigure}
\begin{subfigure}[t]{0.3\textwidth}
\centering
\includegraphics[width = 0.95\linewidth]{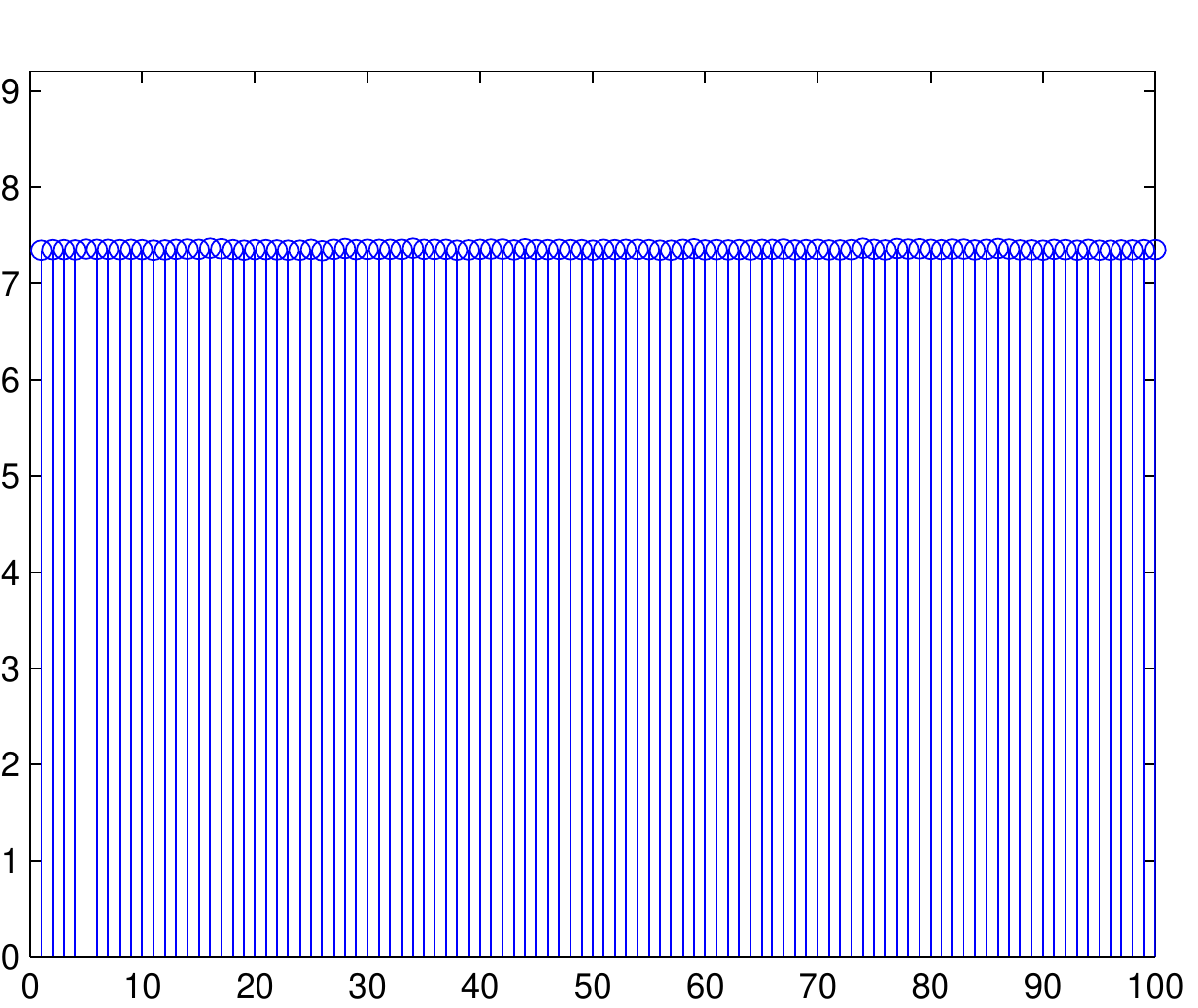}
\end{subfigure}
\begin{subfigure}[t]{0.3\textwidth}
\centering
\includegraphics[width = 0.95\linewidth]{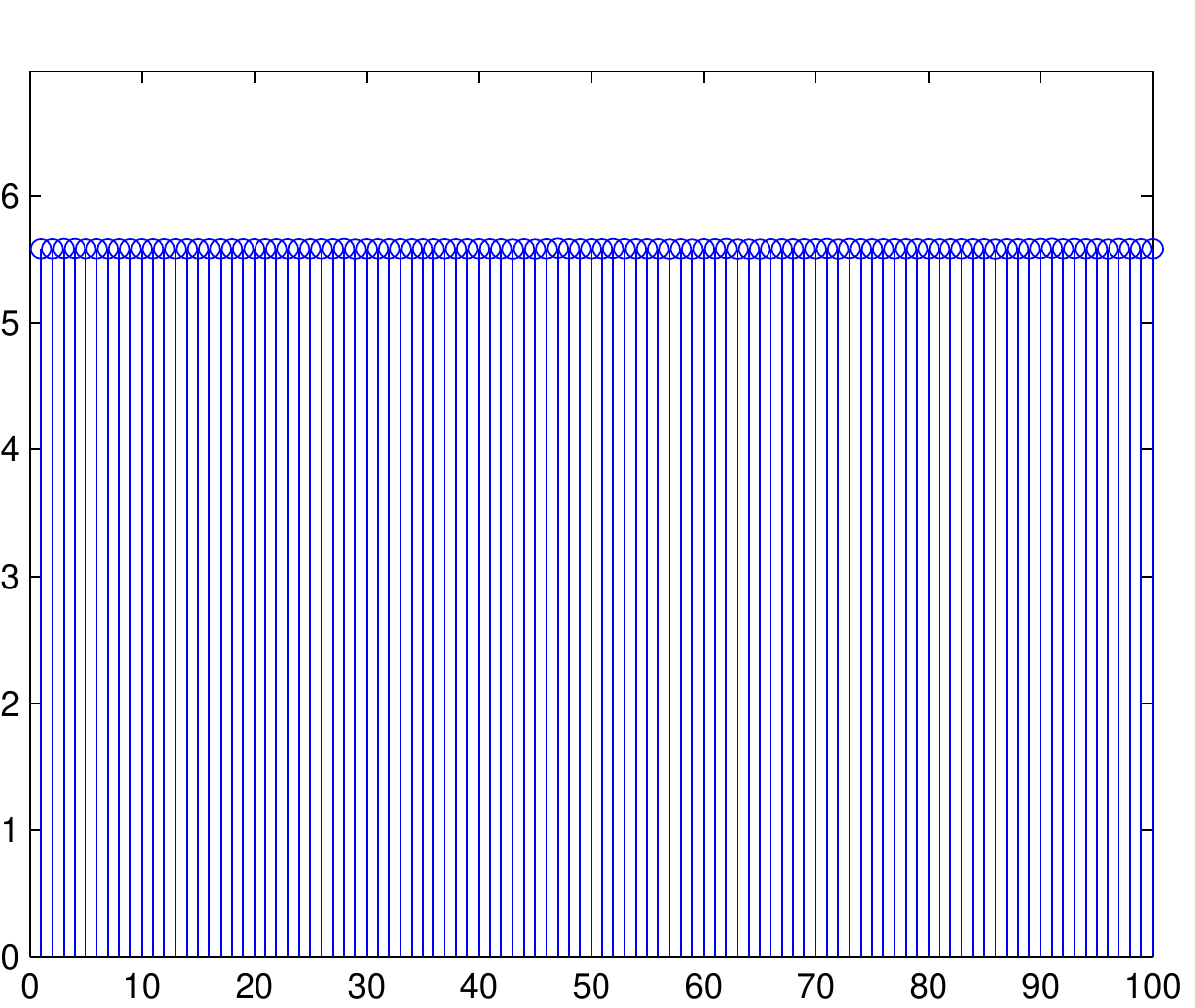}
\end{subfigure}
\caption{\textbf{Alternating direction method for~\eqref{eq:p_l1p} on uncompressed real images seems to always produce the same solution!} \textbf{Top}: Each image is $512 \times 512$ in resolution and encoded in the uncompressed \texttt{pgm} format (uncompressed images to prevent possible bias towards standard bases used for compression, such as DCT or wavelet bases). Each image is evenly divided into $8 \times 8$ non-overlapping image patches ($4096$ in total), and these patches are all vectorized and then stacked as columns of the data matrix $\mb Y$. \textbf{Bottom}: Given each $\mb Y$, we solve~\eqref{eq:p_l1p} $100$ times with independent and randomized (uniform over the orthogonal group) initialization $\mb A_0$. The plots show the values of $\norm{\mb A^*_{\infty} \mb Y}{1}$ across the independent repetitions. They are virtually the same and the relative differences are less than $10^{-3}$! }
\label{fig:odl_examples}
\end{figure}
We provide empirical evidence in support of a positive answer to the above question. Specifically, we learn orthogonal bases (orthobases) for real images patches. Orthobases are of interest because typical hand-designed dictionaries such as discrete cosine (DCT) and wavelet bases are orthogonal, and orthobases seem competitive in performance for applications such as image denoising, as compared to overcomplete dictionaries~\cite{bao2013fast}\footnote{See Section~\ref{sec:intro_rec} for more detailed discussions of this point. \cite{lesage2005learning} also gave motivations and algorithms for learning (union of) orthobases as dictionaries. }. 

We divide a given greyscale image into $8 \times 8$ non-overlapping patches, which are converted into $64$-dimensional vectors and stacked column-wise into a data matrix $\mb Y$. Specializing~\eqref{eq:dl_concept} to this setting, we obtain the optimization problem: 
\begin{align} \label{eq:p_l1p}
\mini_{\mb A \in \R^{n \times n}, \mb X \in \R^{n \times p}}\; \lambda \norm{\mb X}{1} + \frac{1}{2}\norm{\mb A \mb X - \mb Y}{F}^2, \; \st \; \mb A \in O_n. 
\end{align}
To derive a concrete algorithm for~\eqref{eq:p_l1p}, one can deploy the alternating direction method (ADM)\footnote{This method is also called alternating minimization or (block) coordinate descent method. see, e.g., ~\cite{bertsekas1989parallel, tseng2001convergence} for classic results and~\cite{attouch2010proximal,bolte2014proximal} for several interesting recent developments. }, i.e., alternately minimizing the objective function with respect to (w.r.t.) one variable while fixing the other. The iteration sequence actually takes very simple form: for $k = 1, 2, 3, \dots$, 
\begin{align*}
\mb X_k = \mc S_{\lambda}\brac{\mb A^*_{k-1} \mb Y}, \qquad \mb A_k = \mb U \mb V^* \;\text{for} \; \mb U \mb D \mb V^* = \mathtt{SVD}\paren{\mb Y \mb X_k^*}
\end{align*}
where $\mc S_{\lambda}\brac{\cdot}$ denotes the well-known soft-thresholding operator acting elementwise on matrices, i.e.,  $\mc S_{\lambda}\brac{x} \doteq \sign\paren{x} \max\paren{\abs{x} - \lambda, 0}$ for any scalar $x$. 

Figure~\ref{fig:odl_examples} shows what we obtained using the simple ADM algorithm, with \emph{independent and randomized initializations}: 
\begin{quote}
The algorithm seems to always produce the same solution, regardless of the initialization.
\end{quote} 
This observation implies the heuristic ADM algorithm may \emph{always converge to one global minimizer}! \footnote{Technically, the converge to global solutions is surprising because even convergence of ADM to critical points is atypical, see, e.g., \cite{attouch2010proximal,bolte2014proximal} and references therein. Section~\ref{sec:discuss} includes more detailed discussions on this point.} Equally surprising is that the phenomenon has been observed on real images\footnote{Actually the same phenomenon is also observed for simulated data when the coefficient matrix obeys the Bernoulli-Gaussian model, which is defined later. The result on real images supports that previously claimed empirical successes over two decades may be non-incidental. }. One may imagine only random data typically have ``favorable'' structures; in fact, almost all existing theories for DL pertain only to random data~\cite{spielman2012exact,agarwal2013learning,arora2013new,agarwal2013exact,arora2014more,arora2015simple}. 

\subsection{Dictionary Recovery and Our Results} \label{sec:intro_rec}
In this paper, we take a step towards explaining the surprising effectiveness of nonconvex optimization heuristics for DL. We focus on the \emph{dictionary recovery} (DR) setting: given a data matrix $\mb Y$ generated as $\mb Y = \mb A_0 \mb X_0$, where $\mb A_0 \in \mc A \subseteq \R^{n \times m}$ and $\mb X_0 \in \R^{m \times p}$ is ``reasonably sparse'', try to recover $\mb A_0$ and $\mb X_0$. Here recovery means to return any pair $\paren{\mb A_0 \mb \Pi \mb \Sigma, \mb \Sigma^{-1} \mb \Pi^* \mb X_0}$, where $\mb \Pi$ is a permutation matrix and $\mb \Sigma$ is a nonsingular diagonal matrix, i.e., recovering up to sign, scale, and permutation.

To define a reasonably simple and structured problem, we make the following assumptions: 
\begin{itemize}
\item The target dictionary $\mb A_0$ is complete, i.e., square and invertible ($m = n$). In particular, this class includes orthogonal dictionaries. Admittedly overcomplete dictionaries tend to be more powerful for modeling and to allow sparser representations. Nevertheless, most classic hand-designed dictionaries in common use are orthogonal. Orthobases are competitive in performance for certain tasks such as image denoising~\cite{bao2013fast}, and admit faster algorithms for learning and encoding. \footnote{Empirically, there is no systematic evidence supporting that overcomplete dictionaries are strictly necessary for good performance in all published applications (though~\cite{olshausen1997sparse} argues for the necessity from neuroscience perspective). Some of the ideas and tools developed here for complete dictionaries may also apply to certain classes of structured overcomplete dictionaries, such as tight frames. See Section~\ref{sec:discuss} for relevant discussion. }  
\item The coefficient matrix $\mb X_0$ follows the Bernoulli-Gaussian (BG) model with rate $\theta$: $\brac{\mb X_0}_{ij} = \Omega_{ij} V_{ij}$, with $\Omega_{ij} \sim \mathrm{Ber}\paren{\theta}$ and $V_{ij} \sim \mc N\paren{0, 1}$, where all the different random variables are mutually independent. We write compactly $\mb X_0 \sim_{i.i.d.} \mathrm{BG}\paren{\theta}$. 
\end{itemize}
We prove the following result: 
\begin{theorem}[Informal statement of our results]
For any $\theta \in \paren{0, 1/3}$, given $\mb Y = \mb A_0 \mb X_0$ with $\mb A_0$ a complete dictionary and $\mb X_0 \sim_{i.i.d.} \mathrm{BG}\paren{\theta}$, there is a polynomial time algorithm that recovers $\mb A_0$ and $\mb X_0$ with high probability (at least $1-O(p^{-6})$) whenever $p \ge p_{\star}\paren{n, 1/\theta, \kappa\paren{\mb A_0}, 1/\mu}$ for a fixed polynomial $p_\star\paren{\cdot}$, where $\kappa\paren{\mb A_0}$ is the condition number of $\mb A_0$ and $\mu$ is a parameter that can be set as $cn^{-5/4}$ for a fixed positive numerical constant $c$. 
\end{theorem}
Obviously, even if $\mb X_0$ is known, one needs $p \ge n$ to make the identification problem well posed. Under our particular probabilistic model, a simple coupon collection argument implies that one needs $p \ge \Omega\paren{\tfrac{1}{\theta}\log n}$ to ensure all atoms in $\mb A_0$ are observed with high probability (w.h.p.). To ensure that an efficient algorithm exists may demand more. Our result implies when $p$ is polynomial in $n$, $1/\theta$ and $\kappa(\mb A_0)$, recovery with efficient algorithm is possible. 

The parameter $\theta$ controls the sparsity level of $\mb X_0$. Intuitively, the recovery problem is easy for small $\theta$ and becomes harder for large $\theta$.\footnote{Indeed, when $\theta$ is small enough such that columns of $\mb X_0$ are predominately $1$-sparse, one directly observes scaled versions of the atoms (i.e., columns of $\mb X_0$); when $\mb X_0$ is fully dense corresponding to $\theta = 1$, recovery is never possible as one can easily find another complete $\mb A_0'$ and fully dense $\mb X_0'$ such that $\mb Y = \mb A_0' \mb X_0'$ with $\mb A_0'$ not equivalent to $\mb A_0$. 
} It is perhaps surprising that an efficient algorithm can succeed up to constant $\theta$, i.e., linear sparsity in $\mb X_0$. Compared to the case when $\mb A_0$ is known, there is only at most a constant gap in the sparsity level one can deal with. 

For DL, our result gives the first efficient algorithm that provably recovers complete $\mb A_0$ and $\mb X_0$ when $\mb X_0$ has $O(n)$ nonzeros per column under appropriate probability model. Section~\ref{sec:lit_review} provides detailed comparison of our result with other recent recovery results for complete and overcomplete dictionaries. 


\subsection{Main Ingredients and Innovations}
In this section we describe three main ingredients that we use to obtain the stated result. 

\subsubsection{A Nonconvex Formulation}
Since $\mb Y = \mb A_0 \mb X_0$ and $\mb A_0$ is complete, $\mathrm{row}\paren{\mb Y} = \mathrm{row}\paren{\mb X_0}$ ($\mathrm{row}\paren{\cdot}$ denotes the row space of a matrix) and hence rows of $\mb X_0$ are sparse vectors in the known (linear) subspace $\mathrm{row}\paren{\mb Y}$. We can use this fact to first recover the rows of $\mb X_0$, and subsequently recover $\mb A_0$ by solving a system of linear equations. In fact, for $\mb X_0 \sim_{i.i.d.} \mathrm{BG}\paren{\theta}$, rows of $\mb X_0$ are the $n$ \emph{sparsest} vectors (directions) in $\mathrm{row}\paren{\mb Y}$ w.h.p. whenever $p \ge \Omega\paren{n\log n}$~\cite{spielman2012exact}. Thus one might try to recover rows of $\mb X_0$ by solving
\begin{align}
\mini\; \norm{\mb q^* \mb Y}{0}\; \; \st\;\; \mb q \neq \mb 0. 
\end{align}
The objective is discontinuous, and the domain is an open set. In particular, the homogeneous constraint is nonconventional and tricky to deal with. Since the recovery is up to scale, one can remove the homogeneity by fixing the scale of $\mb q$. Known relaxations~\cite{spielman2012exact, demanet2014scaling} fix the scale by setting $\norm{\mb q^* \mb Y}{\infty} = 1$, where $\norm{\cdot}{\infty}$ is the elementwise $\ell^{\infty}$ norm. The optimization problem reduces to a sequence of convex programs, which recover $\paren{\mb A_0, \mb X_0}$ for very sparse $\mb X_0$, but provably break down when columns of $\mb X_0$ has more than $O\paren{\sqrt{n}}$ nonzeros, or $\theta \ge \Omega\paren{1/\sqrt{n}}$. Inspired by our previous image experiment, we work with a \emph{nonconvex} alternative\footnote{A similar formulation has been proposed in~\cite{zibulevsky2001blind} in the context of blind source separation; see also~\cite{qu2014finding}. }:
\begin{align} \label{eq:main_l2}
\mini\;f(\mb q; \widehat{\mb Y}) \doteq \frac{1}{p} \sum_{k=1}^p h_{\mu}\paren{\mb q^* \widehat{\mb y}_k}, \; \st \; \norm{\mb q}{} = 1, 
\end{align}
where $\widehat{\mb Y} \in \R^{n\times p}$ is a proxy for $\mb Y$ (i.e., after appropriate processing), $k$ indexes columns of $\widehat{\mb Y}$, and $\norm{\cdot}{}$ is the usual $\ell^2$ norm for vectors. Here $h_{\mu}\paren{\cdot}$ is chosen to be a convex smooth approximation to $\abs{\cdot}$, namely,  
\begin{align} \label{eq:logexp}
h_{\mu}\paren{z} = \mu \log\paren{\frac{\exp\paren{z/\mu} + \exp\paren{-z/\mu}}{2}} = \mu \log \cosh(z/\mu), 
\end{align} 
which is infinitely differentiable and $\mu$ controls the smoothing level.\footnote{In fact, there is nothing special about this choice and we believe that any valid smooth (twice continuously differentiable) approximation to $\abs{\cdot}$ would work and yield qualitatively similar results. We also have some preliminary results showing the latter geometric picture remains the same for certain nonsmooth functions, such as a modified version of the Huber function, though the analysis involves handling a different set of technical subtleties. The algorithm also needs additional modifications.} The spherical constraint is nonconvex. Hence, a-priori, it is unclear whether \eqref{eq:main_l2} admits efficient algorithms that attain global optima. Surprisingly, simple descent algorithms for \eqref{eq:main_l2} exhibit very striking behavior: on many practical numerical examples\footnote{... not restricted to the model we assume here for $\mb A_0$ and $\mb X_0$. }, they appear to produce global solutions. Our next section will uncover interesting geometrical structures underlying the phenomenon. 

\subsubsection{A Glimpse into High-dimensional Function Landscape} \label{sec:overview_geometry}
\begin{figure}[t]
\centerline{\includegraphics[width=0.3\textwidth]{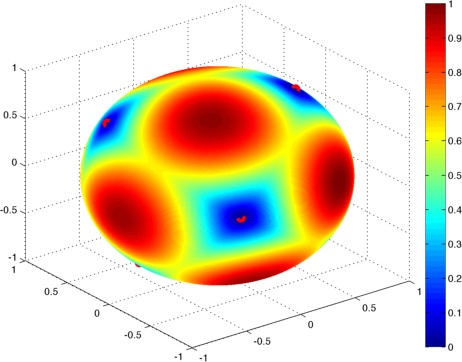} \hspace{.1in} \includegraphics[width=0.3\textwidth]{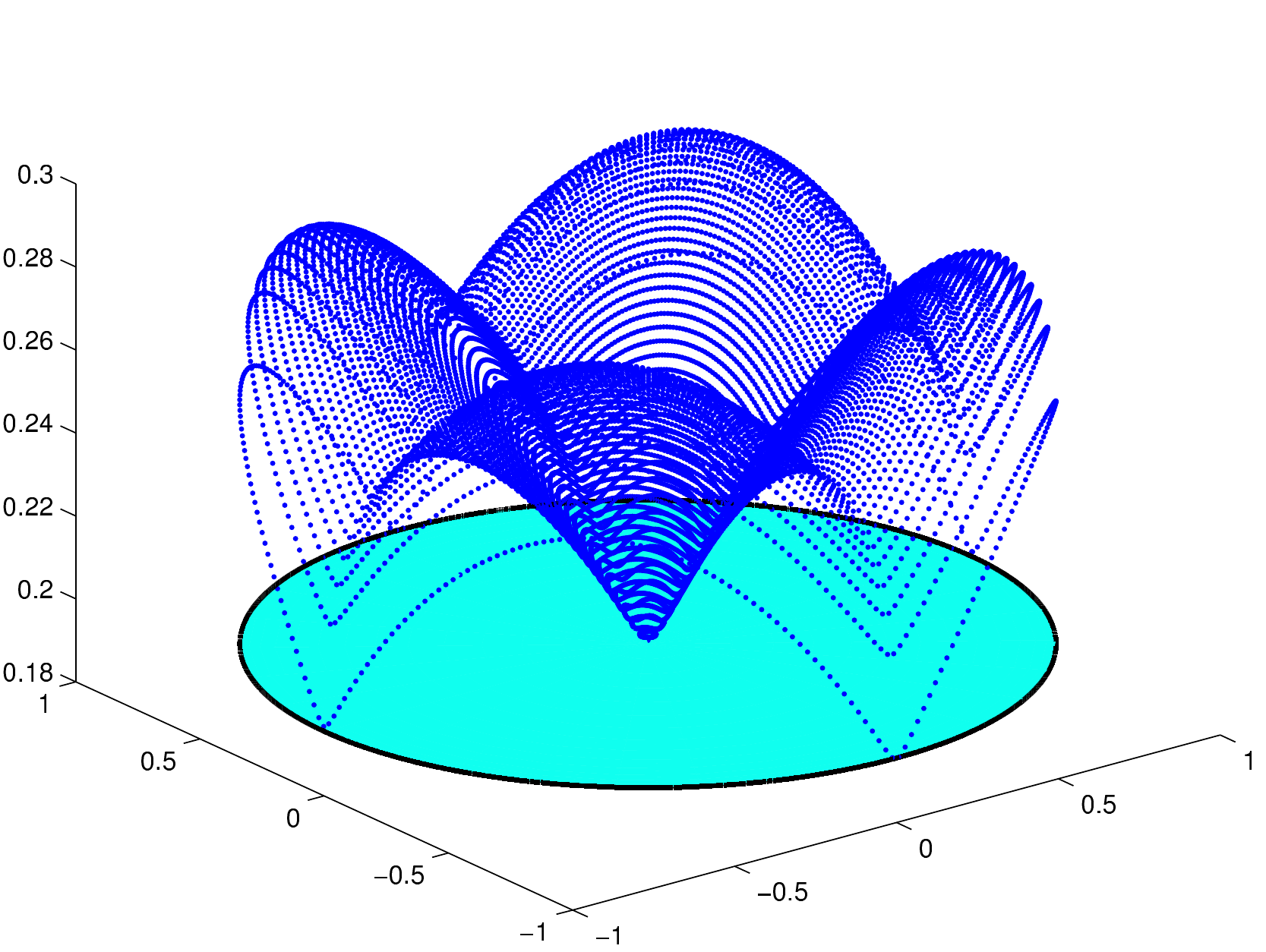} \hspace{.1in} \includegraphics[width=0.2\textwidth]{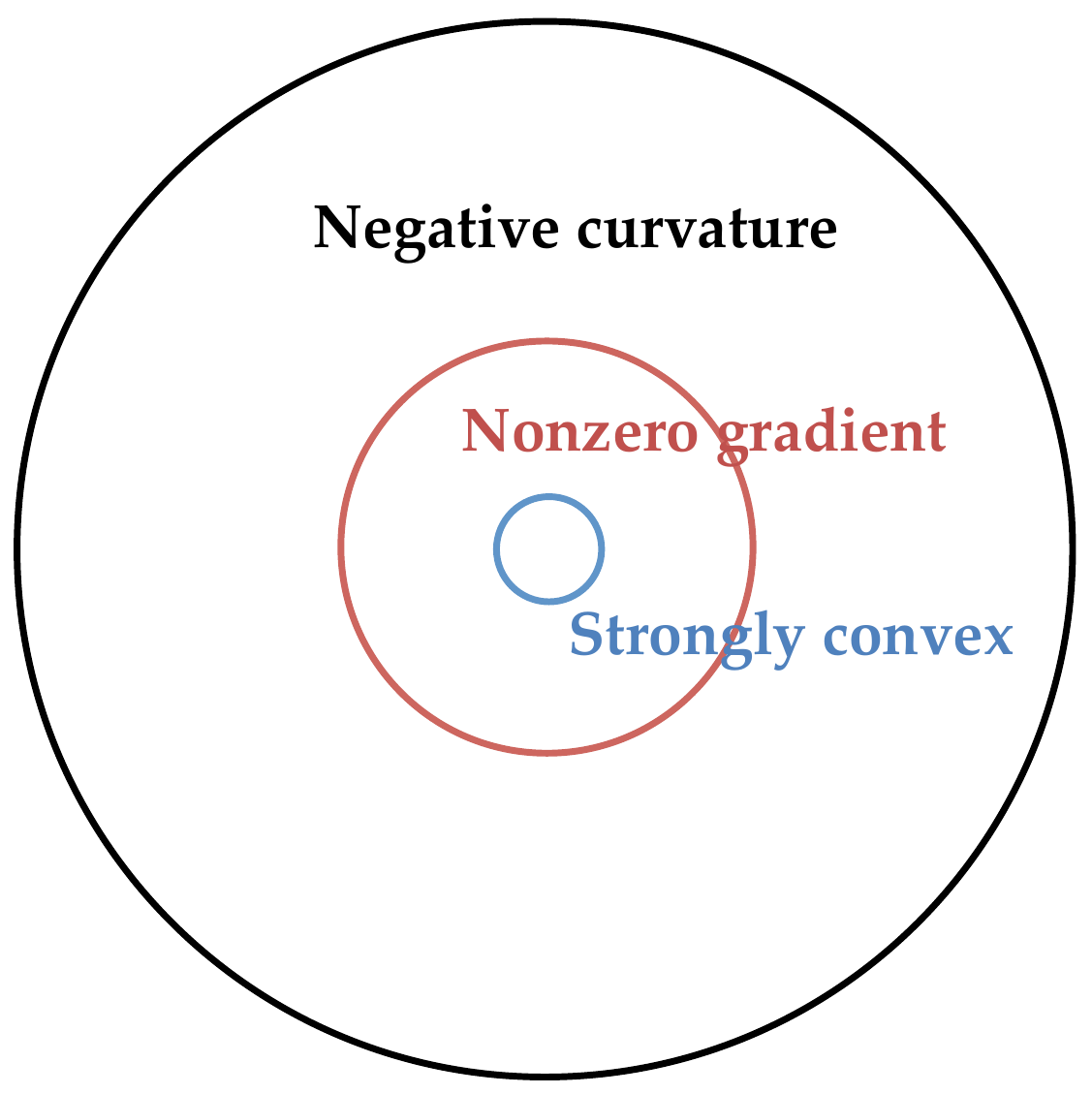} }
\caption{\textbf{Why is dictionary learning over $\bb S^{n-1}$ tractable?} Assume the target dictionary $\mb A_0$ is orthogonal. \textbf{Left:} Large sample objective function $\bb E_{\mb X_0}\brac{f\paren{\mb q}}$. The only local minima are the columns of $\mb A_0$ and their negatives. \textbf{Center:} the same function, visualized as a height above the plane $\mb a_1^\perp$ ($\mb a_1$ is the first column of $\mb A_0$). \textbf{Right:} Around the optimum, the function exhibits a small region of positive curvature, a region of large gradient, and finally a region in which the direction away from $\mb a_1$ is a direction of negative curvature.}
\label{fig:large-sample-sphere}
\end{figure}
For the moment, suppose $\mb A_0$ is orthogonal, and take $\widehat{\mb Y} = \mb Y = \mb A_0 \mb X_0$ in~\eqref{eq:main_l2}. Figure~\ref{fig:large-sample-sphere} (left) plots $\bb E_{\mb X_0}\brac{f\paren{\mb q; \mb Y}}$ over $\mb q \in \bb S^2$ ($n=3$). Remarkably, $\bb E_{\mb X_0}\brac{f\paren{\mb q; \mb Y}}$ has no spurious local minima. In fact, every local minimizer $\widehat{\mb q}$ produces a row of $\mb X_0$: $\widehat{\mb q}^* \mb Y = \alpha \mb e_i^* \mb X_0$ for some $\alpha \neq 0$. 

To better illustrate the point, we take the particular case $\mb A_0 = \mb I$ and project the upper hemisphere above the equatorial plane $\mb e_3^\perp$ onto $\mb e_3^\perp$. The projection is bijective and we equivalently define a reparameterization $g: \mb e_3^\perp \mapsto \R$ of $f$. Figure~\ref{fig:large-sample-sphere} (center) plots the graph of $g$. Obviously the only local minimizers are $\mb 0, \pm \mb e_1, \pm \mb e_2$, and they are also global minimizers. Moreover, the apparent nonconvex landscape has interesting structures around $\mb 0$: when moving away from $\mb 0$, one sees successively a strongly convex region, a nonzero gradient region, and a region where at each point one can always find a direction of negative curvature, as shown schematically in Figure~\ref{fig:large-sample-sphere} (right). This geometry implies that at any nonoptimal point, there is always at least one direction of descent. Thus, any algorithm that can take advantage of the descent directions will likely converge to one global minimizer, irrespective of initialization. 

Two challenges stand out when implementing this idea. For geometry, one has to show similar structure exists for general complete $\mb A_0$, in high dimensions ($n \ge 3$), when the number of observations $p$ is finite (vs.\ the expectation in the experiment). For algorithms, we need to be able to take advantage of this structure without knowing $\mb A_0$ ahead of time. In Section~\ref{sec:overview_alg}, we describe a Riemannian trust region method which addresses the latter challenge. 

\paragraph{Geometry for orthogonal $\mb A_0$.} In this case, we take $\widehat{\mb Y} = \mb Y = \mb A_0 \mb X_0$. Since $f\paren{\mb q; \mb A_0 \mb X_0} \allowbreak = f\paren{\mb A_0^* \mb q; \mb X_0}$, the landscape of $f\paren{\mb q; \mb A_0 \mb X_0}$ is simply a rotated version of that of $f\paren{\mb q; \mb X_0}$, i.e., when $\mb A_0 = \mb I$. Hence we will focus on the case when $\mb A_0 = \mb I$. Among the $2n$ symmetric sections of $\bb S^{n-1}$ centered around the signed basis vectors $\pm \mb e_1, \dots, \pm \mb e_n$, we work with the symmetric section around $\mb e_n$ as an example. The result will carry over to all sections with the same argument; together this provides a complete characterization of the function $f\paren{\mb q; \mb X_0}$ over $\bb S^{n-1}$.    

We again invoke the projection trick described above, this time onto the equatorial plane $\mb e_n^\perp$. This can be formally captured by the reparameterization mapping: 
\begin{align}
\mb q\paren{\mb w} = \paren{\mb w, \sqrt{1-\norm{\mb w}{}^2}}, \; \mb w \in \bb B^{n-1}, 
\end{align}
where $\mb w$ is the new variable in $\mb e_n^\perp \cap \bb B^{n-1}$ and $\bb B^{n-1}$ is the unit ball in $\R^{n-1}$. We first study the composition $g\paren{\mb w; \mb X_0} \doteq f\paren{\mb q\paren{\mb w}; \mb X_0}$ over the set
\begin{align}
\Gamma \doteq \Brac{\mb w: \norm{\mb w}{} < \sqrt{\tfrac{4n-1}{4n}}}.
\end{align}
It can be verified the section we chose to work with is contained in this set\footnote{Indeed, if $\innerprod{\mb q}{\mb e_n} \ge \abs{\innerprod{\mb q}{\mb e_i}}$ for any $i \neq n$, $1 - \norm{\mb w}{}^2 = q_n^2 \ge 1/n$, implying $\norm{\mb w}{}^2 \le \tfrac{n-1}{n} < \tfrac{4n-1}{4n}$. The reason we have defined an open set instead of a closed (compact) one is to avoid potential trivial local minimizers located on the boundary. }. 

Our analysis characterizes the properties of $g\paren{\mb w; \mb X_0}$ by studying three quantities 
\begin{align*}
\nabla^2 g\paren{\mb w; \mb X_0}, \quad \frac{\mb w^* \nabla g\paren{\mb w; \mb X_0}}{\norm{\mb w}{}}, \quad \frac{\mb w^* \nabla^2 g\paren{\mb w; \mb X_0} \mb w}{\norm{\mb w}{}^2}
\end{align*}
respectively over three consecutive regions moving away from the origin, corresponding to the three regions in Figure~\ref{fig:large-sample-sphere} (right). In particular, through typical expectation-concentration style argument, we show that there exists a positive constant $c$ such that 
\begin{align} \label{eq:intro_geo_ineq}
\nabla^2 g\paren{\mb w; \mb X_0} \succeq \frac{1}{\mu} c\theta \mb I,  \quad \frac{\mb w^* \nabla g\paren{\mb w; \mb X_0}}{\norm{\mb w}{}} \ge c \theta, \quad \frac{\mb w^* \nabla^2 g\paren{\mb w; \mb X_0} \mb w}{\norm{\mb w}{}^2} \le -c\theta
\end{align}
over the respective regions w.h.p., confirming our low-dimensional observations described above. In particular, the favorable structure we observed for $n = 3$ persists in high dimensions, w.h.p., even when $p$ is large \emph{yet finite}, for the case $\mb A_0$ is orthogonal. Moreover, the local minimizer of $g\paren{\mb w; \mb X_0}$ over $\Gamma$ is very close to $\mb 0$, within a distance of $O\paren{\mu}$. 

\paragraph{Geometry for complete $\mb A_0$.} For general complete dictionaries $\mb A_0$, we hope that the function $f$  retains the nice geometric structure discussed above. We can ensure this by ``preconditioning'' $\mb Y$ such that the output looks as if being generated from a certain orthogonal matrix, possibly plus a small perturbation. We can then argue that the perturbation does not significantly affect the properties of the graph of the objective function. Write
\begin{align} \label{eq:precond_eq}
\overline{\mb Y} = \paren{\tfrac{1}{p\theta} \mb Y \mb Y^*}^{-1/2} \mb Y. 
\end{align}
Note that for $\mb X_0 \sim_{i.i.d.} \mathrm{BG}\paren{\theta}$, $\expect{\mb X_0 \mb X_0^*}/\paren{p\theta} = \mb I$. Thus, one expects $\tfrac{1}{p\theta} \mb Y \mb Y^* = \tfrac{1}{p\theta} \mb A_0 \mb X_0 \mb X_0^* \mb A_0^*$ to behave roughly like $\mb A_0 \mb A_0^*$ and hence $\overline{\mb Y}$ to behave like 
\begin{align}
\paren{\mb A_0 \mb A_0^*}^{-1/2} \mb A_0 \mb X_0 = \mb U \mb V^* \mb X_0
\end{align}
where we write the SVD of $\mb A_0$ as $\mb A_0 = \mb U \mb \Sigma \mb V^*$. It is easy to see $\mb U \mb V^*$ is an orthogonal matrix. Hence the preconditioning scheme we have introduced is technically sound.  

Our analysis shows that $\overline{\mb Y}$ can be written as 
\begin{align}
\overline{\mb Y} = \mb U \mb V^* \mb X_0 + \mb \Xi \mb X_0, 
\end{align}
where $\mb \Xi$ is a matrix with small magnitude. Simple perturbation argument shows that the constant $c$ in~\eqref{eq:intro_geo_ineq} is at most shrunk to $c/2$ for all $\mb w$ when $p$ is sufficiently large. Thus, the qualitative aspects of the geometry have not been changed by the perturbation. 

\subsubsection{A Second-order Algorithm on Manifold: Riemannian Trust Region Method} \label{sec:overview_alg}
We do not know $\mb A_0$ ahead of time, so our algorithm needs to take advantage of the structure described above without knowledge of $\mb A_0$. Intuitively, this seems possible as the descent direction in the $\mb w$ space appears to also be a local descent direction for $f$ over the sphere. Another issue is that although the optimization problem has no spurious local minima, it does have many saddle points (Figure~\ref{fig:large-sample-sphere}). We can use second-order information to guarantee to escape saddle points. We derive an algorithm based on the Riemannian trust region method (TRM)~\cite{absil2007trust, absil2009} over the sphere for this purpose. 

For a function $f: \R^n \to \R$ and an unconstrained optimization problem 
\begin{align*}
\min_{\mb x \in \R^n} f\paren{\mb x}, 
\end{align*}
typical (second-order) TRM proceeds by successively forming second-order approximations to $f$ at the current iterate, 
\begin{align} \label{eqn:trm_quad_approx}
\widehat{f}\paren{\mb \delta; \mb x^{(k-1)}} \doteq f\paren{\mb x^{(k-1)}} + \nabla^* f\paren{\mb x^{(k-1)}} \mb \delta   + \tfrac{1}{2} \mb \delta^* \mb Q\paren{\mb x^{(k-1)}} \mb \delta, 
\end{align}
where $\mb Q\paren{\mb x^{(k-1)}}$ is a proxy for the Hessian matrix $\nabla^2 f\paren{\mb x^{(k-1)}}$, which encodes the second-order geometry. The next movement direction is determined by seeking a minimum of $\widehat{f}\paren{\mb \delta; \mb x^{(k-1)}}$ over a small region, normally a norm ball $\norm{\mb \delta}{p} \le \Delta$, called the trust region, inducing the well studied trust-region subproblem: 
\begin{align}
\mb \delta^{(k)} \doteq \mathop{\arg\min}_{\mb \delta \in \R^n, \norm{\mb \delta}{p} \le \Delta} \widehat{f}\paren{\mb \delta; \mb x^{(k-1)}}, 
\end{align}
where $\Delta$ is called the trust-region radius that controls how far the movement can be made. A ratio 
\begin{align}
\rho_k \doteq \frac{f\paren{\mb x^{(k-1)}} - f\paren{\mb x^{(k-1)} + \mb \delta^{(k)}}}{\widehat{f}\paren{\mb 0} - \widehat{f}\paren{\mb \delta^{(k-1)}}}
\end{align}
is defined to measure the progress and typically the radius $\Delta$ is updated dynamically according to $\rho_k$ to adapt to the local function behavior. Detailed introductions to the classical TRM can be found in the texts~\cite{ConnGouldToint, NocedalWright}. 

\begin{figure}[!htbp]
\centerline{\includegraphics[width=0.25\textwidth]{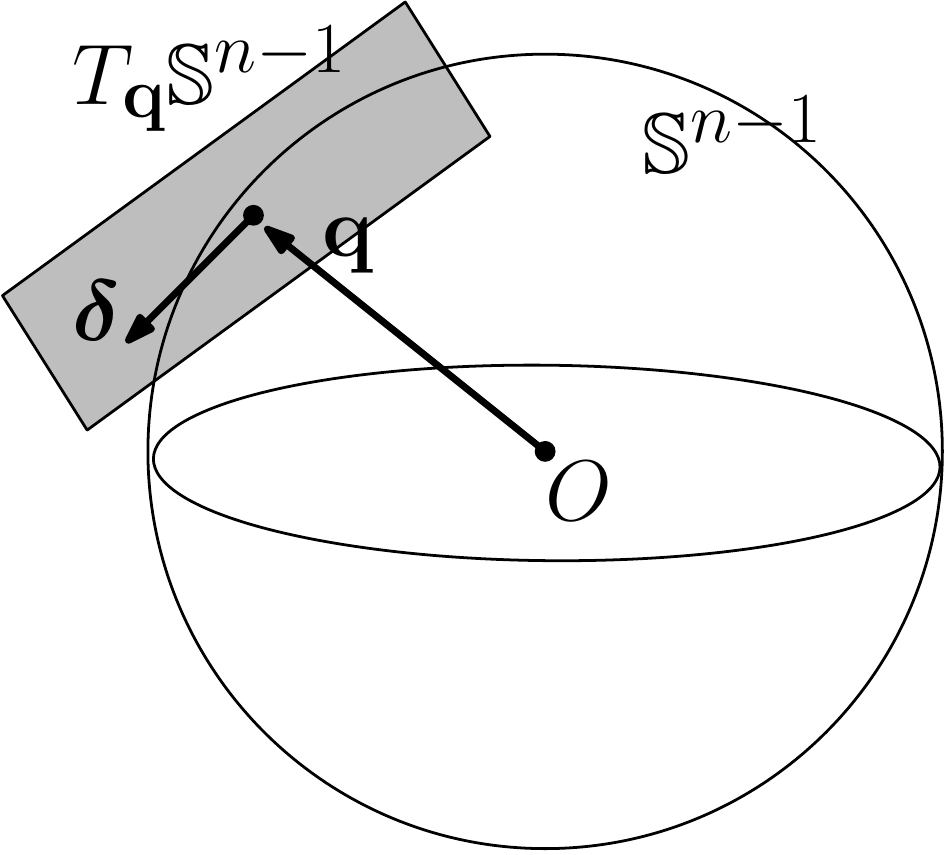} \hspace{.3in} \includegraphics[width=0.3\linewidth]{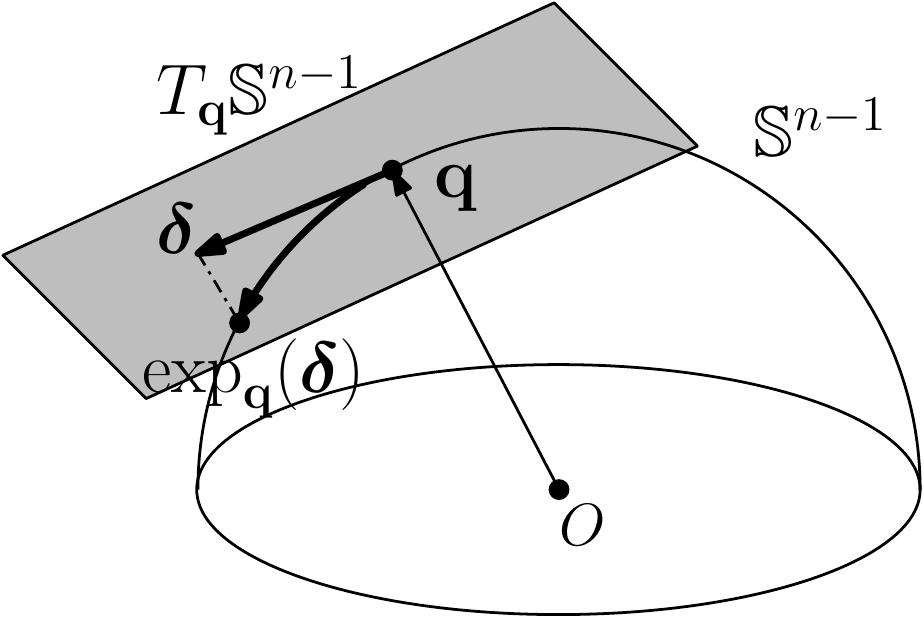}}
\caption{Illustrations of the tangent space $T_{\mb q}\bb S^{n-1}$ and exponential map $\exp_{\mb q}\paren{\mb \delta}$ defined on the sphere $\bb S^{n-1}$.} \label{fig:exp-map}
\end{figure}

To generalize the idea to smooth manifolds, one natural choice is to form the approximation over the tangent spaces~\cite{absil2007trust, absil2009}. Specific to our spherical manifold, for which the tangent space at an iterate $\mb q^{(k)} \in \bb S^{n-1}$ is $T_{\mb q^{(k)}} \bb S^{n-1} \doteq \set{\mb v: \mb v^* \mb q^{(k)} = 0}$ (see Figure~\ref{fig:exp-map}), we work with a ``quadratic'' approximation $\widehat{f}: T_{\mb q^{(k)}} \bb S^{n-1} \mapsto \R$ defined as 
\begin{align}
\widehat{f}(\mb \delta; \mb q^{(k)}) \;\doteq\; f(\mb q^{(k)}) + \innerprod{ \nabla f(\mb q^{(k)}) }{\mb \delta }  + \frac{1}{2} \mb \delta^* \left( \nabla^2 f( \mb q^{(k)}) - \innerprod{ \nabla f(\mb q^{(k)}) }{ \mb q^{(k)} } \mb I \right) \mb \delta.  \label{eqn:f-appx}
\end{align}
To interpret this approximation, let $\mc P_{T_{\mb q^{(k)}} \bb S^{n-1}} \doteq \paren{\mb I - \mb q^{(k)} \paren{\mb q^{(k)}}^*}$ be the orthoprojector onto $T_{\mb q^{(k)}} \bb S^{n-1}$ and write~\eqref{eqn:f-appx} into an equivalent form:  
\begin{multline*}
\widehat{f}(\mb \delta; \mb q^{(k)}) \;\doteq\; f(\mb q^{(k)}) + \innerprod{\mc P_{T_{\mb q^{(k)}} \bb S^{n-1}} \nabla f(\mb q^{(k)}) }{\mb \delta } \\
+ \frac{1}{2} \mb \delta^* \mc P_{T_{\mb q^{(k)}} \bb S^{n-1}} \left( \nabla^2 f( \mb q^{(k)}) - \innerprod{ \nabla f(\mb q^{(k)}) }{ \mb q^{(k)} } \mb I \right) \mc P_{T_{\mb q^{(k)}} \bb S^{n-1}} \mb \delta. 
\end{multline*}
The two terms 
\begin{align*}
\mathrm{grad} f\paren{\mb q^{(k)}} &\doteq \mc P_{T_{\mb q^{(k)}} \bb S^{n-1}} \nabla f(\mb q^{(k)}), \\
\mathrm{Hess} f\paren{\mb q^{(k)}} &\doteq \mc P_{T_{\mb q^{(k)}} \bb S^{n-1}} \left( \nabla^2 f( \mb q^{(k)}) - \innerprod{ \nabla f(\mb q^{(k)}) }{ \mb q^{(k)} } \mb I \right) \mc P_{T_{\mb q^{(k)}} \bb S^{n-1}}
\end{align*} 
are the Riemannian gradient and Riemannian Hessian of $f$ w.r.t. $\bb S^{n-1}$, respectively~\cite{absil2007trust, absil2009}; the above approximation is reminiscent of the usual quadratic approximation described in~\eqref{eqn:trm_quad_approx}. 

Then the Riemannian trust-region subproblem is
\begin{align} \label{eq:trust-region}
\min_{\mb \delta \in T_{\mb q^{(k)}} \bb S^{n-1},\; \norm{\mb \delta}{} \le \Delta} \widehat{f}\paren{\mb \delta; \mb q^{(k)}},  
\end{align}
where we take the simple $\ell^2$ norm ball for the trust region. This can be transformed into a classical trust region subprolem: indeed, taking any orthonormal basis $\mb U_{\mb q^{(k)}}$ for $T_{\mb q^{(k)}} \bb S^{n-1}$, the above problem is equivalent to 
\begin{align} \label{eqn:alg_rie_tsp}
\min_{\norm{\mb \xi}{} \le \Delta} \widehat{f}\paren{\mb U_{\mb q^{(k)}} \mb \xi, \mb q^{(k)}}, 
\end{align}
where the objective is quadratic in $\mb \xi$. This is the classical trust region problem (with $\ell^2$ norm ball constraint) that admits very efficient numerical algorithms~\cite{more1983computing,hazan2014linear}. Once we obtain the minimizer $\mb \xi_\star$, we set $\mb \delta_\star = \mb U \mb \xi_\star$, which solves~\eqref{eq:trust-region}. 

One additional issue as compared to the Euclidean setting is that now $\mb \delta_\star$ is one vector in the tangent space and additive update leads to a point outside the sphere. We resort to the natural exponential map to pull the tangent vector to a point on the sphere: 
\begin{align}
\mb q^{(k+1)} \doteq \exp_{\mb q^{(k)}}\paren{\mb \delta_\star} = \mb q^{(k)} \cos \norm{\mb \delta_\star}{} + \tfrac{\mb \delta_\star}{\norm{\mb \delta_\star}{}} \sin \norm{\mb \delta_\star}{}.  
\end{align}
As seen from Figure~\ref{fig:exp-map}, the movement to the next iterate is ``along the direction"\footnote{Technically, moving along the geodesic whose velocity at time zero is $\mb \delta_\star$. } of $\mb \delta_\star$ while staying over the sphere. 

Using the above geometric characterizations, we prove that w.h.p., the algorithm converges to a local minimizer when the parameter $\Delta$ is sufficiently small\footnote{For simplicity of analysis, we have assumed $\Delta$ is fixed throughout the analysis. In practice, dynamic updates to $\Delta$ lead to faster convergence.}. In particular, we show that (1) the trust region step induces at least a fixed amount of decrease to the objective value in the negative curvature and nonzero gradient region; (2) the trust region iterate sequence will eventually move to and stay in the strongly convex region, and converge to the local minimizer contained in the region with an asymptotic quadratic rate. In short, the geometric structure implies that from \emph{any initialization}, the iterate sequence converges to a close approximation to the target solution in a polynomial number of steps. 

\subsection{Prior Arts and Connections} \label{sec:lit_review}
It is far too ambitious to include here a comprehensive review of the exciting developments of DL algorithms and applications after the pioneer work~\cite{olshausen1996emergence}. We refer the reader to Chapter 12 - 15 of the book~\cite{elad2010sparse} and the survey paper~\cite{mairal2014sparse} for summaries of relevant developments in image analysis and visual recognition. In the following, we focus on reviewing recent developments on the theoretical side of dictionary learning, and draw connections to problems and techniques that are relevant to the current work. 

\paragraph{Theoretical Dictionary Learning.} 
The theoretical study of DL in the recovery setting started only very recently. \cite{aharon2006uniqueness} was the first to provide an algorithmic procedure to correctly extract the generating dictionary. The algorithm requires exponentially many samples and has exponential running time; see also~\cite{hillar2011can}. Subsequent work~\cite{gribonval2010dictionary, geng2011local, schnass2014local, schnass2014identifiability,schnass2015convergence} studied when the target dictionary is a local optimum of natural recovery criteria. These meticulous analyses show that polynomially many samples are sufficient to ensure local correctness under natural assumptions. However, these results do not imply that one can design efficient algorithms to obtain the desired local optimum and hence the dictionary. 

\cite{spielman2012exact} initiated the on-going research effort to provide efficient algorithms that globally solve DR. They showed that one can recover a complete dictionary $\mb A_0$ from $\mb Y = \mb A_0 \mb X_0$ by solving a certain sequence of linear programs, when $\mb X_0$ is a sparse random matrix with $O(\sqrt{n})$ nonzeros per column. \cite{agarwal2013learning, agarwal2013exact} and~\cite{arora2013new, arora2015simple} give efficient algorithms that provably recover overcomplete ($m \ge n$) and incoherent dictionaries, based on a combination of \{clustering or spectral initialization\} and local refinement. These algorithms again succeed when $\mb X_0$ has $\widetilde{O}(\sqrt{n})$ \footnote{The $\widetilde{O}$ suppresses some logarithm factors.} nonzeros per column. Recent work \cite{barak2014dictionary} provides the first polynomial-time algorithm that provably recovers most ``nice'' overcomplete dictionaries when $\mb X_0$ has $O(n^{1-\delta})$ nonzeros per column for any constant $\delta \in (0, 1)$. However, the proposed algorithm runs in super-polynomial time when the sparsity level goes up to $O(n)$. Similarly, \cite{arora2014more} also proposes a super-polynomial (quasipolynomial) time algorithm that guarantees recovery with (almost) $O\paren{n}$ nonzeros per column. By comparison, we give the first \emph{polynomial-time} algorithm that provably recovers complete dictionary $\mb A_0$ when $\mb X_0$ has $O\paren{n}$ nonzeros per column. 

Aside from efficient recovery, other theoretical work on DL includes results on identifiability~\cite{aharon2006uniqueness, hillar2011can, wu2015local}, generalization bounds~\cite{maurer2010dimensional, Vainsencher:2011, Mehta13, gribonval2013sample}, and noise stability~\cite{GribonvalJB14}. 

\paragraph{Finding Sparse Vectors in a Linear Subspace.} 
We have followed~\cite{spielman2012exact} and cast the core problem as finding the sparsest vectors in a given linear subspace, which is also of independent interest. Under a planted sparse model\footnote{... where one sparse vector embedded in an otherwise random subspace.}, \cite{demanet2014scaling} shows solving a sequence of linear programs similar to~\cite{spielman2012exact} can recover sparse vectors with sparsity up to $O\paren{p/\sqrt{n}}$, sublinear in the vector dimension. \cite{qu2014finding} improved the recovery limit to $O\paren{p}$ by solving a nonconvex spherical constrained problem similar to~\eqref{eq:main_l2}\footnote{The only difference is that they chose to work with the Huber function as a proxy of the $\norm{\cdot}{1}$ function. } via an ADM algorithm. The idea of seeking rows of $\mb X_0$ sequentially by solving the above core problem sees precursors in~\cite{zibulevsky2001blind} for blind source separation, and \cite{gottlieb2010matrix} for matrix sparsification. \cite{zibulevsky2001blind} also proposed a nonconvex optimization similar to~\eqref{eq:main_l2} here and that employed in~\cite{qu2014finding}. 

\paragraph{Nonconvex Optimization Problems.} For other nonconvex optimization problems of recovery of structured signals\footnote{This is a body of recent work studying nonconvex recovery up to statistical precision, including, e.g., \cite{loh2011high,loh2013regularized,wang2014nonconvex,balakrishnan2014statistical,wang2014high,loh2014support,loh2015statistical,sun2015provable}. }, 
including low-rank matrix completion/recovery~\cite{keshavan2010matrix, jain2013low, hardt2014understanding, hardt2014fast, netrapalli2014non, jain2014fast, sun2014guaranteed, zheng2015convergent, tu2015low, chen2015fast}, phase retreival~\cite{netrapalli2013phase, candes2015phase, chen2015solving, white2015local}, tensor recovery~\cite{jain2014provable, anandkumar2014guaranteed, anandkumar2014analyzing, anandkumar2015tensor}, mixed regression~\cite{yi2013alternating, lee2013near}, structured element pursuit~\cite{qu2014finding}, and recovery of simultaneously structured signals~\cite{lee2013near}, numerical linear algebra and optimization~\cite{jain2015computing, bhojanapalli2015dropping}, the initialization plus local refinement strategy adopted in theoretical DL~\cite{agarwal2013learning, agarwal2013exact, arora2013new, arora2015simple, arora2014more} is also crucial: nearness to the target solution enables exploiting the local geometry of the target to analyze the local refinement.\footnote{The powerful framework~\cite{attouch2010proximal,bolte2014proximal} to establish local convergence of ADM algorithms to critical points applies to DL/DR also, see, e.g., \cite{bao2014l0, bao2014convergent, bao2014convergence}. However, these results do not guarantee to produce global optima. } By comparison, we provide a complete characterization of the global geometry, which admits efficient algorithms without any special initialization. The idea of separating the geometric analysis and algorithmic design may also prove valuable for other nonconvex problems discussed above. 

\paragraph{Optimization over Riemannian Manifolds.} Our trust-region algorithm on the sphere builds on the extensive research efforts to generalize Euclidean numerical algorithms to (Riemannian) manifold settings. We refer the reader to the monographs~\cite{udriste1994convex,helmke1994optimization,absil2009} for survey of developments in this field. In particular, ~\cite{edelman1998geometry} developed Newton and conjugate-gradient methods for the Stiefel manifolds, of which the spherical manifold is a special case. \cite{absil2007trust} generalized the trust-region methods to Riemannian manifolds. We cannot, however, adopt the existing convergence results that concern either global convergence (convergence to critical points) or local convergence (convergence to a local minimum within a radius). The particular geometric structure forces us to piece together different arguments to obtain the global result. 

\paragraph{Independent Component Analysis (ICA) and Other Matrix Factorization Problems. } 
\begin{figure}[ht]
\centering  
\begin{subfigure}[t]{0.3\textwidth}
\centering
\includegraphics[width = 1\linewidth]{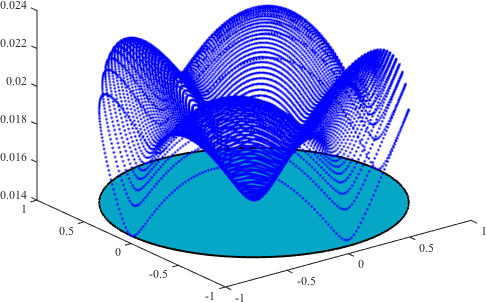} 
\caption{\footnotesize Correlated Gaussian, $\theta = 0.1$}
\label{subfig:gaussian-corr-0.1}
\end{subfigure}
\begin{subfigure}[t]{0.3\textwidth}
\centering
\includegraphics[width = 1\linewidth]{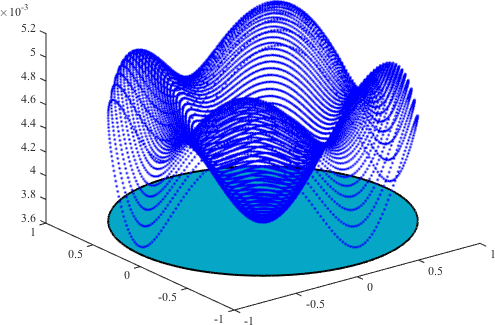} 
\caption{\footnotesize Correlated Uniform, $\theta = 0.1$}
\label{subfig:uniform-corr-0.1}
\end{subfigure} 
\begin{subfigure}[t]{0.3\textwidth}
\centering
\includegraphics[width = 1\linewidth]{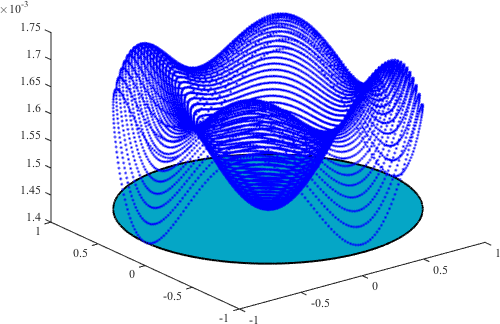}
\caption{\footnotesize Independent Uniform, $\theta = 0.1$}
\label{subfig:uniform-ind-0.1}
\end{subfigure} \\
\begin{subfigure}[t]{0.3\textwidth}
\centering
\includegraphics[width = 1\linewidth]{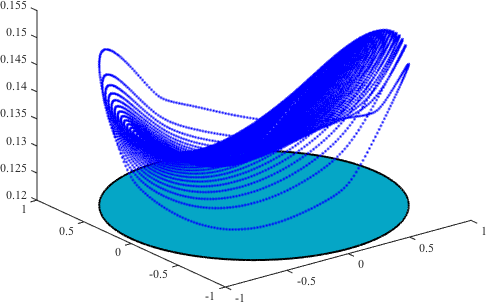} 
\caption{\footnotesize Correlated Gaussian, $\theta = 0.9$}
\label{subfig:gaussian-corr-0.9}
\end{subfigure}
\begin{subfigure}[t]{0.3\textwidth}
\centering
\includegraphics[width = 1\linewidth]{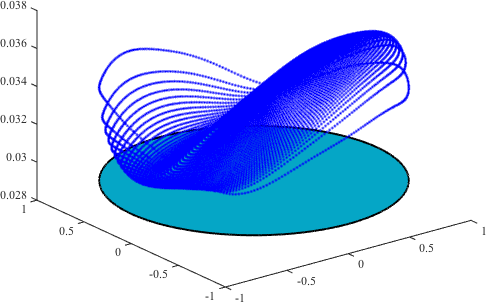}
\caption{\footnotesize Correlated Uniform, $\theta = 0.9$}
\label{subfig:uniform-corr-0.9}
\end{subfigure}
\begin{subfigure}[t]{0.3\textwidth}
\centering
\includegraphics[width = 1\linewidth]{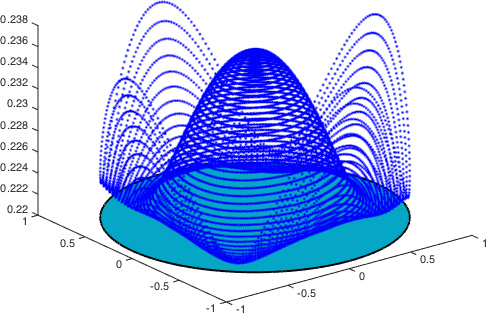}
\caption{\footnotesize Independent Uniform, $\theta = 1$}
\label{subfig:uniform-ind-0.9}
\end{subfigure}
\caption{\textbf{Asymptotic function landscapes when rows of $\mb X_0$ are not independent.} W.l.o.g., we again assume $\mb A_0 = \mb I$. In (a) and (d), $\mb X_0 = \mb \Omega \odot \mb V$, with $\mb \Omega \sim_{i.i.d.} \mathrm{Ber}(\theta)$ and columns of $\mb X_0$ i.i.d. Gaussian vectors obeying $\mb v_i \sim \mc N(\mb 0, \mb \Sigma^2)$ for symmetric $\mb \Sigma$ with $1$'s on the diagonal and i.i.d. off-diagonal entries distributed as $\mc N(0, \sqrt{2}/20)$. Similarly, in (b) and (e), $\mb X_0 = \mb \Omega \odot \mb W$, with $\mb \Omega \sim_{i.i.d.} \mathrm{Ber}(\theta)$ and columns of $\mb X_0$ i.i.d. vectors generated as $\mb w_i = \mb \Sigma \mb u^i$ with $\mb u_i \sim_{i.i.d.} \mathrm{Uniform}[-0.5,0.5]$. For comparison, in (c) and (f), $\mb X_0 = \mb \Omega \odot \mb W$ with $\mb \Omega \sim_{i.i.d.} \mathrm{Ber}(\theta)$ and $\mb W \sim_{i.i.d.} \mathrm{Uniform}[-0.5,0.5]$. Here $\odot$ denote the elementwise product, and the objective function is still based on the $\log\cosh$ function as in~\eqref{eq:main_l2}. }
\label{fig:DL-dependent}
\end{figure}
DL can also be considered in the general framework of matrix factorization problems,  which encompass the classic principal component analysis (PCA), ICA, and clustering, and more recent problems such as nonnegative matrix factorization (NMF), multi-layer neural nets (deep learning architectures). Most of these problems are NP-hard. Identifying tractable cases of practical interest and providing provable efficient algorithms are subject of on-going research endeavors; see, e.g., recent progresses on NMF~\cite{arora2012computing}, and learning deep neural nets~\cite{arora2013provable,sedghi2014provable,neyshabur2013sparse,livni2014computational}. 

ICA factors a data matrix $\mb Y$ as $\mb Y = \mb A \mb X$ such that $\mb A$ is square and rows of $\mb X$ are as independent as possible~\cite{hyvarinen2000independent,hyvarinen2001}. In theoretical study of the recovery problem, it is often assumed that rows of $\mb X_0$ are (weakly) independent (see, e.g., ~\cite{comon1994independent, frieze1996learning, arora2012provable}). Our i.i.d.\ probability model on $\mb X_0$ implies rows of $\mb X_0$ are independent, aligning our problem perfectly with the ICA problem. More interestingly, the $\log\cosh$ objective we analyze here was proposed as a general-purpose \emph{contrast function} in ICA that has not been thoroughly analyzed~\cite{hyvarinen99fast}, and algorithm and analysis with another popular contrast function, the fourth-order cumulants, indeed overlap with ours considerably~\cite{frieze1996learning, arora2012provable}\footnote{Nevertheless, the objective functions are apparently different. Moreover, we have provided a complete geometric characterization of the objective, in contrast to~\cite{frieze1996learning, arora2012provable}. We believe the geometric characterization could not only provide insight to the algorithm, but also help improve the algorithm in terms of stability and also finding all components. }. While this interesting connection potentially helps port our analysis to ICA, it is a fundamental question to ask what is playing the vital role for DR, sparsity or independence. 

Figure~\ref{fig:DL-dependent} helps shed some light in this direction, where we again plot the asymptotic objective landscape with the natural reparameterization as in Section~\ref{sec:overview_geometry}. From the left and central panels, it is evident even without independence, $\mb X_0$ with sparse columns induces the familiar geometric structures we saw in Figure~\ref{fig:large-sample-sphere}; such structures are broken when the sparsity level becomes large. We believe all our later analyses can be generalized to the correlated cases we experimented with. On the other hand, from the right panel\footnote{We have not showed the results on the BG model here, as it seems the structure persists even when $\theta$ approaches $1$. We suspect the ``phase transition'' of the landscape occurs at different points for different distributions and Gaussian is the outlying case where the transition occurs at $1$. }, it seems with independence, the function landscape undergoes a transition as sparsity level grows - target solution goes from minimizers of the objective to the maximizers of the objective. Without adequate knowledge of the true sparsity, it is unclear whether one would like to minimize or maximize the objective.\footnote{For solving the ICA problem, this suggests the $\log \cosh$ contrast function, that works well empirically~\cite{hyvarinen99fast}, may not work for all distributions (rotation-invariant Gaussian excluded of course).} This suggests sparsity, instead of independence, makes our current algorithm for DR work. 

\paragraph{Nonconvex Problems with Similar Geometric Structure} Besides ICA discussed above, it turns out that a handful of other practical problems arising in signal processing and machine learning induce the ``no spurious minimizers, all saddles are second-order'' structure under natural setting, including the eigenvalue problem, generalized phase retrieval~\cite{sun2015geometric}, tensor decomposition~\cite{ge2015escaping}, linear neural nets learning~\cite{baldi1989neural}. \cite{sun2015nonconvex} gave a review of these problems, and discussed how the methodology developed in this and the companion paper~\cite{sun2015complete_b} can be generalized to solve those problems.

\subsection{Notations, Organization, and Reproducible Research}
We use bold capital and small letters such as $\mb X$ and $\mb x$ to denote matrices and vectors, respectively. Small letters are reserved for scalars. Several specific mathematical objects we will frequently work with: $O_k$ for the orthogonal group of order $k$, $\bb S^{n-1}$ for the unit sphere in $\R^n$, $\bb B^n$ for the unit ball in $\R^n$, and $[m] \doteq \set{1, \dots, m}$ for positive integers $m$, $n$, $k$. We use $\paren{\cdot}^*$ for matrix transposition, causing no confusion as we will work entirely on the real field. We use superscript to index rows of a matrix, such as $\mb x^i$ for the $i$-th row of the matrix $\mb X$,  and subscript to index columns, such as $\mb x_j$. All vectors are defaulted to column vectors. So the $i$-th row of $\mb X$ as a row vector will be written as $\paren{\mb x^i}^*$. For norms, $\norm{\cdot}{}$ is the usual $\ell^2$ norm for a vector and to the operator norm (i.e., $\ell^2 \to \ell^2$) for a matrix; all other norms will be indexed by subscript, for example the Frobenius norm $\norm{\cdot}{F}$ for matrices and the element-wise max-norm $\norm{\cdot}{\infty}$. We use $\mb x \sim \mc L$ to mean that the random variable $x$ is distributed according to the law $\mc L$. Let $\mc N$ denote the Gaussian law. Then $\mb x \sim \mc N\paren{\mb 0, \mb I}$ means that $\mb x$ is a standard Gaussian vector. Similarly, we use $\mb x \sim_{i.i.d.} \mc L$ to mean elements of $\mb x$ are independently and identically distributed according to the law $\mc L$. So the fact $\mb x \sim \mc N\paren{\mb 0, \mb I}$ is equivalent to that $\mb x \sim_{i.i.d.} \mc N\paren{0, 1}$. One particular distribution of interest for this paper is the Bernoulli-Gaussian with rate $\theta$: $Z \sim B \cdot G$, with $G \sim \mc N\paren{0, 1}$ and $B \sim \mathrm{Ber}\paren{\theta}$. We also write this compactly as $Z \sim \mathrm{BG}\paren{\theta}$. We frequently use indexed $C$ and $c$ for numerical constants when stating and proving technical results. The scopes of such constants are local unless otherwise noted. We use standard notations for most other cases, with exceptions clarified locally. 

The rest of the paper is organized as follows. In Section~\ref{sec:geometry} we present major technical results for a complete characterization of the geometry sketched in Section~\ref{sec:overview_geometry}. Similarly in Section~\ref{sec:algorithm} we present necessary technical machinery and results for convergence proof of the Riemannian trust-region algorithm over the sphere, corresponding to Section~\ref{sec:overview_alg}. In Section~\ref{sec:main_result}, we discuss the whole algorithmic pipeline for recovering complete dictionaries given $\mb Y$, and present the main theorems. After presenting a simple simulation to corroborate our theory in Section~\ref{sec:exp}, we wrap up the main content in Section~\ref{sec:discuss} by discussing possible improvement and future directions after this work. All major proofs of geometrical and algorithmic results are deferred to Section~\ref{sec:proof_geometry} and Section~\ref{sec:proof_algorithm}, respectively. Section~\ref{sec:proof_main} augments the main results. The appendices cover some recurring technical tools and auxiliary results for the proofs. 

The codes to reproduce all the figures and experimental results can be found online: 
\begin{quote}
\centering
\url{https://github.com/sunju/dl_focm}
\end{quote}

%% file: sec/geometry.tex
\section{High-dimensional Function Landscapes} \label{sec:geometry}
To characterize the function landscape of $f\paren{\mb q; \mb X_0}$ over $\bb S^{n-1}$, we mostly work with the function 
\begin{align}\label{eqn:function-g}
g\paren{\mb w} \doteq f\paren{\mb q\paren{\mb w}; \mb X_0} = \frac{1}{p} \sum_{k=1}^p h_{\mu}\paren{\mb q\paren{\mb w}^* \paren{\mb x_0}_k}, 
\end{align}
induced by the reparametrization
\begin{align}
\mb q\paren{\mb w} = \paren{\mb w, \sqrt{1-\norm{\mb w}{}^2}},  \quad \mb w \in \bb B^{n-1}. 
\end{align}
In particular, we focus our attention to the smaller set 
\begin{align}
\Gamma = \set{\mb w: \norm{\mb w}{} < \sqrt{\frac{4n-1}{4n}}},  
\end{align}
because $\mb q\paren{\Gamma}$ contains all points $\mb q \in \bb S^{n-1}$ with $n \in \mathop{\arg\max}_{i \in \pm [n]} \mb q^* \mb e_i$ and we can characterize other parts of $f$ on $\bb S^{n-1}$ using projection onto other equatorial planes. Note that over $\Gamma$, $q_n = \paren{1-\norm{\mb w}{}^2}^{1/2} \ge \frac{1}{2\sqrt{n}}$. 

\subsection{Main Geometric Theorems} 
\begin{theorem}[High-dimensional landscape - orthogonal dictionary]\label{thm:geometry_orth}
Suppose $\mb A_0 = \mb I$ and hence $\mb Y = \mb A_0 \mb X_0 = \mb X_0$. There exist positive constants $c_\star$ and $C$, such that for any $\theta \in (0,1/2)$ and $\mu < \min\Brac{c_a \theta n^{-1}, c_b n^{-5/4}}$, whenever 
$
p \ge \frac{C}{\mu^2 \theta^2} n^3 \log \frac{n}{\mu \theta}$,
the following hold simultaneously with high probability: 
\begin{align}
\nabla^2 g(\mb w; \mb X_0) &\succeq \frac{c_\star \theta}{\mu} \mb I \quad &\forall \, \mb w \quad \text{s.t.}& \quad \norm{\mb w}{} \le \frac{\mu}{4\sqrt{2}},  \label{eqn:hess-zero-uni-orth} \\
\frac{\mb w^* \nabla g(\mb w; \mb X_0)}{\norm{\mb w}{}} &\ge c_\star \theta \quad &\forall \, \mb w \quad \text{s.t.}& \quad \frac{\mu}{4\sqrt{2}} \le \norm{\mb w}{} \le \frac{1}{20 \sqrt{5}} \label{eqn:grad-uni-orth} \\
\frac{\mb w^* \nabla^2 g(\mb w; \mb X_0) \mb w}{\norm{\mb w}{}^2} &\le - c_\star \theta \quad &\forall \, \mb w \quad \text{s.t.}& \quad \frac{1}{20 \sqrt{5}} \le \norm{\mb w}{} \le \sqrt{\frac{4n-1}{4n}},   \label{eqn:curvature-uni-orth}
\end{align}
{\em and} the function $g(\mb w; \mb X_0)$ has exactly one local minimizer $\mb w_\star$ over the open set $\Gamma \doteq \set{ \mb w: \norm{\mb w}{} < \sqrt{ \tfrac{4n-1}{4n} } }$, which satisfies
\begin{equation}
\norm{\mb w_\star - \mb 0 }{} \;\le\; \min\Brac{\frac{c_c\mu}{\theta} \sqrt{\frac{n \log p}{p}}, \frac{\mu}{16}}.    
\end{equation}
In particular, with this choice of $p$, the probability the claim fails to hold is at most 
$
4np^{-10} + \theta(np)^{-7} + \exp\paren{-0.3\theta n p} + c_d\exp\left( - c_e p \mu^2 \theta^2/n^2\right)
$. 
Here $c_a$ to $c_e$ are all positive numerical constants. 
\end{theorem}
Here $\mb q\paren{\mb 0} = \mb e_n$, which exactly recovers the last row of $\mb X_0$, $\mb x_0^n$. Though the unique local minimizer $\mb w_\star$ may not be $\mb 0$, it is very near to $\mb 0$. Hence the resulting $\mb q\paren{\mb w_\star}$ produces a close approximation to $\mb x_0^n$. Note that $\mb q\paren{\Gamma}$ (strictly) contains all points $\mb q \in \bb S^{n-1}$ such that $n = \mathop{\arg \max}_{i \in \pm [n]} \mb q^* \mb e_i$. We can characterize the graph of the function $f\paren{\mb q; \mb X_0}$ in the vicinity of other signed basis vector $\pm \mb e_i$ simply by changing the plane $\mb e_n^\perp$ to $\mb e_i^\perp$. Doing this $2n$ times (and multiplying the failure probability in Theorem~\ref{thm:geometry_orth} by $2n$), we obtain a characterization of $f\paren{\mb q; \mb X_0}$ over the entirety of $\bb S^{n-1}$.\footnote{In fact, it is possible to pull the very detailed geometry captured in~\eqref{eqn:hess-zero-uni-orth} through~\eqref{eqn:curvature-uni-orth} back to the sphere (i.e., the $\mb q$ space) also; analysis of the Riemannian trust-region algorithm later does part of these. We will stick to this simple global version here. }
 The result is captured by the next corollary. 

\begin{corollary} \label{cor:geometry_orth}
Suppose $\mb A_0 = \mb I$ and hence $\mb Y = \mb A_0 \mb X_0 = \mb X_0$. There exist positive constant $C$, such that for any $\theta \in (0,1/2)$ and $\mu < \min\Brac{c_a \theta n^{-1}, c_b n^{-5/4}}$, whenever $
p \ge \frac{C}{\mu^2 \theta^2} n^3 \log \frac{n}{\mu \theta} $, with probability at least $1 - 8n^2p^{-10} - \theta(np)^{-7} - \exp\paren{-0.3\theta n p} - c_c\exp\paren{- c_d p \mu^2 \theta^2/n^2}$, the function $f\paren{\mb q; \mb X_0}$ has exactly $2n$ local minimizers over the sphere $\bb S^{n-1}$. In particular, there is a bijective map between these minimizers and signed basis vectors $\set{\pm \mb e_i}_i$, such that the corresponding local minimizer $\mb q_\star$ and $\mb b \in \set{\pm \mb e_i}_i$ satisfy 
\begin{align}
\norm{\mb q_\star - \mb b}{} \le \sqrt{2}\min\Brac{\frac{c_c\mu}{\theta} \sqrt{\frac{n \log p}{p}}, \frac{\mu}{16}}. 
\end{align} 
Here $c_a$ to $c_d$ are numerical constants (possibly different from that in the above theorem). \end{corollary}
\begin{proof}
By Theorem~\ref{thm:geometry_orth}, over $\mb q\paren{\Gamma}$, $\mb q\paren{\mb w_\star}$ is the unique local minimizer. Suppose not. Then there exist $\mb q' \in \mb q\paren{\Gamma}$ with $\mb q' \neq \mb q\paren{\mb w_\star}$ and $\eps > 0$, such that $f\paren{\mb q'; \mb X_0} \le f\paren{\mb q; \mb X_0}$ for all $\mb q \in \mb q\paren{\Gamma}$ satisfying $\norm{\mb q' - \mb q}{} < \eps$. Since the mapping $\mb w \mapsto \mb q\paren{\mb w}$ is $2\sqrt{n}$-Lipschitz (Lemma~\ref{lem:lip-h-mu}), $g\paren{\mb w\paren{\mb q'}; \mb X_0} \le g\paren{\mb w\paren{\mb q}; \mb X_0}$ for all $\mb w \in \Gamma$ satisfying $\norm{\mb w\paren{\mb q'} - \mb w\paren{\mb q}}{} < \eps/\paren{2\sqrt{n}}$, implying $\mb w\paren{\mb q'}$ is a local minimizer different from $\mb w_\star$, a contradiction. Let $\norm{\mb w_\star - \mb 0}{} = \eta$. Straightforward calculation shows 
\begin{align*}
\norm{\mb q\paren{\mb w_\star} - \mb e_n}{}^2 = \paren{1-\sqrt{1-\eta^2}}^2 + \eta^2 = 2 - 2\sqrt{1-\eta^2} \le 2\eta^2. 
\end{align*}
Repeating the argument $2n$ times in the vicinity of other signed basis vectors $\pm \mb e_i$ gives $2n$ local minimizers of $f$. Indeed, the $2n$ symmetric sections cover the sphere with certain overlaps, and a simple calculation shows that no such local minimizer lies in the overlapped regions (due to nearness to a signed basis vector). There is no extra local minimizer, as such local minimizer is contained in at least one of the $2n$ symmetric sections, resulting two different local minimizers in one section, contradicting the uniqueness result we obtained above. 
\end{proof}

Though the $2n$ isolated local minimizers may have different objective values, they are equally good in the sense any of them produces a close approximation to a certain row of $\mb X_0$. As discussed in Section~\ref{sec:overview_geometry}, for cases $\mb A_0$ is an orthobasis other than $\mb I$, the landscape of $f\paren{\mb q; \mb Y}$ is simply a rotated version of the one we characterized above. 

\begin{theorem}[High-dimensional landscape - complete dictionary]\label{thm:geometry_comp}
Suppose $\mb A_0$ is complete with its condition number $\kappa\paren{\mb A_0}$. There exist positive constants $c_\star$ and $C$, such that for any $\theta \in (0,1/2)$ and $\mu < \min\Brac{c_a \theta n^{-1}, c_b n^{-5/4}}$, when $p \ge \frac{C}{c_\star^2 \theta} \max\set{\frac{n^4}{\mu^4}, \frac{n^5}{\mu^2}} \kappa^8\paren{\mb A_0} \log^4\paren{\frac{\kappa\paren{\mb A_0} n}{\mu \theta}}$ and $\overline{\mb Y} \doteq \sqrt{p\theta}\paren{\mb Y \mb Y^*}^{-1/2} \mb Y$, $\mb U \mb \Sigma \mb V^* = \mathtt{SVD}\paren{\mb A_0}$, the following hold simultaneously with high probability: 
\begin{align}
\nabla^2 g(\mb w; \mb V \mb U^* \overline{\mb Y}) &\succeq \frac{c_\star \theta}{2\mu} \mb I \quad &\forall \, \mb w \quad \text{s.t.}& \quad \norm{\mb w}{} \le \frac{\mu}{4\sqrt{2}},  \label{eqn:hess-zero-uni-comp} \\
\frac{\mb w^* \nabla g(\mb w; \mb V \mb U^* \overline{\mb Y})}{\norm{\mb w}{}} &\ge \frac{1}{2}c_\star \theta \quad &\forall \, \mb w \quad \text{s.t.}& \quad \frac{\mu}{4\sqrt{2}} \le \norm{\mb w}{} \le \frac{1}{20 \sqrt{5}} \label{eqn:grad-uni-comp} \\
\frac{\mb w^* \nabla^2 g(\mb w; \mb V \mb U^* \overline{\mb Y}) \mb w}{\norm{\mb w}{}^2} &\le -\frac{1}{2} c_\star \theta \quad &\forall \, \mb w \quad \text{s.t.}& \quad \frac{1}{20 \sqrt{5}} \le \norm{\mb w}{} \le \sqrt{\frac{4n-1}{4n}},   \label{eqn:curvature-uni-comp}
\end{align}
{\em and} the function $g(\mb w; \mb V \mb U^* \overline{\mb Y})$ has exactly one local minimizer $\mb w_\star$ over the open set $\Gamma \doteq \set{ \mb w: \norm{\mb w}{} < \sqrt{ \tfrac{4n-1}{4n} } }$, which satisfies
\begin{equation}
\norm{\mb w_\star - \mb 0 }{} \;\le\; \frac{\mu}{7}.    
\end{equation}
In particular, with this choice of $p$, the probability the claim fails to hold is at most 
$
4np^{-10} + \theta(np)^{-7} + \exp\paren{-0.3\theta n p} + p^{-8} + c_d\exp\left( - c_e p \mu^2 \theta^2/n^2\right)
$. 
Here $c_a$ to $c_e$ are all positive numerical constants.
\end{theorem}
\begin{corollary}  \label{cor:geometry_comp}
Suppose $\mb A_0$ is complete with its condition number $\kappa\paren{\mb A_0}$. There exist positive constants $c_\star$ and $C$, such that for any $\theta \in (0,1/2)$ and $\mu < \min\Brac{c_a \theta n^{-1}, c_b n^{-5/4}}$, when $p \ge \frac{C}{c_\star^2 \theta} \max\set{\frac{n^4}{\mu^4}, \frac{n^5}{\mu^2}} \kappa^8\paren{\mb A_0} \allowbreak \log^4\paren{\frac{\kappa\paren{\mb A_0} n}{\mu\theta}}$ and $\overline{\mb Y} \doteq \sqrt{p\theta}\paren{\mb Y \mb Y^*}^{-1/2} \mb Y$, $\mb U \mb \Sigma \mb V^* = \mathtt{SVD}\paren{\mb A_0}$, with probability at least $1 - 8n^2p^{-10} - \theta(np)^{-7} - \exp\paren{-0.3\theta n p} - p^{-8} - c_d\exp\left( - c_e p \mu^2 \theta^2/n^2\right)$, the function $f\paren{\mb q; \mb V \mb U^* \overline{\mb Y}}$ has exactly $2n$ local minimizers over the sphere $\bb S^{n-1}$. In particular, there is a bijective map between these minimizers and signed basis vectors $\set{\pm \mb e_i}_i$, such that the corresponding local minimizer $\mb q_\star$ and $\mb b \in \set{\pm \mb e_i}_i$ satisfy 
\begin{align}
\norm{\mb q_\star - \mb b}{} \le \frac{\sqrt{2}\mu}{7}. 
\end{align} 
Here $c_a$ to $c_d$ are numerical constants (possibly different from that in the above theorem).
\end{corollary}
We will omit the proof as it is almost identical to that of corollary~\ref{cor:geometry_orth}. 

\subsection{Useful Technical Lemmas and Proof Ideas for Orthogonal Dictionaries} \label{sec:geo_results_orth}
The proof of Theorem~\ref{thm:geometry_orth} is conceptually straightforward: one shows that $\bb E_{\mb X_0}\brac{g\paren{\mb w; \mb X_0}}$ has the claimed properties, and then proves that each of the quantities of interest concentrates uniformly about its expectation. The detailed calculations are nontrivial. 

The next three propositions show that in the expected function landscape, we see successively strongly convex region, nonzero gradient region, and directional negative curvature region when moving away from zero, as depicted in Figure~\ref{fig:large-sample-sphere} and sketched in Section~\ref{sec:overview_geometry}. Note that in this case 
\begin{align*}
\bb E_{\mb X_0}\brac{g\paren{\mb q; \mb X_0}} = \bb E_{\mb x \sim_{i.i.d.} \mathrm{BG}\paren{\theta}}\brac{h_{\mu}\paren{\mb q\paren{\mb w}^* \mb x}}. 
\end{align*}

\begin{proposition} \label{prop:geometry_asymp_curvature}
There exists a positive constant $c$, such that for every $\theta \in \paren{0, \tfrac{1}{2}}$ and any $R_h\in\paren{0, \sqrt{\tfrac{4n-1}{4n}}} $, if $\mu \le c\min\Brac{\theta R_h^2 n^{-1}, R_h n^{-5/4}}$, it holds for every $\mb w$ satisfying $R_h\le \norm{\mb w}{} \le \sqrt{\tfrac{4n-1}{4n}}$ that 
\begin{align*}
\frac{\mb w^* \nabla^2_{\mb w} \expect{h_{\mu}\left(\mb q^*\paren{\mb w} \mb x\right)} \mb w}{\norm{\mb w}{}^2}  \le -\frac{\theta }{2\sqrt{2\pi}}. 
\end{align*}
\end{proposition}
\begin{proof}
See Section~\ref{sec:proof_geo_asym_curvature} on Page~\pageref{sec:proof_geo_asym_curvature}. 
\end{proof}

\begin{proposition} \label{prop:geometry_asymp_gradient}
For every $\theta \in \paren{0, \tfrac{1}{2}}$ and every $\mu \le 9/50$, it holds for every $\mb w$ satisfying $r_g \le \norm{\mb w}{} \le R_g$, where $r_g = \tfrac{\mu}{6\sqrt{2}}$ and $R_g = \tfrac{1-\theta}{10\sqrt{5}}$, that 
\begin{align*}
\frac{\mb w^* \nabla_{\mb w} \expect{h_\mu (\mb q^*\paren{\mb w} \mb x)}}{\norm{\mb w}{}} \ge  \frac{\theta}{20\sqrt{2\pi}}. 
\end{align*}
\end{proposition}
\begin{proof}
See Section~\ref{sec:proof_geo_asym_gradient} on Page~\pageref{sec:proof_geo_asym_gradient}. 
\end{proof}

\begin{proposition} \label{prop:geometry_asymp_strong_convexity}
For every $\theta \in \paren{0, \frac{1}{2}}$, and every $\mu \le \frac{1}{20\sqrt{n}}$, it holds for every $\mb w$ satisfying $\norm{\mb w}{} \le \frac{\mu}{4\sqrt{2}}$ that 
\begin{align*}
\bb E\brac{\nabla^2_{\mb w} h_{\mu}\paren{\mb q^*\paren{\mb w}\mb x} } \succeq \frac{\theta}{25\sqrt{2\pi}\mu} \mb I. 
\end{align*}
\end{proposition}
\begin{proof}
See Section~\ref{sec:proof_geo_asym_strcvx} on Page~\pageref{sec:proof_geo_asym_strcvx}. 
\end{proof}

To prove that the above hold qualitatively for finite $p$, i.e., the function $g\paren{\mb w; \mb X_0}$, we will need first prove that for a fixed $\mb w$ each of the quantity of interest concentrate about their expectation w.h.p., and the function is nice enough (Lipschitz) such that we can extend the results to all $\mb w$ via a discretization argument. The next three propositions provide the desired pointwise concentration results.  

\begin{proposition}\label{prop:concentration-hessian-negative}
Suppose $0 <\mu \leq \frac{1}{\sqrt{n}}$. For every $\mb w \in \Gamma$, it holds that for any $t > 0$, 
\begin{align*}
\bb P \brac{\abs{\frac{\mb w^*\nabla^2 g(\mb w; \mb X_0)\mb w}{\norm{\mb w}{}^2} - \bb E\brac{\frac{\mb w^*\nabla^2 g(\mb w; \mb X_0)\mb w}{\norm{\mb w}{}^2}}}\ge t}\leq 4\exp\paren{- \frac{p\mu^2t^2}{512n^2+32n\mu t}}.
\end{align*}
\end{proposition}
\begin{proof}
See Page~\pageref{proof:pt_cn_curvature} under Section~\ref{proof:cn_point}. 
\end{proof}

\begin{proposition} \label{prop:concentration-gradient}
For every $\mb w \in \Gamma$, it holds that for any $t > 0$, 
\begin{align*}
\bb P\brac{\abs{\frac{\mb w^*\nabla g(\mb w; \mb X_0)
}{\norm{\mb w}{}}-\expect{\frac{\mb w^*\nabla g(\mb w; \mb X_0)
}{\norm{\mb w}{}}}}\geq t} \leq 2\exp\paren{-\frac{pt^2}{8n+4\sqrt{n}t}}.
\end{align*}
\end{proposition}
\begin{proof}
See Page~\pageref{proof:pt_cn_gradient} under Section~\ref{proof:cn_point}. 
\end{proof}

\begin{proposition}\label{prop:concentration-hessian-zero}
Suppose $0 <  \mu\leq \frac{1}{\sqrt{n}}$. For every $\mb w \in \Gamma \cap \set{\mb w: \norm{\mb w}{} \le 1/4}$, it holds that for any $t > 0$, 
\begin{align*}
\bb P\brac{\norm{\nabla^2 g(\mb w; \mb X_0) - \bb E\brac{\nabla^2 g(\mb w; \mb X_0)}}{} \geq t} \;\leq\; 4n\exp\paren{-\frac{p\mu^2t^2}{512n^2+32\mu n t}}. 
\end{align*} 
\end{proposition}
\begin{proof}
See Page~\pageref{proof:pt_cn_strcvx} under Section~\ref{proof:cn_point}. 
\end{proof}

The next three propositions provide the desired Lipschitz results. 

\begin{proposition}[Hessian Lipschitz]\label{prop:lip-hessian-negative}
Fix any $\rconcave \in \paren{0, 1}$. Over the set $\Gamma \cap \set{\mb w: \norm{\mb w}{} \ge \rconcave}$, $\tfrac{\mb w^* \nabla^2 g(\mb w; \mb X_0) \mb w}{\norm{\mb w}{}^2}$ is $\Lconcave$-Lipschitz with 
\begin{align*}
\Lconcave \le \frac{16n^3}{\mu^2} \norm{\mb X_0}{\infty}^3 + \frac{8n^{3/2}}{\mu \rconcave} \norm{\mb X_0}{\infty}^2 + \frac{48 n^{5/2} }{\mu} \norm{\mb X_0}{\infty}^2 + 96 n^{5/2} \norm{\mb X_0}{\infty}.
\end{align*}
\end{proposition} 
\begin{proof}
See Page~\pageref{proof:cn_lips_curvature} under Section~\ref{proof:cn_lips}. 
\end{proof}

\begin{proposition}[Gradient Lipschitz]\label{prop:lip-gradient}
Fix any $r_g \in \paren{0, 1}$. Over the set $\Gamma \cap \set{\mb w: \norm{\mb w}{} \ge r_g}$, $\tfrac{\mb w^* \nabla g( \mb w; \mb X_0 )}{\norm{\mb w}{}}$ is $L_g$-Lipschitz with 
\begin{align*}
L_g \le \frac{2 \sqrt{n} \norm{\mb X_0}{\infty}}{r_g} + 8 n^{3/2} \norm{\mb X_0}{\infty} + \frac{4 n^2}{\mu} \norm{\mb X_0}{\infty}^2.
\end{align*}
\end{proposition} 
\begin{proof}
See Page~\pageref{proof:cn_lips_gradient} under Section~\ref{proof:cn_lips}. 
\end{proof}

\begin{proposition}[Lipschitz for Hessian around zero]\label{prop:lip-hessian-zero}
Fix any $\rconvex \in \paren{0, \frac{1}{2}}$. Over the set $\Gamma \cap \set{\mb w: \norm{\mb w}{} \le \rconvex}$, $\nabla^2 g( \mb w; \mb X_0)$ is $\Lconvex$-Lipschitz with  
\begin{align*}
\Lconvex \;\le\;\frac{4n^2}{\mu^2} \norm{\mb X_0}{\infty}^3+\frac{4n}{\mu}\norm{\mb X_0}{\infty}^2 +  \frac{8\sqrt{2}\sqrt{n}}{\mu}\norm{\mb X_0}{\infty}^2 + 8\norm{\mb X_0}{\infty}. 
\end{align*}
\end{proposition}
\begin{proof}
See Page~\pageref{proof:cn_lips_strcvx} under Section~\ref{proof:cn_lips}. 
\end{proof}
Integrating the above pieces, Section~\ref{sec:proof_geometry_orth} provides a complete proof of Theorem~\ref{thm:geometry_orth}. 

\subsection{Extending to Complete Dictionaries} \label{sec:geo_results_comp}
As hinted in Section~\ref{sec:overview_geometry}, instead of proving things from scratch, we build on the results we have obtained for orthogonal dictionaries. In particular, we will work with the preconditioned data matrix 
\begin{align} \label{eq:precon_def}
\overline{\mb Y} \doteq \paren{\frac{1}{p \theta} \mb Y \mb Y^*}^{-1/2} \mb Y
\end{align}
and show that the function landscape $f\paren{\mb q; \overline{\mb Y}}$ looks qualitatively like that of orthogonal dictionaries (up to a global rotation), provided that $p$ is large enough.  

The next lemma shows $\overline{\mb Y}$ can be treated as being generated from an orthobasis with the same BG coefficients, plus small noise. 
\begin{lemma} \label{lem:pert_key_mag}
For any $\theta \in \paren{0, 1/2}$, suppose $\mb A_0$ is complete with condition number $\kappa\paren{\mb A_0}$ and $\mb X_0 \sim_{i.i.d.} \mathrm{BG}\paren{\theta}$. Provided $p \ge C\kappa^4\paren{\mb A_0} \theta n^2 \log (n \theta \kappa\paren{\mb A_0})$, one can write $\overline{\mb Y}$ as defined in~\eqref{eq:precon_def} as 
\begin{align*}
\overline{\mb Y} = \mb U \mb V^* \mb X_0 + \mb \Xi \mb X_0, 
\end{align*}
for a certain $\mb \Xi$ obeying $\norm{\mb \Xi}{} \le 20\kappa^4\paren{\mb A} \sqrt{\frac{\theta n \log p}{p}}$, with probability at least $1-p^{-8}$. Here $\mb U \mb \Sigma \mb V^* = \mathtt{SVD}\paren{\mb A_0}$, and $C$ is a positive numerical constant. 
\end{lemma}
\begin{proof}
See Page~\pageref{proof:comp_pert_bound} under Section~\ref{sec:proof_geometry_comp}. 
\end{proof}

Notice that $\mb U \mb V^*$ above is orthogonal, and that landscape of $f(\mb q; \overline{Y})$ is simply a rotated version of that of $f(\mb q; \mb V \mb U^* \overline{\mb Y})$, or using the notation in the above lemma, that of $f(\mb q; \mb X_0 + \mb V \mb U^* \mb \Xi \mb X_0) = f(\mb q; \mb X_0 + \widetilde{\mb \Xi} \mb X_0)$ assuming $\widetilde{\mb \Xi} \doteq \mb V \mb U^* \mb \Xi$. So similar to the orthogonal case, it is enough to consider this ``canonical'' case, and its ``canonical'' reparametrization: 
\begin{align*}
g\paren{\mb w; \mb X_0 + \widetilde{\mb \Xi} \mb X_0} = \frac{1}{p}\sum_{k=1}^p h_{\mu}\paren{\mb q^*\paren{\mb w} \paren{\mb x_0}_k + \mb q^*\paren{\mb w} \widetilde{\mb \Xi} \paren{\mb x_0}_k}. 
\end{align*}
The following lemma provides quantitative comparison between the gradient and Hessian of $g\paren{\mb w; \mb X_0 + \widetilde{\mb \Xi} \mb X_0}$ and that of $g\paren{\mb w; \mb X_0}$. 
\begin{lemma} \label{lem:pert_key_grad_hess}
There exist positive constants $C_a$ and $C_b$, such that for all $\mb w \in \Gamma$, 
\begin{align*}
\norm{\nabla_{\mb w} g(\mb w; \mb X_0 + \widetilde{\mb \Xi} \mb X_0) - \nabla_{\mb w{}} g\paren{\mb w; \mb X_0} }{} & \le C_a\frac{n}{\mu} \log\paren{np} \|\widetilde{\mb \Xi}\|, \\
\norm{\nabla_{\mb w}^2 g(\mb w; \mb X_0 + \widetilde{\mb \Xi} \mb X_0) - \nabla_{\mb w}^2 g\paren{\mb w; \mb X_0}}{} & \le C_b \max\set{\frac{n^{3/2}}{\mu^2}, \frac{n^2}{\mu}} \log^{3/2}\paren{np} \|\widetilde{\mb \Xi}\|
\end{align*}
with probability at least $1 - \theta\paren{np}^{-7} - \exp\paren{-0.3 \theta np}$. 
\end{lemma}
\begin{proof}
See Page~\pageref{proof:comp_pert_bound2} under Section~\ref{sec:proof_geometry_comp}. 
\end{proof}
Combining the above two lemmas, it is easy to see when $p$ is large enough, $\|\widetilde{\mb \Xi}\| = \norm{\mb \Xi}{}$ is then small enough (Lemma~\ref{lem:pert_key_mag}), and hence the changes to the gradient and Hessian caused by the perturbation are small. This gives the results presented in Theorem~\ref{thm:geometry_comp}; see Section~\ref{sec:proof_geometry_comp} for the detailed proof. In particular, for the $p$ chosen in Theorem~\ref{thm:geometry_comp}, it holds that 
\begin{align} \label{eq:pert_upper_bound}
\norm{\widetilde{\mb \Xi}}{} \le c c_\star \theta \paren{\max\set{\frac{n^{3/2}}{\mu^2}, \frac{n^2}{\mu}} \log^{3/2}\paren{np}}^{-1}
\end{align}
for a certain constant $c$ which can be made arbitrarily small by making the constant $C$ in $p$ large.

%% file: sec/algorithm.tex
\section{Finding One Local Minimizer via the Riemannian Trust-Region Method} \label{sec:algorithm}
The above geometric results show every local minimizer of $f(\mb q; \widehat{\mb Y})$ over $\bb S^{n-1}$ approximately recovers one row of $\mb X_0$. So the crucial problem left now is how to efficiently obtain one of the local minimizers. The presence of saddle points have motivated us to develop a (second-order)
Riemannian trust-region algorithm over the sphere; the existence of descent directions at nonoptimal points drives the trust-region iteration sequence towards one of the minimizers asymptotically. We will prove that under our modeling assumptions, this algorithm efficiently produces an accurate approximation\footnote{By ``accurate'' we mean one can achieve an arbitrary numerical accuracy $\eps > 0$ with a reasonable amount of time. Here the running time of the algorithm is on the order of $\log \log ( 1/\eps )$ in the target accuracy $\eps$, and polynomial in other problem parameters. } to one of the minimizers. Throughout the exposition, basic knowledge of Riemannian geometry is assumed. We will try to keep the technical requirement minimal possible; the reader can consult the excellent monograph~\cite{absil2009} for relevant background and details. 

\subsection{The Riemannian Trust-Region Algorithm over the Sphere}\label{subsec:TRM-algorithm}
We are interested to seek one local minimizer of the problem
\begin{align} \label{eq:main_obj}
\mini\; \quad f(\mb q; \widehat{\mb Y}) \doteq \frac{1}{p}\sum_{k=1}^p h_{\mu}( \mb q^* \widehat{\mb y}_i) \quad \st \quad \mb q \in \bb S^{n-1}.
\end{align}
For a function $f$ in the Euclidean space, the typical TRM starts from some initialization $\mb q^{(0)} \in \bb R^{n}$, and produces a sequence of iterates $\mb q^{(1)}, \mb q^{(2)}, \dots$, by repeatedly minimizing a quadratic approximation $\widehat{f}$ to the objective function $f(\mb q)$, over a ball centered about the current iterate. 

Here, we are interested in the restriction of $f$ to the unit sphere $\bb S^{n-1}$. Instead of directly approximating the function in $\bb R^n$, we form quadratic approximations of $f$ in the tangent space of $\bb S^{n-1}$. Recall that the tangent space of a sphere at a point $\mb q\in \bb S^{n-1}$ is $T_{\mb q}\bb S^{n-1} =\Brac{\mb \delta \in \bb R^n\;|\;\mb q^*\mb \delta = 0 } $, i.e., the set of vectors that are orthogonal to $\mb q$. Consider $\mb \delta \in T_{\mb q} \bb S^{n-1}$ with $\norm{\mb \delta}{} = 1$. The map $\gamma\paren{t}: t \mapsto \mb q \cos t + \mb \delta \sin t$ defines a smooth curve on the sphere that satisfies $\gamma\paren{0} = \mb q$ and $\dot{\gamma}\paren{0} = \mb \delta$. The function $f\circ \gamma\paren{t}$ obviously is smooth and we expect Taylor expansion around $0$ a good approximation of the function, at least in the vicinity of $0$. Taylor's theorem gives  
\begin{align*}
f \circ \gamma\paren{t} = f\paren{\mb q} + t\innerprod{\nabla f\paren{\mb q}}{\mb \delta} + \frac{t^2}{2}\paren{\mb \delta^* \nabla^2 f\paren{\mb q} \mb \delta - \innerprod{\nabla f\paren{\mb q}}{\mb q}} + O\paren{t^3}. 
\end{align*}
We therefore form the ``quadratic'' approximation $\widehat{f}\paren{\mb \delta; \mb q}: T_{\mb q} \bb S^{n-1} \mapsto \R$ as 
\begin{equation}
\widehat{f}(\mb \delta; \mb q, \widehat{\mb Y}) \;\doteq\; f(\mb q) + \innerprod{ \nabla f(\mb q; \widehat{\mb Y}) }{\mb \delta } + \frac{1}{2} \mb \delta^* \left( \nabla^2 f( \mb q; \widehat{\mb Y}) - \innerprod{ \nabla f(\mb q; \widehat{\mb Y}) }{ \mb q } \mb I \right) \mb \delta. \label{eqn:f-appx}
\end{equation}
Given the previous iterate $\mb q^{(k-1)}$, the TRM produces the next iterate by generating a solution $\widehat{\mb \delta}$ to 
\begin{equation} \label{eqn:subproblem-1}
\min_{\mb \delta \in T_{\mb q^{(k-1)}} \bb S^{n-1}, \; \norm{\mb \delta}{} \le \Delta} \quad \widehat{f}(\mb \delta; \mb q^{(k-1)}),
\end{equation}
and then ``pull'' the solution $\widehat{\mb \delta}$ from $T_{\mb q} \bb S^{n-1}$ back to $\bb S^{n-1}$. Moreover, for any vector $\mb \delta \in T_{\mb q} \bb S^{n-1}$, the exponential map $\exp_{\mb q}\paren{\mb \delta}: T_{\mb q} \bb S^{n-1} \mapsto \bb S^{n-1}$ is 
\begin{align*}
\exp_{\mb q}\paren{\mb \delta} = \mb q \cos \norm{\mb \delta}{} + \frac{\mb \delta}{\norm{\mb \delta}{}} \sin \norm{\mb \delta}{}. 
\end{align*}
If we choose the exponential map to pull back the movement $\widehat{\mb \delta}$\footnote{The exponential map is only one of the many possibilities; also for general manifolds other retraction schemes may be more practical. See exposition on retraction in Chapter 4 of~\cite{absil2009}. }, the next iterate then reads 
\begin{align}
\mb q^{(k)} = \mb q^{(k-1)} \cos \|\widehat{\mb \delta}\| + \frac{\widehat{\mb \delta}}{\|\widehat{\mb \delta}\|} \sin \|\widehat{\mb \delta}\|. 
\end{align}
We have motivated~\eqref{eqn:f-appx} and hence the algorithm in an intuitive way from the Taylor approximation to the function $f$ over $\bb S^{n-1}$. To understand its properties, it is useful to interpret it as a \emph{Riemannian trust-region method} over the manifold $\bb S^{n-1}$. The class of algorithm is discussed in detail in the monograph~\cite{absil2009}. In particular, the quadratic approximation~\eqref{eqn:f-appx} can be obtained by noting that the function $f \circ \exp_{\mb q}(\mb \delta; \widehat{\mb Y}): T_{\mb q} \bb S^{n-1} \mapsto \R$ obeys  
\begin{align*}
f \circ \exp_{\mb q}(\mb \delta; \widehat{\mb Y}) = f(\mb q; \widehat{\mb Y}) + \innerprod{\mb \delta}{\grad f(\mb q; \widehat{\mb Y})} + \frac{1}{2} \mb \delta^* \Hess f(\mb q; \widehat{\mb Y}) \mb \delta + O(\norm{\mb \delta}{}^3), 
\end{align*}
where $\grad f(\mb q; \widehat{\mb Y})$ and $\Hess f(\mb q; \widehat{\mb Y})$ are the Riemannian gradient and Riemannian Hessian~\cite{absil2009} respectively, defined as \begin{align*}
\grad f(\mb q; \widehat{\mb Y}) & \doteq \mc P_{T_{\mb q} \bb S^{n-1}} \nabla f(\mb q; \widehat{\mb Y}), \\
\Hess f(\mb q; \widehat{\mb Y}) & \doteq \mc P_{T_{\mb q} \bb S^{n-1}} \paren{\nabla^2 f(\mb q; \widehat{\mb Y}) - \innerprod{\nabla f(\mb q; \widehat{\mb Y})}{\mb q} \mb I}\mc P_{T_{\mb q} \bb S^{n-1}}, 
\end{align*}
with $\mc P_{T_{\mb q} \bb S^{n-1}} \doteq \mb I - \mb q \mb q^*$ the orthoprojector onto the tangent space $T_{\mb q} \bb S^{n-1}$. We will use these standard notions in analysis of the algorithm. 

To solve the subproblem \eqref{eqn:subproblem-1} numerically, we can take any matrix $\mb U \in \bb R^{n \times (n-1)}$ whose columns form an orthonormal basis for $T_{\mb q^{(k-1)}}\bb S^{n-1}$, and produce a solution $\widehat{\mb \xi}$ to
\begin{equation}
\min_{\norm{\mb \xi}{} \le \Delta} \quad \widehat{f}(\mb U \mb \xi; \mb q^{(k-1)}),    \label{eqn:trsp} 
\end{equation}
where by~\eqref{eqn:f-appx}, 
\begin{multline*}
\widehat{f}(\mb U \mb \xi; \mb q^{(k-1)}) = f(\mb q) + \innerprod{\mb U^* \nabla f(\mb q^{(k-1)})}{\mb \xi} + \\
\frac{1}{2} \mb \xi^* \paren{\mb U^* \nabla^2 f(\mb q^{(k-1)}; \widehat{\mb Y}) \mb U - \innerprod{\nabla f(\mb q^{(k-1)}; \widehat{\mb Y})}{\mb q^{(k-1)}} \mb I_{n-1} } \mb \xi. 
\end{multline*}
Solution to~\eqref{eqn:subproblem-1} can then be recovered as $\widehat{\mb \delta} = \mb U\widehat{\mb \xi}$. The problem \eqref{eqn:trsp} is an instance of the classic {\em trust region subproblem}, i.e., minimizing a quadratic function subject to a single quadratic constraint, which can be solved in polynomial time, either by root finding methods~\cite{more1983computing, conn2000trust} or by semidefinite programming (SDP)~\cite{rendl1997semidefinite, ye2003new, fortin2004trust, hazan2014linear}. As the root finding methods numerically suffer from the so-called ``hard case''~\cite{more1983computing}, we deploy the SDP approach here. We introduce
\begin{align}
\tilde{\mb \xi} = \brac{\mb \xi^*,1}^*,~\mb \Theta = \tilde{\mb \xi} \tilde{\mb \xi} ^*,~\mb M = \brac{\begin{array}{ll}
\mb A &\mb b\\
\mb b^* &0
\end{array}},
\end{align}
where $\mb A = \mb U^*(\nabla^2 f(\mb q^{(k-1)}; \widehat{\mb Y}) - \innerprod{\nabla f(\mb q^{(k-1)}; \widehat{\mb Y})}{\mb q^{(k-1)}}\mb I) \mb U$ and $\mb b = \mb U^* \nabla f(\mb q^{(k-1)}; \widehat{\mb Y})$. The resulting SDP to solve is 
\begin{align}
\mini_{\; \mb \Theta} \innerprod{\mb M}{\mb \Theta},~\st~\trace({\mb \Theta})\le \Delta^2+ 1,~\innerprod{\mb E_{n+1}}{\mb \Theta}=1,~\mb \Theta \succeq \mb 0, \label{eqn:SDP_relaxation}
\end{align}
where $\mb E_{n+1} = \mb e_{n+1} \mb e_{n+1}^*$. Once the problem \eqref{eqn:SDP_relaxation} is solved to its optimal $\mb \Theta_\star$, one can provably recover the optimal solution $\mb \xi_\star$ of \eqref{eqn:trsp} by computing the SVD of $\mb \Theta_\star= \widetilde{\mb U}\mb \Sigma \widetilde{\mb V}^*$, and extract as a subvector by the first $n-1$ coordinates of the principal eigenvector $\widetilde{\mb u}_1$ (see Appendix B of \cite{boyd2004optimization}). 
 
The choice of trust region size $\Delta$ is important both for the convergence theory and practical effectiveness of TRMs. Following standard recommendations (see, e.g., Chapter 4 of \cite{NocedalWright}), we use a backtracking approach which modifies $\Delta$ from iteration to iteration based on the accuracy of the approximation $\widehat{f}$. The whole algorithmic procedure is described as pseudocode as Algorithm \ref{alg:trm}.
\begin{algorithm}
\caption{Riemannian TRM Algorithm for Finding One Local Minimizer}
\label{alg:trm}
\begin{algorithmic}[1]
\renewcommand{\algorithmicrequire}{\textbf{Input:}}
\renewcommand{\algorithmicensure}{\textbf{Output:}}
\Require{Data matrix $\mb Y \in \bb R^{n \times p}$, smoothing parameter $\mu$ and parameters $\eta_{vs},~\eta_s,~\gamma_i,~\gamma_d, \Delta_{\max}, ~\Delta_{\min}$}
\Ensure{$\widehat{\mb q} \in \bb S^{n-1}$}
\State Initialize $\mb q^{(0)} \in \bb S^{n-1}$, $\Delta^{(0)}$ and $k=1$,
\While{not converged}
\State Set $\mb U \in \bb R^{n \times (n-1)}$ to be an orthonormal basis for $T_{\mb q^{(k-1)}} \bb S^{n-1}$
\State Solve the trust region subproblem \
\begin{equation*}
\widehat{\mb \xi}=\mathop{\arg\min}_{\norm{\mb \xi}{} \le \Delta^{(k-1)}} \widehat{f}(\mb U \mb \xi; \mb q^{(k-1)}, \widehat{\mb Y} )   
\end{equation*} 
\State Set \
\begin{equation*}
\widehat{\mb \delta} \leftarrow \mb U \widehat{\mb \xi},\quad \widehat{\mb q} \leftarrow \mb q^{(k-1)} \cos \|\widehat{\mb \delta}\| + \frac{\widehat{\mb \delta}}{\|\widehat{\mb \delta}\|} \sin \|\widehat{\mb \delta}\|.  
\end{equation*}
\State Set
\begin{equation*}
\rho_k \leftarrow \frac{f(\mb q^{(k-1)}; \widehat{\mb Y}) - f( \widehat{\mb q}; \widehat{\mb Y}) }{f(\mb q^{(k-1)}; \widehat{\mb Y}) - \widehat{f}(\widehat{\mb \delta}; \mb q^{(k-1)}, \widehat{\mb Y}) }
\end{equation*}
\If{$\rho_k\ge \eta_{vs}$ and $\|\hat{\mb \xi}\| = \Delta^{(k-1)}$ }
\State Set $\mb q^{(k)} \leftarrow  \widehat{\mb q}$ and $\Delta^{(k)} \leftarrow  \min\paren{\gamma_i\Delta^{(k-1)},\Delta_{\max}}$. \Comment{very successful}
\ElsIf{$\rho_k\ge \eta_{s}$}
\State Set $\mb q^{(k)} \leftarrow  \widehat{\mb q}$ and $\Delta^{(k)} \leftarrow  \Delta^{(k-1)}$. \Comment{successful}
\Else 
\State Set $\mb q^{(k)} \leftarrow  \mb q^{(k-1)} $ and $\Delta^{(k)} \leftarrow  \max\paren{\gamma_d\Delta^{(k-1)},\Delta_{\min}}$. \Comment{unsuccessful}
\EndIf
\State Set $k = k+1$.
\EndWhile
\end{algorithmic}
\end{algorithm}
In our numerical implementation, we randomly initialize $\mb q^{(0)}$ and set $\Delta^{(0)}=0.1, \eta_{vs} = 0.9,~\eta_s = 0.1,~\gamma_d = 1/2,~\gamma_i = 2$, $\Delta_{\max}= 1$ and $\Delta_{\min} = 10^{-16}$, and the algorithm is stopped when $\paren{f(\widehat{\mb q}) - f(\mb q^{(k-1)})}/\|\widehat{\mb \delta}\| \le 10^{-6}$.

\subsection{Main Convergence Results}
By using general results on the Riemannian TRM (see, e.g., Chapter 7 of \cite{absil2009}), it is not difficult to prove that the iterates $\mb q^{(k)}$ produced by Algorithm \ref{alg:trm} converge to a critical point of the objective $f(\mb q)$ over $\bb S^{n-1}$. In this section, we show that under our probabilistic assumptions, this claim can be strengthened. In particular, the algorithm is guaranteed to produce an accurate approximation to a local minimizer of the objective function, in a number of iterations that is polynomial in the problem size. The arguments described in Section~\ref{sec:geometry} show that with high probability every local minimizer of $f$ produces a close approximation of one row of $\mb X_0$. Taken together, this implies that the algorithm efficiently produces a close approximation to one row of $\mb X_0$. 

Our next two theorems summarize the convergence results for orthogonal and complete dictionaries, respectively.
 
\begin{theorem}[TRM convergence - orthogonal dictionary] \label{thm:trm_orth}
Suppose the dictionary $\mb A_0$ is orthogonal. Then there exists a positive constant $C$, such that for all $\theta \in \paren{0, 1/2}$, and $\mu < \min \set{c_a \theta n^{-1}, c_b n^{-5/4}}$, whenever $\exp(n) \ge p \ge C n^3 \log \tfrac{n}{\mu \theta}/(\mu^2 \theta^2)$, with probability at least
$
1 - 8 n^2 p^{-10} - \theta (np)^{-7} - \exp\paren{-0.3\theta n p} - p^{-10} - c_c\exp\left( - c_d p \mu^2 \theta^2/n^2\right), 
$
the Riemannian trust-region algorithm with input data matrix $\widehat{\mb Y} = \mb Y$, any initialization $\mb q^{(0)}$ on the sphere, and a step size satisfying 
\begin{align}
\Delta \le \min\set{\frac{c_e c_\star \theta \mu^2}{n^{5/2} \log^{3/2}\paren{np}}, \frac{c_f c_\sharp^3 \theta^3 \mu}{n^{7/2}\log^{7/2}\paren{np}}}.
\end{align}
returns a solution $\widehat{\mb q}\in \bb S^{n-1} $ which is $\eps$ near to one of the local minimizers $\mb q_\star$ (i.e., $\norm{\widehat{\mb q}-\mb q_\star}{}\leq \eps$) in 
\begin{align}
\max\set{\frac{c_g n^6 \log^3\paren{np}}{c_{\star}^3 \theta^3 \mu^4 }, \frac{c_h n }{c_\sharp^2 \theta^2 \Delta^2}}\paren{f(\mb q^{(0)}) - f(\mb q_\star)} + \log\log \frac{c_i c_\star \theta \mu}{\eps n^{3/2} \log^{3/2}\paren{np}}
\end{align}
iterations. Here $c_\star$, $c_\sharp$ as defined in Theorem~\ref{thm:geometry_orth} and Lemma~\ref{lem:alg_strcvx_lb} respectively ($c_\star$ and $c_\sharp$ can be set to the same constant value), and $c_a$, $c_b$ are the same numerical constants as defined in Theorem~\ref{thm:geometry_orth}, $c_c$ to $c_i$ are other positive numerical constants. 
\end{theorem}

\begin{theorem}[TRM convergence - complete dictionary] \label{thm:trm_comp}
Suppose the dictionary $\mb A_0$ is complete with condition number $\kappa\paren{\mb A_0}$. There exists a positive constant $C$, such that for all $\theta \in \paren{0, 1/2}$, and $\mu < \min \set{c_a \theta n^{-1}, c_b n^{-5/4}}$, whenever $\exp(n) \ge p \ge \frac{C}{c_\star^2 \theta} \max\set{\frac{n^4}{\mu^4}, \frac{n^5}{\mu^2}} \kappa^8\paren{\mb A_0} \log^4\paren{\frac{\kappa\paren{\mb A_0} n}{\mu \theta}}$, with probability at least
$
1 - 8 n^2 p^{-10} - \theta (np)^{-7} - \exp\paren{-0.3\theta n p} -2p^{-8} - c_c\exp\left( - c_d p \mu^2 \theta^2/n^2\right), 
$
the Riemannian trust-region algorithm with input data matrix $\overline{\mb Y} \doteq \sqrt{p\theta}\paren{\mb Y \mb Y^*}^{-1/2} \mb Y$ where $\mb U \mb \Sigma \mb V^* = \mathtt{SVD}\paren{\mb A_0}$, any initialization $\mb q^{(0)}$ on the sphere and a step size satisfying 
\begin{align}
\Delta \le \min\set{\frac{c_e c_\star \theta \mu^2}{n^{5/2} \log^{3/2}\paren{np}}, \frac{c_f c_\sharp^3 \theta^3 \mu}{n^{7/2}\log^{7/2}\paren{np}}}.
\end{align}
returns a solution $\widehat{\mb q}\in \bb S^{n-1} $ which is $\eps$ near to one of the local minimizers $\mb q_\star$ (i.e., $\norm{\widehat{\mb q}-\mb q_\star}{}\leq \eps$) in 
\begin{align}
\max\set{\frac{c_g n^6 \log^3\paren{np}}{c_{\star}^3 \theta^3 \mu^4 }, \frac{c_h n }{c_\sharp^2 \theta^2 \Delta^2}}\paren{f(\mb q^{(0)}) - f(\mb q_\star)} + \log\log \frac{c_i c_\star \theta \mu}{\eps n^{3/2} \log^{3/2}\paren{np}}
\end{align}
iterations. Here $c_\star$, $c_\sharp$ as defined in Theorem~\ref{thm:geometry_orth} and Lemma~\ref{lem:alg_strcvx_lb} respectively ($c_\star$ and $c_\sharp$ can be set to the same constant value), and $c_a$, $c_b$ are the same numerical constants as defined in Theorem~\ref{thm:geometry_orth}, $c_c$ to $c_i$ are other positive numerical constants. 
\end{theorem}
Our convergence result shows that for any target accuracy $\eps > 0$ the algorithm terminates within polynomially many steps. Our estimate of the number of steps is pessimistic: our analysis has assumed a fixed step size $\Delta$ and the running time is relatively large degree polynomial in $p$ and $n$, while on typical numerical examples (e.g., $\mu = 10^{-2}$, $n \sim 100$, and $\eps= O(\mu)$), the algorithm with adaptive step size as described in Algorithm~\ref{alg:trm} produces an accurate solution in relatively few ($20$-$50$) iterations. Nevertheless, our goal in stating the above results is not to provide a tight analysis, but to prove that the Riemannian TRM algorithm finds a local minimizer in polynomial time. For nonconvex problems, this is not entirely trivial -- results of~\cite{murty1987some} show that in general it is NP-hard to find a local minimum of a nonconvex function. 

\subsection{Useful Technical Results and Proof Ideas for Orthogonal Dictionaries} \label{sec:alg_orth}
The reason that our algorithm is successful derives from the geometry depicted in Figure \ref{fig:large-sample-sphere} and formalized in Theorem~\ref{thm:geometry_orth}. Basically, the sphere $\bb S^{n-1}$ can be divided into three regions. Near each local minimizer, the function is strongly convex, and the algorithm behaves like a standard (Euclidean) TRM algorithm applied to a strongly convex function -- in particular, it exhibits a quadratic asymptotic rate of convergence. Away from local minimizers, the function always exhibits either a strong gradient, or a direction of negative curvature (an eigenvalue of the Hessian which is bounded below zero). The Riemannian TRM aglorithm is capable of exploiting these quantities to reduce the objective value by at least a constant in each iteration. The total number of iterations spent away from the vicinity of the local minimizers can be bounded by comparing this constant to the initial objective value. Our proofs follow exactly this line and make the various quantities precise.   

\subsubsection{Basic Facts about the Sphere}
For any point $\mb q \in \bb S^{n-1}$, the tangent space $T_{\mb q} \bb S^{n -1}$ and the orthoprojector $\mc P_{T_{\mb q} \bb S^{n -1}}$ onto $T_{\mb q} \bb S^{n -1}$ are given by
\begin{align*}
	T_{\mb q} \bb S^{n -1 } &= \set{ \mb \delta\in \bb R^n  \mid \mb q^* \mb \delta = 0 }, \\
	\mc P_{T_{\mb q} \bb S^{n -1}} &= (\mb I - \mb q \mb q^* ) = \mb U \mb U^*,
\end{align*}
where $\mb U \in \bb R^{n\times (n-1)}$ is an arbitrary orthonormal basis for $T_{\mb q}\bb S^{n-1}$ (note that the orthoprojector is independent of the basis $\mb U$ we choose). Moreover, for any $\mb \delta \in T_{\mb q} \bb S^{n -1}$, the exponential map $\exp_{\mb q}(\mb \delta): T_{\mb q}\bb S^{n-1} \mapsto \bb S^{n-1} $ is given by 
\begin{align*}
	\exp_{\mb q}(\mb \delta) = \mb q \cos\norm{\mb \delta}{} + \frac{\mb \delta}{\norm{\mb \delta}{}} \sin\norm{\mb \delta}{}.
\end{align*}
Let $\nabla f(\mb q)$ and $\nabla^2 f(\mb q)$ denote the usual (Euclidean) gradient and Hessian of $f$ w.r.t. $\mb q$ in $\R^n$. For our specific $f$ defined in~\eqref{eq:main_obj}, it is easy to check that
\begin{align}
\nabla f\paren{\mb q; \widehat{\mb Y}} & = \frac{1}{p}\sum_{k=1}^p \tanh\paren{\frac{\mb q^*\widehat{\mb y}_k}{\mu}} \widehat{\mb y}_k, \label{eq:fq_grad}\\
\nabla^2 f\paren{\mb q; \widehat{\mb Y}} & = \frac{1}{p} \sum_{k=1}^p \frac{1}{\mu}\brac{1-\tanh^2\paren{\frac{\mb q^* \widehat{\mb y}_k}{\mu}}} \widehat{\mb y}_k \widehat{\mb y}^*_k. \label{eq:fq_hess}
\end{align}
Since $\bb S^{n-1}$ is an embedded submanifold of $\reals^n$, the Riemannian gradient and Riemannian Hessian defined on $T_{\mb q} \bb S^{n -1 }$ are given by 
\begin{align} 
\grad f (\mb q; \wh{\mb Y}) & = \mc P_{T_{\mb q} \bb S^{n-1}} \nabla f(\mb q; \wh{\mb Y}), \label{eq:fq_rie_grad} \\
\Hess f (\mb q; \wh{\mb Y}) &= \mc P_{T_{\mb q} \bb S^{n-1}} \left( \nabla^2 f(\mb q; \wh{\mb Y}) - \innerprod{\nabla f(\mb q; \wh{\mb Y}) }{\mb q} \mb I \right) \mc P_{T_{\mb q} \bb S^{n-1}}; \label{eq:fq_rie_hess}
 \end{align}
so the second-order Taylor approximation for the function $f$ is
\begin{align*}
	\widehat{f}\paren{\mb \delta; \mb q, \widehat{\mb Y}} = f(\mb q; \wh{\mb Y}) + \innerprod{\mb \delta}{\grad f(\mb q; \wh{\mb Y})}+ \frac{1}{2} \mb \delta^* \Hess f(\mb q; \wh{\mb Y})\mb \delta,\qquad \forall~\mb \delta \in T_{\mb q} \bb S^{n-1}.
\end{align*}
The first order necessary condition for {\em unconstrained} minimization of function $\widehat{f}$ over $T_{\mb q} \bb S^{n-1}$ is
\begin{align}\label{eqn:ts-optimal-solution-1}
	\grad f(\mb q; \wh{\mb Y}) + \Hess f(\mb q; \wh{\mb Y}) \mb \delta_\star = \mb 0; 
\end{align}
if $\Hess f(\mb q)$ is positive semidefinite and has full rank $n-1$ (hence ``nondegenerate"\footnote{Note that the $n \times n$ matrix $\Hess f(\mb q; \wh{\mb Y})$ has rank at most $n-1$, as the nonzero $\mb q$ obviously is in its null space. When  $\Hess f(\mb q; \wh{\mb Y})$ has rank $n-1$, it has no null direction in the tangent space. Thus, in this case it acts on the tangent space like a full-rank matrix. }),  
the unique solution $\mb \delta_\star$ is
\begin{align*}
	\mb \delta_\star = -\mb U \paren{\mb U^*\brac{\Hess f(\mb q)} \mb U }^{-1} \mb U^* \grad f(\mb q), 
\end{align*}
which is also invariant to the choice of basis $\mb U$. Given a tangent vector $\mb \delta \in T_{\mb q}\bb S^{n-1}$, let $\gamma(t) \doteq \exp_{\mb q}(t \mb \delta)$ denote a geodesic curve on $\bb S^{n-1}$. Following the notation of \cite{absil2009}, let 
 \begin{align*}
 	\mc P_{\gamma}^{\tau \leftarrow 0} : T_{\mb q} \bb S^{n-1} \to T_{\gamma(\tau)} \bb S^{n-1}
 \end{align*}
denotes the parallel translation operator, which translates the tangent vector $\mb \delta$ at $\mb q = \gamma(0)$ to a tangent vector at $\gamma(\tau)$, in a ``parallel'' manner. In the sequel, we identify $\mc P_{\gamma}^{\tau \leftarrow 0}$ with the following $n \times n$ matrix, whose restriction to $T_{\mb q} \bb S^{n-1}$ is the parallel translation operator (the detailed derivation can be found in Chapter 8.1 of \cite{absil2009}):
 \begin{eqnarray} \label{eq:alg_tsp_op}
 \mc P_{\gamma}^{\tau \leftarrow 0} &=& \left( \mb I - \frac{\mb \delta \mb \delta^*}{\norm{\mb \delta}{}^2} \right) - \mb q \sin\left( \tau \norm{\mb \delta}{} \right) \frac{\mb \delta^*}{\norm{\mb \delta}{}} + \frac{\mb \delta}{\norm{\mb \delta}{}} \cos\left( \tau \norm{\mb \delta }{} \right) \frac{\mb \delta^* }{\norm{\mb \delta}{}} \nonumber \\
 &=& \mb I + \left( \cos( \tau \norm{\mb \delta}{} ) - 1 \right) \frac{\mb \delta \mb \delta^*}{\norm{\mb \delta }{}^2} -  \sin\left( \tau \norm{\mb \delta }{} \right) \frac{\mb q \mb \delta^*}{\norm{\mb \delta}{}}. 
 \end{eqnarray}
Similarly, following the notation of \cite{absil2009}, we denote the inverse of this matrix by $\mc P_{\gamma}^{0 \leftarrow \tau}$, where its restriction to $T_{\gamma(\tau)} \bb S^{n-1}$ is the inverse of the parallel translation operator $\mc P_{\gamma}^{\tau \leftarrow 0}$.

\subsubsection{Key Steps towards the Proof}
Note that for any orthogonal $\mb A_0$, $f\paren{\mb q; \mb A_0 \mb X_0} = f\paren{\mb A_0^* \mb q; \mb X_0}$. In words, this is the above established fact that the function landscape of $f(\mb q; \mb A_0 \mb X_0)$ is a rotated version of that of $f(\mb q; \mb X_0)$. Thus, any local minimizer $\mb q_\star$ of $f(\mb q; \mb X_0)$ is rotated to $\mb A_0 \mb q_\star$, one minimizer of $f(\mb q; \mb A_0 \mb X_0)$. Also if our algorithm generates iteration sequence $\mb q_0, \mb q_1, \mb q_2, \dots$ for $f(\mb q; \mb X_0)$ upon initialization $\mb q_0$, it will generate the iteration sequence $\mb A_0  \mb q_0, \mb A_0 \mb q_1, \mb A_0 \mb q_2, \dots$ for $f\paren{\mb q; \mb A_0 \mb X_0}$. So w.l.o.g. it is adequate that we prove the convergence results for the case $\mb A_0 = \mb I$. So in this section (Section~\ref{sec:alg_orth}), we write $f(\mb q)$ to mean $f(\mb q; \mb X_0)$. 

We partition the sphere into three regions, for which we label as $\rI$, $\rII$, $\rIII$, corresponding to the strongly convex, nonzero gradient, and negative curvature regions, respectively (see Theorem~\ref{thm:geometry_orth}). That is, $\rI$ consists of a union of $2n$ spherical caps
of radius $\tfrac{\mu}{4\sqrt{2}}$, each centered around a signed standard basis vector $\pm \mb e_i$. $\rII$ consist of the set difference of a union of $2n$ spherical caps of radius $\tfrac{1}{20\sqrt{5}}$, centered around the standard basis vectors $\pm \mb e_i$, and $\rI$. Finally, $\rIII$ covers the rest of the sphere. We say a trust-region step takes an $\rI$ step if the current iterate is in $\rI$; similarly for $\rII$ and $\rIII$ steps. Since we use the geometric structures derived in Theorem~\ref{thm:geometry_orth} and Corollary~\ref{cor:geometry_orth}, the conditions 
\begin{align} \label{eq:trm_proof_assumed_cond}
\theta \in (0,1/2), \quad \mu < \min\Brac{c_a \theta n^{-1}, c_b n^{-5/4}}, \quad p \ge \frac{C}{\mu^2 \theta^2} n^3 \log \frac{n}{\mu \theta}
\end{align}
are always in force. 

At each step $k$ of the algorithm, suppose $\mb \delta^{(k)}$ is the minimizer of the trust-region subproblem~\eqref{eqn:subproblem-1}. We call the step ``{\em constrained}'' if $\norm{\mb \delta^{(k)}}{} = \Delta$ (the minimizer lies on the boundary and hence the constraint is active), and call it ``{\em unconstrained}'' if $\|\mb \delta^{(k)}\| < \Delta$ (the minimizer lies in the relative interior and hence the constraint is not in force). Thus, in the unconstrained case the optimality condition is~\eqref{eqn:ts-optimal-solution-1}. 

The next lemma provides some estimates about $\nabla f$ and $\nabla^2 f$ that are useful in various contexts. 
\begin{lemma} \label{lem:mag_lip_fq}
We have the following estimates about $\nabla f$ and $\nabla^2 f$: 
\begin{align*}
\sup_{\mb q \in \bb S^{n-1}}\norm{\nabla f\paren{\mb q}}{} & \doteq M_{\nabla} \le \sqrt{n} \norm{\mb X_0}{\infty} , \\
\sup_{\mb q \in \bb S^{n-1}}\norm{\nabla^2 f\paren{\mb q}}{} & \doteq M_{\nabla^2} \le \frac{n}{\mu}\norm{\mb X_0}{\infty}^2, \\
\sup_{\mb q, \mb q' \in \bb S^{n-1}, \mb q \neq \mb q'} \frac{\norm{\nabla f\paren{\mb q} - \nabla f\paren{\mb q'}}{}}{\norm{\mb q - \mb q'}{}} & \doteq L_{\nabla} \le \frac{n}{\mu} \norm{\mb X_0}{\infty}^2, \\
\sup_{\mb q, \mb q' \in \bb S^{n-1}, \mb q \neq \mb q'} \frac{\norm{\nabla^2 f\paren{\mb q} - \nabla^2 f\paren{\mb q'}}{}}{\norm{\mb q - \mb q'}{}} & \doteq L_{\nabla^2} \le \frac{2}{\mu^2} n^{3/2} \norm{\mb X_0}{\infty}^3. 
\end{align*}
\end{lemma}
\begin{proof}
See Page~\pageref{proof:lem_mag_lip_fq} under Section~\ref{sec:proof_algorithm}. 
\end{proof}
Our next lemma says if the trust-region step size $\Delta$ is small enough, one Riemannian trust-region step reduces the objective value by a certain amount when there is any descent direction.  
\begin{lemma} \label{lem:alg_approx_bd2}
Suppose that the trust region size $\Delta \le 1$, and there exists a tangent vector $\mb \delta \in T_{\mb q} \bb S^{n-1}$ with $\norm{\mb \delta}{} \le \Delta$, such that
\begin{equation*}
f( \exp_{\mb q}(\mb \delta) ) \;\le\; f(\mb q) - s
\end{equation*}
for some positive scalar $s\in \reals $. Then the trust region subproblem produces a point $\mb \delta_\star$ with 
\begin{equation*}
f(\exp_{\mb q}(\mb \delta_\star)) \;\le\; f(\mb q) - s + \frac{1}{3}\eta_f \Delta^3, 
\end{equation*}
where $\eta_f \doteq M_{\nabla} + 2M_{\nabla^2} + L_{\nabla} + L_{\nabla^2}$ and $M_{\nabla}$, $M_{\nabla^2}$, $L_{\nabla}$, $L_{\nabla^2}$ are the quantities defined in Lemma~\ref{lem:mag_lip_fq}. 
\end{lemma}
\begin{proof}
See Page~\pageref{proof:lem_alg_approx_bd2} under Section~\ref{sec:proof_algorithm}. 
\end{proof}

To show decrease in objective value for $\rII$ and $\rIII$, now it is enough to exhibit a descent direction for each point in these regions. The next two lemmas help us almost accomplish the goal. For convenience again we choose to state the results for the ``canonical'' section that is in the vicinity of $\mb e_n$ and the projection map $\mb q\paren{\mb w} = [\mb w; (1-\norm{\mb w}{}^2)^{1/2}]$, with the idea that similar statements hold for other symmetric sections. 

\begin{lemma} \label{lem:alg_gradient_func}
Suppose that the trust region size $\Delta \le 1$, 
$\mb w^* \nabla g(\mb w)/\norm{\mb w}{} \ge \beta_g$ 
 for some scalar $\beta_g$, and that $\mb w^* \nabla g(\mb w)/\norm{\mb w}{}$ is $L_g$-Lipschitz on an open ball $\mc B\left(\mb w, \frac{3\Delta}{2\pi\sqrt{n}}\right)$ centered at $\mb w$. Then there exists a tangent vector $\mb \delta \in T_{\mb q} \bb S^{n-1}$ with $\norm{\mb \delta}{} \le \Delta$, such that
\begin{equation*}
f(\exp_{\mb q}(\mb \delta)) \;\le\; f(\mb q) -  \min \set{ \frac{\beta_g^2}{2 L_g}, \frac{3\beta_g \Delta}{4\pi\sqrt{n}} }. 
\end{equation*}
\end{lemma}
\begin{proof}
See Page~\pageref{proof:lem_alg_gradient_func} under Section~\ref{sec:proof_algorithm}. 
\end{proof}

\begin{lemma}  \label{lem:alg_neg_cuv_func}
Suppose that the trust-region size $\Delta \le 1$, 
$
\mb w^* \nabla^2 g(\mb w) \mb w/\norm{\mb w}{}^2 \le - \betaconcave,
$
for some $\betaconcave$, and that $\mb w^* \nabla^2 g(\mb w) \mb w/\norm{\mb w}{}^2$ is $\Lconcave$ Lipschitz on the open ball $\mc B\left(\mb w, \frac{3\Delta}{2\pi \sqrt{n}} \right)$ centered at $\mb w$. Then there exists a tangent vector $\mb \delta \in T_{\mb q} \bb S^{n-1}$ with $\norm{\mb \delta}{} \le \Delta$, such that
\begin{equation*}
f( \exp_{\mb q}(\mb \delta) ) \;\le\; f( \mb q) - \min \set{ \frac{2 \betaconcave^3}{3 \Lconcave^2}, \frac{3\Delta^2 \betaconcave}{8\pi^2 n} }. 
\end{equation*}
\end{lemma}
\begin{proof}
See Page~\pageref{proof:lem_alg_neg_cuv_func} under Section~\ref{sec:proof_algorithm}. 
\end{proof}
One can take $\beta_g = \betaconcave = c_\star \theta$ as shown in Theorem~\ref{thm:geometry_orth}, and take the Lipschitz results in Section~\ref{sec:geo_results_orth} (note that $\norm{\mb X_0}{\infty} \le 4 \log^{1/2} (np)$ w.h.p. by Lemma~\ref{lem:X-infinty-tail-bound}), repeat the argument for other $2n-1$ symmetric regions, and conclude that w.h.p. the objective value decreases by at least a constant amount. The next proposition summarizes the results. 

\begin{proposition}\label{lem:TRM-lemma-ii}
Assume~\eqref{eq:trm_proof_assumed_cond}.	In regions $\rII$ and $\rIII$, each trust-region step reduces the objective value by at least 
\begin{align} \label{eq:alg_thm_r23_dec}
\dII = \frac{1}{2}  \min\paren{\frac{c_\star^2 c_a \theta^2 \mu}{ n^2 \log\paren{np}}, \frac{3\Delta c_\star \theta}{4\pi\sqrt{n}}},  \quad \text{and} \quad \dIII = \frac{1}{2} \min\paren{\frac{c_\star^3c_b\theta^3 \mu^4}{n^6 \log^3\paren{np}}, \frac{3\Delta^2 c_\star \theta}{8\pi^2 n}}
\end{align}
respectively, provided that
\begin{align} \label{eqn:TRM-lemma-ii-iii-0}
	\Delta < \frac{c_c c_\star \theta \mu^2}{n^{5/2} \log^{3/2}\paren{np}},
\end{align}
where $c_a$ to $c_c$ are positive numerical constants, and $c_\star$ is as defined in Theorem~\ref{thm:geometry_orth}. 
\end{proposition}
\begin{proof}
We only consider the symmetric section in the vicinity of $\mb e_n$ and the claims carry on to others by symmetry. If the current iterate $\mb q^{(k)}$ is in the region $\rII$, by Theorem \ref{thm:geometry_orth}, w.h.p., we have $\mb w^* g\paren{\mb w}/\norm{\mb w}{} \ge c_\star \theta$ for the constant $c_\star$. By Proposition~\ref{prop:lip-gradient} and Lemma~\ref{lem:X-infinty-tail-bound}, w.h.p., $\mb w^* g\paren{\mb w}/\norm{\mb w}{}$ is $C_2n^2\log\paren{np}/\mu$-Lipschitz. Therefore, By Lemma~\ref{lem:alg_approx_bd2} and Lemma~\ref{lem:alg_gradient_func}, a trust-region step decreases the objective value by at least
    \begin{align*}
    \dII \doteq \min\paren{\frac{c_\star^2 \theta^2 \mu}{2C_2 n^2 \log\paren{np}}, \frac{3 c_\star \theta\Delta}{4\pi\sqrt{n}}} - \frac{c_0n^{3/2}\log^{3/2}\paren{np}}{3\mu^2} \Delta^3. 
\end{align*}
	Similarly, if $\mb q^{(k)}$ is in the region $\rIII$, by Proposition \ref{prop:lip-hessian-negative}, Theorem \ref{thm:geometry_orth} and Lemma~\ref{lem:X-infinty-tail-bound}, w.h.p., $\mb w^* \nabla^2 g\paren{\mb w} \mb w/\norm{\mb w}{}^2$ is $C_3n^3\log^{3/2}\paren{np}/\mu^2$-Lipschitz and upper bounded by $-c_\star \theta$. By Lemma~\ref{lem:alg_approx_bd2} and Lemma~\ref{lem:alg_neg_cuv_func}, a trust-region step decreases the objective value by at least
	\begin{align*}
\dIII \doteq \min\paren{\frac{2c_\star^3\theta^3 \mu^4}{3C_3^2n^6 \log^3\paren{np}}, \frac{3\Delta^2 c_\star \theta}{8\pi^2 n}} - \frac{c_0n^{3/2}\log^{3/2}\paren{np}}{3\mu^2} \Delta^3. 
\end{align*}
It can be easily verified that when $\Delta$ obeys~\eqref{eq:alg_thm_r23_dec}, \eqref{eqn:TRM-lemma-ii-iii-0} holds. 
\end{proof}

The analysis for $\rI$ is slightly trickier. In this region, near each local minimizer, the objective function is strongly convex. So we still expect each trust-region step decreases the objective value. On the other hand, it is very unlikely that we can provide a universal lower bound for the amount of decrease - as the iteration sequence approaches one local minimizer, the movement is expected to be diminishing. Nevertheless, close to the minimizer the trust-region algorithm takes ``unconstrainted'' steps. For constrained $\rI$ steps, we will again show reduction in objective value by at least a fixed amount; for unconstrained step, we will show the distance between the iterate and the nearest local minimizer drops down rapidly. 

The next lemma concerns the function value reduction for constrained $\rI$ steps. 
\begin{lemma}  \label{lem:alg_strcvx_func}
Suppose the trust-region size $\Delta \le 1$, and that at a given iterate $k$, $\Hess f\paren{\mb q^{(k)}} \succeq m_H \mc P_{T_{\mb q^{(k)}}\bb S^{n-1} }$, and $\norm{\Hess f\paren{\mb q^{(k)}}}{} \le M_H$. Further assume the optimal solution $\mb \delta_\star\in T_{\mb q^{(k)}}\bb S^{n-1}$ to the trust-region subproblem~\eqref{eqn:subproblem-1} satisfies $\norm{\mb \delta_\star}{} = \Delta$, i.e., the norm constraint is active. Then there exists a tangent vector $\mb \delta \in T_{\mb q^{(k)}} \bb S^{n-1}$ with $\norm{\mb \delta}{} \le \Delta$, such that
\begin{equation*}
f( \exp_{\mb q^{(k)}}(\mb \delta) ) \;\le\; f\paren{\mb q^{(k)}} - \frac{m_H^2 \Delta^2}{M_H} + \frac{1}{6} \eta_f\Delta^3,
\end{equation*}
where $\eta_f$ is defined the same as Lemma~\ref{lem:alg_approx_bd2}.
\end{lemma}
\begin{proof}
See Page~\pageref{proof:lem_alg_strcvx_func} under Section~\ref{sec:proof_algorithm}. 
\end{proof}
The next lemma provides an estimate of $m_H$. 
Again we will only state the result for the ``canonical'' section with the ``canonical'' $\mb q(\mb w)$ mapping.  
\begin{lemma} \label{lem:alg_strcvx_lb}
There exist positive constants $C$ and $c_\sharp$, such that for all $\theta \in \paren{0, 1/2}$ and $\mu < \theta/10$, whenever $p \ge Cn^3 \log \frac{n}{\theta \mu}/(\mu \theta^2)$, it holds with probability at least $1 -  \theta \paren{np}^{-7} - \exp\paren{-0.3\theta np} - p^{-10}$ that for all $\mb q$ with $\norm{\mb w\paren{\mb q}}{} \le \frac{\mu}{4\sqrt{2}}$, 
\begin{align*}
\Hess f\paren{\mb q} \succeq c_\sharp \frac{\theta}{\mu} \mc P_{T_{\mb q} \bb S^{n-1}}.  
\end{align*}
\end{lemma}
\begin{proof}
See Page~\pageref{proof:lem_alg_strcvx_lb} under Section~\ref{sec:proof_algorithm}. 
\end{proof}
We know that $\norm{\mb X_0}{\infty} \le 4\log^{1/2} (np)$ w.h.p., and hence by the definition of Riemannian Hessian and Lemma~\ref{lem:mag_lip_fq}, 
\begin{align*}
M_H \doteq \norm{\Hess f(\mb q)}{} & \le \norm{\nabla^2 f(\mb q)}{} + \norm{\nabla f(\mb q)}{} \le M_{\nabla^2} + M_{\nabla} \le \frac{2n}{\mu} \norm{\mb X_0}{\infty} ^2 \le \frac{16n}{\mu} \log (np), 
\end{align*}
Combining this estimate and Lemma~\ref{lem:alg_strcvx_lb}, and Lemma~\ref{lem:alg_approx_bd2}, we obtain a concrete lower bound for the reduction of objective value for each constrained $\rI$ step. 
\begin{proposition}\label{lem:TRM-lemma-iii}
Assume~\eqref{eq:trm_proof_assumed_cond}. Each constrained $\rI$ trust-region step (i.e., $\norm{\mb \delta}{} = \Delta$) reduces the objective value by at least 
\begin{align}\label{eqn:TRM-decrease-d-i}
	\dI = \frac{c c_\star^2\theta^2 }{\mu n\log (np) }\Delta^2,
\end{align}
provided 
\begin{align}\label{eqn:TRM-lemma-ii-0}
 \Delta \leq \frac{c'c_\sharp^2 \theta^2 \mu  }{n^{5/2} \log^{5/2} (np) }. 
\end{align}
The constant $c_\sharp$ is as defined in Lemma~\ref{lem:alg_strcvx_lb} and $c, c'$ are a positive numerical constants.
\end{proposition}
\begin{proof}
We only consider the symmetric section in the vicinity of $\mb e_n$ and the claims carry on to others by symmetry. We have that w.h.p. 
\begin{align*}
\norm{\Hess f(\mb q)}{} \le \frac{16n}{\mu} \log(np), \quad \text{and} \quad  \Hess f(\mb q) \succeq c_\sharp \frac{\theta}{\mu} \mc P_{T_{\mb q} \bb S^{n-1}}, 
\end{align*}
where $c_\sharp$ is as defined in Lemma~\ref{lem:alg_strcvx_lb}. Combining these estimates with Lemma~\ref{lem:alg_approx_bd2} and Lemma~\ref{lem:alg_strcvx_func}, one trust-region step will find next iterate $\mb q^{(k+1)}$ that decreases the objective value by at least 
\begin{align*}
\dI \doteq \frac{c_\sharp^2 \theta^2/\mu^2 }{2n\log\paren{np}/\mu} \Delta^2 - \frac{c_0n^{3/2} \log^{3/2}\paren{np}}{\mu^2} \Delta^3.
\end{align*}
Finally, by the condition on $\Delta$ in \eqref{eqn:TRM-lemma-ii-0} and the assumed conditions~\eqref{eq:trm_proof_assumed_cond}, we obtain 
\begin{align*}
	\dI \geq \frac{c_\sharp^2\theta^2 }{2\mu n\log (np) }\Delta^2 - \frac{c_0n^{3/2} \log^{3/2}\paren{np}}{\mu^2} \Delta^3 \geq \frac{c_\sharp^2\theta^2 }{4\mu n\log (np) }\Delta^2, 
\end{align*}
as desired.
\end{proof}

By the proof strategy for $\rI$ we sketched before Lemma~\ref{lem:alg_strcvx_func}, we expect the iteration sequence ultimately always takes unconstrained steps when it moves very near to a local minimizer. We will show that the following is true: when $\Delta$ is small enough, once the iteration sequence starts to take unconstrained $\rI$ step, it will take consecutive unconstrained $\rI$ steps afterwards. It takes two steps to show this: (1) upon an unconstrained $\rI$ step, the next iterate will stay in $\rI$. It is obvious we can make $\Delta \in O(1)$ to ensure the next iterate stays in $\rI \cup \rII$. To strengthen the result, we use the gradient information. From Theorem~\ref{thm:geometry_orth}, we expect the magnitudes of the gradients in $\rII$ to be lower bounded; on the other hand, in $\rI$ where points are near local minimizers, continuity argument implies that the magnitudes of gradients should be upper bounded. We will show that when $\Delta$ is small enough, there is a gap between these two bounds, implying the next iterate stays in $\rI$; (2) when $\Delta$ is small enough, the step is in fact unconstrained. Again we will only state the result for the ``canonical'' section with the ``canonical'' $\mb q(\mb w)$ mapping. The next lemma exhibits an absolute lower bound for magnitudes of gradients in $\rII$. 
\begin{lemma} \label{lem:alg_gradient_lb}
For all $\mb q$ satisfying $\frac{\mu}{4\sqrt{2}} \le \norm{\mb w\paren{\mb q}}{} \le \frac{1}{20\sqrt{5}}$, it holds that 
\begin{align*}
\norm{\grad f\paren{\mb q}}{} \ge \frac{9}{10} \frac{\mb w^* \nabla g\paren{\mb w}}{\norm{\mb w}{}}. 
\end{align*}
\end{lemma}
\begin{proof}
See Page~\pageref{proof:lem_alg_gradient_lb} under Section~\ref{sec:proof_algorithm}. 
\end{proof}
Assuming~\eqref{eq:trm_proof_assumed_cond}, Theorem~\ref{thm:geometry_orth} gives that w.h.p. $\mb w^* \nabla g(\mb w)/\norm{\mb w}{} \ge c_\star \theta$. Thus, w.h.p, $\norm{\grad f(\mb q)}{} \ge 9c_\star\theta/10$ for all $\mb q \in \rII$. The next lemma compares the magnitudes of gradients before and after taking one unconstrained $\rI$ step. This is crucial to providing upper bound for magnitude of gradient for the next iterate, and also to establishing the ultimate (quadratic) sequence convergence.
\begin{lemma}\label{lem:TR-step} 
Suppose the trust-region size $\Delta \le 1$, and at a given iterate $k$, $\Hess f \paren{\mb q^{(k)}} \succeq m_H \mc P_{ T_{\mb q^{(k)}}\bb S^{n-1}}$, and that the unique minimizer $\mb \delta_\star\in T_{\mb q^{(k)}}\bb S^{n-1} $ to the trust region subproblem \eqref{eqn:subproblem-1} satisfies $\norm{\mb \delta_\star}{} < \Delta$ (i.e., the constraint is inactive). Then, for $\mb q^{(k+1)} = \exp_{\mb q^{(k)}}\paren{\mb \delta_\star} $, we have 
\begin{equation*}
\|\grad f (\mb q^{(k+1)})\| \;\le\; \frac{L_H}{2m_H^2} \|\grad f(\mb q^{(k)}) \|^2, 
\end{equation*}
where $L_H \doteq \frac{5}{2\mu^2}n^{3/2} \norm{\mb X_0}{\infty}^3 + \frac{9}{\mu}n \norm{\mb X_0}{\infty}^2 + 9\sqrt{n} \norm{\mb X_0}{\infty}$.
\end{lemma}
\begin{proof}
See Page~\pageref{proof:lem_TR-step} under Section~\ref{sec:proof_algorithm}. 
\end{proof}
We can now bound the Riemannian gradient of the next iterate as 
\begin{align*}
\|\grad f (\mb q^{(k+1)})\| 
& \le \frac{L_H}{2m_H^2} \|\grad f(\mb q^{(k)}) \|^2 \\
& \le \frac{L_H}{2m_H^2} \|[\mb U^* \Hess f(\mb q^{(k)}) \mb U] [\mb U^* \Hess f(\mb q^{(k)}) \mb U]^{-1} \grad f(\mb q^{(k)}) \|^2 \\
& \le  \frac{L_H}{2m_H^2}  \norm{\Hess f(\mb q^{(k)})}{}^2 \Delta^2 = \frac{L_H M_H^2}{2m_H^2} \Delta^2. 
\end{align*}
Obviously, one can make the upper bound small by tuning down $\Delta$. Combining the above lower bound for $\norm{\grad f(\mb q)}{}$ for $\mb q \in \rII$, one can conclude that when $\Delta$ is small, the next iterate $\mb q^{(k+1)}$ stays in $\rI$. Another application of the optimality condition~\eqref{eqn:ts-optimal-solution-1} gives conditions on $\Delta$ that guarantees the next trust-region step is also unconstrained. Detailed argument can be found in proof of the following proposition. 
\begin{proposition}\label{lem:TRM-lemma-iv}
Assume~\eqref{eq:trm_proof_assumed_cond}. W.h.p, once the trust-region algorithm takes an unconstrained $\rI$ step (i.e., $\norm{\mb \delta}{}<\Delta$), it always takes unconstrained $\rI$ steps, provided that
\begin{align}
 \Delta \le  \frac{c c_\sharp^3 \theta^3 \mu}{n^{7/2} \log^{7/2}\paren{np}}, 
\end{align} 	
Here $c$ is a positive numerical constant, and $c_\sharp$ is as defined in Lemma~\ref{lem:alg_strcvx_lb}.
\end{proposition}
\begin{proof}
We only consider the symmetric section in the vicinity of $\mb e_n$ and the claims carry on to others by symmetry. Suppose that step $k$ is an unconstrained $\rI$ step. Then 
\begin{align*}
\|\mb w(\mb q^{(k+1)}) - \mb w(\mb q^{(k)})\| 
& \le \|\mb q^{(k+1)} - \mb q^{(k)}\|  = \|\exp_{\mb q^{(k)}(\mb \delta)} - \mb q^{(k)} \| \\ 
& = \sqrt{2-2\cos \norm{\mb \delta}{}} = 2\sin(\norm{\mb \delta}{}/2) \le \norm{\mb \delta}{} < \Delta. 
\end{align*}
Thus, if $\Delta \le \tfrac{1}{20\sqrt{5}} - \tfrac{\mu}{4\sqrt{2}}$, $\mb q^{(k+1)}$ will be in $\rI \cup \rII$. Next, we show that if $\Delta$ is sufficiently small, $\mb q^{(k+1)}$ will be indeed in $\rI$. By Lemma \ref{lem:TR-step}, 
\begin{align}
\norm{\grad f \paren{\mb q^{(k+1)} } }{} 
& \le \frac{L_H}{2m_H^2} \norm{ \grad f \paren{\mb q^{(k)}} }{}^2 \nonumber \\
& \le \frac{L_H M_H^2}{2m_H^2} \norm{ \brac{\mb U^* \Hess f \paren{\mb q^{(k)}} \mb U}^{-1} \mb U^*\grad f \paren{\mb q^{(k)}} }{}^2 
 \le \frac{L_H M_H^2}{2m_H^2} \Delta^2,  \label{eqn:TRM-lemma-iii-3}
 \end{align}
where we have used the fact that 
\begin{align*}
\norm{\mb \delta^{(k)}}{} = \norm{ \brac{\mb U^* \Hess f \paren{\mb q^{(k)}} \mb U}^{-1} \mb U^*\grad f \paren{\mb q^{(k)}} }{} < \Delta,  
\end{align*}
as the step is unconstrained. On the other hand, by Theorem~\ref{thm:geometry_orth} and Lemma~\ref{lem:alg_gradient_lb}, w.h.p. 
\begin{align}
 \norm{\grad f\paren{\mb q}}{} \ge \beta_{\grad} \doteq \frac{9}{10} c_\star \theta,  \quad \forall \; \mb q \in \rII. \label{eqn:TRM-lemma-iii-4}
\end{align}
Hence, provided 
\begin{align}\label{eqn:Delta-bound-1}
	\Delta < \frac{m_H}{M_H}\sqrt{ \frac{2\betagrad}{L_H}},
\end{align}
we have $\mb q^{(k+1)} \in \rI$. 

We next show that when $\Delta$ is small enough, the next step is also unconstrained. Straight forward calculations give 
\begin{align*}
\norm{\mb U \brac{\mb U^* \Hess f \paren{\mb q^{(k+1)}} \mb U}^{-1} \mb U^*\grad f \paren{\mb q^{(k+1)}} }{} \leq \frac{L_H M_H^2}{2m_H^3} \Delta^2. 
\end{align*}
Hence, provided that 
 \begin{equation} \label{eqn:Delta-bound-2}
 \Delta < \frac{2m_H^3}{L_H M_H^2},
 \end{equation}
we will have
\begin{align*}
\norm{\mb U \brac{\mb U^* \Hess f \paren{\mb q^{(k+1)}} \mb U}^{-1} \mb U^*\grad f \paren{\mb q^{(k+1)}} }{} < \Delta; 
\end{align*}
in words, the minimizer to the trust-region subproblem for the next step lies in the relative interior of the trust region - the constraint is inactive. By Lemma~\ref{lem:TR-step} and Lemma \ref{lem:X-infinty-tail-bound}, we have 
\begin{align}
L_H \;=\; C_1 n^{3/2} \log^{3/2}\paren{np}/\mu^2, \label{eqn:lip-value-bound}
\end{align}
w.h.p. for some numerical constant $C_1$. Combining this and our previous estimates of $m_H$, $M_H$, we conclude whenever
\begin{align*}
\Delta 
\leq  \min \set{\frac{1}{20\sqrt{5}} - \frac{\mu}{4\sqrt{2}}, \frac{c_1\mu c_\sharp c_\star^{1/2} \theta^{3/2} }{n^{7/4} \log^{7/4}\paren{np}}, \frac{c_2\mu c_\sharp^3 \theta^3 }{n^{7/2} \log^{7/2}\paren{np}}}. 
\end{align*}
for some positive numerical constants $c_1$ and $c_2$, w.h.p. our next trust-region step is also an unconstrained $\rI$ step. Noting that $c_\star$ and $c_\sharp$ can be made the same by our definition, we make the claimed simplification on $\Delta$. This completes the proof. 
\end{proof}

Finally, we want to show that ultimate unconstrained $\rI$ iterates actually converges to one nearby local minimizer rapidly. Lemma~\ref{lem:TR-step} has established the gradient is diminishing. The next lemma shows the magnitude of gradient serves as a good proxy for distance to the local minimizer. 
\begin{lemma}\label{lem:TR-grad-opt} 
Let $\mb q_\star \in \bb S^{n-1}$ such that $\grad f (\mb q_\star) = \mb 0$, and $\mb \delta \in T_{\mb q_\star}\bb S^{n-1}$. Consider a geodesic $\gamma(t) = \exp_{\mb q_\star}(t \mb \delta)$, and suppose that on $[0, \tau]$, $\Hess f (\gamma(t)) \succeq m_H \mc P_{T_{\gamma(t)} \bb S^{n-1}}$. Then 
\begin{equation*}
\norm{ \grad f (\gamma(\tau)) }{} \;\ge\; m_H \tau \norm{\mb \delta}{}. 
\end{equation*}
\end{lemma}
\begin{proof}
See Page~\pageref{proof:lem_TR-grad-opt} under Section~\ref{sec:proof_algorithm}. 
\end{proof}
To see this relates the magnitude of gradient to the distance away from the critical point, w.l.o.g., one can assume $\tau = 1$ and consider the point $\mb q = \exp_{\mb q_\star}(\mb \delta)$. Then 
\begin{align*}
\norm{\mb q_\star - \mb q}{} = \norm{\exp_{\mb q_\star}(\mb \delta) - \mb q}{} = \sqrt{2 - 2\cos \norm{\mb \delta}{}} = 2\sin (\norm{\mb \delta}{}/2) \le \norm{\mb \delta}{} \le \norm{\grad f(\mb q)}{}/m_H, 
\end{align*}
where at the last inequality above we have used Lemma~\ref{lem:TR-grad-opt}. Hence, combining this observation with Lemma~\ref{lem:TR-step}, we can derive the asymptotic sequence convergence result as follows. 
\begin{proposition}\label{lem:TRM-lemma-v}
Assume~\eqref{eq:trm_proof_assumed_cond} and the conditions in Lemma \ref{lem:TRM-lemma-iv}. Let $\mb q^{(k_0)} \in \rI$ and the $k_0$-th step the first unconstrained $\rI$ step and $\mb q_\star$ be the unique local minimizer of $f$ over one connected component of $\rI$ that contains $\mb q^{(k_0)}$. Then w.h.p., for any positive integer $k'\geq 1$, 
\begin{align}
	\norm{\mb q^{(k_0 + k')} - \mb q_\star }{} \; \leq\;  \frac{cc_\sharp \theta \mu}{n^{3/2} \log^{3/2} \paren{np}} 2^{- 2^{k'}},
\end{align}
provided that 
\begin{align}
	\Delta \leq \frac{c'c_\sharp^2 \theta^2 \mu}{n^{5/2} \log^{5/2} (np) }. 
\end{align}
Here $c_\sharp$ is as defined in Lemma~\ref{lem:alg_strcvx_lb} that can be made equal to $c_s\star$ as defined in Theorem~\ref{thm:geometry_orth}, and $c$, $c'$ are positive numerical constants. 
\end{proposition}
\begin{proof}
By the geometric characterization in Theorem~\ref{thm:geometry_orth} and corollary~\ref{cor:geometry_orth}, $f$ has $2n$ separated local minimizers, each located in $\rI$ and within distance $\sqrt{2}\mu/16$ of one of the $2n$ signed basis vectors $\{\pm \mb e_i\}_{i \in [n]}$. Moreover, it is obvious when $\mu \le 1$, $\rI$ consists of $2n$ disjoint connected components. We only consider the symmetric component in the vicinity of $\mb e_n$ and the claims carry on to others by symmetry. 

Suppose that $k_0$ is the index of the first unconstrained iterate in region $\rI$, i.e., $\mb q^{(k_0)} \in \rI$. By Lemma \ref{lem:TR-step}, for any integer $k'\geq 1$, we have
	 \begin{align}\label{eqn:TRM-lemma-v-1}
	 	\norm{\grad f \paren{\mb q^{(k_0 + k')} } }{} \;\leq\;  \frac{2 m_H^2}{L_H}\left( \frac{L_H}{2 m_H^2}  \norm{\grad f \paren{\mb q^{(k_0)}} }{} \right)^{2^{k'}}.
	 \end{align}
	 where $L_H$ is as defined in Lemma~\ref{lem:TR-step}, $m_H$ as the strong convexity parameter for $\rI$ defined above.   
	 
	 Now suppose $\mb q_\star$ is the unique local minimizer of $f$, lies in the same $\rI$ component that $q^{(k_0)}$ is located. Let $\gamma_{k'}(t) = \exp_{\mb q_\star}\paren{t\mb \delta}$ to be the unique geodesic that connects $\mb q_\star$ and $\mb q^{(k_0+k')}$ with $\gamma_{k'}(0) = \mb q_\star$ and $\gamma_{k'}(1) = \mb q^{(k_0+k')}$. We have 
	 \begin{align*}
	 \norm{\mb q^{(k_0 + k')} - \mb q_\star}{} 
	 & \le \norm{\exp_{\mb q_\star}(\mb \delta) - \mb q_\star}{} = \sqrt{2- 2\cos \norm{\mb \delta}{}} = 2 \sin(\norm{\mb \delta}{}/2) \\
	 & \le \norm{\mb \delta}{} \le \frac{ 1 }{ m_H }\norm{ \grad f \paren{\mb q^{(k_0 + k')}} }{}
 \leq  \frac{2m_H}{L_H} \left( \frac{L_H}{2m_H^2} \norm{\grad f \paren{\mb q^{(k_0)}} }{} \right)^{2^{k'}}, 
	 \end{align*}
where at the second line we have repeatedly applied Lemma \ref{lem:TR-grad-opt}. 

By the optimality condition~\eqref{eqn:ts-optimal-solution-1} and the fact that $\norm{\mb \delta^{(k_0)}}{} < \Delta$, we have
\begin{align*}
\frac{L_H}{2 m_H^2} \norm{\grad f \paren{\mb q^{(k_0)} } }{} 
& \le \frac{L_H}{2 m_H^2} M_H \norm{ \brac{\mb U^* \Hess f \paren{\mb q^{(k_0)}} \mb U}^{-1} \mb U^* \grad f \paren{\mb q^{(k_0)} } }{} 
\le \frac{L_H M_H}{2m_H^2} \Delta.
\end{align*}
Thus, provided
\begin{align}\label{eqn:trust-region-size-v}
	\Delta < \frac{m_H^2}{L_H M_H},
\end{align} 
we can combine the above results and obtain 
\begin{align*}
	\norm{\mb q^{(k_0 + k')} - \mb q_\star }{} \;\le\; \frac{2m_H}{L_H} 2^{- 2^{k'}}.
\end{align*}
Based on the previous estimates for $m_H$, $M_H$ and $L_H$, we obtain that w.h.p., 
\begin{align*}
	\norm{\mb q^{(k_0 + k')} - \mb q_\star }{} \; \leq\;  \frac{c_1 c_\sharp \theta \mu}{n^{3/2} \log^{3/2} \paren{np}} 2^{- 2^{k'}}.
\end{align*}
Moreover, by \eqref{eqn:trust-region-size-v}, w.h.p., it is sufficient to have the trust region size
\begin{align*}
	\Delta \leq \frac{c_2 c_\sharp^2 \theta^2 \mu }{n^{5/2} \log^{5/2} (np) }.
\end{align*}
Thus, we complete the proof.
\end{proof}

Now we are ready to piece together the above technical proposition to prove Theorem~\ref{thm:trm_orth}. 

\begin{proof}[of Theorem \ref{thm:trm_orth}]
Assuming~\eqref{eq:trm_proof_assumed_cond} and in addition that 
	\begin{align*}
	 \Delta < \min\set{\frac{c_1c_\star  \theta \mu^2}{n^{5/2} \log^{3/2}\paren{np}}, \frac{c_2 c_\sharp^3 \theta^3 \mu}{n^{7/2}\log^{7/2}\paren{np}}}
	\end{align*}
for small enough numerical constants $c_1$ and $c_2$ and $c_\star$, $c_\sharp$ as defined in Theorem~\ref{thm:geometry_orth} and Lemma~\ref{lem:alg_strcvx_lb} respectively ($c_\star$ and $c_\sharp$ can be set to the same constant value), it can be verified that the conditions of all the above propositions are satisfied. Since each of the local minimizers is contained in the relative interior of one connected component of $\rI$ (comparing distance of local minimizers to their respective signed basis vector, as stated in Corollary~\ref{cor:geometry_orth}, with size of each connected $\rI$ component yields this ), we can define a threshold value 
\begin{align*}
\zeta \doteq \min \set{\min_{\mb q \;\in\; \overline{\rII \cup \rIII}} f\paren{\mb q}, \max_{\mb q \;\in\; \overline{\rI}} f\paren{\mb q}} 
\end{align*}
where overline $\overline{\cdot}$ here denotes set closure. Obviously $\zeta$ is well-defined as the function $f$ is continuous, and both sets $\overline{\rII \cup \rIII}$ and $\overline{\rI}$ are compact. Also for any of the local minimizers, say $\mb q_\star$, it holds that $\zeta > f(\mb q_\star)$. 

By the four propositions above, a step will either be $\rIII$, $\rII$, or constrained $\rI$ step that decreases the objective value by at least a certain fixed amount (we call this \emph{Type A}), or be an unconstrained $\rI$ step (\emph{Type B}), such that all future steps are unconstrained $\rI$ and the sequence converges to one local minimizer quadratically. Hence, regardless the initialization, the whole iteration sequence consists of consecutive Type A steps, followed by consecutive Type B steps. Depending on the initialization, either the Type A phase or the Type B phase can be absent. In any case, \emph{in a finite number of steps, the function value must drops below $\zeta$ and all future iterates stay in $\rI$}. Indeed, if the function value never drops below $\zeta$, by continuity the whole sequence must be of entirely Type A - whereby either the finite-length sequence converges to one local minimizer, or every iterate of the infinite sequence steadily decreases the objective value by at least a \emph{fixed} amount - in either case, the objective value should ever drop below $\zeta$ in finitely many steps; hence contradiction arises. Once the function value drops below $\zeta$, type A future steps decreases the objective value further down below $\zeta$ - by definition of $\zeta$, these iterates stay within $\rI$, and type B future steps, aka unconstrained $\rI$ steps obviously keep all subsequent iterates in $\rI$.   

There are three possibilities after the objective value drop below $\zeta$ and all future iterates stay in $\rI$. Assume $\mb q_\star$ is the unique local minimizer in the same connected component of $\rI$ as the current iterate: (1) the sequence always take constrained $\rI$ steps and hits $\mb q_\star$ exactly in finitely many steps;  (2) the sequence takes constrained $\rI$ steps until reaching certain point $\mb q' \in \rI$ such that $f(\mb q') < f(\mb q_\star) + \dI$, where $\dI$ is as defined in Proposition~\ref{lem:TRM-lemma-iii}. Since each constrained $\rI$ step must decrease the objective value by at least $\dI$, the next and all future steps must be unconstrained $\rI$ steps and the sequence converges to $\mb q_\star$;  (3) the sequence starts to take unconstrained $\rI$ steps at a certain point $\mb q'' \in \rI$ such that $f(\mb q'') \ge f(\mb q_\star) + \dI$. In any case, the sequence converges to the local minimizer $\mb q_\star$. By Proposition~ \ref{lem:TRM-lemma-ii}, Proposition~\ref{lem:TRM-lemma-iii}, and Proposition~ \ref{lem:TRM-lemma-v}, the number of iterations to obtain an $\eps$-near solution to $\mb q_\star$ can be grossly bounded by
\begin{align*}
\#\text{Iter} &\;\leq\;  \frac{f\paren{\mb q^{(0)}} - f\paren{\mb q_\star}}{\min\Brac{\dI, \dII, \dIII }} \;+\;\log\log \paren{\frac{c_5c_\sharp \theta \mu}{\eps n^{3/2} \log^{3/2}\paren{np}}} \nonumber\\
 &\;\leq\;  \brac{\min\set{\frac{c_3c_{\star}^3 \theta^3 \mu^4}{n^6 \log^3\paren{np}}, \frac{c_4c_\sharp^2 \theta^2 }{n}\Delta^2 }}^{-1}\paren{f\paren{\mb q^{(0)}} - f\paren{\mb q_\star}} \;+\; \log\log \paren{\frac{c_5c_\sharp \theta \mu}{\eps n^{3/2} \log^{3/2}\paren{np}}}, 
\end{align*}
where we have assumed $p \le \exp(n)$ when comparing the various bounds. Finally, the claimed failure probability comes from a simple union bound with careful bookkeeping. 
\end{proof}

\subsection{Extending to Convergence for Complete Dictionaries} \label{sec:alg_comp}
\input{sec/proof_trm_comp}

%% file: sec/proof_trm_comp.tex
Note that for any complete $\mb A_0$ with condition number $\kappa\paren{\mb A_0}$, from Lemma~\ref{lem:pert_key_mag} we know when $p$ is large enough, w.h.p. one can write the preconditioned $\ol{\mb Y}$ as 
\begin{align*}
\ol{\mb Y} = \mb U \mb V^* \mb X_0 + \mb \Xi \mb X_0
\end{align*}
for a certain $\mb \Xi$ with small magnitude, and $\mb U \mb \Sigma \mb V^* = \mathtt{SVD}\paren{\mb A_0}$. Since $\mb U \mb V^*$ is orthogonal, 
\begin{align*}
f\paren{\mb q; \mb U \mb V^* \mb X_0 + \mb \Xi \mb X_0} = f\paren{\mb V \mb U^* \mb q; \mb X_0 + \mb V \mb U^* \mb \Xi \mb X_0}. 
\end{align*}
In words, the function landscape of $f(\mb q; \mb U \mb V^* \mb X_0 + \mb \Xi \mb X_0)$ is a rotated version of that of $f(\mb q; \mb X_0 + \mb V \mb U^* \mb \Xi \mb X_0)$. Thus, any local minimizer $\mb q_\star$ of $f(\mb q; \mb X_0 + \mb V \mb U^* \mb \Xi \mb X_0)$ is rotated to $\mb U \mb V^* \mb q_\star$, one minimizer of $f(\mb q; \mb U \mb V^* \mb X_0 + \mb \Xi \mb X_0)$. Also if our algorithm generates iteration sequence $\mb q_0, \mb q_1, \mb q_2, \dots$ for $f(\mb q; \mb X_0 + \mb V \mb U^* \mb \Xi \mb X_0)$ upon initialization $\mb q_0$, it will generate the iteration sequence $\mb U \mb V^* \mb q_0$, $\mb U \mb V^* \mb q_1$, $\mb U \mb V^* \mb q_2, \dots$ for $f\paren{\mb q;  \mb U \mb V^* \mb X_0 + \mb \Xi \mb X_0}$. So w.l.o.g. it is adequate that we prove the convergence results for the case $f(\mb q; \mb X_0 + \mb V \mb U^* \mb \Xi \mb X_0)$, corresponding to $\bm A_0 = \mb I$ with perturbation $\wt{\mb \Xi} \doteq \mb V \mb U^* \mb \Xi$. So in this section (Section~\ref{sec:alg_comp}), we write $f(\mb q; \wt{\mb X_0})$ to mean $f(\mb q; \mb X_0 + \wt{\mb \Xi} \mb X_0)$. 

Theorem~\ref{thm:geometry_comp} has shown that when 
\begin{align} \label{eq:trm_proof_assumed_cond_comp}
\theta \in \paren{0, \frac{1}{2}}, \; \mu \le \min\set{\frac{c_a\theta}{n}, \frac{c_b}{n^{5/4}}}, \; p \ge \frac{C}{c_\star^2 \theta} \max\set{\frac{n^4}{\mu^4}, \frac{n^5}{\mu^2}} \kappa^8 \paren{\mb A_0} \log^4\paren{\frac{\kappa\paren{\mb A_0}n}{\mu \theta}}, 
\end{align}
the geometric structure of the landscape is qualitatively unchanged and the $c_\star$ constant can be replaced with $c_\star/2$. Particularly, for this choice of $p$, Lemma~\ref{lem:pert_key_mag} implies 
\begin{align} \label{eq:trm_conv_comp_pert_bound}
\|\wt{\mb \Xi}\| = \| \mb V \mb U^*\mb \Xi\| \le \norm{\widetilde{\mb \Xi}}{} \le c c_\star \theta \paren{\max\set{\frac{n^{3/2}}{\mu^2}, \frac{n^2}{\mu}} \log^{3/2}\paren{np}}^{-1} 
\end{align}
for a constant $c$ that can be made arbitrarily small by setting the constant $C$ in $p$ sufficiently large. The whole proof is quite similar to that of orthogonal case in the last section. We will only sketch the major changes below. To distinguish with the corresponding quantities in the last section, we use $\wt{\cdot}$ to denote the corresponding perturbed quantities here. 
\begin{itemize}
\item Lemma~\ref{lem:mag_lip_fq}: Note that  
\begin{align*}
\|\mb X_0 + \wt{\mb \Xi} \mb X_0\|_{\infty} \le \norm{\mb X_0}{\infty} + \|\wt{\mb \Xi}\mb X_0\|_{\infty} \le \|\mb X_0\|_\infty + \sqrt{n} \|\wt{\mb \Xi}\| \|\mb X_0\|_\infty \le 3\|\mb X_0\|_\infty/2, 
\end{align*}
where by~\eqref{eq:trm_conv_comp_pert_bound} we have used $\|\wt{\mb \Xi}\| \le 1/(2\sqrt{n})$ to simplify the above result. So we obtain  
\begin{align*}
\wt{M}_{\nabla} \le \frac{3}{2} M_{\nabla}, \; \wt{M}_{\nabla^2} \le \frac{9}{4} M_{\nabla^2}, \; \wt{L}_{\nabla} \le \frac{9}{4} L_{\nabla}, \; \wt{L}_{\nabla^2} \le \frac{27}{8} L_{\nabla^2}. 
\end{align*}

\item Lemma~\ref{lem:alg_approx_bd2}: Now we have
\begin{align*}
\wt{\eta}_f \doteq \wt{M}_{\nabla} + 2\wt{M}_{\nabla^2} + \wt{L}_{\nabla} + \wt{L}_{\nabla^2} \le 4 \eta_f. 
\end{align*}

\item Lemma~\ref{lem:alg_gradient_func} and Lemma~\ref{lem:alg_neg_cuv_func} are generic and nothing changes. 

\item Proposition~\ref{lem:TRM-lemma-ii}: We have now $\mb w^* \mb g(\mb w)/\norm{\mb w}{} \ge c_\star \theta/2$ by Theorem~\ref{thm:geometry_comp} and w.h.p. $\mb w^* \nabla g(\mb w)/\norm{\mb w}{}$ is $C_1n^2 \log(np)/\mu$-Lipschitz by Proposition~\ref{prop:lip-gradient} and the fact $\norm{\mb X_0 + \wt{\mb \Xi} \mb X_0}{\infty} \le 3\norm{\mb X_0}{\infty}/2$ shown above. Similarly, $\mb w^* \mb g(\mb w)/\norm{\mb w}{} \le -c_\star \theta/2$ by Theorem~\ref{thm:geometry_comp} and $\mb w^* \nabla^2 g(\mb w) \mb w/\norm{\mb w}{}^2$ is $C_2 n^3 \log^{3/2}(np) /\mu^2$-Lipschitz. Moreover, $\wt{\eta}_f \le 4\eta_f$ as shown above. Since there are only multiplicative constant changes to the various quantities, we conclude 
\begin{align}
\wt{\dII} = c_1 \dII, \quad \wt{\dIII} = c_1 \dIII
\end{align}
provided  
\begin{align}
\Delta < \frac{c_2 c_\star \theta \mu^2}{n^{5/2} \log^{3/2}\paren{np}}. 
\end{align}

\item Lemma~\ref{lem:alg_strcvx_func}: $\eta_f$ is changed to $\wt{\eta}_f$ with $\wt{\eta}_f \le 4\eta_f$ as shown above. 

\item Lemma~\ref{lem:alg_strcvx_lb}: By~\eqref{eq:fq_hess}, we have 
\begin{multline*}
\norm{\nabla^2 f(\mb q; \mb X_0) - \nabla^2 f(\mb q; \wt{\mb X_0}) }{} \le \frac{1}{p} \sum_{k=1}^p \Brac{L_{\ddot{h}} \|\wt{\mb \Xi}\| \norm{\mb x_k}{}^2 + \frac{1}{\mu}\norm{\mb x_k \mb x_k^* - \wt{\mb x}_k \wt{\mb x}_k^*}{}} \\
 \le \|\wt{\mb \Xi}\| \paren{L_{\ddot{h}} + 2/\mu + \|\wt{\mb \Xi}\|/\mu } \sum_{k=1}^p \norm{\mb x_k}{}^2 \le \|\wt{\mb \Xi}\| \paren{L_{\ddot{h}} + 3/\mu }n\norm{\mb X_0}{\infty}^2, 
\end{multline*}
where $L_{\ddot{h}}$ is the Lipschitz constant for the function $\ddot{h}_{\mu}\paren{\cdot}$ and we have used the fact that $\|\wt{\mb \Xi}\| \le 1$. Similarly, by~\ref{eq:fq_grad}, 
\begin{align*}
\norm{\nabla f(\mb q; \mb X_0) - \nabla f(\mb q; \wt{\mb X_0}) }{}
\le \frac{1}{p}\sum_{k=1}^p \Brac{L_{\dot{h}_\mu} \|\wt{\mb \Xi}\| \norm{\mb x_k}{} + \|\wt{\mb \Xi}\| \norm{\mb x_k}{}  } \le \paren{L_{\dot{h}_\mu} +1} \|\wt{\mb \Xi}\| \sqrt{n} \norm{\mb X_0}{\infty}, 
\end{align*}
where $L_{\dot{h}}$ is the Lipschitz constant for the function $\dot{h}_{\mu}\paren{\cdot}$. Since $L_{\ddot{h}} \le 2/\mu^2$ and $L_{\dot{h}} \le 1/\mu$, and $\norm{\mb X_0}{\infty} \le 4\sqrt{\log(np)}$ w.h.p. (Lemma~\ref{lem:X-infinty-tail-bound}). By~\eqref{eq:trm_conv_comp_pert_bound}, w.h.p. we have 
\begin{align*}
\norm{\nabla f(\mb q; \mb X_0) - \nabla f(\mb q; \wt{\mb X_0}) }{} \le \frac{1}{2} c_\sharp \theta, \quad\text{and}\quad  \norm{\nabla^2 f(\mb q; \mb X_0) - \nabla^2 f(\mb q; \wt{\mb X_0}) }{}  \le \frac{1}{2} c_\sharp \theta, 
\end{align*}
provided the constant $C$ in~\eqref{eq:trm_proof_assumed_cond_comp} for $p$ is large enough. Thus, by~\eqref{eq:fq_rie_hess} and the above estimates we have 
\begin{align*}
\norm{\Hess f(\mb q; \mb X_0) - \Hess f(\mb q; \wt{\mb X}_0)}{} & \le \norm{\nabla f(\mb q; \mb X_0) - \nabla f(\mb q; \wt{\mb X_0})}{} + \norm{\nabla^2 f(\mb q; \mb X_0) - \nabla^2 f(\mb q; \wt{\mb X_0})}{} \\
& \le c_\sharp \theta \le \frac{1}{2} c_\sharp \frac{\theta}{\mu}, 
\end{align*}
provided $\mu \le 1/2$. So we conclude 
\begin{align}
\Hess f(\mb q; \wt{\mb X}_0) \succeq \frac{1}{2}c_\sharp \frac{\theta}{\mu} \mc P_{T_{\mb q}} \bb S^{n-1} \Longrightarrow \wt{m_H} \ge \frac{1}{2}c_\sharp \frac{\theta}{\mu}. 
\end{align}

\item Proposition~\ref{lem:TRM-lemma-iii}: From the estimate of $M_H$ above Proposition~\ref{lem:TRM-lemma-iii} and the last point, we have 
\begin{align*}
\norm{\Hess f(\mb q; \wt{\mb X}_0)}{} \le \frac{36}{\mu} \log(np), \quad \text{and} \quad  \Hess f(\mb q; \wt{\mb X}_0) \succeq \frac{1}{2}c_\sharp \frac{\theta}{\mu} \mc P_{T_{\mb q}} \bb S^{n-1}. 
\end{align*}
Also since $\wt{\eta}_f \le 4\eta_f$ in Lemma~\ref{lem:alg_approx_bd2} and Lemma~\ref{lem:alg_strcvx_func}, there are only multiplicative constant change to the various quantities. We conclude that 
\begin{align}
\wt{\dI} = c_3 \dI
\end{align}
provided that 
\begin{align}
\Delta \le \frac{c_4 c_\sharp^2 \theta^2 \mu  }{n^{5/2} \log^{5/2} (np) }. 
\end{align}
\item Lemma~\ref{lem:alg_gradient_lb} is generic and nothing changes. 
\item Lemma~\ref{lem:TR-step}: $\wt{L}_H \le 27L_H/8$. 
\item Proposition~\ref{lem:TRM-lemma-iv}: All the quantities involved in determining $\Delta$, $m_H$, $M_H$, and $L_H$, $\beta_{\grad}$ are modified by at most constant multiplicative factors and changed to their respective tilde version, so we conclude that the RTM algorithm always takes unconstrained $\rI$ step after taking one, provided that
\begin{align}
\Delta \le \frac{c_5 c_\sharp^3 \theta^3 \mu}{n^{7/2} \log^{7/2}\paren{np}}. 
\end{align}
\item Lemma~\ref{lem:TR-grad-opt}:is generic and nothing changes. 

\item Proposition~\ref{lem:TRM-lemma-v}: Again $m_H$, $M_H$, $L_H$ are changed to $\wt{m_H}$, $\wt{M_H}$, and $\wt{L_H}$, respectively, differing by at most constant multiplicative factors. So we conclude for any integer $k' \ge 1$, 
\begin{align}
	\norm{\mb q^{(k_0 + k')} - \mb q_\star }{} \; \leq\;  \frac{c_6 c_\sharp \theta \mu}{n^{3/2} \log^{3/2} \paren{np}} 2^{- 2^{k'}},
\end{align}
provided 
\begin{align}
	\Delta \leq \frac{c_7 c_\sharp^2 \theta^2 \mu}{n^{5/2} \log^{5/2} (np) }. 
\end{align}
\end{itemize}
The final proof to Theorem~\ref{thm:geometry_comp} is almost identical to that of Theorem~\ref{thm:geometry_orth}, except for 
\begin{align}
\Delta & \le \min\set{\frac{c_8 c_\star  \theta \mu^2}{n^{5/2} \log^{3/2}\paren{np}}, \frac{c_9 c_\sharp^3 \theta^3 \mu}{n^{7/2}\log^{7/2}\paren{np}}}, \\
\wt{\zeta} & \doteq \min \set{\min_{\mb q \;\in\; \overline{\rII \cup \rIII}} f\paren{\mb q; \wt{\mb X_0}}, \max_{\mb q \;\in\; \overline{\rI}} f\paren{\mb q; \wt{\mb X}_0}},  
\end{align}
and hence all $\zeta$ is now changed to $\wt{\zeta}$, and also $\dI$, $\dII$, and $\dIII$ are changed to $\wt{\dI}$, $\wt{\dII}$, and $\wt{\dIII}$ as defined above, respectively. The final iteration complexity to each an $\eps$-near solution is hence 
\begin{align*}
\#\text{Iter} 
 & \le \brac{\min\set{\frac{c_{10} c_{\star}^3 \theta^3 \mu^4}{n^6 \log^3\paren{np}}, \frac{c_{11} c_\sharp^2 \theta^2 }{n}\Delta^2 }}^{-1}\paren{f\paren{\mb q^{(0)}} - f\paren{\mb q_\star}} \;+\; \log\log \paren{\frac{c_{12} c_\sharp \theta \mu}{\eps n^{3/2} \log^{3/2}\paren{np}}}. 
\end{align*}
Hence overall the qualitative behavior of the algorithm is not changed, as compared to that for the orthogonal case. Above $c_1$ through $c_{12}$ are all numerical constants.

%% file: sec/main_result.tex
\section{Complete Algorithm Pipeline and Main Results} \label{sec:main_result} 

For orthogonal dictionaries, from Theorem \ref{thm:geometry_orth} and its corollary, we know that all the minimizers $\wh{\mb q}_\star$ are $O(\mu)$ away from their respective nearest ``target'' $\mb q_\star$, with $\mb q_\star^* \wh{\mb Y} = \alpha \mb e_i^* \mb X_0$ for certain $\alpha \ne 0$ and $i \in [n]$; in Theorem~\ref{thm:trm_orth}, we have shown that w.h.p.\ the Riemannian TRM algorithm produces a solution $\widehat{\mb q}\in \bb S^{n-1}$ that is $\eps$ away to one of the minimizers, say $\wh{\mb q}_\star$. Thus, the $\wh{\mb q}$ returned by the TRM algorithm is $O(\eps + \mu)$ away from $\mb q_\star$. For exact recovery, we use a simple linear programming rounding procedure, which guarantees to exactly produce the optimizer $\mb q_\star$. We then use deflation to sequentially recover other rows of $\mb X_0$. Overall, w.h.p.\ both the dictionary $\mb A_0$ and sparse coefficient $\mb X_0$ are exactly recovered up to sign permutation, when $\theta \in \Omega(1)$, for orthogonal dictionaries. We summarize relevant technical lemmas and main results in Section~\ref{sec:main_orth}. The same procedure can be used to recover complete dictionaries, though the analysis is slightly more complicated; we present the results in Section~\ref{sec:main_comp}. Our overall algorithmic pipeline for recovering orthogonal dictionaries is sketched as follows. 
\begin{leftbar} 
\begin{enumerate}
\item \textbf{Estimating one row of $\mb X_0$ by the Riemannian TRM algorithm.} By Theorem \ref{thm:geometry_orth} (resp. Theorem~\ref{thm:geometry_comp}) and Theorem \ref{thm:trm_orth} (resp. Theorem~\ref{thm:trm_comp}), starting from any, when the relevant parameters are set appropriately (say as $\mu_\star$ and $\Delta_\star$), w.h.p., our Riemannian TRM algorithm finds a local minimizer $\widehat{\mb q}$, with $\mb q_\star$ the nearest target that exactly recovers one row of $\mb X_0$ and $\norm{\wh{\mb q} - \mb q_\star}{} \in O(\mu)$ (by setting the target accuracy of the TRM as, say,  $\eps = \mu$).

\item \textbf{Recovering one row of $\mb X_0$ by rounding.} To obtain the target solution $\mb q_\star$ and hence recover (up to scale) one row of $\mb X_0$, we solve the following linear program:
\begin{align}\label{eqn:LP_rounding}
	\mini_{\mb q} \norm{\mb q^*\wh{\mb Y}}{1},\quad \st \quad \innerprod{\mb r}{\mb q} = 1, 
\end{align}
with $\mb r = \widehat{\mb q}$. We show in Lemma~\ref{lem:alg_rounding_orth} (resp. Lemma~\ref{lem:alg_rounding_comp}) that when $\innerprod{\wh{\mb q}}{\mb q_\star}$ is sufficiently large, implied by $\mu$ being sufficiently small, w.h.p. the minimizer of \eqref{eqn:LP_rounding} is exactly $\mb q_\star$, and hence one row of $\mb X_0$ is recovered by $\mb q_\star^* \wh{\mb Y}$.

\item \textbf{Recovering all rows of $\mb X_0$ by deflation.} Once $\ell$ rows of $\mb X_0$ ($1 \le \ell \le n-2$) have been recovered, say, by unit vectors $\mb q_\star^1, \dots, \mb q_\star^\ell$, one takes an orthonormal basis $\mb U$ for $[\mathrm{span}\paren{\mb q_\star^1, \dots, \mb q_\star^\ell}]^\perp$, and minimizes the new function $h(\mb z) \doteq f(\mb U \mb z; \wh{\mb Y})$ on the sphere $\bb S^{n-\ell-1}$ with the Riemannian TRM algorithm (though conservative, one can again set parameters as $\mu_\star$, $\Delta_\star$, as in Step $1$) to produce a $\wh{\mb z}$. Another row of $\mb X_0$ is then recovered via the LP rounding~\eqref{eqn:LP_rounding} with input $\mb r = \mb U \wh{\mb z}$ (to produce $\mb q_\star^{\ell+1}$). Finally, by repeating the procedure until depletion, one can recover all the rows of $\mb X_0$.

\item \textbf{Reconstructing the dictionary $\mb A_0$.} By solving the linear system $\mb Y = \mb A\mb X_0$, one can obtain the dictionary $\mb A_0 = \mb Y \mb X_0^* \paren{\mb X_0 \mb X_0^*}^{-1}$.  
\end{enumerate}
\end{leftbar}

\subsection{Recovering Orthogonal Dictionaries} \label{sec:main_orth}
\begin{theorem}[Main theorem - recovering orthogonal dictionaries]\label{thm:main_orth}
Assume the dictionary $\mb A_0$ is orthogonal and we take $\wh{\mb Y} = \mb Y$. Suppose $\theta \in \paren{0,1/3}$, $\mu_\star <  \min\Brac{c_a\theta n^{-1},c_b n^{-5/4}}$, and $p \ge Cn^3 \log \frac{n}{\mu_\star \theta} /\paren{\mu_\star^2\theta^2}$. The above algorithmic pipeline with parameter setting
\begin{align}
\Delta_\star \le \min\set{\frac{c_c c_\star \theta \mu_\star^2}{n^{5/2} \log^{5/2}\paren{np}}, \frac{c_d c_\star^3 \theta^3 \mu_\star}{n^{7/2}\log^{7/2}\paren{np}}}, 
\end{align}
recovers the dictionary $\mb A_0$ and $\mb X_0$ in polynomial time, with failure probability bounded by $c_e p^{-6}$. Here $c_\star$ is as defined in Theorem~\ref{thm:geometry_orth}, and $c_a$ through $c_e$, and $C$ are all positive numerical constants. 
\end{theorem}

Towards a proof of the above theorem, it remains to be shown the correctness of the rounding and deflation procedures.

\paragraph{Proof of LP rounding.} The following lemma shows w.h.p.\ the rounding will return the desired $\mb q_\star$, provided the estimated $\wh{\mb q}$ is already near to it. 
\begin{lemma} [LP rounding - orthogonal dictionary]\label{lem:alg_rounding_orth}
There exists a positive constant $C$, such that for all $\theta \in \paren{0,1/3}$, and $p \ge Cn^2\log(n/\theta)/\theta$, with probability at least 
$
	1-2p^{-10} - \theta (n-1)^{-7}p^{-7} - \exp\paren{-0.3\theta (n-1) p},  
$
the rounding procedure~\eqref{eqn:LP_rounding} returns $\mb q_\star$ for any input vector $\mb r$ that satisfies
\begin{align*}
	\innerprod{\mb r}{\mb q_\star} \ge 249/250. 
\end{align*}
\end{lemma}
\begin{proof}
See Page~\pageref{proof:lem_alg_rounding_orth} under Section~\ref{sec:proof_main}. 
\end{proof}
Since $\innerprod{\wh{\mb q}}{\mb q_\star} = 1-\|\wh{\mb q} - \mb q_\star\|^2/2$, and $\norm{\wh{\mb q} - \mb q_\star}{} \in O(\mu)$, it is sufficient when $\mu$ is smaller than some small constant. 

\paragraph{Proof sketch of deflation.} We show the deflation works by induction. To understand the deflation procedure, it is important to keep in mind that the ``target'' solutions $\Brac{\mb q_\star^i}_{i=1}^n$ are orthogonal to each other. W.l.o.g., suppose we have found the first $\ell$ unit vectors $\mb q_\star^1, \dots, \mb q_\star^\ell$ which recover the first $\ell$ rows of $\mb X_0$. Correspondingly, we partition the target dictionary $\mb A_0$  and $\mb X_0$ as
\begin{align}\label{eqn:matrices-partition}
	\mb A_0 = [\mb V, \mb V^\perp],\quad \mb X_0 = \brac{\begin{smallmatrix}
\mb X_0^{[\ell]} \\
\mb X_0^{[n-\ell]}
\end{smallmatrix} },
\end{align}
where $\mb V \in \R^{n \times \ell}$, and $\mb X_0^{[\ell]}\in \bb R^{\ell\times n} $ denotes the submatrix with the first $\ell$ rows of $\mb X_0$. Let us define a function: $f_{n -\ell}^{\downarrow}: \R^{n-\ell} \mapsto \R$ by
\begin{align}\label{eqn:func-(n-l)}
f_{n -\ell}^{\downarrow}(\mb z; \mb W) \doteq \frac{1}{p}\sum_{k=1}^p h_{\mu}(\mb z^* \mb w_k), 
\end{align}
for any matrix $\mb W \in \R^{(n -\ell) \times p}$. Then by \eqref{eq:main_l2}, our objective function is equivalent to  
\begin{align*}
	h(\mb z) = f(\mb U \mb z; \mb A_0 \mb X_0) = f_{n-\ell}^{\downarrow}(\mb z;\mb U^* \mb A_0\mb X_0) = f_{n-\ell}^{\downarrow}(\mb z; \mb U^*\mb V\mb X_0^{[\ell]} + \mb U^*\mb V^\perp \mb X_0^{[n-\ell]}).
\end{align*}
Since the columns of the orthogonal matrix $\mb U\in \bb R^{n\times (n-\ell)}$ forms the orthogonal complement of $\text{span}\paren{\mb q_\star^1,\cdots,\mb q_\star^\ell}$, it is obvious that $\mb U^*\mb V=\mb 0$. Therefore, we obtain
\begin{align*}
h(\mb z) = f_{n-\ell}^{\downarrow}(\mb z; \mb U^*\mb V^\perp \mb X_0^{[n-\ell]}).
\end{align*}
Since $\mb U^* \mb V^\perp$ is orthogonal and $\mb X_0^{[n-\ell]} \sim_{i.i.d.} \mathrm{BG}(\theta)$, this is another instance of orthogonal dictionary learning problem with reduced dimension. If we keep the parameter settings $\mu_\star$ and $\Delta_\star$ as Theorem \ref{thm:main_orth}, the conditions of Theorem~\ref{thm:geometry_orth} and Theorem~\ref{thm:trm_orth} for all cases with reduced dimensions are still valid. So w.h.p., the TRM algorithm returns a $\wh{\mb z}$ such that $\norm{\wh{\mb z} - \mb z_\star}{} \in O(\mu_\star)$ where $\mb z_\star$ is a ``target'' solution that recovers a row of $\mb X_0$: 
\begin{align*}
\mb z_\star^*\mb U^* \mb V^\perp \mb X_0^{[n-\ell]} = \mb z_\star^*\mb U^* \mb A_0\mb X_0 = \alpha \mb e_i^* \mb X_0,\quad \text{for some }i \not \in [\ell].
\end{align*}
So pulling everything back in the original space, the effective target is $\mb q_\star^{\ell+1} \doteq \mb U \mb z_\star$, and $\mb U \wh{\mb z}$ is our estimation obtained from the TRM algorithm. Moreover, 
\begin{align*}
\norm{\mb U \wh{\mb z} - \mb U \mb z_\star}{} = \norm{\wh{\mb z} - \mb z_\star}{} \in O(\mu_\star). 
\end{align*}
Thus, by Lemma~\ref{lem:alg_rounding_orth}, one successfully recovers $\mb U \mb z_\star$ from $\mb U \wh{\mb z}$ w.h.p. when $\mu_\star$ is smaller than a constant. The overall failure probability can be obtained via a simple union bound and simplification of the exponential tails with inverse polynomials in $p$. 

\subsection{Recovering Complete Dictionaries} \label{sec:main_comp}
By working with the preconditioned data samples $\wh{\mb Y} = \overline{\mb Y} \doteq \sqrt{\theta p}\paren{\mb Y\mb Y^*}^{-1/2} \mb Y$,\footnote{In practice, the parameter $\theta$ might not be know beforehand. However, because it only scales the problem, it does not affect the overall qualitative aspect of results.} we can use a similar procedure described above to recover complete dictionaries.

\begin{theorem}[Main theorem - recovering complete dictionaries]\label{thm:main_comp}
Assume the dictionary $\mb A_0$ is complete with condition number $\kappa\paren{\mb A_0}$ and we take $\wh{\mb Y} = \ol{\mb Y}$. Suppose $\theta \in \paren{0,1/3}$, $\mu_\star <  \min\Brac{c_a\theta n^{-1},c_b n^{-5/4}}$, and $p \ge \frac{C}{c_\star^2 \theta} \max\set{\frac{n^4}{\mu^4}, \frac{n^5}{\mu^2}} \kappa^8\paren{\mb A_0} \log^4\paren{\frac{\kappa\paren{\mb A_0} n}{\mu \theta}}$. The algorithmic pipeline with parameter setting
\begin{align}
\Delta_\star \le \min\set{\frac{c_c c_\star \theta \mu_\star^2}{n^{5/2} \log^{5/2}\paren{np}}, \frac{c_d c_\star^3 \theta^3 \mu_\star}{n^{7/2}\log^{7/2}\paren{np}}}, 
\end{align}
recovers the dictionary $\mb A_0$ and $\mb X_0$ in polynomial time, with failure probability bounded by $c_e p^{-6}$. Here $c_\star$ is as defined in Theorem~\ref{thm:geometry_orth}, and $c_a$ through $c_f$, and $C$ are all positive numerical constants. 
\end{theorem}

Similar to the orthogonal case, we need to show the correctness of the rounding and deflation procedures so that the theorem above holds.

\paragraph{Proof of LP rounding}
The result of the LP rounding is only slightly different from that of the orthogonal case in Lemma \ref{lem:alg_rounding_orth}, so is the proof.
\begin{lemma} [LP rounding - complete dictionary]\label{lem:alg_rounding_comp}
There exists a positive constant $C$, such that for all $\theta \in \paren{0,1/3}$, and $p \ge \frac{C}{c_\star^2 \theta} \max\set{\frac{n^4}{\mu^4}, \frac{n^5}{\mu^2}} \kappa^8\paren{\mb A_0} \log^4\paren{\frac{\kappa\paren{\mb A_0} n}{\mu \theta}}$, with probability at least 
$
	1-3p^{-8} - \theta (n-1)^{-7}p^{-7} - \exp\paren{-0.3\theta (n-1) p},  
$
the rounding procedure~\eqref{eqn:LP_rounding} returns $\mb q_\star$ for any input vector $\mb r$ that satisfies
\begin{align*}
	\innerprod{\mb r}{\mb q_\star} \ge 249/250.  
\end{align*}
\end{lemma}
\begin{proof}
See Page~\pageref{proof:lem_alg_rounding_comp} under Section~\ref{sec:proof_main}. 
\end{proof}

\paragraph{Proof sketch of deflation.} We use a similar induction argument to show the deflation works. Compared to the orthogonal case, the tricky part here is that the target vectors $\Brac{\mb q_\star^i}_{i=1}^n$ are not necessarily orthogonal to each other, but they are almost so. W.l.o.g., let us again assume that $\mb q_\star^1, \dots, \mb q_\star^\ell$ recover the first $\ell$ rows of $\mb X_0$, and similarly partition the matrix $\mb X_0$ as in \eqref{eqn:matrices-partition}.

By Lemma~\ref{lem:pert_key_mag} and~\eqref{eq:pert_upper_bound}, we can write $\ol{\mb Y} = (\mb Q + \mb \Xi) \mb X_0$ for some orthogonal matrix $\mb Q$ and small perturbation $\mb \Xi$ with $\norm{\mb \Xi}{} \le \delta < 1/10$ for some large $p$ as usual. Similar to the orthogonal case, we have
\begin{align*}
	h(\mb z) = f(\mb U \mb z; (\mb Q + \mb \Xi) \mb X_0) = f_{n -\ell}^{\downarrow}(\mb z; \mb U^* (\mb Q + \mb \Xi) \mb X_0),
\end{align*}
where $f_{n -\ell}^{\downarrow}$ is defined the same as in \eqref{eqn:func-(n-l)}. Next, we show that the matrix $\mb U^* (\mb Q + \mb \Xi) \mb X_0$ can be decomposed as $\mb U^*\mb V\mb X_0^{[n-\ell]} + \mb \Delta$, where $\mb V\in \bb R^{(n-\ell)\times n }$ is orthogonal and $\mb \Delta$ is a small perturbation matrix. More specifically, we show that

\begin{lemma}\label{lem:deflation-bound}
Suppose the matrices $\mb U\in \bb R^{n\times (n-\ell)}$, $\mb Q \in \bb R^{n\times n}$ are orthogonal as defined above, $\mb \Xi$ is a perturbation matrix with $\norm{\mb \Xi}{}\leq 1/20$, then
	\begin{align}
		\mb U^*\paren{\mb Q+\mb \Xi}\mb X_0 = \mb U^* \mb V\mb X_0^{[n-\ell]} + \mb \Delta,
	\end{align}
	where $\mb V\in \bb R^{n\times (n-\ell)} $ is a orthogonal matrix spans the same subspace as that of $\mb U$, and the norms of $\mb \Delta$ is bounded by
	\begin{align}
		\norm{\mb \Delta}{\ell^1\rightarrow \ell^2 } \leq 16\sqrt{n} \norm{\mb \Xi}{} \norm{\mb X_0}{\infty}, \quad \norm{\mb \Delta}{} \leq 16 \norm{\mb \Xi}{} \norm{\mb X_0}{},
	\end{align}
	where $\norm{\mb W}{\ell^1\rightarrow \ell^2} = \sup_{\norm{\mb z}{1}=1} \norm{\mb W\mb z}{} = \max_k \norm {\mb w_k}{} $ denotes the max column $\ell^2$-norm of a matrix $\mb W$. 
\end{lemma}
\begin{proof}
See Page~\pageref{proof:lem_deflation-bound} under Section~\ref{sec:proof_main}. 
\end{proof}
Since $\mb U\mb V$ is orthogonal and $\mb X_0^{[n-\ell]}\sim_{i.i.d.} \text{BG}(\theta)$, we come into another instance of perturbed dictionary learning problem with reduced dimension
\begin{align*}
	h(\mb z) = f_{n-\ell}^{\downarrow}(\mb z; \mb U^*\mb V \mb X_0^{[n-\ell]} + \mb \Delta).
\end{align*}
Since our perturbation analysis in proving Theorem~\ref{thm:geometry_comp} and Theorem~\ref{thm:trm_comp} solely relies on the fact that $\norm{\mb \Delta}{\ell^1 \rightarrow \ell^2} \leq C \norm{\mb \Xi}{}\sqrt{n}\norm{\mb X_0}{\infty}$, it is enough to make $p$ large enough so that the theorems are still applicable for the reduced version $f_{n-\ell}^{\downarrow}(\mb z; \mb U^*\mb V \mb X_0^{[n-\ell]} + \mb \Delta)$. Thus, by invoking Theorem~\ref{thm:geometry_comp} and Theorem~\ref{thm:trm_comp}, the TRM algorithm provably returns one $\wh{\mb z}$ such that $\wh{\mb z}$ is near to a perturbed optimal $\wh{\mb z}_\star$ with
\begin{align}\label{eqn:comp-solution}
	\wh{\mb z}_\star^* \mb U^*\mb V\mb X_0^{[n-\ell]} = \mb z_\star^* \mb U^*\mb V\mb X_0^{[n-\ell]} +\mb z_\star^*\mb \Delta =\alpha \mb e_i^* \mb X_0,\quad \text{for some } i\not \in [\ell],
\end{align}
where $\mb z_\star$ with $\norm{\mb z_\star}{}=1$ is the exact solution. More specifically, Corollary~\ref{cor:geometry_comp} implies
\begin{align*}
	\norm{\wh{\mb z} - \wh{\mb z}_\star }{} \leq \sqrt{2}\mu_\star /7.
\end{align*}
Next, we show that $\wh{\mb z}$ is also very near to the exact solution $\mb z_\star$. Indeed, the identity \eqref{eqn:comp-solution} suggests
\begin{align}
	&\paren{\wh{\mb z}_\star -\mb z_\star }^*\mb U^*\mb V\mb X_0^{[n-\ell]} = \mb z_\star^* \mb \Delta \nonumber\\
	\Longrightarrow\;& \wh{\mb z}_\star - \mb z_\star = \brac{(\mb X_0^{[n-\ell]})^* \mb V^*\mb U }^\dagger \mb \Delta^*\mb z_\star = \mb U^*\mb V \brac{(\mb X_0^{[n-\ell]})^* }^\dagger \mb \Delta^*\mb z_\star \label{eqn:comp-z-distance}
\end{align}
where $\mb W^\dagger = (\mb W^*\mb W)^{-1}\mb W^*$ denotes the pseudo inverse of a matrix $\mb W$ with full column rank. Hence, by \eqref{eqn:comp-z-distance} we can bound the distance between $\wh{\mb z}_\star$ and $\mb z_\star$ by
\begin{align*}
	\norm{\wh{\mb z}_\star - \mb z_\star }{}  \leq \norm{\brac{(\mb X_0^{[n-\ell]})^* }^\dagger }{} \norm{\mb \Delta}{} \leq \sigma_{\min}^{-1}(\mb X_0^{[n-\ell]} ) \norm{\mb \Delta}{}
\end{align*}
By Lemma~\ref{lem:bg_identity_diff}, when $p \ge \Omega(n^2 \log n)$, w.h.p., 
\begin{align*}
	 \theta p/2\leq \sigma_{\min}(\mb X_0^{[n-\ell]} (\mb X_0^{[n-\ell]})^* ) \leq \norm{\mb X_0^{[n-\ell]}(\mb X_0^{[n-\ell]})^* }{} \leq \norm{\mb X_0\mb X_0^*}{}\leq 3\theta p/2.
\end{align*}
Hence, combined with Lemma \ref{lem:deflation-bound}, we obtain
\begin{align*}
	\sigma_{\min}^{-1}(\mb X_0^{[n-\ell]}) \leq \sqrt{\frac{2}{\theta p}},\quad  \norm{\mb \Delta }{} \leq  28\sqrt{ \theta p} \norm{\mb \Xi}{}/\sqrt{2},
\end{align*}
which implies that $\norm{\wh{\mb z}_\star -\mb z_\star }{}\leq 28 \norm{\mb \Xi}{}$. Thus, combining the results above, we obtain
\begin{align*}
	\norm{\wh{\mb z} - \mb z_\star}{} \le \norm{\wh{\mb z} - \wh{\mb z}_\star}{} + \norm{\wh{\mb z}_\star - \mb z_\star}{} \le \sqrt{2}\mu_\star/7 + 28\norm{\mb \Xi}{}.
\end{align*}
Lemma~\ref{lem:pert_key_mag}, and in particular~\eqref{eq:pert_upper_bound}, for our choice of $p$ as in Theorem~\ref{thm:geometry_comp}, $\norm{\mb \Xi}{} \le c\mu_\star^2 n^{-3/2}$, where $c$ can be made smaller by making the constant in $p$ larger. For $\mu_\star$ sufficiently small, we conclude that 
\begin{align*}
\norm{\mb U \wh{\mb z} - \mb U\mb z_\star}{} = \norm{\wh{\mb z} - \mb z_\star}{} \le 2\mu_\star/7. 
\end{align*} 
In words, the TRM algorithm returns a $\wh{\mb z}$ such that $\mb U \wh{\mb z}$ is very near to one of the unit vectors $\Brac{\mb q_\star^i}_{i=1}^n$, such that $(\mb q_\star^i)^* \ol{\mb Y} = \alpha \mb e_i^*\mb X_0$ for some $\alpha \ne 0$. For $\mu_\star$ smaller than a fixed constant, one will have 
\begin{align*}
\innerprod{\mb U \wh{\mb z}}{\mb q_\star^i} \ge 249/250, 
\end{align*}
and hence by Lemma~\ref{lem:alg_rounding_comp}, the LP rounding exactly returns the optimal solution $\mb q_\star^i$ upon the input $\mb U \wh{\mb z}$. 

The proof sketch above explains why the recursive TRM plus rounding works. The overall failure probability can be obtained via a simple union bound and simplifications of the exponential tails with inverse polynomials in $p$.

%% file: sec/exp.tex
\section{Simulations} \label{sec:exp}
To corroborate our theory, we experiment with dictionary recovery on simulated data. For simplicity, we focus on recovering orthogonal dictionaries and we declare success once a single row of the coefficient matrix is recovered. 

Since the problem is invariant to rotations, w.l.o.g.\ we set the dictionary as $\mb A_0 = \mb I \in \R^{n \times n}$. We fix $p = 5n^3$, and each column of the coefficient matrix $\mb X_0 \in \R^{n \times p}$ has exactly $k$ nonzero entries, chosen uniformly random from $\binom{[n]}{k}$. These nonzero entries are i.i.d. standard normals. This is slightly different from the Bernoulli-Gaussian model we assumed for analysis. For $n$ reasonably large, these two models produce similar behavior. For the sparsity surrogate defined in~\eqref{eq:logexp}, we fix the parameter $\mu = 10^{-2}$. We implement Algorithm~\ref{alg:trm} with adaptive step size instead of the fixed step size in our analysis. 

To see how the allowable sparsity level varies with the dimension, which our theory primarily is about, we vary the dictionary dimension $n$ and the sparsity $k$ both between $1$ and $120$; for every pair of $(k,n)$ we repeat the simulations independently for $T=5$ times. Because the optimal solutions are signed coordinate vectors $\Brac{\mb e_i}_{i=1}^n$, for a solution $\wh{\mb q}$ returned by the TRM algorithm, we define the reconstruction error (RE) to be
\begin{align}
\mathtt{RE} =  \min_{1\leq i\leq n} \paren{\norm{\wh{\mb q} - \mb e_i}{},\norm{\wh{\mb q} + \mb e_i}{} }. 
\end{align}
The trial is determined to be a success once $\mathtt{RE} \le \mu$, with the idea that this indicates $\wh{\mb q}$ is already very near the target and the target can likely be recovered via the LP rounding we described (which we do not implement here). 
\begin{figure}
\centerline{\includegraphics[width=0.5\linewidth]{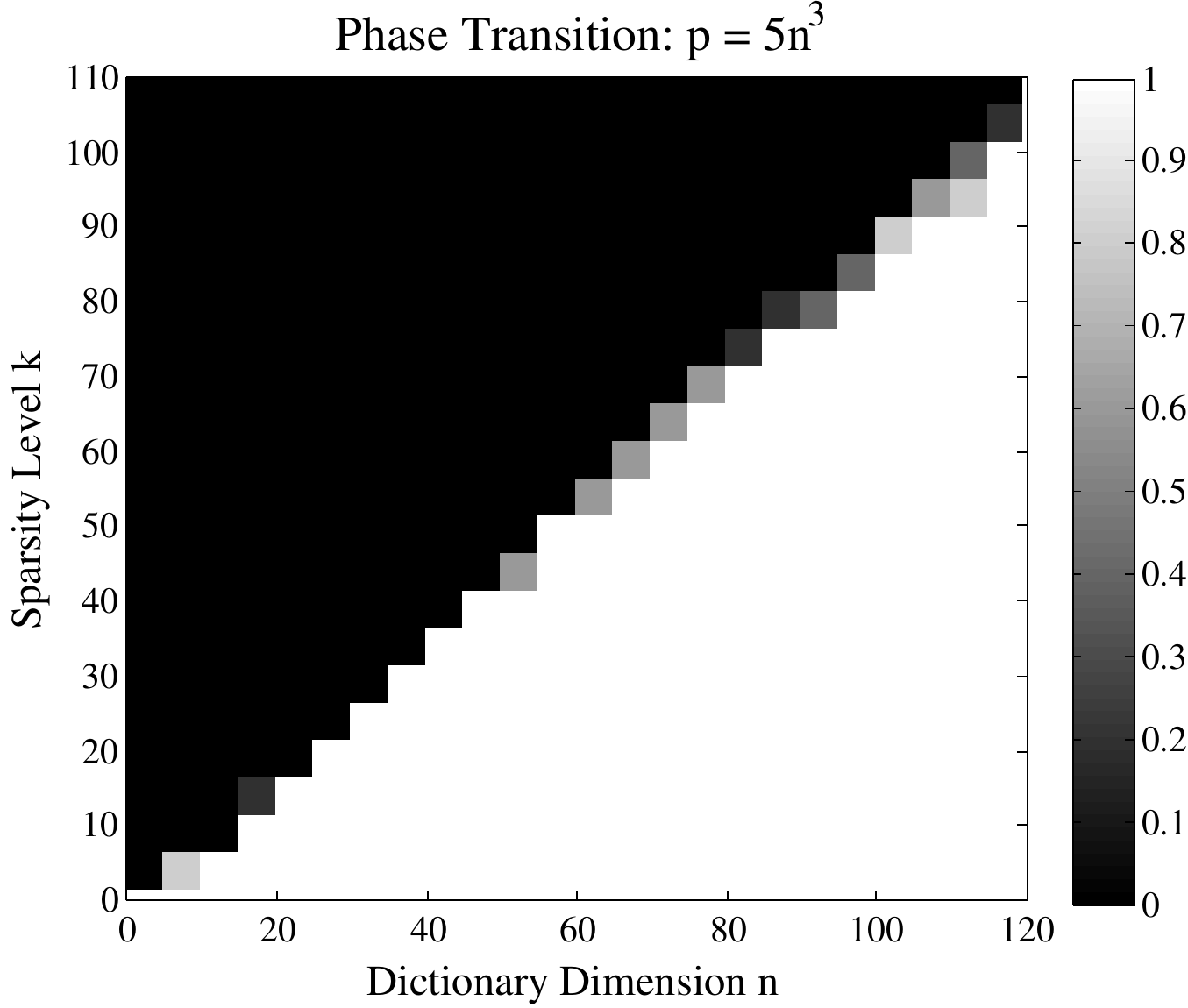}}
\caption{Phase transition for recovering a single sparse vector under the dictionary learning model with the sample complexity $p=5n^3$} \label{fig:exp}
\end{figure}
Figure~\ref{fig:exp} shows the phase transition in the $(n, k)$ plane for the orthogonal case. It is obvious that our TRM algorithm can work well into the linear region whenever $p \in O(n^3)$. Our analysis is tight up to logarithm factors, and also the polynomial dependency on $1/\mu$, which under the theory is polynomial in $n$.

%% file: sec/discuss.tex
\section{Discussion} \label{sec:discuss}
For recovery of complete dictionaries, the LP program approach in~\cite{spielman2012exact} that works with $\theta \le O(1/\sqrt{n})$ only demands $p \ge \Omega(n^2 \log n^2)$, which is recently improved to $p \ge \Omega(n \log^4 n)$~\cite{luh2015dictionary}, almost matching the lower bound $\Omega(n \log n)$ (i.e., when $\theta \sim 1/n$). The sample complexity stated in Theorem~\ref{thm:main_comp} is obviously much higher. It is interesting to see whether such growth in complexity is intrinsic to working in the linear regime. Though our experiments seemed to suggest the necessity of $p \sim O(n^3)$ even for the orthogonal case, there could be other efficient algorithms that demand much less. Tweaking these three points will likely improve the complexity: (1) The $\ell^1$ proxy. The derivative and Hessians of the $\log\cosh$ function we adopted entail the $\tanh$ function, which is not amenable to effective approximation and affects the sample complexity; (2) Geometric characterization and algorithm analysis. It seems working directly on the sphere (i.e., in the $\mb q$ space) could simplify and possibly improve certain parts of the analysis; (3) treating the complete case directly, rather than using (pessimistic) bounds to treat it as a perturbation of the orthogonal case. Particularly, general linear transforms may change the space significantly, such that preconditioning and comparing to the orthogonal transforms may not be the most efficient way to proceed. 

It is possible to extend the current analysis to other dictionary settings. Our geometric structures and algorithms allow plug-and-play noise analysis. Nevertheless, we believe a more stable way of dealing with noise is to directly extract the whole dictionary, i.e., to consider geometry and optimization (and perturbation) over the orthogonal group. This will require additional nontrivial technical work, but likely feasible thanks to the relatively complete knowledge of the orthogonal group~\cite{edelman1998geometry, absil2009}. A substantial leap forward would be to extend the methodology to recovery of \emph{structured} overcomplete dictionaries, such as tight frames. Though there is no natural elimination of one variable, one can consider the marginalization of the objective function wrt the coefficients and work with hidden functions. \footnote{This recent work~\cite{arora2015simple} on overcomplete DR has used a similar idea. The marginalization taken there is near to the global optimum of one variable, where the function is well-behaved. Studying the global properties of the marginalization may introduce additional challenges.} For the coefficient model, as we alluded to in Section~\ref{sec:lit_review}, our analysis and results likely can be carried through to coefficients with statistical dependence and physical constraints. 

The connection to ICA we discussed in Section~\ref{sec:lit_review} suggests our geometric characterization and algorithms can be modified for the ICA problem. This likely will provide new theoretical insights and computational schemes to ICA. In the surge of theoretical understanding of nonconvex heuristics~\cite{keshavan2010matrix, jain2013low, hardt2014understanding, hardt2014fast, netrapalli2014non, jain2014fast, netrapalli2013phase, candes2015phase, jain2014provable, anandkumar2014guaranteed, yi2013alternating, lee2013near, qu2014finding, lee2013near, agarwal2013learning, agarwal2013exact, arora2013new, arora2015simple, arora2014more}, the initialization plus local refinement strategy mostly differs from practice, whereby random initializations seem to work well, and the analytic techniques developed are mostly fragmented and highly specialized. The analytic and algorithmic we developed here hold promise to provide a coherent account of these problems. It is interesting to see to what extent we can streamline and generalize the framework. 

Our motivating experiment on real images in Section~\ref{sec:intro_exp} remains mysterious. If we were to believe that real image data are ``nice'' and our objective there does not have spurious local minima either, it is surprising ADM would escape all other critical points -- this is not predicted by classic or modern theories. One reasonable place to start is to look at how gradient descent algorithms with generic initializations can escape local maxima and saddle points (at least with high probability). The recent work~\cite{ge2015escaping} has showed that randomly perturbing each iterate can help gradient algorithm to escape saddle points with high probability.  
It would be interesting to know whether similar results can be obtained for gradient descent algorithms with random initialization. The continuous counterpart seems well understood; see, e.g., ~\cite{helmke1994optimization} for discussions of Morse-Bott theorem and gradient flow convergence. 

%% file: sec/proof_geometry.tex
\section{Proofs of Main Technical Results for High Dimensional Geometry} \label{sec:proof_geometry}
In this section, we provide complete proofs for technical results stated in Section~\ref{sec:geometry}. Before that, let us introduce some notations and common results that will be used later throughout this section. Since we deal with BG random variables and random vectors, it is often convenient to write such vector explicitly as $\mb x = \brac{\Omega_1 v_1, \dots, \Omega_n v_n} = \bm \Omega \odot \mb v$, where $\Omega_1, \dots, \Omega_n$ are i.i.d.\ Bernoulli random variables and $v_1, \dots, v_n$ are i.i.d.\ standard normal. For a particular realization of such random vector, we will denote the support as $\mc I \subset [n]$. Due to the particular coordinate map in use, we will often refer to subset $\mc J \doteq \mc I \setminus\Brac{n}$ and the random vectors $\overline{\mb x} \doteq \brac{\Omega_1 v_1, \dots, \Omega_{n-1} v_{n-1}}$ and $\overline{\mb v} \doteq \brac{v_1, \dots, v_{n-1}}$ in $\R^{n-1}$. By Lemma \ref{lem:derivatives_basic_surrogate}, it is not hard to see that 
\begin{align}
\nabla_{\mb w} h_{\mu}\paren{\mb q^*\paren{\mb w} \mb x} & = \tanh\paren{\frac{\mb q^*\paren{\mb w} \mb x}{\mu}} \paren{\overline{\mb x} - \frac{x_n}{q_n\paren{\mb w}} \mb w},\label{eqn:lse-gradient} \\
\nabla^2_{\mb w} h_{\mu}\paren{\mb q^*\paren{\mb w} \mb x} & = \frac{1}{\mu} \brac{1-\tanh^2\paren{\frac{\mb q^*\paren{\mb w} \mb x}{\mu}}} \paren{\overline{\mb x} - \frac{x_n}{q_n\paren{\mb w}} \mb w} \paren{\overline{\mb x} - \frac{x_n}{q_n\paren{\mb w}} \mb w}^* \nonumber \\
& \qquad - x_n \tanh\paren{\frac{\mb q^*\paren{\mb w} \mb x}{\mu}} \paren{\frac{1}{q_n\paren{\mb w}} \mb I + \frac{1}{q_n^3\paren{\mb w}} \mb w \mb w^*}. \label{eqn:lse-hessian}
\end{align}

\subsection{Proofs for Section~\ref{sec:geo_results_orth}}
\input{sec/proof_asymp} \label{sec:proof_geo_exp}
\input{sec/proof_finite_concentration}
\input{sec/proof_finite_lipschitz}

\subsection{Proofs of Theorem~\ref{thm:geometry_orth}} \label{sec:proof_geometry_orth}
\input{sec/proof_finite_union}

\subsection{Proofs for Section~\ref{sec:geo_results_comp} and Theorem~\ref{thm:geometry_comp}} \label{sec:proof_geometry_comp}
\input{sec/proof_geo_comp}

%% file: sec/proof_asymp.tex
\subsubsection{Proof of Proposition~\ref{prop:geometry_asymp_curvature}} \label{sec:proof_geo_asym_curvature}
The proof involves some delicate analysis, particularly polynomial approximation of the function $f\paren{t} = \frac{1}{\paren{1+t}^2}$ over $t \in \left[0, 1\right]$. This is naturally induced by the $1-\tanh^2\paren{\cdot}$ function. The next lemma characterizes one polynomial approximation of $f\paren{t}$. 

\begin{lemma} \label{lem:neg_curvature_norm_bound} 
Consider $f(t) = \frac{1}{(1+t)^2}$ for $t \in \brac{0,1}$. For every $T > 1$, there is a sequence $b_0, b_1, \dots$, with $\norm{\mb b}{\ell^1} = T < \infty$, such that the polynomial $p(t) = \sum_{k=0}^\infty b_k t^k$ satisfies 
\begin{align*}
\norm{f-p}{L^1[0,1]} \;\le\; \frac{1}{2\sqrt{T}}, \quad \norm{f-p}{L^\infty[0,1]} \;\le\; \frac{1}{\sqrt{T}},
\end{align*}
In particular, one can choose $b_k = (-1)^k(k+1)\beta^k$ with $\beta = 1-1/\sqrt{T} < 1$ such that 
\begin{align*}
p\paren{t} = \frac{1}{\paren{1+\beta t}^2} = \sum_{k=0}^\infty (-1)^k(k+1)\beta^k t^k. 
\end{align*}
Moreover, such sequence satisfies $0<\sum_{k=0}^\infty \frac{b_k}{(1+k)^3}<\sum_{k=0}^\infty \frac{\abs{b_k}}{(1+k)^3}<2$. 
\end{lemma}

\begin{lemma}\label{lem:neg_curvature_tanh_square_1} 
Let $X\sim \mc N\paren{0,\sigma_X^2}$ and $Y\sim \mc N\paren{0,\sigma_Y^2}$. We have
\begin{multline*}
\bb E\brac{\paren{1-\tanh^2\paren{\frac{X+Y}{\mu}} } X^2 \indicator{X+Y>0} } \le \\
 \frac{1}{\sqrt{2\pi}} \frac{\mu \sigma_X^2 \sigma_Y^2}{\paren{\sigma_X^2 + \sigma_Y^2}^{3/2}} + \frac{\mu^3 \sigma_X^2 \sigma_Y^2}{\paren{\sigma_X^2 + \sigma_Y^2}^{3/2}} + \frac{3}{4\sqrt{2\pi}} \frac{\sigma_X^2 \mu^3}{\paren{\sigma_X^2 + \sigma_Y^2}^{5/2}} \paren{3\mu^2 + 4 \sigma_X^2}.
\end{multline*}
\end{lemma}

\begin{proof}
For $x + y > 0$, let $z = \exp\paren{-2\frac{x + y}{\mu} } \in [0,1]$, then $1-\tanh^2\paren{\frac{x + y}{\mu}} = \frac{4z}{\paren{1+z}^2}$. Fix any $T > 1$ to be determined later, by Lemma~\ref{lem:neg_curvature_norm_bound}, we choose the polynomial $p_{\beta}\paren{z} = \frac{1}{\paren{1+\beta z}^2}$ with $\beta = 1 - 1/\sqrt{T}$ to upper bound $f\paren{z} = \frac{1}{\paren{1+z}^2}$. So we have 
\begin{align*}
\bb E\brac{\paren{1-\tanh^2\paren{\frac{X+Y}{\mu}} } X^2 \indicator{X+Y>0} } 
& = 4\bb E\brac{Zf\paren{Z} X^2 \indicator{X+Y>0} } \nonumber \\
& \le 4\bb E\brac{Zp_{\beta}\paren{Z} X^2 \indicator{X+Y>0} } \nonumber \\
& = 4\sum_{k=0}^\infty \Brac{b_k\bb E\brac{Z^{k+1}X^2\indicator{X+Y>0} }  }, 
\end{align*}
where $b_k = (-1)^k(k+1)\beta^k$, and the exchange of infinite summation and expectation above is justified in view that
\begin{align*}
\sum_{k=0}^\infty \abs{b_k}\bb E\brac{Z^{k+1}X^2\indicator{X+Y>0} } \le \sum_{k=0}^\infty \abs{b_k}\bb E\brac{X^2\indicator{X+Y>0} } \le \sigma_X^2  \sum_{k=0}^\infty \abs{b_k} < \infty
\end{align*} 
and the dominated convergence theorem (see, e.g., theorem 2.24 and 2.25 of~\cite{folland1999real}). By Lemma~\ref{lem:aux_asymp_proof_a}, we have 
\begin{align*}
& \sum_{k=0}^\infty \Brac{b_k\bb E\brac{Z^{k+1}X^2\indicator{X+Y>0} }  } \\
=\; & \sum_{k=0}^\infty \paren{-\beta}^k \paren{k+1} \left[\paren{\sigma_X^2+\frac{4\paren{k+1}^2}{\mu^2}\sigma_X^4}\exp\paren{\frac{2\paren{k+1}^2}{\mu^2}\paren{\sigma_X^2 + \sigma_Y^2}}\Phi^c\paren{\frac{2\paren{k+1}}{\mu}\sqrt{\sigma_X^2 + \sigma_Y^2}} \right. \\
& \qquad \left. - \frac{2\paren{k+1}}{\mu} \frac{\sigma_X^4}{\sqrt{2\pi}\sqrt{\sigma_X^2 + \sigma_Y^2}}\right] \\
\le \; & \frac{1}{\sqrt{2\pi}}\sum_{k=0}^\infty \paren{-\beta}^k \paren{k+1}\brac{\frac{\sigma_X^2 \mu}{2\paren{k+1}\sqrt{\sigma_X^2 + \sigma_Y^2}} - \frac{\sigma_X^2 \mu^3}{8\paren{k+1}^3 \paren{\sigma_X^2 + \sigma_Y^2}^{3/2}} - \frac{\mu \sigma_X^4}{2\paren{k+1}\paren{\sigma_X^2 + \sigma_Y^2}^{3/2}}} \nonumber \\
& \qquad + \frac{3}{\sqrt{2\pi}}\sum_{k=0}^\infty \beta^k \paren{k+1} \paren{\sigma_X^2+\frac{4\paren{k+1}^2}{\mu^2}\sigma_X^4} \frac{\mu^5}{32\paren{k+1}^5 \paren{\sigma_X^2 + \sigma_Y^2}^{5/2}}, 
\end{align*}
where we have applied Type I upper and lower bounds for $\Phi^c\paren{\cdot}$ to even $k$ and odd $k$ respectively and rearrange the terms to obtain the last line. Using the following estimates (see Lemma~\ref{lem:neg_curvature_norm_bound})
\begin{align*}
\sum_{k=0}^\infty \paren{-\beta}^k = \frac{1}{1+\beta}, \quad \sum_{k=0}^\infty \frac{b_k}{\paren{k+1}^3} \ge 0, \quad  \sum_{k=0}^\infty \frac{\abs{b_k}}{\paren{k+1}^5} \le \sum_{k=0}^\infty \frac{\abs{b_k}}{\paren{k+1}^3} \le 2, 
\end{align*}
we obtain 
\begin{multline*}
\sum_{k=0}^\infty \Brac{b_k\bb E\brac{Z^{k+1}X^2\indicator{X+Y>0} } } 
 \le \\
 \frac{1}{2\sqrt{2\pi}} \frac{\mu \sigma_X^2 \sigma_Y^2}{\paren{\sigma_X^2 + \sigma_Y^2}^{3/2}} \frac{1}{1+\beta} + \frac{3}{16\sqrt{2\pi}} \frac{\sigma_X^2 \mu^3}{\paren{\sigma_X^2 + \sigma_Y^2}^{5/2}} \paren{3\mu^2 + 4 \sigma_X^2}. 
\end{multline*}
Noticing $\frac{1}{1+\beta} < \frac{1}{2} + \frac{1}{2\sqrt{T}}$ and choosing $T = \mu^{-4}$, we obtain the desired result. 
\end{proof}

\begin{lemma} \label{lem:neg_curvature_tanh_square_2}
Let $X\sim \mc N\paren{0,\sigma_X^2}$ and $Y\sim \mc N\paren{0,\sigma_Y^2}$. We have 
\begin{multline*}
\bb E\left[ \tanh\left( \frac{X + Y }{\mu} \right) X \right] \ge \\
\frac{2\sigma_X^2}{\sqrt{2\pi}\sqrt{\sigma_X^2 + \sigma_Y^2}} - \frac{4\mu^2\sigma_X^2}{\sqrt{2\pi}\sqrt{\sigma_X^2 + \sigma_Y^2}} - \frac{2\sigma_X^2\mu^2}{\sqrt{2\pi} \paren{\sigma_X^2 + \sigma_Y^2}^{3/2}} - \frac{3\sigma_X^2\mu^4}{2\sqrt{2\pi}\paren{\sigma_X^2 + \sigma_Y^2}^{5/2}}. 
\end{multline*}
\end{lemma}

\begin{proof}
By Lemma \ref{lem:aux_asymp_proof_a}, we know 
\begin{align*}
\bb E \brac{\tanh\paren{\frac{X+Y}{\mu}} X  } \;=\; \frac{\sigma_X^2}{\mu } \bb E\brac{1- \tanh^2\paren{\frac{X+Y}{\mu}} }
\end{align*}
Similar to the proof of the above lemma, for $x + y > 0$, let $z  = \exp\paren{-2\frac{x + y}{\mu}}$ and $f\paren{z} = \frac{1}{\paren{1+z}^2}$. Fixing any $T > 1$, we will use $4zp_{\beta}\paren{z} = \frac{4z}{\paren{1+\beta z}^2}$ to approximate the $1-\tanh^2\paren{\frac{x+y}{\mu}} = 4zf\paren{z}$ function from above, where again $\beta = 1 - 1/\sqrt{T}$. So we obtain 
\begin{align*}
\bb E\brac{1- \tanh^2\paren{\frac{X+Y}{\mu}} } 
& = 8\expect{f\paren{Z} Z \indicator{X + Y > 0}} \nonumber \\
& = 8\expect{p_{\beta}\paren{Z} Z \indicator{X + Y > 0}} - 8\expect{\paren{p_{\beta}\paren{Z} - f\paren{Z} }Z \indicator{X + Y > 0}}. 
\end{align*}
Now for the first term, we have 
\begin{align*}
\expect{p_{\beta}\paren{Z} Z \indicator{X + Y > 0}} = \sum_{k=0}^\infty b_k \bb E\brac{ Z^{k+1} \indicator{X+Y>0} }, 
\end{align*}
justified as $\sum_{k=0}^\infty \abs{b_k} \bb E\brac{ Z^{k+1} \indicator{X+Y>0} } \le \sum_{k=0}^\infty \abs{b_k} < \infty$ making the dominated convergence theorem (see, e.g., theorem 2.24 and 2.25 of~\cite{folland1999real}) applicable. To proceed, from Lemma~\ref{lem:aux_asymp_proof_a}, we obtain
\begin{align*}
& \sum_{k=0}^\infty  b_k \bb E\brac{ Z^{k+1} \indicator{X+Y>0} } \nonumber \\
=\; & \sum_{k=0}^\infty \paren{-\beta}^k \paren{k+1} \exp\paren{\frac{2}{\mu^2}\paren{k+1}^2 \paren{\sigma_X^2 + \sigma_Y^2}} \Phi^c\paren{\frac{2}{\mu}\paren{k+1} \sqrt{\sigma_X^2 + \sigma_Y^2}} \\
\ge\; & \frac{1}{\sqrt{2\pi}} \sum_{k = 0}^\infty \paren{-\beta}^k \paren{k+1}\paren{\frac{\mu}{2\paren{k+1}\sqrt{\sigma_X^2 + \sigma_Y^2}} - \frac{\mu^3}{8\paren{k+1}^3\paren{\sigma_X^2 + \sigma_Y^2}^{3/2}}} \\
& \qquad - \frac{3}{\sqrt{2\pi}} \sum_{k=0}^\infty \beta^k \paren{k+1} \frac{\mu^5}{32\paren{k+1}^5\paren{\sigma_X^2 + \sigma_Y^2}^{5/2}}, 
\end{align*}
where we have applied Type I upper and lower bounds for $\Phi^c\paren{\cdot}$ to odd $k$ and even $k$ respectively and rearrange the terms to obtain the last line. Using the following estimates (see Lemma~\ref{lem:neg_curvature_norm_bound})
\begin{align*}
\sum_{k=0}^\infty \paren{-\beta}^k = \frac{1}{1+\beta}, \quad  0 \le \sum_{k=0}^\infty \frac{b_k}{\paren{k+1}^3} \le \sum_{k=0}^\infty \frac{\abs{b_k}}{\paren{k+1}^5} \le \sum_{k=0}^\infty \frac{\abs{b_k}}{\paren{k+1}^3} \le 2, 
\end{align*}
we obtain 
\begin{multline*}
\sum_{k=0}^\infty  b_k \bb E\brac{ Z^{k+1} \indicator{X+Y>0} } \ge \\
\frac{\mu}{2\sqrt{2\pi} \sqrt{\sigma_X^2 + \sigma_Y^2}} \frac{1}{1+\beta} - \frac{\mu^3}{4\sqrt{2\pi}\paren{\sigma_X^2 + \sigma_Y^2}^{3/2}} - \frac{3\mu^5}{16\sqrt{2\pi}\paren{\sigma_X^2 + \sigma_Y^2}^{5/2}}. 
\end{multline*}
To proceed, by Lemma \ref{lem:aux_asymp_proof_a} and Lemma \ref{lem:neg_curvature_norm_bound}, we have
\begin{align*}
\bb E\brac{\paren{p_{\beta}(Z)-f(Z)} Z\indicator{X+Y>0}} \le \norm{p-f}{L^\infty[0,1]} \bb E\brac{Z\indicator{X+Y>0}}\le \frac{\mu}{2\sqrt{2\pi T} \sqrt{\sigma_X^2+\sigma_Y^2}}, 
\end{align*}
where we have also used Type I upper bound for $\Phi^c\paren{\cdot}$. Combining the above estimates, we get
\begin{multline*}
\bb E\brac{\tanh\paren{\frac{X+Y}{\mu}} X} \ge \\
\frac{4\sigma_X^2}{\sqrt{2\pi}\sqrt{\sigma_X^2 + \sigma_Y^2}} \paren{\frac{1}{1+\beta} - \frac{1}{\sqrt{T}}} - \frac{2\sigma_X^2\mu^2}{\sqrt{2\pi} \paren{\sigma_X^2 + \sigma_Y^2}^{3/2}} - \frac{3\sigma_X^2\mu^4}{2\sqrt{2\pi}\paren{\sigma_X^2 + \sigma_Y^2}^{5/2}}. 
\end{multline*}
Noticing $\frac{1}{1+\beta} > \frac{1}{2}$ and taking $T = \mu^{-4}$, we obtain the claimed result. 
\end{proof}

\begin{proof}[of Proposition~\ref{prop:geometry_asymp_curvature}]
For any $i \in [n-1]$, we have 
\begin{align*}
\int_0^1 \int_{\mb x} \abs{\frac{\partial}{\partial w_i} h_{\mu}\paren{\mb q^*\paren{\mb w} \mb x}} \mu\paren{d\mb x}\; d w_i \le \int_0^1 \int_{\mb x} \paren{\abs{x_i} + \abs{x_n} \frac{1}{q_n\paren{\mb w}}} \mu\paren{d\mb x}\; d w_i < \infty. 
\end{align*}
Hence by Lemma~\ref{lemma:exchange_diff_int} we obtain $\frac{\partial }{\partial w_i}\expect{h_{\mu}\paren{\mb q^*\paren{\mb w} \mb x}} = \expect{\frac{\partial }{\partial w_i} h_{\mu}\paren{\mb q^*\paren{\mb w} \mb x}}$. Moreover for any $j \in [n-1]$, 
\begin{multline*}
\int_0^1 \int_{\mb x} \abs{\frac{\partial^2}{\partial w_j \partial w_i} h_{\mu}\paren{\mb q^*\paren{\mb w} \mb x}} \mu\paren{d\mb x}\; d w_j \le \\
\int_0^1 \int_{\mb x} \brac{\frac{1}{\mu}\paren{\abs{x_i} + \frac{\abs{x_n}}{q_n\paren{\mb w}}}\paren{\abs{x_j} + \frac{\abs{x_n}}{q_n\paren{\mb w}}} + \abs{x_n}\paren{\frac{1}{q_n\paren{\mb w}} + \frac{1}{q_n^3\paren{\mb w}}}} \mu\paren{d\mb x}\; d w_i < \infty. 
\end{multline*}
Invoking Lemma~\ref{lemma:exchange_diff_int} again we obtain 
\begin{align*}
\frac{\partial^2}{\partial w_j \partial w_i} \expect{h_{\mu}\paren{\mb q^*\paren{\mb w} \mb x}} = \frac{\partial}{\partial w_j}\expect{\frac{\partial }{\partial w_i} h_{\mu}\paren{\mb q^*\paren{\mb w} \mb x}} = \expect{\frac{\partial^2}{\partial w_j \partial w_i} h_{\mu}\paren{\mb q^*\paren{\mb w} \mb x}}. 
\end{align*}
The above holds for any pair of $i, j \in [n-1]$, so it follows that 
\begin{align*}
\nabla^2_{\mb w}\expect{h_{\mu}\paren{\mb q^*\paren{\mb w} \mb x}} = \expect{\nabla^2_{\mb w} h_{\mu}\paren{\mb q^*\paren{\mb w} \mb x}}. 
\end{align*}
Hence it is easy to see that 
\begin{align*}
& \mb w^* \nabla^2_{\mb w} \expect{h_{\mu}\left(\mb q^*\paren{\mb w} \mb x\right)} \mb w  \nonumber \\
=\; &  \frac{1}{\mu}\expect{\left(1-\tanh^2\left(\frac{\mb q^*\paren{\mb w} \mb x}{\mu}\right)\right) \left(\mb w^* \overline{\mb x} - \frac{x_n}{q_n\paren{\mb w}}\norm{\mb w}{}^2\right)^2} - \expect{\tanh\left(\frac{\mb q^*\paren{\mb w} \mb x}{\mu}\right) \frac{x_n}{q_n^3\paren{\mb w}}\norm{\mb w}{}^2}. 
\end{align*}
Now the first term is 
\begin{align*}
& \frac{1}{\mu}\expect{\left(1-\tanh^2\left(\frac{\mb q^*\paren{\mb w} \mb x}{\mu}\right)\right) \left(\mb w^* \overline{\mb x} - \frac{x_n}{q_n\paren{\mb w}}\norm{\mb w}{}^2\right)^2} \\
=\; & \frac{2\paren{1-\theta}}{\mu} \expect{\left(1-\tanh^2\left(\frac{\mb w^*\overline{\mb x}}{\mu}\right)\right) \paren{\mb w^* \overline{\mb x}}^2 \indicator{\mb w^* \overline{\mb x} > 0}} \\
& \qquad - \frac{4\theta}{\mu} \frac{\norm{\mb w}{}^2}{q_n^2\paren{\mb w}} \expect{\left(1-\tanh^2\left(\frac{\mb w^*\overline{\mb x} + q_n\paren{\mb w} x_n}{\mu}\right)\right) \paren{\mb w^* \overline{\mb x}}\paren{q_n\paren{\mb w} x_n} \indicator{\mb w^*\overline{\mb x} + q_n\paren{\mb w} x_n > 0}} \\
& \qquad + \frac{2\theta}{\mu} \bb E_{\mc J}\bb E_{\mb v}\brac{\left(1-\tanh^2\left(\frac{\mb w^*_{\mc J}\overline{\mb v} + q_n\paren{\mb w} v_n}{\mu}\right)\right) \left(\mb w^*_{\mc J} \overline{\mb v}\right)^2 \indicator{\mb w^*_{\mc J}\overline{\mb v} + q_n\paren{\mb w} v_n > 0}} \\
& \qquad + \frac{2\theta}{\mu} \frac{\norm{\mb w}{}^4}{q_n^4\paren{\mb w}} \bb E_{\mc J} \bb E_{\mb v}\brac{\left(1-\tanh^2\left(\frac{\mb w^*_{\mc J}\overline{\mb v} + q_n\paren{\mb w} v_n}{\mu}\right)\right) \left(q_n\paren{\mb w} v_n\right)^2 \indicator{\mb w^*_{\mc J}\overline{\mb v} + q_n\paren{\mb w} v_n > 0}} \\
\le \; & \frac{8\paren{1-\theta}}{\mu} \expect{\exp\paren{-2\frac{\mb w^* \overline{\mb x}}{\mu}} \paren{\mb w^* \overline{\mb x}}^2 \indicator{\mb w^* \overline{\mb x} > 0}} \\
& \qquad + \frac{8\theta}{\mu} \frac{\norm{\mb w}{}^2}{q_n^2\paren{\mb w}} \expect{ \exp\paren{-\frac{2}{\mu} \paren{\mb w^*\overline{\mb x} + q_n\paren{\mb w} x_n}}\paren{\mb w^* \overline{\mb x} + q_n\paren{\mb w} x_n }^2 \indicator{\mb w^*\overline{\mb x} + q_n\paren{\mb w} x_n > 0} } \\
& \qquad + \frac{2\theta}{\mu} \bb E_{\mc J}\bb E_{X, Y}\brac{\left(1-\tanh^2\left(\frac{X + Y}{\mu}\right)\right) Y^2 \indicator{X + Y > 0}} \\
& \qquad + \frac{2\theta}{\mu} \frac{\norm{\mb w}{}^4}{q_n^4\paren{\mb w}} \bb E_{\mc J} \bb E_{X, Y}\brac{\left(1-\tanh^2\left(\frac{X + Y}{\mu}\right)\right) X^2 \indicator{X + Y > 0}}, 
\end{align*}
where conditioned on each support set $\mc J$, we let $X \doteq q_n\paren{\mb w} v_n \sim \mc N\paren{0, q_n^2\paren{\mb w}}$ and $Y \doteq \mb w_{\mc J}^* \overline{\mb v} \sim \mc N\paren{0, \norm{\mb w_{\mc J}}{}^2}$. Noticing the fact $t \mapsto \exp\paren{-2t/\mu}t^2$ for $t > 0$ is maximized at $t = \mu$ with maximum value $\exp\paren{-2}\mu^2$, and in view of the estimate in Lemma~\ref{lem:neg_curvature_tanh_square_1}, we obtain 
\begin{align*}
& \frac{1}{\mu}\expect{\left(1-\tanh^2\left(\frac{\mb q^*\paren{\mb w} \mb x}{\mu}\right)\right) \left(\mb w^* \overline{\mb x} - \frac{x_n}{q_n\paren{\mb w}}\norm{\mb w}{}^2\right)^2} \nonumber \\
\le\; & 8\exp\paren{-2}\paren{1-\theta + \frac{\norm{\mb w}{}^2}{q_n^2\paren{\mb w}}\theta} \mu \nonumber \\
& \quad + \frac{2\theta}{\mu} \bb E_{\mc J}\brac{\frac{1}{\sqrt{2\pi}} \frac{\mu \norm{\mb w_{\mc J}}{}^2 q_n^2\paren{\mb w}}{\norm{\mb q_{\mc I}}{}^3} + \frac{\mu^3 \norm{\mb w_{\mc J}}{}^2 q_n^2\paren{\mb w}}{\norm{\mb q_{\mc I}}{}^3}+ \frac{3}{4\sqrt{2\pi}} \frac{\norm{\mb w_{\mc J}}{}^2 \mu^3}{\norm{\mb q_{\mc I}}{}^5} \paren{3\mu^2 + 4 \norm{\mb w_{\mc J}}{}^2}} \nonumber \\
& \quad + \frac{2\theta}{\mu} \frac{\norm{\mb w}{}^4}{q_n^4\paren{\mb w}} \bb E_{\mc J} \brac{\frac{1}{\sqrt{2\pi}} \frac{\mu \norm{\mb w_{\mc J}}{}^2 q_n^2\paren{\mb w}}{\norm{\mb q_{\mc I}}{}^3} + \frac{\mu^3 \norm{\mb w_{\mc J}}{}^2 q_n^2\paren{\mb w}}{\norm{\mb q_{\mc I}}{}^3} + \frac{3}{4\sqrt{2\pi}} \frac{q_n^2\paren{\mb w} \mu^3}{\norm{\mb q_{\mc I}}{}^5} \paren{3\mu^2 + 4 q_n^2\paren{\mb w}} } \nonumber \\
\le\; & \frac{2\theta}{\sqrt{2\pi} q_n^2\paren{\mb w}} \bb E_{\mc J}\brac{\frac{\norm{\mb w_{\mc J}}{}^2}{\norm{\mb q_{\mc I}}{}^3}} + \frac{11}{20}\mu\paren{2+\frac{1}{q_n^2\paren{\mb w}}} + 2\theta \mu^2\paren{1 + \frac{3}{\sqrt{2\pi}q_n\paren{\mb w}} + \frac{1}{q_n^3\paren{\mb w}} + \frac{3}{\sqrt{2\pi} q_n^5\paren{\mb w}}}, 
\end{align*}
where we have used $\mu < q_n\paren{\mb w} \le \norm{\mb q_{\mc I}}{}$ and $\norm{\mb w_{\mc J}}{} \le \norm{\mb q_{\mc I}}{}$ and $\norm{\mb w}{} \le 1$ and $\theta \in \paren{0, 1/2}$ to simplify the intermediate quantities to obtain the last line. Similarly for the second term, we obtain 
\begin{align*}
& \expect{\tanh\left(\frac{\mb q^*\paren{\mb w} \mb x}{\mu}\right) \frac{x_n}{q_n^3\paren{\mb w}}\norm{\mb w}{}^2} \nonumber \\
=\; & \frac{\norm{\mb w}{}^2\theta}{q_n^4\paren{\mb w}} \bb E_{\mc J} \bb E_{\mb v}\brac{\tanh\paren{\frac{\mb w^*_{\mc J} \overline{\mb v} + q_n\paren{\mb w} v_n}{\mu}} x_n q_n\paren{\mb w}} \nonumber \\
\ge\; & \frac{\norm{\mb w}{}^2\theta}{q_n^4\paren{\mb w}} \bb E_{\mc J} \brac{\frac{2q_n^2\paren{\mb w}}{\sqrt{2\pi}\norm{\mb q_{\mc I}}{}} - \frac{4\mu^2 q_n^2\paren{\mb w}}{\sqrt{2\pi}\norm{\mb q_{\mc I}}{}} - \frac{2q_n^2\paren{\mb w}\mu^2}{\sqrt{2\pi} \norm{\mb q_{\mc I}}{}^3} - \frac{3q_n^2\paren{\mb w}\mu^4}{2\sqrt{2\pi}\norm{\mb q_{\mc I}}{}^5}} \nonumber \\
\ge\; & \sqrt{\frac{2}{\pi}} \frac{\theta}{q_n^2\paren{\mb w}} \bb E_{\mc J}\brac{\frac{\norm{\mb w}{}^2}{\norm{\mb q_{\mc I}}{}}} - \frac{4\theta \mu^2}{\sqrt{2\pi}}\paren{\frac{1}{q_n^3\paren{\mb w}} + \frac{1}{q_n^5\paren{\mb w}}}. 
\end{align*}
Collecting the above estimates, we obtain 
\begin{align}
& \mb w^* \nabla^2_{\mb w}\expect{h_{\mu}\paren{\mb q^*\paren{\mb w} \mb x}} \mb w \nonumber \\
\le \; & \sqrt{\frac{2}{\pi}} \frac{\theta}{q_n^2\paren{\mb w}} \bb E_{\mc J}\brac{\frac{\norm{\mb w_{\mc J}}{}^2}{\norm{\mb q_{\mc I}}{}^3} - \frac{\norm{\mb w}{}^2\paren{\norm{\mb w_{\mc J}}{}^2 + q_n^2\paren{\mb w}}}{\norm{\mb q_{\mc I}}{}^3}} \nonumber \\
& \qquad + \frac{11}{20}\mu\paren{2+\frac{1}{q_n^2\paren{\mb w}}} + 2\theta \mu^2\paren{1 + \frac{3}{\sqrt{2\pi}q_n\paren{\mb w}} + \frac{2}{q_n^3\paren{\mb w}} + \frac{5}{\sqrt{2\pi} q_n^5\paren{\mb w}}} \nonumber \\
\le\; &  -\sqrt{\frac{2}{\pi}}\theta \expect{\frac{\norm{\mb w_{\mc J^c}}{}^2}{\norm{\mb q_{\mc I}}{}^3}} + \frac{11}{10}\mu + \frac{11}{20}\frac{\mu}{q_n\paren{\mb w}} + 2\theta \mu^2\paren{1 + \frac{6}{q_n^5\paren{\mb w}}} \nonumber \\
\le \; & -\sqrt{\frac{2}{\pi}}\theta \paren{1-\theta} \norm{\mb w}{}^2 \expect{\frac{1}{\norm{\mb q_{\mc I}}{}^3}} + \frac{11}{10}\mu + \frac{11}{20}\frac{\mu}{q_n^2\paren{\mb w}} + 2\theta \mu^2\paren{1 + \frac{6}{q_n^5\paren{\mb w}}}, \label{eq:neg_curvature_final_form}
\end{align}
where to obtain the last line we have invoked the association inequality in Lemma~\ref{lemma:harris_ineq}, as both $\norm{\mb w_{\mc J^c}}{}^2$ and $1/\norm{\mb q_{\mc I}}{}^3$ both coordinatewise nonincreasing w.r.t. the index set. Substituting the upper bound for $\mu$ into~\eqref{eq:neg_curvature_final_form} and noting $R_h \le \norm{\mb w}{}$ and also noting the fact $q_n\paren{\mb w} \ge \frac{1}{2\sqrt{n}}$ (implied by the assumption $\norm{\mb w}{} \le \sqrt{\frac{4n-1}{4n}}$), we obtain the claimed result. 
\end{proof}

\subsubsection{Proof of Proposition~\ref{prop:geometry_asymp_gradient}} \label{sec:proof_geo_asym_gradient}
\begin{proof}
By similar consideration as proof of the above proposition, the following is justified:  
\begin{align*}
\nabla_{\mb w}\expect{h_{\mu}\paren{\mb q^*\paren{\mb w} \mb x}} = \expect{\nabla_{\mb w} h_{\mu}\paren{\mb q^*\paren{\mb w} \mb x}}. 
\end{align*}
Now consider 
\begin{align}
\mb w^* \nabla \expect{h_\mu (\mb q^*\paren{\mb w} \mb x)} 
& = \nabla \expect{\mb w^* h_\mu (\mb q^*\paren{\mb w} \mb x)} \nonumber \\
& = \expect{\tanh\paren{\frac{\mb q^*\paren{\mb w} \mb x}{\mu}} \paren{\mb w^* \bar{\mb x} }} - \frac{\norm{\mb w}{}^2}{q_n}\expect{\tanh\paren{\frac{\mb q^*\paren{\mb w} \mb x}{\mu}} x_n}. \label{eqn:proof_grad_1}
\end{align}
For~\eqref{eqn:proof_grad_1}, we next provide a lower bound for the first expectation and an upper bound for the second expectation. For the first, we have 
\begin{align*}
& \expect{\tanh\paren{\frac{\mb q^*\paren{\mb w} \mb x}{\mu}} \paren{\mb w^* \overline{\mb x} }} \\
=\; & \theta \bb E_{\mc J} \brac{ \bb E_{\mb v} \brac{ \tanh\paren{\frac{\mb w_{\mc J}^* \overline{\mb v} + q_n\paren{\mb w} v_n}{\mu} } \paren{\mb w_{\mc J}^* \overline{\mb v} }}}  + (1-\theta) \bb E_{\mc J} \brac{ \bb E_{\mb v} \brac{ \tanh\paren{\frac{\mb w_{\mc J}^* \overline{\mb v}}{\mu}} \paren{\mb w_{\mc J}^* \overline{\mb v} }}} \\
=\; &\theta \bb E_{\mc J}\brac{\bb E_{X, Y}\brac{ \tanh\paren{\frac{X+Y}{\mu}}Y } }\;+\; (1-\theta) \bb E_{\mc J}\brac{\bb E_{Y}\brac{\tanh\paren{\frac{Y}{\mu}}Y} }, 
\end{align*}
where $X \doteq q_n\paren{\mb w} v_n \sim \mc N\paren{0, q_n^2\paren{\mb w}}$ and $Y \doteq \mb w^*_{\mc J} \overline{\mb v} \sim \mc N\paren{0, \norm{\mb w_{\mc J}}{}^2}$. Now by Lemma \ref{lemma:harris_ineq} we obtain 
\begin{align*}
\bb E \brac{ \tanh\paren{\frac{X+Y}{\mu} } Y} \ge \bb E\brac{\tanh\paren{\frac{X+Y}{\mu}}} \bb E\brac{Y} = 0, 
\end{align*} 
as $\tanh\paren{\frac{X+Y}{\mu} } $ and $X$ are both coordinatewise nondecreasing function of $X$ and $Y$. Using the $\tanh\paren{z} \geq \paren{1-\exp\paren{-2z}}/2$ lower bound for $z > 0$ and integral results in Lemma~\ref{lem:aux_asymp_proof_a}, we obtain
\begin{align*}
\bb E\brac{\tanh\paren{\frac{Y}{\mu}}Y} 
& = 2\bb E\brac{\tanh\paren{\frac{Y}{\mu}}Y \indicator{Y > 0}} \\
& \ge \bb E\brac{\paren{1-\exp\paren{-\frac{2Y}{\mu}} }Y\indicator{Y>0} } \\
& = \frac{2\sigma_Y^2}{\mu}\exp\paren{\frac{2\sigma_Y^2}{\mu^2}} \Phi^c\paren{\frac{2\sigma_Y}{\mu}}  \\
& \ge \frac{2\sigma_Y^2}{\mu \sqrt{2\pi}} \paren{\sqrt{1+\frac{\sigma_Y^2}{\mu^2}} - \frac{\sigma_Y}{\mu}} \\
& \ge \frac{2\sigma_Y^2}{\mu\sqrt{2\pi}} \paren{ \sqrt{1+\frac{\norm{\mb w}{}^2}{\mu^2}} - \frac{\norm{\mb w}{}}{\mu}}, 
\end{align*}
where at the second last inequality we have used Type III lower bound for Gaussian upper tail $\Phi^c\paren{\cdot}$ (Lemma~\ref{lem:gaussian_tail_est}), and at the last we have used the fact that $t \mapsto \sqrt{1+t^2}-t$ is a monotonic decreasing function over $t > 0$ and that $\sigma_Y = \norm{\mb w_{\mc J}}{}\le \norm{\mb w}{}$. Collecting the above estimates, we have 
\begin{align}
\expect{\tanh\paren{\frac{\mb q^*\paren{\mb w} \mb x}{\mu}} \paren{\mb w^* \overline{\mb x} }} 
& \ge \paren{1- \theta} \bb E_{\mc J}\brac{\frac{2\norm{\mb w_{\mc J}}{}^2}{\mu\sqrt{2\pi}} \paren{ \sqrt{1+\frac{\norm{\mb w}{2}^2}{\mu^2}} - \frac{\norm{\mb w}{}}{\mu}}} \nonumber \\
& \ge  \paren{1- \theta} \bb E_{\mc J}\brac{\frac{2\norm{\mb w_{\mc J}}{}^2}{\mu\sqrt{2\pi}} \frac{\mu}{10\norm{\mb w}{}}} \nonumber \\
& \ge \frac{\theta\paren{1-\theta}\norm{\mb w}{}}{5\sqrt{2\pi}}, \label{eqn:proof_grad_2}
\end{align}
where at the second line we have used the assumption that $\norm{\mb w}{} \ge \frac{\mu}{6\sqrt{2}}$ and also the fact that $\sqrt{1+x^2} \ge x + \frac{1}{10x}$ for $x \ge \frac{1}{6\sqrt{2}}$. 

For the second expectation of~\eqref{eqn:proof_grad_1}, we have
\begin{align}
\bb E\brac{\tanh\paren{\frac{\mb q^*\paren{\mb w} \mb x}{\mu}} x_n} 
\le \theta \bb E\brac{\abs{\tanh\paren{\frac{\mb q^*\paren{\mb w} \mb x}{\mu}}} \abs{v_n}} 
\le \theta \sqrt{\frac{2}{\pi}}, 
\label{eqn:proof_grad_3}
\end{align}
as $\tanh\paren{\cdot}$ is bounded by one in magnitude. Plugging the results of~\eqref{eqn:proof_grad_2} and~\eqref{eqn:proof_grad_3} into~\eqref{eqn:proof_grad_1} and noticing that $q_n\paren{\mb w}^2 + \norm{\mb w}{}^2 = 1$ we obtain 
\begin{align*}
\mb w^* \nabla \expect{h_\mu (\mb q^*\paren{\mb w} \mb x)} \;\ge\; 
\frac{\theta\norm{\mb w}{}}{\sqrt{2\pi}}\brac{ \frac{1-\theta}{5} -\frac{2\norm{\mb w}{}}{\sqrt{1-\norm{\mb w}{}^2}}} \ge \frac{\theta\paren{1-\theta}\norm{\mb w}{}}{10\sqrt{2\pi}}, 
\end{align*}
where we have invoked the assumption that $\norm{\mb w}{} \le \frac{1}{10\sqrt{5}}\paren{1-\theta}$ to provide the upper bound $\frac{2\norm{\mb w}{}}{\sqrt{1-\norm{\mb w}{}^2}} \le \frac{1}{10}\paren{1-\theta}$. We then choose the particular ranges as stated for $\mu$ and $\theta$ to ensure $r_g < R_g$, completing the proof. 
\end{proof}

\subsubsection{Proof of Proposition~\ref{prop:geometry_asymp_strong_convexity}} \label{sec:proof_geo_asym_strcvx}

\begin{proof}
By consideration similar to proof of Proposition~\ref{prop:geometry_asymp_curvature}, we can exchange the hessian and expectation, i.e., 
\begin{align*}
\nabla^2_{\mb w} \expect{h_{\mu}\paren{\mb q^*\paren{\mb w} \mb x}} = \expect{\nabla^2_{\mb w} h_{\mu}\paren{\mb q^*\paren{\mb w} \mb x}}. 
\end{align*}
We are interested in the expected Hessian matrix
\begin{align*}
\nabla^2_{\mb w} \bb E\brac{h_{\mu}\paren{\mb q^*\paren{\mb w}\mb x}} 
&=\frac{1}{\mu}\bb E\brac{\paren{1-\tanh^2\paren{\frac{\mb q^*\paren{\mb w} x}{\mu}} } \paren{\overline{\mb x} - \frac{x_n}{q_n\paren{\mb w}}\mb w }\paren{\overline{\mb x} - \frac{x_n}{q_n\paren{\mb w}}\mb w }^*  } \nonumber \\
&- \bb E\brac{\tanh\paren{\frac{\mb q^*\paren{\mb w}\mb x}{\mu}} \paren{\frac{x_n}{q_n\paren{\mb w}} \mb I +\frac{x_n}{q_n^3\paren{\mb w}} \mb w\mb w^*}}
\end{align*}
in the region that $0 \le \norm{\mb w}{}\le \frac{\mu}{4\sqrt{2}} $. 

When $\mb w=\mb 0$, by Lemma \ref{lem:aux_asymp_proof_a}, we have 
\begin{align*}
& \left. \bb E\brac{\nabla^2_{\mb w} h_{\mu}\paren{\mb q^*\paren{\mb w}\mb x} }\right|_{\mb w=0}\\ 
=\; & \frac{1}{\mu}\bb E\brac{\paren{1-\tanh^2\paren{\frac{x_n}{\mu}}}\overline{\mb x}\; \overline{\mb x}^* } - \bb E\brac{\tanh\paren{\frac{x_n}{\mu}}x_n}\mb I \\
=\; & \frac{\theta(1-\theta)}{\mu}\mb I + \frac{\theta^2}{\mu} \bb E_{v_n}\brac{1- \tanh^2\paren{\frac{v_n}{\mu}}}\mb I - \frac{\theta}{\mu} \bb E_{v_n}\brac{1-\tanh^2\paren{\frac{v_n}{\mu}}}\mb I  \\
=\; & \frac{\theta(1-\theta)}{\mu} \bb E_{v_n}\brac{\tanh^2\paren{\frac{q_n\paren{\mb w}v_n}{\mu}}} \mb I. 
\end{align*}
Simple calculation based on Lemma~\ref{lem:aux_asymp_proof_a} shows 
\begin{align*}
\bb E_{v_n}\brac{\tanh^2\paren{\frac{v_n}{\mu}}} \ge 2\paren{1-4\exp\paren{\frac{2}{\mu^2}} \Phi^c\paren{\frac{2}{\mu}}} \ge 2\paren{1-\frac{2}{\sqrt{2\pi}} \mu}. 
\end{align*}
Invoking the assumptions $\mu \le \frac{1}{20 \sqrt{n}} \le 1/20$ and $\theta < 1/2$, we obtain 
\begin{align*}
\left. \bb E\brac{\nabla^2_{\mb w} h_{\mu}\paren{\mb q^*\paren{\mb w}\mb x} }\right|_{\mb w=0} \succeq \frac{\theta\paren{1-\theta}}{\mu} \paren{2 - \frac{4}{\sqrt{2\pi}}\mu} \mb I 
\succeq \frac{\theta}{\mu} \paren{1-\frac{1}{10\sqrt{2\pi}}} \mb I.  
\end{align*}

When $0 <\norm{\mb w}{}\le \frac{\mu}{4\sqrt{2}}$, we aim to derive a semidefinite lower bound for 
\begin{align}
& \bb E\brac{\nabla^2_{\mb w} h_{\mu}\paren{\mb q^*\paren{\mb w}\mb x}} \nonumber \\
=\; &  \frac{1}{\mu}\bb E\brac{\paren{1-\tanh^2\paren{\frac{\mb q^*\paren{\mb w}\mb x}{\mu}}}\overline{\mb x}\; \overline{\mb x}^*} - \frac{1}{q_n^2\paren{\mb w}}\bb E\brac{\tanh\paren{\frac{\mb q^*\paren{\mb w}\mb x}{\mu}}q_n\paren{\mb w}x_n}\mb I\nonumber \\
&-\frac{1}{\mu q_n^2\paren{\mb w}}\bb E\brac{\paren{1-\tanh^2\paren{\frac{\mb q^*\paren{\mb w}\mb x}{\mu}}}q_n\paren{\mb w}x_n\paren{\mb w \overline{\mb x}^*+\overline{\mb x}\mb w^*}} \nonumber \\
& + \frac{1}{q_n^4\paren{\mb w}}\Brac{ \frac{1}{\mu}\bb E\brac{\paren{1-\tanh^2\paren{\frac{\mb q^*\paren{\mb w}\mb x}{\mu}}}(q_n\paren{\mb w}x_n)^2 } -\bb E\brac{\tanh\paren{\frac{\mb q^*\paren{\mb w}\mb x}{\mu}}q_n\paren{\mb w}x_n}  }\mb w\mb w^*  \label{eqn:thm-str-cvx-1}. 
\end{align}
We will first provide bounds for the last two lines and then tackle the first which is slightly more tricky. For the second line, we have 
\begin{align*}
& \frac{1}{\mu q_n^2\paren{\mb w}}\norm{ \bb E\brac{\paren{1-\tanh^2\paren{\frac{\mb q^*\paren{\mb w}\mb x}{\mu}}}q_n\paren{\mb w}x_n\paren{\mb w \overline{\mb x}^*+\overline{\mb x}\mb w^*}} }{} \nonumber \\
\le\; &  \frac{2}{\mu q_n^2\paren{\mb w}}\norm{ \bb E\brac{\paren{1-\tanh^2\paren{\frac{\mb q^*\paren{\mb w}\mb x}{\mu}}}q_n\paren{\mb w}x_n \bar{\mb x} } \mb w^*}{} \nonumber \\
\le\; &  \frac{2}{\mu q_n^2\paren{\mb w}}\norm{ \bb E\brac{\paren{1-\tanh^2\paren{\frac{\mb q^*\paren{\mb w}\mb x}{\mu}}}q_n\paren{\mb w}x_n \overline{\mb x} } }{} \norm{\mb w}{} \nonumber \\
\le\; &  \frac{2}{\mu q_n\paren{\mb w}}\theta^2 \bb E\brac{\abs{v_n}} \expect{\norm{\overline{\mb v}}{}} \norm{\mb w}{}  \nonumber \\
\le\; &  \frac{4\theta^2}{\pi\mu q_n\paren{\mb w}} \sqrt{n}\norm{\mb w}{} \le \frac{\theta}{\mu} \frac{4\theta \sqrt{n}\norm{\mb w}{}}{\pi \sqrt{1-\norm{\mb w}{}^2}} \le \frac{\theta}{\mu}\frac{1}{40\pi}, 
\end{align*}
where from the third to the fourth line we have used $\norm{1-\tanh^2\paren{\frac{\mb q^*\paren{\mb w} \mb x}{\mu}}}{} \le 1$, Jensen's inequality for the $\norm{\cdot}{}$ function, and independence of $x_n$ and $\overline{\mb x}$, and to obtain the last bound we have invoked the $\norm{\mb w}{} \le \frac{\mu}{4\sqrt{2}}$, $\mu \le \frac{1}{20\sqrt{n}}$, and $\theta < \frac{1}{2}$ assumptions. For the third line in \eqref{eqn:thm-str-cvx-1}, by Lemma \ref{lem:derivatives_basic_surrogate} and Lemma \ref{lem:aux_asymp_proof_a},
\begin{align*}
&\abs{\frac{1}{\mu}\bb E\brac{\paren{1-\tanh^2\paren{\frac{\mb q^*\paren{\mb w}\mb x}{\mu}}}(q_n\paren{\mb w}x_n)^2 } -\bb E\brac{ \tanh\paren{\frac{\mb q^*\paren{\mb w}\mb x}{\mu}}q_nx_n}} \nonumber \\
=\; & \left|\frac{\theta}{\mu} \bb E_{\mc J}\bb E_{\mb v}\brac{\paren{1-\tanh^2\paren{\frac{\mb w_{\mc J}^*\overline{\mb v}+q_n\paren{\mb w}v_n}{\mu}} \paren{q_n\paren{\mb w}v_n}^2} } \right. \nonumber \\
& \qquad \left. -  \theta \bb E_{\mc J}\bb E_{\mb v}\brac{\tanh\paren{\frac{\mb w_{\mc J}^*\overline{\mb v}+q_n\paren{\mb w}v_n}{\mu}}q_n\paren{\mb w}v_n} \right| \nonumber \\
=\; & \frac{\theta}{\mu}\bb E_{\mc J}\bb E_{\mb v}\brac{\paren{1-\tanh^2\paren{\frac{\mb w_{\mc J}^*\overline{\mb v}+q_n\paren{\mb ws}v_n}{\mu}}}\paren{(q_n\paren{\mb w}v_n)^2 + q_n^2\paren{\mb w}}}  \nonumber \\
\le \;& \frac{8\theta}{\mu}\bb E_{\mc J}\bb E_{\mb v}\brac{\exp\paren{-\frac{2}{\mu}\paren{\mb w_{\mc J}^*\overline{\mb v}+q_n\paren{\mb w}v_n}}\paren{(q_n\paren{\mb w}v_n)^2 + q_n^2\paren{\mb w}} \indicator{\mb w_{\mc J}^*\overline{\mb v}+q_n\paren{\mb w}v_n>0}} \nonumber \\
\le \;& \frac{8\theta}{\sqrt{2\pi}}\bb E_{\mc J}\brac{ \frac{q_n^2\paren{\mb w}}{\sqrt{q_n^2\paren{\mb w}+\norm{\mb w_{\mc J} }{}^2}}} \; \le \; \frac{8\theta q_n\paren{\mb w}}{\sqrt{2\pi}}. 
\end{align*}
Thus, we have
\begin{align*}
& \frac{1}{q_n^4\paren{\mb w}}\Brac{\frac{1}{\mu}\bb E\brac{\paren{1-\tanh^2\paren{\frac{\mb q^*\mb x}{\mu}}}(q_nx_n)^2 } -\bb E\brac{ \tanh\paren{\frac{\mb q^*\mb x}{\mu}}q_nx_n}}  \mb w\mb w^* \nonumber \\
\succeq\; & - \frac{8\theta }{q_n^3\paren{\mb w}\sqrt{2\pi}} \norm{\mb w}{}^2 \mb I \succeq -\frac{\theta}{\mu}\paren{\frac{64 n^{3/2}\mu \norm{\mb w}{}^2}{q_n^3\paren{\mb w} \sqrt{2\pi}}} \mb I \succeq -\frac{\theta}{\mu}\frac{1}{4000\sqrt{2\pi}} \mb I,  \label{eqn:thm-str-cvx-6}
\end{align*}
where we have again used $\norm{\mb w}{} \le \frac{\mu}{4\sqrt{2}}$, $\mu \le \frac{1}{20\sqrt{n}}$, and $q_n\paren{\mb w} \ge \frac{1}{2\sqrt{n}}$ assumptions to simplify the final bound.

To derive a lower bound for the first line of~\eqref{eqn:thm-str-cvx-1}, we lower bound the first term and upper bound the second. The latter is easy: using Lemma~\ref{lem:derivatives_basic_surrogate} and Lemma~\ref{lem:aux_asymp_proof_a}, 
\begin{align*}
& \frac{1}{q_n^2\paren{\mb w}}\bb E\brac{\tanh\paren{\frac{\mb q^*\paren{\mb w}\mb x}{\mu}}q_n\paren{\mb w}x_n} \nonumber \\
=\; & \frac{\theta}{\mu} \bb E_{\mc J} \bb E_{\mb v}\brac{1-\tanh^2\brac{\frac{\mb w^*_{\mc J} \overline{\mb v} + q_n\paren{\mb w} v_n}{\mu}}} \nonumber \\
\le\; & \frac{8\theta}{\mu} \bb E_{\mc J} \bb E_{\mb v} \brac{\exp\paren{-2\frac{\mb w^*_{\mc J} \overline{\mb v} + q_n\paren{\mb w} v_n}{\mu}} \indicator{\mb w^*_{\mc J} \overline{\mb v} + q_n\paren{\mb w} v_n > 0}} \nonumber \\
\le\; & \frac{4\theta}{\sqrt{2\pi} q_n\paren{\mb w}} \le \frac{\theta}{\mu} \frac{8\sqrt{n}\mu}{\sqrt{2\pi}} \le \frac{\theta}{\mu}\frac{2}{5\sqrt{2\pi}}, 
\end{align*}
where we have again used assumptions that $q_n\paren{\mb w} \ge \frac{1}{2\sqrt{n}}$ and $\mu \le \frac{1}{20\sqrt{n}}$ to simplify the last bound. To lower bound the first term, first note that 
\begin{align*}
\frac{1}{\mu}\bb E\brac{\paren{1-\tanh^2\paren{\frac{\mb q^*\paren{\mb w}\mb x}{\mu}}}\overline{\mb x}\; \overline{\mb x}^*} \succeq \frac{1-\theta}{\mu} \bb E_{\overline{\mb x}}\brac{\paren{1-\tanh^2\paren{\frac{\mb w^* \overline{\mb x}}{\mu}}}\overline{\mb x}\; \overline{\mb x}^*}. 
\end{align*}
We set out to lower bound the expectation as 
\begin{align*}
\bb E_{\overline{\mb x}}\brac{\paren{1-\tanh^2\paren{\frac{\mb w^* \overline{\mb x}}{\mu}}}\overline{\mb x}\; \overline{\mb x}^*} \succeq \theta \beta \mb I
\end{align*}
for some scalar $\beta > 0$. Suppose $\mb w$ has $k \in [n-1]$ nonzeros, w.l.o.g., further assume the first $k$ elements of $\mb w$ are these nonzeros. It is easy to see the expectation above has a block diagonal structure $\diag\paren{\bm \Sigma; \alpha  \theta \mb I_{n-1-k}}$, 
where 
\begin{align*}
\alpha \doteq \bb E_{\overline{\mb x}}\brac{\paren{1-\tanh^2\paren{\frac{\mb w^* \overline{\mb x}}{\mu}}}}. 
\end{align*}
So in order to derive the $\theta \beta \mb I$ lower bound as desired, it is sufficient to show $\bm \Sigma \succeq \theta \beta \mb I$ for some $0 < \beta < 1$, i.e., letting $\widetilde{\mb w} \in \R^k$ be the subvector of nonzero elements, 
\begin{align*}
\bb E_{\widetilde{\mb x} \sim_{i.i.d.} \mathrm{BG}\paren{\theta}}\brac{\paren{1-\tanh^2\paren{\frac{\widetilde{\mb w}^* \widetilde{\mb x}}{\mu}}}\widetilde{\mb x}\; \widetilde{\mb x}^*} \succeq \theta \beta \mb I, 
\end{align*} 
which is equivalent to that for all $\mb z \in \R^k$ such that $\norm{\mb z}{} = 1$, 
\begin{align*}
\bb E_{\widetilde{\mb x} \sim_{i.i.d.} \mathrm{BG}\paren{\theta}}\brac{\paren{1-\tanh^2\paren{\frac{\widetilde{\mb w}^* \widetilde{\mb x}}{\mu}}} \paren{\widetilde{\mb x}^* \mb z}^2} \ge \theta \beta. 
\end{align*}
It is then sufficient to show that for any nontrivial support set $\mc S \subset [k]$ and any vector $\mb z \in \R^k$ such that $\supp\paren{\mb z} = \mc S$ with $\norm{\mb z}{} = 1$, 
\begin{align*}
\bb E_{\widetilde{\mb v} \sim_{i.i.d.} \mc N\paren{0, 1}}\brac{\paren{1-\tanh^2\paren{\frac{\widetilde{\mb w}^*_{\mc S} \widetilde{\mb v}}{\mu}}} \paren{\widetilde{\mb v}^* \mb z}^2} \ge \beta. 
\end{align*}
To see the implication, suppose the latter claimed holds, then for any $\mb z$ with unit norm, 
\begin{align*}
& \bb E_{\widetilde{\mb x} \sim_{i.i.d.} \mathrm{BG}\paren{\theta}}\brac{\paren{1-\tanh^2\paren{\frac{\widetilde{\mb w}^* \widetilde{\mb x}}{\mu}}} \paren{\widetilde{\mb x}^* \mb z}^2}  \\
=\; & \sum_{s = 1}^k \theta^s \paren{1-\theta}^{k-s} \sum_{\mc S \in \binom{[k]}{s}} \bb E_{\widetilde{\mb v} \sim_{i.i.d.} \mc N\paren{0, 1}}\brac{\paren{1-\tanh^2\paren{\frac{\widetilde{\mb w}^*_{\mc S} \widetilde{\mb v}}{\mu}}} \paren{\widetilde{\mb v}^* \mb z_{\mc S}}^2} \\
\ge\; & \sum_{s = 1}^k \theta^s \paren{1-\theta}^{k-s} \sum_{\mc S \in \binom{[k]}{s}} \beta \norm{\mb z_{\mc S}}{}^2 = \beta \bb E_{\mc S}\brac{\norm{\mb z_{\mb S}}{}^2} = \theta \beta. 
\end{align*}
Now for any fixed support set $\mc S \subset [k]$, $\mb z = \mc P_{\widetilde{\mb w}_{\mc S}} \mb z + \paren{\mb I - \mc P_{\widetilde{\mb w}_{\mc S}}} \mb z$. So we have
\begin{align*}
& \bb E_{\widetilde{\mb v} \sim_{i.i.d.} \mc N\paren{0, 1}}\brac{\paren{1-\tanh^2\paren{\frac{\widetilde{\mb w}^*_{\mc S} \widetilde{\mb v}}{\mu}}} \paren{\widetilde{\mb v}^* \mb z}^2} \\
=\; & \bb E_{\widetilde{\mb v}}\brac{\paren{1-\tanh^2\paren{\frac{\widetilde{\mb w}^*_{\mc S} \widetilde{\mb v}}{\mu}}} \paren{\widetilde{\mb v}^* \mc P_{\widetilde{\mb w}_{\mc S}} \mb z}^2} + \bb E_{\widetilde{\mb v}}\brac{\paren{1-\tanh^2\paren{\frac{\widetilde{\mb w}^*_{\mc S} \widetilde{\mb v}}{\mu}}} \paren{\widetilde{\mb v}^* \paren{\mb I - \mc P_{\widetilde{\mb w}_{\mc S}}} \mb z}^2} \\
=\; & \frac{\paren{{\widetilde{\mb w}_{\mc S}}^* \mb z}^2}{\norm{\mb w_{\mc S}}{}^4} \bb E_{\widetilde{\mb v}}\brac{\paren{1-\tanh^2\paren{\frac{\widetilde{\mb w}^*_{\mc S} \widetilde{\mb v}}{\mu}}} \paren{\widetilde{\mb v}^* \widetilde{\mb w}_{\mc S}}^2} + \bb E_{\widetilde{\mb v}}\brac{\paren{1-\tanh^2\paren{\frac{\widetilde{\mb w}^*_{\mc S} \widetilde{\mb v}}{\mu}}}} \bb E_{\widetilde{\mb v}}\brac{\paren{\widetilde{\mb v}^* \paren{\mb I - \mc P_{\widetilde{\mb w}_{\mc S}}} \mb z}^2} \\
\ge\; & 2\frac{\paren{{\widetilde{\mb w}_{\mc S}}^* \mb z}^2}{\norm{\mb w_{\mc S}}{}^4} \bb E_{\widetilde{\mb v}}\brac{\exp\paren{-\frac{2\widetilde{\mb w}^*_{\mc S} \widetilde{\mb v}}{\mu}} \paren{\widetilde{\mb v}^* \widetilde{\mb w}_{\mc S}}^2 \indicator{\widetilde{\mb v}^* \widetilde{\mb w}_{\mc S} > 0}} 
+ 2\bb E_{\widetilde{\mb v}}\brac{\exp\paren{-\frac{2\widetilde{\mb w}^*_{\mc S} \widetilde{\mb v}}{\mu}} \indicator{\widetilde{\mb w}^*_{\mc S} \widetilde{\mb v} > 0}} \norm{\paren{\mb I - \mc P_{\widetilde{\mb w}_{\mc S}}} \mb z}{}^2. 
\end{align*}
Using expectation result from Lemma~\ref{lem:aux_asymp_proof_a}, and applying Type III lower bound for Gaussian tails, we obtain 
\begin{align*}
& \bb E_{\widetilde{\mb v} \sim_{i.i.d.} \mc N\paren{0, 1}}\brac{\paren{1-\tanh^2\paren{\frac{\widetilde{\mb w}^*_{\mc S} \widetilde{\mb v}}{\mu}}} \paren{\widetilde{\mb v}^* \mb z}^2} \nonumber \\
\ge\; & \frac{1}{\sqrt{2\pi}}\paren{\sqrt{4+\frac{4\norm{\widetilde{\mb w}_{\mc S}}{}^2}{\mu^2}}- \frac{2\norm{\widetilde{\mb w}_{\mc S}}{}}{\mu}} - \frac{4\paren{{\widetilde{\mb w}_{\mc S}}^* \mb z}^2}{\mu\sqrt{2\pi}\norm{\widetilde{\mb w}_{\mc S}}{}} \nonumber \\
\ge\; & \frac{1}{\sqrt{2\pi}} \paren{2 - \frac{3}{4} \sqrt{2}}, 
\end{align*}
where we have used Cauchy-Schwarz to obtain $\paren{\widetilde{\mb v}^* \mb z}^2 \le \norm{\widetilde{\mb v}^*}{}^2$ and invoked the assumption $\norm{\mb w}{} \le \frac{\mu}{4\sqrt{2}}$ to simplify the last bound. On the other hand, we similarly obtain 
\begin{align*}
\alpha = \bb E_{\mc J} \bb E_{Z \sim \mc N(0, \|\mb w_{\mc J}\|^2)} [1-\tanh^2 (Z/\mu)] \ge \frac{2}{\sqrt{2\pi}} \frac{\sqrt{4 \|\mb w\|^2/\mu^2 + 4} - 2\|\mb w\|/\mu}{2} \ge \frac{1}{\sqrt{2\pi}} \paren{2 - \frac{1}{2} \sqrt{2}}. 
\end{align*}
So we can take $\beta = \frac{1}{\sqrt{2\pi}} \paren{2 - \frac{3}{4} \sqrt{2}} < 1$. 

Putting together the above estimates for the case $\mb w \neq \mb 0$, we obtain 
\begin{align*}
\expect{\nabla^2_{\mb w} h_{\mu}\paren{\mb q^*\paren{\mb w} \mb x}} \succeq \frac{\theta}{\mu\sqrt{2\pi}} \paren{1-\frac{3}{8}\sqrt{2} - \frac{\sqrt{2\pi}}{40\pi} - \frac{1}{4000} - \frac{2}{5}} \mb I \succeq \frac{1}{25\sqrt{2\pi}} \frac{\theta}{\mu} \mb I. 
\end{align*}
Hence for all $\mb w$, we can take the $\frac{1}{25\sqrt{2\pi}} \frac{\theta}{\mu}$ as the lower bound, completing the proof. 
\end{proof}

%% file: sec/proof_finite_concentration.tex
\subsubsection{Proof of Pointwise Concentration Results} \label{proof:cn_point}
To avoid clutter of notations, in this subsection we write $\mb X$ to mean $\mb X_0$; similarly $\mb x_k$ for $\paren{\mb x_0}_k$, the $k$-th column of $\mb X_0$. The function $g\paren{\mb w}$ means $g\paren{\mb w; \mb X_0}$. We first establish a useful comparison lemma between random i.i.d. Bernoulli random vectors random i.i.d. normal random vectors. 

\begin{lemma}\label{lem:U-moments-bound}
Suppose $\mb z, \mb z' \in \R^n$ are independent and obey $\mb z \sim_{i.i.d.} \mathrm{BG}\paren{\theta}$ and $\mb z' \sim_{i.i.d.} \mc N\paren{0, 1}$. Then, for any fixed vector $\mb v \in \R^n$, it holds that 
\begin{align*}
\expect{\abs{\mb v^* \mb z}^m} & \le \expect{\abs{\mb v^* \mb z'}^m } = \bb E_{Z \sim \mc N\paren{0, \norm{\mb v}{}^2}}\brac{\abs{Z}^m}, \\
\expect{\norm{\mb z}{}^m} & \le \expect{\norm{\mb z'}{}^m}, 
\end{align*}
for all integers $m \ge 1$. 
\end{lemma}

Now, we are ready to prove Proposition~\ref{prop:concentration-hessian-negative} to Proposition~\ref{prop:concentration-hessian-zero} as follows.

\begin{proof}[of Proposition \ref{prop:concentration-hessian-negative}] \label{proof:pt_cn_curvature}
Let 
\begin{align*}
Y_k = \frac{1}{\norm{\mb w}{}^2}\mb w^*\nabla^2 h_\mu\paren{\mb q(\mb w)^*\mb x_k}\mb w, 
\end{align*}
then $\frac{\mb w^*\nabla^2 g(\mb w)\mb w}{\norm{\mb w}{}^2} = \frac{1}{p} \sum_{k=1}^p Y_k$. For each $Y_k$ ($k \in [p]$), from \eqref{eqn:lse-hessian}, we know that
\begin{align*}
Y_k \;&=\; \frac{1}{\mu}\paren{1-\tanh^2\paren{\frac{\mb q(\mb w)^*\mb x_k}{\mu}}}
\paren{\frac{\mb w^*\overline{\mb x}_k}{\norm{\mb w}{}}- \frac{x_k\paren{n}\norm{\mb w}{}}{q_n(\mb w)}}^2 - \tanh\paren{\frac{\mb q(\mb w)^*\mb x_k}{\mu}}\frac{x_k\paren{n}}{q_n^3(\mb w)}. 
\end{align*}
Writing $Y_k = W_k + V_k$, where 
\begin{align*}
W_k & =  \frac{1}{\mu}\paren{1-\tanh^2\paren{\frac{\mb q(\mb w)^*\mb x_k}{\mu}}}
\paren{\frac{\mb w^*\overline{\mb x}_k}{\norm{\mb w}{}}- \frac{x_k\paren{n}\norm{\mb w}{}}{q_n(\mb w)}}^2, \\
V_k & = - \tanh\paren{\frac{\mb q(\mb w)^*\mb x_k}{\mu}}\frac{x_k\paren{n}}{q_n^3(\mb w)}. 
\end{align*}
Then by similar argument as in proof to Proposition~\ref{prop:concentration-gradient}, we have for all integers $m \ge 2$ that 
\begin{align*}
\expect{\abs{W_k}^m} 
& \le \frac{1}{\mu^m} \bb E\brac{\abs{\frac{\mb w^*\overline{\mb x}_k}{\norm{\mb w}{}}- \frac{x_k\paren{n}\norm{\mb w}{}}{q_n(\mb w)}}^{2m}} \le  \frac{1}{\mu^m} \bb E_{Z \sim \mc N\paren{0, 1/q_n^2\paren{\mb w}}}\brac{\abs{Z}^{2m}} \\
& \le \frac{1}{\mu^m}(2m-1)!!(4n)^m \le  \frac{m!}{2} \paren{\frac{4n}{\mu}}^m,  \\
\bb E\brac{\abs{V_k}^m} 
& \le \frac{1}{q_n^{3m}(\mb w)}\bb E\brac{\abs{v_k\paren{n}}^m} \le \paren{2\sqrt{n}}^{3m} (m-1)!! \le \frac{m!}{2} \paren{8n\sqrt{n}}^m, 
\end{align*}
where we have again used the assumption that $q_n\paren{\mb w} \ge \frac{1}{2\sqrt{n}}$ to simplify the result. Taking $\sigma_W^2 = 16n^2/\mu^2 \geq \bb E\brac{W_k^2}$, $R_W = 4n/\mu$ and $\sigma_V^2 = 64n^3\geq \bb E\brac{V_k^2}$, $R_V = 8n\sqrt{n}$, and considering $S_W = \frac{1}{p}\sum_{k=1}^p W_k$ and $S_V = \frac{1}{p}\sum_{k=1}^p V_k$, then by Lemma \ref{lem:mc_bernstein_scalar}, we obtain 
\begin{align*}
\bb P\brac{\abs{S_W - \bb E\brac{S_W}}\geq \frac{t}{2} } \;&\leq\; 2\exp\paren{- \frac{p\mu^2t^2}{128n^2 + 16n\mu t}}, \\
\bb P\brac{\abs{S_V - \bb E\brac{S_V}}\geq \frac{t}{2} } \;&\leq\; 2\exp\paren{- \frac{pt^2}{512n^3+32n\sqrt{n} t}}. 
\end{align*}
Combining the above results, we obtain
\begin{align*}
\bb P\brac{\abs{\frac{1}{p} \sum_{k=1}^p X_k - \expect{X_k}} \geq t} \;&=\; \bb P\brac{\abs{S_W - \bb E\brac{S_W}+ S_V - \bb E\brac{S_V} }\geq t } \nonumber \\
\; &\leq \; \bb P\brac{\abs{S_W - \bb E\brac{S_W}}\geq \frac{t}{2} }+ \bb P\brac{\abs{S_V - \bb E\brac{S_V}}\geq \frac{t}{2} } \nonumber \\
\;&\leq \; 2\exp\paren{- \frac{p\mu^2t^2}{128n^2+16n\mu t}}+2\exp\paren{- \frac{pt^2}{512n^3+32n\sqrt{n} t}} \nonumber \\
\;&\leq\; 4\exp\paren{- \frac{p\mu^2t^2}{512n^2+32n\mu t}}, 
\end{align*}
provided that $\mu \leq \frac{1}{\sqrt{n}}$, as desired. 
\end{proof}

\begin{proof}[of Proposition~\ref{prop:concentration-gradient} ] \label{proof:pt_cn_gradient}
Let 
\begin{align*}
X_k = \frac{\mb w^*}{\norm{\mb w}{2}}\nabla h_\mu\paren{\mb q(\mb w)^*\mb x_k}, 
\end{align*}
then $\frac{\mb w^*\nabla g(\mb w)}{\norm{\mb w}{2}}= \frac{1}{p} \sum_{k=1}^p X_k$. 
For each $X_k, k \in [p]$, from \eqref{eqn:lse-gradient}, we know that 
\begin{align*}
\abs{X_k} = \abs{\tanh\paren{\frac{\mb q(\mb w)^*\mb x_k}{\mu } }\paren{\frac{\mb w^*\overline{\mb x}_k}{\norm{\mb w}{}} - \frac{\norm{\mb w}{2}x_k\paren{n}}{q_n\paren{\mb w}}}} \le \abs{\frac{\mb w^*\overline{\mb x}_k}{\norm{\mb w}{}} - \frac{\norm{\mb w}{2}x_k\paren{n}}{q_n\paren{\mb w}}}, 
\end{align*}
as the magnitude of $\tanh\paren{\cdot}$ is bounded by one. Because $\frac{\mb w^*\overline{\mb x}_k}{\norm{\mb w}{2}} - \frac{\norm{\mb w}{}x_k\paren{n}}{q_n\paren{\mb w}} = \paren{\frac{\mb w}{\norm{\mb w}{}}, - \frac{\norm{\mb w}{}}{q_n\paren{\mb w}}}^* \mb x_k$ and $\mb x_k \sim_{i.i.d.} \mathrm{BG}\paren{\theta}$, invoking Lemma~\ref{lem:U-moments-bound}, we obtain for every integer $m \ge 2$ that
\begin{align*}
\expect{\abs{X_k}^m} \le \bb E_{Z \sim \mc N\paren{0, 1/q_n^2\paren{\mb w}}}\brac{\abs{Z}^m} \le \frac{1}{q_n\paren{\mb w}^m} (m-1)!!\; \le\; \frac{m!}{2} \paren{4n} \paren{2\sqrt{n}}^{m-2}, 
\end{align*}
where the Gaussian moment can be looked up in Lemma~\ref{lem:guassian_moment} and we used the fact that $(m-1)!! \leq m!/2$ and the assumption that $q_n\paren{\mb w} \ge \frac{1}{2\sqrt{n}}$ to get the result. Thus, by taking $\sigma^2 = 4n\geq \bb E\brac{X_k^2}$ and $R = 2\sqrt{n}$, and we obtain the claimed result by invoking Lemma \ref{lem:mc_bernstein_scalar}.
\end{proof} 

\begin{proof}[of Proposition \ref{prop:concentration-hessian-zero}] \label{proof:pt_cn_strcvx}
Let $\mb Z_k =\nabla^2_{\mb w} h_\mu\paren{\mb q(\mb w)^*\mb x_k}$, then $\nabla^2_{\mb w} g\paren{\mb w} = \frac{1}{p} \sum_{k=1}^p \mb Z_k$. From \eqref{eqn:lse-hessian}, we know that
\begin{align*}
\mb Z_k \;=\; \mb W_k + \mb V_k
\end{align*}
where 
\begin{align*}
\mb W_k \;&=\; \frac{1}{\mu} \paren{1 - \tanh^2\paren{\frac{\mb q(\mb w)^*\mb x_k}{\mu }}}\paren{\overline{\mb x}_k - \frac{x_k\paren{n}\mb w}{q_n(\mb w)}}\paren{\overline{\mb x}_k - \frac{x_k\paren{n}\mb w}{q_n(\mb w)}}^* \\
\mb V_k \;&=\; -\tanh\paren{\frac{\mb q(\mb w)^*\mb x_k}{\mu}}\paren{\frac{x_k\paren{n}}{q_n(\mb w)}\mb I + \frac{x_k\paren{n}\mb w\mb w^*}{q_n^3(\mb w)}}. 
\end{align*}
For $\mb W_k$, we have
\begin{align*}
\mb 0 \preceq \bb E\brac{\mb W_k^m} 
&\preceq \frac{1}{\mu^m} \bb E\brac{ \norm{\overline{\mb x}_k - \frac{x_k\paren{n}\mb w}{q_n(\mb w)}}{}^{2m-2} \paren{\overline{\mb x}_k - \frac{x_k\paren{n}\mb w}{q_n(\mb w)}}\paren{\overline{\mb x}_k - \frac{x_k\paren{n}\mb w}{q_n(\mb w)}}^* } \nonumber \\
& \preceq \frac{1}{\mu^m} \bb E\brac{ \norm{\overline{\mb x}_k - \frac{x_k\paren{n}\mb w}{q_n(\mb w)}}{}^{2m}} \mb I \nonumber \\
\;&\preceq\; \frac{2^m}{\mu^m} \bb E\brac{ \paren{\norm{\overline{\mb x}_k}{}^2 + \frac{x^2_k\paren{n}\norm{\mb w}{}^2}{q^2_n(\mb w)}}^{m} } \mb I \nonumber \\
\;&\preceq\;\frac{2^m}{\mu^m} \bb E\brac{ \norm{\mb x_k}{}^{2m} } \mb I
\;\preceq\; \frac{2^m}{\mu^m}\bb E_{Z \sim \chi^2\paren{n}}\brac{Z^m}\mb I,
\end{align*}
where we have used the fact that $\norm{\mb w}{}^2/q_n^2(\mb w) = \norm{\mb w}{}^2/(1-\norm{\mb w}{}^2)\leq 1 $ for $\norm{\mb w}{2}\leq \frac{1}{4}$ and Lemma~\ref{lem:U-moments-bound} to obtain the last line. By Lemma \ref{lem:chi_sq_moment}, we obtain 
\begin{align*}
\mb 0 \preceq \bb E\brac{\mb W_k^m} \;\preceq\; \paren{\frac{2}{\mu}}^m \frac{m!}{2} \paren{2n}^m \mb I \;=\; \frac{m!}{2}\paren{\frac{4n}{\mu}}^{m} \mb I. 
\end{align*}
Taking $R_W = \frac{4n}{\mu}$ and $\mb \sigma_W^2 =\frac{16n^2}{\mu^2} \ge \expect{\mb W_k^2}$, and letting $\mb S_{W} \doteq \frac{1}{p} \sum_{k=1}^p \mb W_k$, by Lemma~\ref{lem:mc_bernstein_matrix}, we obtain 
\begin{align*}
\bb P\brac{\norm{\mb S_W - \bb E\brac{\mb S_W}}{}\geq \frac{t}{2}}\;&\leq\; 2n\exp\paren{-\frac{p\mu^2t^2}{128n^2+16\mu n t}}.
\end{align*}
Similarly, for $\mb V_k$, we have
\begin{align*}
\bb E\brac{\mb V_k^m} \;&\preceq\; \paren{\frac{1}{q_n(\mb w)} + \frac{\norm{\mb w}{}^2}{q_n^3(\mb w)}}^m\bb E\brac{ \abs{x_k\paren{n}}^m }\mb I \nonumber \\
\;&\preceq\;  \paren{8n\sqrt{n}}^m \paren{m-1}!! \mb I \nonumber \\
\;&\preceq\; \frac{m!}{2}\paren{8n\sqrt{n}}^{m} \mb I,  
\end{align*}
where we have used the fact $q_n\paren{\mb w} \ge \frac{1}{2\sqrt{n}}$ to simplify the result. Similar argument also shows $-\bb E\brac{\mb V_k^m} \preceq m!\paren{8n\sqrt{n}}^m \mb I /2$. Taking $R_V = 8 n\sqrt{n}$ and $\mb \sigma_V^2 = 64n^3$, and letting $\mb S_V \doteq \frac{1}{p} \sum_{k=1}^p \mb V_k$, again by Lemma \ref{lem:mc_bernstein_matrix}, we obtain 
\begin{align*}
\bb P\brac{\norm{\mb S_V - \bb E\brac{\mb S_V}}{}\geq \frac{t}{2}}\;&\leq\; 2n\exp\paren{- \frac{pt^2}{512n^3+32n\sqrt{n}t}}.
\end{align*}
Combining the above results, we obtain 
\begin{align*}
\bb P\brac{\norm{\frac{1}{p}\sum_{k=1}^p \mb Z_k - \expect{\mb Z_k} }{}\geq t} \;&=\; \bb P\brac{\norm{\mb S_W -\bb E\brac{\mb S_W}+\mb S_V -\bb E\brac{\mb S_V} }{}\geq t} \nonumber \\
\;&\leq \; \bb P\brac{\norm{\mb S_W -\bb E\brac{\mb S_W} }{}\geq \frac{t}{2}} + \bb P\brac{\norm{\mb S_V -\bb E\brac{\mb S_V} }{}\geq \frac{t}{2}} \nonumber \\
\;&\leq \; 2n\exp\paren{-\frac{p\mu^2t^2}{128n^2+16\mu n t}}+ 2n\exp\paren{- \frac{pt^2}{512n^3+32n\sqrt{n}t}} \nonumber \\
\;&\leq \; 4n\exp\paren{-\frac{p\mu^2t^2}{512n^2+32\mu n t}},
\end{align*}
where we have simplified the final result based on the fact that $\mu \leq \frac{1}{\sqrt{n}}$.
\end{proof}

%% file: sec/proof_finite_lipschitz.tex
\subsubsection{Proof of Lipschitz Results} \label{proof:cn_lips}
To avoid clutter of notations, in this subsection we write $\mb X$ to mean $\mb X_0$; similarly $\mb x_k$ for $\paren{\mb x_0}_k$, the $k$-th column of $\mb X_0$. The function $g\paren{\mb w}$ means $g\paren{\mb w; \mb X_0}$. We need the following lemmas to prove the Lipschitz results.

\begin{lemma}\label{lem:composition} Suppose that $\varphi_1 : U \to V$ is an $L$-Lipschitz map from a normed space $U$ to a normed space $V$, and that $\varphi_2 : V \to W$ is an $L'$-Lipschitz map from $V$ to a normed space $W$. Then the composition $\varphi_2 \circ \varphi_1 : U \to W$ is $LL'$-Lipschitz. 
\end{lemma}

\begin{lemma}\label{lem:Lip-combined} Fix any $\mc D \subseteq \reals^{n-1}$. Let $g_1, g_2 : \mc D \to \reals$, and assume that $g_1$ is $L_1$-Lipschitz, and $g_2$ is $L_2$-Lipschitz, and that $g_1$ and $g_2$ are bounded over $\mc D$, i.e., $\abs{g_1(\mb x)}\le M_1$ and $\abs{g_2(\mb x)} \le M_2$ for all $x\in \mc D$ with some constants $M_1>0$ and $M_2>0$. Then the function $h(\mb x) = g_1(\mb x) g_2(\mb x)$ is $L$-Lipschitz, with 
\begin{align*}
L \;=\; M_1 L_2 + M_2 L_1. 
\end{align*} 
\end{lemma}

\begin{lemma}\label{lem:lip-h-mu}
For every $\mb w, \mb w^\prime \in \Gamma$, and every fixed $\mb x$, we have
\begin{align*}
\abs{\dot{h}_\mu\paren{\mb q(\mb w)^*\mb x} -\dot{h}_\mu\paren{\mb q(\mb w^\prime)^*\mb x} }\;&\le\; \frac{2\sqrt{n}}{\mu}\norm{\mb x}{} \norm{\mb w-\mb w^\prime}{}, \\
\abs{\ddot{h}_\mu\paren{\mb q(\mb w)^*\mb x} -\ddot{h}_\mu\paren{\mb q(\mb w^\prime)^*\mb x} }\;&\le\; \frac{4\sqrt{n}}{\mu^2}\norm{\mb x}{} \norm{\mb w-\mb w^\prime}{}.
\end{align*}
\end{lemma}
\begin{proof}
We have 
\begin{align*}
\abs{q_n\paren{\mb w} - q_n\paren{\mb w'}} 
& =  \abs{\sqrt{1-\norm{\mb w}{}^2} - \sqrt{1-\norm{\mb w'}{}^2}}
  = \frac{\norm{\mb w + \mb w'}{} \norm{\mb w - \mb w'}{}}{\sqrt{1-\norm{\mb w}{}^2} + \sqrt{1-\norm{\mb w'}{}^2}} \\
& \le \frac{\max\paren{\norm{\mb w}{}, \norm{\mb w'}{}}}{\min\paren{q_n\paren{\mb w}, q_n\paren{\mb w'}}} \norm{\mb w - \mb w'}{}. 
\end{align*}
Hence it holds that 
\begin{align*}
\norm{\mb q\paren{\mb w} - \mb q\paren{\mb w'}}{}^2 
& = \norm{\mb w - \mb w'}{}^2 + \abs{q_n\paren{\mb w} - q_n\paren{\mb w'}}^2 \le \paren{1+\frac{\max\paren{\norm{\mb w}{}^2, \norm{\mb w'}{}^2}}{\min\paren{q_n^2\paren{\mb w}, q_n^2\paren{\mb w'}}}} \norm{\mb w - \mb w'}{}^2 \\
& = \frac{1}{\min\paren{q_n^2\paren{\mb w}, q_n^2\paren{\mb w'}}} \norm{\mb w - \mb w'}{}^2 \le 4n \norm{\mb w - \mb w'}{}^2, 
\end{align*}
where we have used the fact $q_n\paren{\mb w} \ge \frac{1}{2\sqrt{n}}$ to get the final result. Hence the mapping $\mb w \mapsto \mb q(\mb w)$ is $2\sqrt{n}$-Lipschitz over $\Gamma$. Moreover it is easy to see $\mb q \mapsto \mb q^*\mb x$ is $\norm{\mb x}{2}$-Lipschitz. By Lemma \ref{lem:derivatives_basic_surrogate} and the composition rule in Lemma \ref{lem:composition}, we obtain the desired claims.
\end{proof}

\begin{lemma}\label{lem:lip-g} 
For any fixed $\mb x$, consider the function
\begin{align*}
t_{\mb x}(\mb w) \doteq \frac{\mb w^*\overline{\mb x}}{\norm{\mb w}{}} - \frac{x_n}{q_n(\mb w)}\norm{\mb w}{}
\end{align*}
defined over $\mb w \in \Gamma$. Then, for all $\mb w, \mb w'$ in $\Gamma$ such that $\norm{\mb w}{} \ge r$ and $\norm{\mb w'}{} \ge r$ for some constant $r \in \paren{0, 1}$, it holds that 
\begin{align*}
\abs{ t_{\mb x}(\mb w)- t_{\mb x}(\mb w^\prime)}\;&\le \; 2\paren{\frac{\norm{\mb x}{}}{r}+ 4n^{3/2} \norm{\mb x}{\infty}} \norm{\mb w-\mb w^\prime}{} , \\
 \abs{ t_{\mb x}(\mb w) }\; &\le \; 2\sqrt{n}\norm{\mb x}{}, \\
\abs{ t^2_{\mb x}(\mb w)- t^2_{\mb x}(\mb w^\prime)}\;&\le \;8\sqrt{n} \norm{\mb x}{} \paren{\frac{\norm{\mb x}{}}{r}+ 4n^{3/2} \norm{\mb x}{\infty}} \norm{\mb w-\mb w^\prime}{},  \\
 \abs{ t^2_{\mb x}(\mb w) }\; &\le \;  4n\norm{\mb x}{}^2. 
\end{align*}
\end{lemma}

\begin{proof}
First of all, we have 
\begin{align*}
\abs{ t_{\mb x}(\mb w) }\; =\; \brac{\frac{\mb w^*}{\norm{\mb w}{}},-\frac{\norm{\mb w}{}}{q_n(\mb w)}}\mb x \;\le\; \norm{\mb x}{}\paren{1+ \frac{\norm{\mb w}{}^2}{q_n^2(\mb w)}}^{1/2} =\frac{\norm{\mb x}{}}{\abs{q_n(\mb w)}}\le 2\sqrt{n}\norm{\mb x}{}, 
\end{align*}
where we have used the assumption that $q_n\paren{\mb w} \ge \frac{1}{2\sqrt{n}}$ to simplify the final result. The claim about $\abs{t_{\mb x}^2\paren{\mb w}}$ follows immediately. Now 
\begin{align*}
\abs{ t_{\mb x}(\mb w)- t_{\mb x}(\mb w^\prime)} \le \abs{\paren{\frac{\mb w}{\norm{\mb w}{}} -\frac{\mb w^\prime}{\norm{\mb w^\prime}{}}}^*\overline{\mb x}  } + \abs{x_n}\abs{ \frac{\norm{\mb w}{}}{q_n(\mb w)}- \frac{\norm{\mb w^\prime}{}}{q_n(\mb w^\prime)} }. 
\end{align*}
Moreover we have 
\begin{align*}
\abs{\paren{\frac{\mb w}{\norm{\mb w}{}} -\frac{\mb w^\prime}{\norm{\mb w^\prime}{}}}^*\overline{\mb x}  } 
& \le \norm{\overline{\mb x}}{}\norm{\frac{\mb w}{\norm{\mb w}{}} -\frac{\mb w^\prime}{\norm{\mb w^\prime}{}}}{} \le \norm{\mb x}{} \frac{\norm{\mb w - \mb w'}{} \norm{\mb w'}{} + \norm{\mb w'}{} \abs{\norm{\mb w}{} - \norm{\mb w'}{}}  }{\norm{\mb w}{} \norm{\mb w'}{}} \\
& \le \frac{2\norm{\mb x}{}}{r} \norm{\mb w - \mb w'}{},  
\end{align*}
where we have used the assumption that $\norm{\mb w}{} \ge r$ to simplify the result. Noticing that $t \mapsto t/\sqrt{1-t^2}$ is continuous over $\brac{a, b}$ and differentiable over $\paren{a, b}$ for any $0 < a < b < 1$, by mean value theorem, 
\begin{align*}
\abs{ \frac{\norm{\mb w}{}}{q_n(\mb w)}- \frac{\norm{\mb w^\prime}{}}{q_n(\mb w^\prime)} } \;\le\; \sup_{\mb w\; \in\; \Gamma } \frac{1}{\paren{1-\norm{\mb w}{}^2}^{3/2}} \norm{\mb w-\mb w^\prime}{} \; \le \; 8n^{3/2} \norm{\mb w-\mb w^\prime}{}, 
\end{align*}
where we have again used the assumption that $q_n\paren{\mb w} \ge \frac{1}{2\sqrt{n}}$ to simplify the last result. Collecting the above estimates, we obtain 
\begin{align*}
\abs{ t_{\mb x}(\mb w)- t_{\mb x}(\mb w^\prime)} \le \paren{2\frac{\norm{\mb x}{}}{r}+ 8n^{3/2} \norm{\mb x}{\infty}} \norm{\mb w-\mb w^\prime}{}, 
\end{align*}
as desired. For the last one, we have 
\begin{align*}
\abs{ t^2_{\mb x}(\mb w)- t^2_{\mb x}(\mb w^\prime)} \;&=\; \abs{ t_{\mb x}(\mb w)- t_{\mb x}(\mb w^\prime)} \abs{ t_{\mb x}(\mb w)+t_{\mb x}(\mb w^\prime)} \nonumber \\
\;&\le \; 2\sup_{\mb s\; \in\; \Gamma} \abs{t_{\mb x}(\mb s)}\abs{ t_{\mb x}(\mb w)- t_{\mb x}(\mb w^\prime)}, 
\end{align*}
leading to the claimed result once we substitute estimates of the involved quantities. 
\end{proof}

\begin{lemma}\label{lem:lp-Phi}
For any fixed $\mb x$, consider the function 
\begin{align*}
\mb \Phi_{\mb x}(\mb w) = \frac{x_n}{q_n(\mb w)}\mb I + \frac{x_n}{q_n^3(\mb w)}\mb w\mb w^*
\end{align*}
defined over $\mb w \in \Gamma$. Then, for all $\mb w, \mb w' \in \Gamma$ such that $\norm{\mb w}{} < r$ and $\norm{\mb w'}{} < r$ with some constant $r \in \paren{0, \frac{1}{2}}$, it holds that 
\begin{align*}
 \norm{\mb \Phi_{\mb x}(\mb w) }{} \;&\le\;2\norm{\mb x}{\infty},  \\
\norm{\mb \Phi_{\mb x}(\mb w) - \mb \Phi_{\mb x}(\mb w^\prime)}{}\;&\le\;  4\norm{\mb x}{\infty}\norm{\mb w - \mb w^\prime}{}.  
\end{align*}
\end{lemma}
\begin{proof}
Simple calculation shows 
\begin{align*}
\norm{\mb \Phi_{\mb x}(\mb w) }{} \le \norm{\mb x}{\infty} \paren{ \frac{1}{q_n(\mb w)}+ \frac{\norm{\mb w}{}^2}{q_n^3(\mb w)} } = \frac{\norm{\mb x}{\infty}}{q_n^3(\mb w)}\le \frac{\norm{\mb x}{\infty}}{(1-r^2)^{3/2}} \le 2\norm{\mb x}{\infty}. 
\end{align*}
For the second one, we have 
\begin{align*}
\norm{\mb \Phi_{\mb x}(\mb w)-\mb \Phi_{\mb x}(\mb w^\prime)}{}
&\le \norm{\mb x}{\infty} \norm{ \frac{1}{q_n(\mb w)}\mb I + \frac{1}{q_n^3(\mb w)}\mb w\mb w^* -  \frac{1}{q_n(\mb w^\prime)}\mb I - \frac{1}{q_n^3(\mb w^\prime)}\mb w^\prime(\mb w^\prime)^*}{} \\
&\le \norm{\mb x}{\infty} \paren{\abs{\frac{1}{q_n\paren{\mb w}} - \frac{1}{q_n\paren{\mb w'}}} + \abs{\frac{\norm{\mb w}{}^2}{q_n^3\paren{\mb w}} - \frac{\norm{\mb w'}{}^2}{q_n^3\paren{\mb w'}}}}. 
\end{align*}
Now 
\begin{align*}
\abs{\frac{1}{q_n\paren{\mb w}} - \frac{1}{q_n\paren{\mb w'}}} = \frac{\abs{q_n\paren{\mb w} - q_n\paren{\mb w'}}}{q_n\paren{\mb w} q_n\paren{\mb w'}} \le \frac{\max\paren{\norm{\mb w}{}, \norm{\mb w'}{}}}{\min\paren{q_n^3\paren{\mb w}, q_n^3\paren{\mb w'}}} \norm{\mb w - \mb w'}{} \le \frac{4}{3\sqrt{3}} \norm{\mb w - \mb w'}{}, 
\end{align*}
where we have applied the estimate for $\abs{q_n\paren{\mb w} - q_n\paren{\mb w'}}$ as established in Lemma~\ref{lem:lip-h-mu} and also used $\norm{\mb w}{} \le 1/2$ and $\norm{\mb w'}{} \le 1/2$ to simplify the above result. Further noticing $t \mapsto t^2/\paren{1-t^2}^{3/2}$ is differentiable over $t \in \paren{0, 1}$, we apply the mean value theorem and obtain 
\begin{align*}
\abs{\frac{\norm{\mb w}{}^2}{q_n^3\paren{\mb w}} - \frac{\norm{\mb w'}{}^2}{q_n^3\paren{\mb w'}}} \le \sup_{\mb s \in \Gamma, \norm{\mb s}{} \le r < \frac{1}{2} } \frac{\norm{\mb s}{}^3 + 2\norm{\mb s}{}}{\paren{1-\norm{\mb s}{}^2}^{5/2}} \norm{\mb w - \mb w'}{} \le \frac{4}{\sqrt{3}} \norm{\mb w - \mb w'}{}. 
\end{align*}
Combining the above estimates gives the claimed result. 
\end{proof}

\begin{lemma}\label{lem:lp-zeta}
For any fixed $\mb x$, consider the function
\begin{align*}
\mb \zeta_{\mb x}(\mb w) = \overline{\mb x} - \frac{x_n}{q_n(\mb w)}\mb w
\end{align*}
defined over $\mb w \in \Gamma$. Then, for all $\mb w, \mb w' \in \Gamma$ such that $\norm{\mb w}{} \le r$ and $\norm{\mb w'}{} \le r$ for some constant $r \in \paren{0, \frac{1}{2}}$, it holds that 
\begin{align*}
 \norm{\mb \zeta_{\mb x}(\mb w) \mb \zeta_{\mb x}(\mb w)^*  }{} \;&\le\;2n \norm{\mb x}{\infty}^2,  \\
\norm{\mb \zeta_{\mb x}(\mb w) \mb \zeta_{\mb x}(\mb w)^* - \mb \zeta_{\mb x}(\mb w^\prime) \mb \zeta_{\mb x}(\mb w^\prime)^* }{}\;&\le\; 8\sqrt{2}\sqrt{n}  \norm{\mb x}{\infty}^2 \norm{\mb w - \mb w^\prime}{}.
\end{align*}
\end{lemma}
\begin{proof}
We have $\norm{\mb w}{}^2/q_n^2\paren{\mb w} \le 1/3$ when $\norm{\mb w}{} \le r < 1/2$, hence it holds that 
\begin{align*}
\norm{\mb \zeta_{\mb x}(\mb w) \mb \zeta_{\mb x}(\mb w)^*  }{} \le \norm{\mb \zeta_{\mb x}(\mb w)}{}^2 \le 2\norm{\overline{\mb x}}{}^2 + 2x_n^2\frac{\norm{\mb w}{}}{q_n^2\paren{\mb w}} \le 2n \norm{\mb x}{\infty}^2. 
\end{align*}
For the second, we first estimate  
\begin{align*}
\norm{\mb \zeta(\mb w) - \mb \zeta(\mb w^\prime)}{} 
& = \norm{x_n\paren{\frac{\mb w}{q_n\paren{\mb w}}  - \frac{\mb w'}{q_n\paren{\mb w'}}}}{} \le \norm{\mb x}{\infty} \norm{\frac{\mb w}{q_n\paren{\mb w}}  - \frac{\mb w'}{q_n\paren{\mb w'}}}{} \\
&\le  \norm{\mb x}{\infty} \paren{ \frac{1}{q_n(\mb w)}\norm{\mb w - \mb w^\prime}{}+  \norm{\mb w^\prime}{} \abs{\frac{1}{q_n(\mb w)} - \frac{1}{q_n(\mb w^\prime)}} } \\
&\le  \norm{\mb x}{\infty} \paren{ \frac{1}{q_n(\mb w)}+ \frac{ \norm{\mb w^\prime}{}}{\min\Brac{q_n^3(\mb w),q_n^3(\mb w^\prime)}} }\norm{\mb w - \mb w^\prime}{}  \\
&\le  \norm{\mb x}{\infty} \paren{\frac{2}{\sqrt{3}}+ \frac{4}{3\sqrt{3}}}\norm{\mb w - \mb w^\prime}{} \le 4\norm{\mb x}{\infty}\norm{\mb w - \mb w^\prime}{}. 
\end{align*}
Thus, we have 
\begin{align*}
\norm{\mb \zeta_{\mb x}(\mb w) \mb \zeta_{\mb x}(\mb w)^* - \mb \zeta_{\mb x}(\mb w^\prime) \mb \zeta_{\mb x}(\mb w^\prime)^* }{}\;&\le\; \norm{\mb \zeta(\mb w)}{} \norm{\mb \zeta(\mb w) - \zeta(\mb w^\prime)}{}+ \norm{\mb \zeta(\mb w) - \zeta(\mb w^\prime)}{}\norm{\mb \zeta(\mb w^\prime)}{} \nonumber \\
\;&\le \; 8\sqrt{2}\sqrt{n}  \norm{\mb x}{\infty}^2 \norm{\mb w - \mb w^\prime}{}, 
\end{align*}
as desired.
\end{proof}

Now, we are ready to prove all the Lipschitz propositions.

\begin{proof}[of Proposition \ref{prop:lip-hessian-negative}] \label{proof:cn_lips_curvature}
Let 
\begin{align*}
F_k(\mb w)=\ddot{h}_{\mu}\paren{\mb q(\mb w)^* \mb x_k} t^2_{\mb x_k}(\mb w) + \dot{h}_{\mu}\paren{\mb q(\mb w)^* \mb x_k} \frac{x_k\paren{n}}{q_n^3(\mb w)}. 
\end{align*}
Then, $\frac{1}{\norm{\mb w}{}^2}\mb w^*\nabla^2 g(\mb w) \mb w= \frac{1}{p} \sum_{k=1}^p F_k(\mb w)$. Noticing that $\ddot{h}_{\mu}\paren{\mb q(\mb w)^* \mb x_k}$ is bounded by $1/\mu$ and $\dot{h}_{\mu}\paren{\mb q(\mb w)^* \mb x_k}$ is bounded by $1$, both in magnitude. Applying Lemma \ref{lem:Lip-combined}, Lemma \ref{lem:lip-h-mu} and Lemma \ref{lem:lip-g}, we can see $F_k(\mb w)$ is $\Lconcave^k$-Lipschitz with
\begin{align*}
\Lconcave^k &= 4n \norm{\mb x_k}{}^2 \frac{4\sqrt{n}}{\mu^2} \norm{\mb x_k}{}+ \frac{1}{\mu}8\sqrt{n} \norm{\mb x_k}{} \paren{\frac{\norm{\mb x_k}{}}{\rconcave} + 4n^{3/2} \norm{\mb x_k}{\infty}} \nonumber \\
&+(2\sqrt{n})^3\norm{\mb x_k}{\infty} \frac{2\sqrt{n}}{\mu} \norm{\mb x_k}{}+ \sup_{\rconcave < a < \sqrt{\frac{2n-1}{2n}}} \frac{3}{\paren{1-a^2}^{5/2}} \norm{\mb x_k}{\infty} \nonumber \\
&=\frac{16n^{3/2}}{\mu^2} \norm{\mb x_k}{}^3 + \frac{8\sqrt{n}}{\mu \rconcave} \norm{\mb x_k}{}^2 + \frac{48 n^2}{\mu}\norm{\mb x_k}{} \norm{\mb x_k}{\infty} + 96 n^{5/2} \norm{\mb x_k}{\infty}. 
\end{align*}
Thus, $\frac{1}{\norm{\mb w}{2}}\mb w^*\nabla^2 g(\mb w) \mb w$ is $\Lconcave$-Lipschitz with
\begin{align*}
\Lconcave \le \frac{1}{p}\sum_{k=1}^p \Lconcave^k \le \frac{16n^3}{\mu^2} \norm{\mb X}{\infty}^3 + \frac{8n^{3/2}}{\mu \rconcave} \norm{\mb X}{\infty}^2 + \frac{48 n^{5/2} }{\mu} \norm{\mb X}{\infty}^2 + 96 n^{5/2} \norm{\mb X}{\infty}, 
\end{align*}
as desired. 
\end{proof}

\begin{proof} [of Proposition \ref{prop:lip-gradient} ] \label{proof:cn_lips_gradient}
We have 
\begin{align*}
\norm{\frac{\mb w^*}{\norm{\mb w}{}}\nabla g(\mb w) - \frac{\mb w'^*}{\norm{\mb w'}{}}\nabla g(\mb w^\prime) }{} \le \frac{1}{p}\sum_{k=1}^p \norm{\dot{h}_{\mu}\paren{\mb q(\mb w)^* \mb x_k} t_{\mb x_k}\paren{\mb w}  - \dot{h}_{\mu}\paren{\mb q(\mb w^\prime)^* \mb x_k} t_{\mb x_k}\paren{\mb w'}}{}
\end{align*}
where $\dot{h}_\mu (t) = \tanh(t/\mu)$ is bounded by one in magnitude, and $t_{\mb x_k}(\mb w)$ and $t_{\mb x_k^\prime}(\mb w)$ is defined as in Lemma \ref{lem:lip-g}. By Lemma \ref{lem:Lip-combined}, Lemma \ref{lem:lip-h-mu} and Lemma \ref{lem:lip-g}, we know that $\dot{h}_{\mu}\paren{\mb q(\mb w)^* \mb x_k} t_{\mb x_k}\paren{\mb w}$ is $L_k$-Lipschitz with constant
\begin{align*}
L_k = \frac{2\norm{\mb x_k}{}}{r_g}+ 8n^{3/2} \norm{\mb x_k}{\infty} + \frac{4n}{\mu} \norm{\mb x_k}{}^2. 
\end{align*}
Therefore, we have
\begin{align*}
\norm{\frac{\mb w^*}{\norm{\mb w}{}}\nabla g(\mb w) - \frac{\mb w^*}{\norm{\mb w}{}}\nabla g(\mb w^\prime) }{} 
& \le \frac{1}{p}\sum_{k=1}^p\paren{ \frac{2\norm{\mb x_k}{}}{r_g}+ 8n^{3/2} \norm{\mb x_k}{\infty} + \frac{4n}{\mu} \norm{\mb x_k}{}^2} \norm{\mb w-\mb w^\prime}{} \nonumber \\
&\le \paren{\frac{2\sqrt{n}}{r_g}\norm{\mb X}{\infty}+ 8n^{3/2} \norm{\mb X}{\infty} + \frac{4n^2}{\mu}\norm{\mb X}{\infty}^2  }\norm{\mb w-\mb w^\prime}{}, 
\end{align*}
as desired. 
\end{proof}

\begin{proof}[of Proposition \ref{prop:lip-hessian-zero}] \label{proof:cn_lips_strcvx}
Let 
\begin{align*}
\mb F_k(\mb w) &= \ddot{h}_\mu (\mb q(\mb w)^*\mb x_k)\mb \zeta_k(\mb w) \mb \zeta_k(\mb w)^*- \dot{h}_\mu\paren{\mb q(\mb w)^* \mb x_k}\bm \Phi_k(\mb w)
\end{align*}
with $\mb \zeta_k(\mb w) = \overline{\mb x}_k - \frac{x_k\paren{n}}{q_n(\mb w)}\mb w$ and $\mb \Phi_k(\mb w) = \frac{x_k\paren{n}}{q_n(\mb w)} \mb I + \frac{x_{n,k}}{q_n(\mb w)}\mb w\mb w^*$. Then, $\nabla^2 g(\mb w) =  \frac{1}{p}\sum_{k=1}^p \mb F_k(\mb w)$. Using Lemma \ref{lem:Lip-combined}, Lemma \ref{lem:lip-h-mu}, Lemma \ref{lem:lp-Phi} and Lemma \ref{lem:lp-zeta}, and the facts that $\ddot{h}_\mu(t)$ is bounded by $1/\mu$ and that $\ddot{h}_\mu(t)$ is bounded by $1$ in magnitude, we can see $\mb F_k(\mb w)$ is $\Lconvex^k$-Lipschitz continuous with 
\begin{align*}
\Lconvex^k 
& = \frac{1}{\mu}\times 8\sqrt{2}\sqrt{n}\norm{\mb x_k}{\infty}^2+ \frac{2\sqrt{n}}{\mu^2}\norm{\mb x_k}{} \times 2n \norm{\mb x_k}{\infty}^2 + 4\norm{\mb x_k}{\infty}+ \frac{2\sqrt{n}}{\mu}\norm{\mb x_k}{} \times 2\norm{\mb x_k}{\infty} \nonumber \\
& \le \frac{4n^{3/2}}{\mu^2} \norm{\mb x_k}{} \norm{\mb x_k}{\infty}^2+\frac{4\sqrt{n}}{\mu}\norm{\mb x_k}{}\norm{\mb x_k}{\infty} +  \frac{8\sqrt{2}\sqrt{n}}{\mu}\norm{\mb x_k}{\infty}^2 + 4\norm{\mb x_k}{\infty}. 
\end{align*}
Thus, we have
\begin{align*}
\Lconvex \le \frac{1}{p}\sum_{k=1}^p \Lconvex^k \le \frac{4n^2}{\mu^2} \norm{\mb X}{\infty}^3+\frac{4n}{\mu}\norm{\mb X}{\infty}^2 +  \frac{8\sqrt{2}\sqrt{n}}{\mu}\norm{\mb X}{\infty}^2 + 8\norm{\mb X}{\infty}, 
\end{align*}
as desired. 
\end{proof}

%% file: sec/proof_finite_union.tex
To avoid clutter of notations, in this subsection we write $\mb X$ to mean $\mb X_0$; similarly $\mb x_k$ for $\paren{\mb x_0}_k$, the $k$-th column of $\mb X_0$. The function $g\paren{\mb w}$ means $g\paren{\mb w; \mb X_0}$.
Before proving Theorem \ref{thm:geometry_orth}, we record one useful lemma. 

\begin{lemma}\label{lem:X-infinty-tail-bound}
For any $\theta \in \paren{0, 1}$, consider the random matrix $\mb X \in \R^{n_1 \times n_2}$ with $\mb X \sim_{i.i.d.} \mathrm{BG}\paren{\theta}$. Define the event $\event_\infty \doteq \Brac{ 1 \le \norm{\mb X}{\infty}\leq 4\sqrt{\log\paren{n p}}}$. It holds that 
\begin{align*}
\prob{\event_\infty^c} \leq \theta \paren{np}^{-7} + \exp\paren{-0.3\theta np}. 
\end{align*}
\end{lemma}
For convenience, we define three regions for the range of $\mb w$:
\begin{align*}
R_1 & \doteq \condset{ \mb w}{\norm{\mb w}{} \le \frac{\mu}{4\sqrt{2}}}, \qquad 
R_2 \doteq \condset{ \mb w}{\frac{\mu}{4\sqrt{2}} \le \norm{\mb w}{} \le \frac{1}{20\sqrt{5}}}, \\
R_3 & \doteq \condset{ \mb w}{\frac{1}{20\sqrt{5}} \le \norm{\mb w}{} \le \sqrt{\frac{4n-1}{4n}}}. 
\end{align*}

\begin{proof}[of Theorem \ref{thm:geometry_orth}] 
We will focus on deriving the qualitative result and hence be sloppy about constants. All indexed capital $C$ or small $c$ are numerical constants. 
\paragraph{Strong convexity in region $R_1$.} Proposition~\ref{prop:geometry_asymp_strong_convexity} shows that for any $\mb w \in R_1$, $\bb E \left[ \nabla^2 g(\mb w) \right] \succeq \frac{c_1\theta }{\mu} \mb I$. For any $\eps \in (0, \mu/\paren{4\sqrt{2}} )$, $R_1$ has an $\eps$-net $N_1$ of size at most $(3 \mu / \paren{4\sqrt{2}\eps})^n$. On $\event_\infty$, $\nabla^2 g$ is 
\begin{equation*}
L_1 \doteq \frac{C_2 n^2}{\mu^2} \log^{3/2}(np)
\end{equation*}
Lipschitz by Proposition~\ref{prop:lip-hessian-zero}. Set $\eps = \frac{c_1 \theta}{3 \mu L_1}$, so 
\begin{equation*}
\# N_1 \le \exp\left(2 n \log \left( \frac{ C_3 n \log( n p ) }{\theta} \right) \right). 
\end{equation*}
Let $\event_1$ denote the event
\begin{equation*}
\event_1 = \set{ \max_{\mb w \in N_1} \norm{\nabla^2 g(\mb w) - \bb E\left[ \nabla^2 g(\mb w) \right] }{} \le \frac{c_1 \theta }{3 \mu} }.
\end{equation*}
On $\event_1 \cap \event_\infty$, 
\begin{equation*}
\sup_{\norm{\mb w}{} \le \mu/\paren{4\sqrt{2}}} \norm{\nabla^2 g(\mb w) - \bb E \left[ \nabla^2 g(\mb w) \right] }{} \;\le\; \frac{2c_1 \theta}{3 \mu},
\end{equation*}
and so on $\event_1 \cap \event_\infty$, \eqref{eqn:hess-zero-uni-orth} holds for any constant $c_\star \le c_1 / 3$. Setting $t = c_1 \theta / 3 \mu$ in Proposition \ref{prop:concentration-hessian-zero}, we obtain that for any fixed $\mb w$,
\begin{equation*}
\prob{ \norm{ \nabla^2 g(\mb w) - \bb E\left[ \nabla^2 g(\mb w) \right] }{} \ge \frac{c_1 \theta}{3 \mu} } \le 4 n \exp\left( - \frac{c_4 p \theta^2}{n^2} \right).
\end{equation*}
Taking a union bound, we obtain that 
\begin{equation*}
\prob{ \event_1^c } \;\le\; 4 n \exp\left( - \frac{c_4 p \theta^2}{n^2} + C_5 n \log(n) + C_5 n \log \log(p) \right). 
\end{equation*}

\paragraph{Large gradient in region $R_2$.} Similarly, for the gradient quantity, for $\mb w \in R_2$, Proposition~\ref{prop:geometry_asymp_gradient} shows that 
\begin{equation*}
\bb E \left [\frac{\mb w^* \nabla g(\mb w) }{\norm{\mb w}{}} \right] \;\ge\; c_6 \theta. 
\end{equation*}
Moreover, on $\event_\infty$, $\frac{\mb w^* \nabla g(\mb w) }{\norm{\mb w}{}}$
is 
\begin{equation*}
L_2 \doteq \frac{C_7 n^2}{\mu} \log(np)
\end{equation*}
Lipschitz by Proposition~\ref{prop:lip-gradient}. For any $\eps < \frac{1}{20 \sqrt{5}}$, the set $R_2$ has an $\eps$-net $N_2$ of size at most $\left(\frac{3}{20 \eps \sqrt{5}}\right)^n$. Set $\eps = \frac{c_6 \theta}{3 L_2}$, so 
\begin{equation*}
\# N_2 \;\le\; \exp\left( n \log \left( \frac{C_8 n^2 \log(np)}{\theta \mu} \right) \right). 
\end{equation*}
Let $\event_2$ denote the event 
\begin{equation*}
\event_2 = \set{ \max_{\mb w \in N_2} \magnitude{ \frac{\mb w^* \nabla g(\mb w) }{\norm{\mb w}{}} - \bb E \left [\frac{\mb w^* \nabla g(\mb w) }{\norm{\mb w}{}} \right]  } \;\le\; \frac{c_6 \theta}{3} }.
\end{equation*}
On $\event_2 \cap \event_\infty$, 
\begin{equation}
\sup_{\mb w \in R_2} \magnitude{ \frac{\mb w^* \nabla g(\mb w) }{\norm{\mb w}{}} - \bb E \left [\frac{\mb w^* \nabla g(\mb w) }{\norm{\mb w}{}} \right]  } \;\le\; \frac{2c_6 \theta}{3},
\end{equation}
and so on $\event_2 \cap \event_\infty$, \eqref{eqn:grad-uni-orth} holds for any constant $c_\star \le c_6 / 3$. Setting $t = c_6 \theta / 3$ in Proposition \ref{prop:concentration-gradient}, we obtain that for any fixed $\mb w \in R_2$, 
\begin{equation*}
\bb P \left[ \magnitude{ \frac{\mb w^* \nabla g(\mb w) }{\norm{\mb w}{}} - \bb E \left [\frac{\mb w^* \nabla g(\mb w) }{\norm{\mb w}{}} } \right] \right]
 \;\le\; 2 \exp\left( - \frac{c_9 p \theta^2 }{n} \right), 
\end{equation*}
and so 
\begin{equation}
\bb P\left[ \event_2^c \right] \;\le\; 2 \exp\left( - \frac{c_9 p \theta^2}{n} + n \log\left( \frac{C_8 n^2 \log(np)}{\theta \mu} \right) \right). 
\end{equation}

\paragraph{Existence of negative curvature direction in $R_3$.} Finally, for any $\mb w \in R_3$, Proposition~\ref{prop:geometry_asymp_curvature} shows that 
\begin{equation*}
\bb E\left[ \frac{ \mb w^* \nabla^2 g(\mb w) \mb w }{\norm{\mb w}{}^2} \right] \;\le\; -c_9 \theta. 
\end{equation*}
On $\event_\infty$, $\frac{ \mb w^* \nabla^2 g(\mb w) \mb w }{\norm{\mb w}{}^2}$ is 
\begin{equation*}
L_3 = \frac{C_{10} n^3}{\mu^2} \log^{3/2}(np) 
\end{equation*}
Lipschitz by Proposition~\ref{prop:lip-hessian-negative}. As above, for any $\eps \le \sqrt{ \frac{4n - 1}{4n} }$, $R_3$ has an $\eps$-net $N_3$ of size at most $(3/\eps)^n$. Set $\eps = c_9 \theta / 3 L_3$. Then 
\begin{equation*}
\# N_3 \;\le\; \exp\left( n \log \left( \frac{C_{11} n^3 \log^{3/2}(np)}{\theta \mu^2 } \right) \right).
\end{equation*}
Let $\event_3$ denote the event 
\begin{equation*}
\event_3 = \set{ \max_{\mb w \in N_3} \magnitude{ \frac{ \mb w^* \nabla^2 g(\mb w) \mb w }{\norm{\mb w}{}^2} - \bb E\left[ \frac{ \mb w^* \nabla^2 g(\mb w) \mb w }{\norm{\mb w}{}^2} \right]  } \le \frac{c_9 \theta}{3} }
\end{equation*}
On $\event_3 \cap \event_\infty$, 
\begin{equation*}
\sup_{\mb w \in R_3} \magnitude{ \frac{ \mb w^* \nabla^2 g(\mb w) \mb w }{\norm{\mb w}{}^2} - \bb E\left[ \frac{ \mb w^* \nabla^2 g(\mb w) \mb w }{\norm{\mb w}{}^2} \right]  } \;\le\; \frac{ 2c_9 \theta}{3},
\end{equation*}
and \eqref{eqn:curvature-uni-orth} holds with any constant $c_\star < c_9 / 3$. Setting $t = c_9 \theta / 3$ in Proposition \ref{prop:concentration-hessian-negative} and taking a union bound, we obtain 
\begin{equation*}
\prob{\event_3^c} \;\le\; 4 \exp\left( - \frac{c_{12} p \mu^2 \theta^2}{n^2} + n \log\left( \frac{C_{11} n^3 \log^{3/2}(np)}{\theta \mu^2} \right) \right). 
\end{equation*}


\paragraph{The unique local minimizer located near $\mb 0$. } Let $\event_g$ be the event that the bounds \eqref{eqn:hess-zero-uni-orth}-\eqref{eqn:curvature-uni-orth} hold. On $\event_g$, the function $g$ is $\frac{c_\star\theta}{\mu}$-strongly convex over $R_1 = \condset{\mb w}{\norm{\mb w}{} \le \mu/\paren{4\sqrt{2}}}$. This implies that $f$ has at most one local minimum on $R_1$. It also implies that for any $\mb w \in R_1$,
\begin{eqnarray*}
g(\mb w) \ge g(\mb 0) + \innerprod{\nabla g(\mb 0) }{\mb w} + \frac{c\theta}{2 \mu} \norm{\mb w}{}^2 \ge g(\mb 0) - \norm{\mb w}{} \norm{\nabla g(\mb 0) }{} + \frac{c_\star\theta}{2 \mu} \norm{\mb w}{}^2. 
\end{eqnarray*}
So, if $g(\mb w) \le g(\mb 0)$, we necessarily have 
\begin{equation*}
\norm{\mb w}{} \;\le\; \frac{2 \mu}{c_\star \theta} \norm{\nabla g(\mb 0)}{}.
\end{equation*}
Suppose that 
\begin{equation}
\norm{\nabla g(\mb 0)}{} \le \frac{c_\star \theta}{32}. \label{eqn:grad-zero-bound}
\end{equation}
Then $g(\mb w) \le g(\mb 0)$ implies that $\norm{\mb w}{} \le \mu / 16$. By Wierstrass's theorem, $g(\mb w)$ has at least one minimizer $\mb w_\star$ over the compact set $S = \condset{\mb w}{\norm{\mb w}{} \le \mu / 10}$. By the above reasoning, $\norm{\mb w_\star}{} \le \mu / 16$, and hence $\mb w_\star$ does not lie on the boundary of $S$. This implies that $\mb w_\star$ is a local minimizer of $g$. Moreover, as above, 
\begin{equation*}
\norm{\mb w_\star}{} \;\le\; \frac{2 \mu}{c_\star \theta} \norm{\nabla g(\mb 0)}{}.
\end{equation*}

We now use the vector Bernstein inequality to show that with our choice of $p$, \eqref{eqn:grad-zero-bound} is satisifed with high probability. Notice that 
\begin{equation*}
\nabla g(\mb 0) = \frac{1}{p} \sum_{i = 1}^p \dot{h}_\mu( x_i(n) ) \overline{\mb x}_i, 
\end{equation*}
and $\dot{h}_\mu$ is bounded by one in magnitude, so for any integer $m \ge 2$,
\begin{eqnarray*}
\expect{ \norm{ \dot{h}_\mu( x_i(n) ) \overline{\mb x}_i}{}^m }
\le \expect{ \norm{\mb x_i}{}^m } 
\le \bb E_{Z \sim \chi(n)} \left[ Z^m \right] \le m!n^{m/2}, 
\end{eqnarray*}
where we have applied the moment estimate for the $\chi\paren{n}$ distribution shown in Lemma~\ref{lem:chi_moment}. Applying the vector Bernstein inequality in Corollary \ref{cor:vector-bernstein} with $R = \sqrt{n}$ and $\sigma^2 = 2n$, we obtain 
\begin{equation*}
\prob{ \norm{ \nabla g(\mb 0) }{} \;\ge\; t } \;\le\; 2 (n+1) \exp\left( - \frac{pt^2}{ 4 n + 2 \sqrt{n} t } \right) 
\end{equation*}
for all $t > 0$. Using this inequality, it is not difficult to show that there exist constants $C_{13}, C_{14} > 0$ such that when $p \ge C_{13} n \log n$, with probability at least $1- 4n p^{-10}$, 
\begin{equation}
\norm{\nabla g(\mb 0)}{} \;\le\; C_3 \sqrt{\frac{n \log p}{p}}. \label{eqn:grad-zero}
\end{equation}
When $\frac{p}{\log p} \ge \frac{C_{14} n}{\theta^2}$, for appropriately large $C_{14}$, \eqref{eqn:grad-zero} implies \eqref{eqn:grad-zero-bound}. Summing up failure probabilities completes the proof. 
\end{proof}

%% file: sec/proof_geo_comp.tex
\begin{proof}[of Lemma~\ref{lem:pert_key_mag}] \label{proof:comp_pert_bound}
By the generative model, 
\begin{align*}
\overline{\mb Y} = \paren{\frac{1}{p\theta}\mb Y \mb Y^*}^{-1/2} \mb Y = \paren{\frac{1}{p\theta}\mb A_0 \mb X_0 \mb X_0^* \mb A_0^*}^{-1/2} \mb A_0 \mb X_0.
\end{align*} 
Since $\expect{\mb X_0 \mb X_0^*/\paren{p \theta}} = \mb I$, we will compare $\paren{\frac{1}{p\theta}\mb A_0 \mb X_0 \mb X_0^* \mb A_0^*}^{-1/2} \mb A_0$ with $\paren{\mb A_0 \mb A_0^*}^{-1/2} \mb A_0 = \mb U \mb V^*$. By Lemma~\ref{lem:half_inverse_pert}, we have
\begin{align*}
& \norm{\paren{\frac{1}{p\theta} \mb A_0 \mb X_0 \mb X_0^* \mb A_0^*}^{-1/2} \mb A_0 - \paren{\mb A_0 \mb A_0^*}^{-1/2} \mb A_0}{} \nonumber \\
\le\; & \norm{\mb A_0}{} \norm{\paren{\frac{1}{p\theta}\mb A_0 \mb X_0 \mb X_0^* \mb A_0^*}^{-1/2} - \paren{\mb A_0 \mb A_0^*}^{-1/2}}{} \nonumber \\
\le\; & \norm{\mb A_0}{} \frac{2\norm{\mb A_0}{}^3}{\sigma_{\min}^4\paren{\mb A_0}} \norm{\frac{1}{p\theta} \mb X_0 \mb X_0^* - \mb I}{}  = 2\kappa^4\paren{\mb A_0} \norm{\frac{1}{p\theta} \mb X_0 \mb X_0^* - \mb I}{}
\end{align*}
provided 
\begin{align*}
\norm{\mb A_0}{}^2\norm{\frac{1}{p\theta} \mb X_0 \mb X_0^* - \mb I}{} \le \frac{\sigma_{\min}^2\paren{\mb A_0}}{2} \Longleftrightarrow \norm{\frac{1}{p\theta} \mb X_0 \mb X_0^* - \mb I}{} \le \frac{1}{2\kappa^2\paren{\mb A_0}}.  
\end{align*}
On the other hand, by Lemma~\ref{lem:bg_identity_diff}, when $p \ge C_1 n^2 \log n$ for some large constant $C_1$,  $\norm{\tfrac{1}{p\theta} \mb X_0 \mb X_0^* - \mb I}{} \le 10 \sqrt{\tfrac{\theta n \log p}{p}}$ with probability at least $1-p^{-8}$. Thus, when $p \ge C_2 \kappa^4\paren{\mb A_0} \theta n^2 \log (n \theta \kappa\paren{\mb A_0})$, 
\begin{align*}
\norm{\paren{\frac{1}{p\theta} \mb A_0 \mb X_0 \mb X_0^* \mb A_0^*}^{-1/2} \mb A_0 - \paren{\mb A_0 \mb A_0^*}^{-1/2} \mb A_0}{}  \le 20 \kappa^4\paren{\mb A_0} \sqrt{\frac{\theta n \log p}{p}},  
\end{align*}
as desired. 
\end{proof}

\begin{proof}[of Lemma~\ref{lem:pert_key_grad_hess}] \label{proof:comp_pert_bound2}
To avoid clutter in notation, we write $\mb X$ to mean $\mb X_0$, and $\mb x_k$ to mean $\paren{\mb x_0}_k$ in this proof. We also let $\widetilde{\mb Y} \doteq \mb X_0 + \widetilde{\mb \Xi} \mb X_0$. Note the Jacobian matrix for the mapping $\mb q\paren{\mb w}$ is $\nabla_{\mb w} \mb q\paren{\mb w} = \brac{\mb I, -\mb w/\sqrt{1-\norm{\mb w}{}^2}}$. Hence for any vector $\mb z \in \R^{n}$ and all $\mb w \in \Gamma$,  
\begin{align*}
\norm{\nabla_{\mb w}\mb q\paren{\mb w} \mb z}{} \le \sqrt{n-1}\norm{\mb z}{\infty} + \frac{\norm{\mb w}{}}{\sqrt{1-\norm{\mb w}{}^2}} \norm{\mb z}{\infty} \le 3\sqrt{n} \norm{\mb z}{\infty}. 
\end{align*}
Now we have 
\begin{align*}
& \norm{\nabla_{\mb w} g\paren{\mb w; \widetilde{\mb Y}} - \nabla_{\mb w{}} g\paren{\mb w; \mb X} }{}\\
=\; & \norm{\frac{1}{p}\sum_{k=1}^p \dot{h}_{\mu}\paren{\mb q^*\paren{\mb w} \mb x_k + \mb q^*\paren{\mb w} \widetilde{\mb \Xi} \mb x_k}\nabla_{\mb w}\mb q\paren{\mb w} \paren{\mb x_k + \widetilde{\mb \Xi} \mb x_k} - \frac{1}{p}\sum_{k=1}^p \dot{h}_{\mu}\paren{\mb q^*\paren{\mb w} \mb x_k}\nabla_{\mb w}\mb q\paren{\mb w} \mb x_k}{} \\
\le\; & \norm{\frac{1}{p}\sum_{k=1}^p \dot{h}_{\mu}\paren{\mb q^*\paren{\mb w} \mb x_k + \mb q^*\paren{\mb w} \widetilde{\mb \Xi} \mb x_k}\nabla_{\mb w}\mb q\paren{\mb w} \paren{\mb x_k + \widetilde{\mb \Xi} \mb x_k - \mb x_k}}{} \\
& \qquad + \norm{\frac{1}{p}\sum_{k=1}^p \brac{\dot{h}_{\mu}\paren{\mb q^*\paren{\mb w} \mb x_k + \mb q^*\paren{\mb w} \widetilde{\mb \Xi} \mb x_k} - \dot{h}_{\mu}\paren{\mb q^*\paren{\mb w} \mb x_k}}\nabla_{\mb w}\mb q\paren{\mb w} \mb x_k}{}\\
\le\;& \norm{\widetilde{\mb \Xi}}{}\paren{\max_{t} \dot{h}_{\mu}\paren{t} 3n \norm{\mb X}{\infty} + L_{\dot{h}_{\mu}} 3n\norm{\mb X}{\infty}^2},  
\end{align*}
where $L_{\dot{h}_{\mu}}$ denotes the Lipschitz constant for $\dot{h}_{\mu}\paren{\cdot}$. Similarly, suppose $\norm{\widetilde{\mb \Xi}}{} \le \tfrac{1}{2n}$, and also notice that 
\begin{align*}
\norm{\frac{\mb I}{q_n\paren{\mb w}} + \frac{\mb w \mb w^*}{q_n^3\paren{\mb w}}}{} \le \frac{1}{q_n\paren{\mb w}} + \frac{\norm{\mb w}{}^2}{q_n^3\paren{\mb w}} = \frac{1}{q_n^3\paren{\mb w}} \le 2\sqrt{2} n^{3/2},
\end{align*} 
we obtain that  
\begin{align*}
& \norm{\nabla_{\mb w}^2 g\paren{\mb w; \widetilde{\mb Y}} - \nabla_{\mb w}^2 g\paren{\mb w; \mb X}}{} \\
\le\; & \norm{\frac{1}{p}\sum_{k=1}^p \brac{\ddot{h}\paren{\mb q^*\paren{\mb w} \widetilde{\mb y}_k} \nabla_{\mb w}\mb q\paren{\mb w} \widetilde{\mb y}_k \widetilde{\mb y}_k^* \paren{\nabla_{\mb w} \mb q\paren{\mb w}}^* - \ddot{h}\paren{\mb q^*\paren{\mb w} \mb x_k} \nabla_{\mb w}\mb q\paren{\mb w} \mb x_k \mb x_k^* \paren{\nabla_{\mb w} \mb q\paren{\mb w}}^*}}{} \\
& \qquad + \norm{\frac{1}{p}\sum_{k=1}^p \brac{\dot{h}\paren{\mb q^*\paren{\mb w} \widetilde{\mb y}_k}\paren{\frac{\mb I}{q_n\paren{\mb w}} + \frac{\mb w \mb w^*}{q_n^3}} \widetilde{\mb y}_k\paren{n} - \dot{h}\paren{\mb q^*\paren{\mb w} \mb x_k}\paren{\frac{\mb I}{q_n\paren{\mb w}} + \frac{\mb w \mb w^*}{q_n^3}} \mb x_k\paren{n}}}{}\\
\le\; & \tfrac{45}{2}L_{\ddot{h}_{\mu}} n^{3/2} \norm{\mb X}{\infty}^3 \norm{\widetilde{\mb \Xi}}{} + \max_t \ddot{h}_{\mu}\paren{t} \paren{18n^{3/2} \norm{\mb X}{\infty}^2 \norm{\widetilde{\mb \Xi}}{}  + 10 n^2 \norm{\mb X}{\infty}^2\norm{\widetilde{\mb \Xi}}{}^2} \\
& \qquad + 3\sqrt{2} L_{\dot{h}_{\mu}} n^2 \norm{\widetilde{\mb \Xi}}{} \norm{\mb X}{\infty}^2 + \max_t \dot{h}\paren{t} 2\sqrt{2} n^2 \norm{\widetilde{\mb \Xi}}{} \norm{\mb X}{\infty}, 
\end{align*}
where $L_{\ddot{h}_{\mu}}$ denotes the Lipschitz constant for $\ddot{h}_{\mu}\paren{\cdot}$. Since  
\begin{align*}
\max_{t} \dot{h}_{\mu}\paren{t}  \le 1, & \quad \max_{t} \ddot{h}_{\mu}\paren{t} \le \frac{1}{\mu}, \quad L_{h_{\mu}} \le 1, \quad L_{\dot{h}_{\mu}} \le \frac{1}{\mu}, \quad L_{\ddot{h}_{\mu}} \le \frac{2}{\mu^2},
\end{align*}
and by Lemma~\ref{lem:X-infinty-tail-bound}, $\norm{\mb X}{\infty} \le 4\sqrt{\log\paren{np}}$ with probability at least $1-\theta \paren{np}^{-7} -\exp\paren{-0.3\theta np}$, we obtain 
\begin{align*}
\norm{\nabla_{\mb w} g\paren{\mb w; \widetilde{\mb Y}} - \nabla_{\mb w{}} g\paren{\mb w; \mb X} }{} & \le C_1\frac{n}{\mu} \log\paren{np} \norm{\widetilde{\mb \Xi}}{}, \\
\norm{\nabla_{\mb w}^2 g\paren{\mb w; \widetilde{\mb Y}} - \nabla_{\mb w}^2 g\paren{\mb w; \mb X}}{} & \le C_2 \max\set{\frac{n^{3/2}}{\mu^2}, \frac{n^2}{\mu}} \log^{3/2}\paren{np} \norm{\widetilde{\mb \Xi}}{}
\end{align*}
for numerical constants $C_1, C_2 > 0$.
\end{proof}

\begin{proof}[of Theorem~\ref{thm:geometry_comp}]
Assume the constant $c_\star$ as defined in Theorem~\ref{thm:geometry_orth}. By Lemma~\ref{lem:pert_key_mag}, when 
\begin{align*} 
p \ge \frac{C_1}{c_\star^2 \theta} \max\set{\frac{n^4}{\mu^4}, \frac{n^5}{\mu^2}} \kappa^8\paren{\mb A_0} \log^4 \paren{\frac{\kappa\paren{\mb A_0}n}{\mu \theta}},
\end{align*}
the magnitude of the perturbation is bounded as 
\begin{align*}
\norm{\widetilde{\mb \Xi}}{} \le C_2 c_\star \theta \paren{\max\set{\frac{n^{3/2}}{\mu^2}, \frac{n^2}{\mu}} \log^{3/2}\paren{np}}^{-1}, 
\end{align*}
where $C_2$ can be made arbitrarily small by making $C_1$ large. 
Combining this result with Lemma~\ref{lem:pert_key_grad_hess}, we obtain that for all $\mb w \in \Gamma$, 
\begin{align*}
\norm{\nabla_{\mb w} g\paren{\mb w; \mb X_0 + \widetilde{\mb \Xi} \mb X_0} - \nabla_{\mb w{}} g\paren{\mb w; \mb X} }{} & \le \frac{c_\star \theta}{2} \nonumber \\
\norm{\nabla_{\mb w}^2 g\paren{\mb w; \mb X_0 + \widetilde{\mb \Xi} \mb X_0} - \nabla_{\mb w}^2 g\paren{\mb w; \mb X}}{} & \le \frac{c_\star \theta}{2}, 
\end{align*}
with probability at least $1-p^{-8} - \theta\paren{np}^{-7} - \exp\paren{-0.3\theta n p}$. In view of~\eqref{eqn:curvature-uni-comp} in Theorem~\ref{thm:geometry_orth}, we have 
\begin{align*}
\frac{\mb w^* g\paren{\mb w; \mb X_0 + \widetilde{\mb \Xi} \mb X_0} \mb w}{\norm{\mb w}{}^2} 
& = \frac{\mb w^* g\paren{\mb w; \mb X_0} \mb w}{\norm{\mb w}{}^2} + \frac{\mb w^* g\paren{\mb w; \mb X_0 + \widetilde{\mb \Xi} \mb X_0} \mb w}{\norm{\mb w}{}^2} - \frac{\mb w^* g\paren{\mb w; \mb X_0} \mb w}{\norm{\mb w}{}^2} \\
& \le - c_\star \theta + \norm{\nabla_{\mb w}^2 g\paren{\mb w; \mb X_0 + \widetilde{\mb \Xi} \mb X_0} - \nabla_{\mb w}^2 g\paren{\mb w; \mb X}}{} \le -\frac{1}{2} c_\star \theta. 
\end{align*}
By similar arguments, we obtain~\eqref{eqn:hess-zero-uni-comp} through~\eqref{eqn:curvature-uni-comp} in Theorem~\ref{thm:geometry_comp}. 

To show the unique local minimizer over $\Gamma$ is near $\mb 0$, we note that (recall the last part of proof of Theorem~\ref{thm:geometry_orth} in Section~\ref{sec:proof_geometry_orth}) $g\paren{\mb w; \mb X_0 + \widetilde{\mb \Xi} \mb X_0}$ being $\frac{c_\star \theta}{2\mu}$ strongly convex near $\mb 0$ implies that 
\begin{align*}
\norm{\mb w_\star}{} \le \frac{4\mu}{c_\star \theta} \norm{\nabla g\paren{\mb 0; \mb X_0 + \widetilde{\mb \Xi} \mb X_0}}{}. 
\end{align*}
The above perturbation analysis implies there exists $C_3 > 0$ such that when 
\begin{align*}
p \ge \frac{C_3}{c_\star^2 \theta} \max\set{\frac{n^4}{\mu^4}, \frac{n^5}{\mu^2}} \kappa^8\paren{\mb A_0} \log^4 \paren{\frac{\kappa\paren{\mb A_0}n}{\mu \theta}}, 
\end{align*}
it holds that 
\begin{align*}
\norm{\nabla_{\mb w} g\paren{\mb 0; \mb X_0 + \widetilde{\mb \Xi} \mb X_0} - \nabla_{\mb w{}} g\paren{\mb 0; \mb X} }{} & \le \frac{c_\star \theta}{400}, 
\end{align*}
which in turn implies 
\begin{align*}
\norm{\mb w_\star}{} \le \frac{4\mu}{c_\star \theta} \norm{\nabla g\paren{\mb 0; \mb X_0}}{} + \frac{4\mu}{c_\star \theta} \frac{c_\star \theta}{400} \le \frac{\mu}{8} + \frac{\mu}{100} < \frac{\mu}{7}, 
\end{align*}
where we have recall the result that $\frac{2\mu}{c_\star \theta} \norm{\nabla g\paren{\mb 0; \mb X_0}}{} \le \mu /16$ from proof of Theorem~\ref{thm:geometry_orth}. A simple union bound with careful bookkeeping gives the success probability. 
\end{proof}

%% file: sec/proof_algorithm.tex
\section{Proof of Convergence for the Trust-Region Algorithm}\label{sec:proof_algorithm}

\begin{proof}[of Lemma~\ref{lem:mag_lip_fq}] \label{proof:lem_mag_lip_fq}
Using the fact $\tanh\paren{\cdot}$ and $1-\tanh^2\paren{\cdot}$ are bounded by one in magnitude, by \eqref{eq:fq_grad} and \eqref{eq:fq_hess} we have 
\begin{align*}
\norm{\nabla f\paren{\mb q}}{} & \le \frac{1}{p}\sum_{k=1}^p \norm{\mb x_k}{} \le \sqrt{n} \norm{\mb X}{\infty}, \\
\norm{\nabla^2 f\paren{\mb q}}{} & \le \frac{1}{p}\sum_{k=1}^p \frac{1}{\mu} \norm{\mb x_k}{}^2 \le \frac{n}{\mu} \norm{\mb X}{\infty}^2,
\end{align*}
for any $\mb q \in \bb S^{n-1}$. Moreover, 
\begin{align*}
\sup_{\mb q, \mb q' \in \bb S^{n-1}, \mb q \neq \mb q'} \frac{\norm{\nabla f\paren{\mb q} - \nabla f\paren{\mb q'}}{}}{\norm{\mb q - \mb q'}{}} 
& \le \frac{1}{p} \sum_{k=1}^p \norm{\mb x_k}{} \sup_{\mb q, \mb q' \in \bb S^{n-1}, \mb q \neq \mb q'}\frac{\abs{\tanh\paren{\frac{\mb q^* \mb x_k}{\mu}} - \tanh\paren{\frac{\mb q'^* \mb x_k}{\mu}}}}{\norm{\mb q - \mb q'}{}} \\
& \le \frac{1}{p} \sum_{k=1}^p \norm{\mb x_k}{} \frac{\norm{\mb x_k}{}}{\mu} \le \frac{n}{\mu} \norm{\mb X}{\infty}^2, 
\end{align*} 
where at the last line we have used the fact the mapping $\mb q \mapsto \mb q^* \mb x_k/\mu$ is $\norm{\mb x_k}{}/\mu$ Lipschitz, and $\mb x \mapsto \tanh\paren{x}$ is $1$-Lipschitz, and the composition rule in Lemma~\ref{lem:composition}. Similar argument yields the final bound. 
\end{proof}

\begin{proof}[of Lemma~\ref{lem:alg_approx_bd2}] \label{proof:lem_alg_approx_bd2}
Suppose we can establish 
\begin{equation*}
\left| f\paren{ \exp_{\mb q}(\mb \delta) } - \widehat{f}\paren{\mb q, \mb \delta} \right| \;\le\; \frac{1}{6}\eta_f \norm{\mb \delta}{}^3. 
\end{equation*}
Applying this twice we obtain 
\begin{align*}
f( \exp_{\mb q}(\mb \delta_\star) ) 
& \le \widehat{f}(\mb q,\mb \delta_\star) + \frac{1}{6}\eta_f \Delta^3 
\le \widehat{f}(\mb q, \mb \delta) + \frac{1}{6}\eta_f \Delta^3
\le f(\exp_{\mb q}(\mb \delta)) + \frac{1}{3}\eta_f\Delta^3 
\le f(\mb q) - s + \frac{1}{3}\eta_f \Delta^3,
\end{align*}
as claimed. Next we establish the first result. Let $\mb \delta_0 = \frac{\mb \delta}{\norm{\mb \delta}{}}$, and $t = \norm{\mb \delta}{}$. Consider the composite function
\begin{align*}
	\zeta(t) \doteq f( \exp_{\mb q}(t \mb \delta_0)) = f( \mb q \cos(t) + \mb \delta_0 \sin(t) ),
\end{align*} 
and also 
\begin{align*}
\dot{\zeta}(t) &= \innerprod{ \nabla f\left( \mb q \cos(t) + \mb \delta_0 \sin(t) \right) }{ - \mb q \sin(t) + \mb \delta_0 \cos(t) } \\
\ddot{\zeta}(t) &= \innerprod{ \nabla^2 f\left( \mb q \cos(t) + \mb \delta_0 \sin(t) \right) ( - \mb q \sin(t) + \mb \delta_0 \cos(t) ) }{ - \mb q \sin(t) + \mb \delta_0 \cos(t) } \nonumber \\
& + \quad \innerprod{ \nabla f\left( \mb q \cos(t) + \mb \delta_0 \sin(t) \right) }{ - \mb q \cos(t) - \mb \delta_0 \sin(t) }. 
\end{align*}
In particular, this gives that 
\begin{align*}
	\zeta(0) &= f( \mb q) \\
\dot{\zeta}(0) &= \innerprod{ \mb \delta_0 }{\nabla f(\mb q) } \\
\ddot{\zeta}(0) &= \mb \delta_0^* \left( \nabla^2 f(\mb q) - \innerprod{ \nabla f(\mb q) }{\mb q} \mb I \right) \mb \delta_0.
\end{align*}

We next develop a bound on $\magnitude{ \ddot{\zeta}(t) - \ddot{\zeta}(0) }$. Using the triangle inequality, we can casually bound this difference as 
\begin{align*}
&\magnitude{ \ddot{\zeta}(t) - \ddot{\zeta}(0) } \\
\le\; & \magnitude{  \innerprod{ \nabla^2 f\left( \mb q \cos(t) + \mb \delta_0 \sin(t) \right) ( - \mb q \sin(t) + \mb \delta_0 \cos(t) ) }{ - \mb q \sin(t) + \mb \delta_0 \cos(t)   }  - \mb \delta_0^* \nabla^2 f(\mb q) \mb \delta_0} \nonumber \\ 
& \qquad + \qquad \magnitude{ \innerprod{ \nabla f\left( \mb q \cos(t) + \mb \delta_0 \sin(t) \right) }{ - \mb q \cos(t) - \mb \delta_0 \sin(t) } + \innerprod{ \nabla f(\mb q) }{\mb q} } \nonumber \\
\le\; & \magnitude{ \innerprod{ \left[ \nabla^2 f( \mb q \cos(t) + \mb \delta_0 \sin(t) ) - \nabla^2 f(\mb q) \right] \left( - \mb q \sin(t) + \mb \delta_0 \cos(t) \right) }{ - \mb q \sin(t) + \mb \delta_0 \cos(t) } }  \nonumber \\
& \qquad + \qquad \magnitude{ \innerprod{ \nabla^2 f(\mb q) \left( - \mb q \sin(t) + \mb \delta_0 \cos(t) - \mb \delta_0 \right) }{ - \mb q \sin(t) + \mb \delta_0 \cos(t) } } \nonumber \\
& \qquad + \qquad \magnitude{ \innerprod{ \nabla^2 f(\mb q) \mb \delta_0 }{ - \mb q \sin(t) + \mb \delta_0 \cos(t) - \mb \delta_0 } } \nonumber \\ 
& \qquad + \qquad \magnitude{ \innerprod{ \nabla f( \mb q \cos(t) + \mb \delta_0 \sin(t) ) }{ - \mb q \cos(t) - \mb \delta_0 \sin(t) } + \innerprod{ \nabla f(\mb q \cos(t) + \mb \delta_0 \sin(t) ) }{ \mb q } } \nonumber \\
& \qquad + \qquad \magnitude{ \innerprod{ \nabla f(\mb q \cos(t) + \mb \delta_0 \sin(t)) }{\mb q } - \innerprod{ \nabla f(\mb q) }{\mb q} } \\
\le\; & L_{\nabla^2} \norm{ \mb q \cos(t) + \mb \delta_0 \sin(t) - \mb q }{} \nonumber  \\
& \qquad + M_{\nabla^2} \norm{ - \mb q \sin(t) + \mb \delta_0 \cos(t) - \mb \delta_0 }{} \nonumber \\
& \qquad + M_{\nabla^2} \norm{ - \mb q \sin(t) + \mb \delta_0 \cos(t) - \mb \delta_0 }{} \nonumber \\
& \qquad + M_\nabla \norm{ - \mb q \cos(t) - \mb \delta_0 \sin(t) + \mb q }{} \nonumber \\
& \qquad + L_\nabla \norm{ \mb q \cos(t) + \mb \delta_0 \sin(t) - \mb q }{} \\
=\; & \left( L_{\nabla^2} + 2 M_{\nabla^2} + M_{\nabla} + L_{\nabla} \right) \sqrt{ (1 - \cos(t))^2 + \sin^2(t) } \\
=\; & \eta_f \sqrt{ 2 - 2\cos t} \le \eta_f \sqrt{ 4 \sin^2\paren{t/2}} \le \eta_f t, 
\end{align*}
where in the final line we have used the fact $1-\cos x = 2\sin^2\paren{x/2}$ and that $\sin x \le x$ for $x\in \brac{0, 1}$, and $M_{\nabla}$, $M_{\nabla^2}$, $L_{\nabla}$ and $L_{\nabla^2}$ are the quantities defined in Lemma~\ref{lem:mag_lip_fq}. By the integral form of Taylor's theorem in Lemma \ref{lem:Taylor-integral-form} and the result above, we have 
\begin{align*}
\abs{f\paren{\exp_{\mb q}(\mb \delta)} - \widehat{f}\paren{\mb q,\mb \delta} } &=
\magnitude{ \zeta(t) - \left( \zeta(0) + t \dot{\zeta}(0) + \tfrac{t^2}{2} \ddot{\zeta}(0) \right) } \\
& = \magnitude{ t^2\int_0^1 \paren{1-s} \ddot{\zeta}\paren{st}\; ds - \tfrac{t^2}{2} \ddot{\zeta}(0) }  \\
& =  t^2\magnitude{ \int_0^1 \paren{1-s} \brac{\ddot{\zeta}\paren{st} - \ddot{\zeta}\paren{0} } \; ds }   \\
& \le t^2 \int_0^1 \paren{1-s}st \eta_f\; ds = \frac{\eta_f t^3}{6},
\end{align*}
with $t = \norm{\mb \delta}{}$ we obtain the desired result. 
\end{proof}

\begin{proof}[of Lemma~\ref{lem:alg_gradient_func}] \label{proof:lem_alg_gradient_func}
By the integral form of Taylor's theorem in Lemma \ref{lem:Taylor-integral-form}, for any $t \in \brac{0, \frac{3\Delta}{2\pi\sqrt{n}}}$, we have 
\begin{align*}
& g\left( \mb w - t \frac{\mb w}{\norm{\mb w}{}} \right) \\
=\; & g(\mb w) - t\int_0^1 \innerprod{\nabla g\paren{\mb w -st\frac{\mb w}{\norm{\mb w}{}}}}{\frac{\mb w}{\norm{\mb w}{}}}\; ds \\
=\; & g\paren{\mb w} - t\frac{\mb w^* \nabla g\paren{\mb w}}{\norm{\mb w}{}} + t\int_0^1 \innerprod{\nabla g\paren{\mb w} - \nabla g\paren{\mb w -st\frac{\mb w}{\norm{\mb w}{}}}}{\frac{\mb w}{\norm{\mb w}{}}}\; ds \\
=\; & g\paren{\mb w} - t\frac{\mb w^* \nabla g\paren{\mb w}}{\norm{\mb w}{}} + t\int_0^1 \paren{\innerprod{\nabla g\paren{\mb w}}{\frac{\mb w}{\norm{\mb w}{}}} - \innerprod{\nabla g\paren{\mb w - st \frac{\mb w}{\norm{\mb w}{}}}}{\frac{\mb w - st\mb w/\norm{\mb w}{}}{\norm{\mb w - st \mb w/\norm{\mb w}{}}{}}}}\; ds \\
\le\; & g\paren{\mb w} - t\frac{\mb w^* \nabla g\paren{\mb w}}{\norm{\mb w}{}} + \frac{L_g}{2} t^2 \le g\paren{\mb w} - t\beta_g + \frac{L_g}{2}t^2. 
\end{align*}
Minimizing this function over $t \in \left[0, \frac{3\Delta }{ 2\pi \sqrt{n}} \right]$, we obtain that there exists a $\mb w' \in \mc B\left(\mb w,\frac{3\Delta }{ 2\pi \sqrt{n} } \right)$ such that 
\begin{align*}
g(\mb w') \;\le\; g(\mb w) - \min \set{ \frac{\beta_g^2}{2 L_g}, \frac{3\beta_g \Delta}{4\pi\sqrt{n}} }.
\end{align*}
Given such a $\mb w'\in \mc B\left(\mb w,\frac{3\Delta }{ 2\pi \sqrt{n} } \right)$, there must exist some $\mb \delta\in T_{\mb q}\bb S^{n-1}$ such that $\mb q(\mb w') = \exp_{\mb q}(\mb \delta)$. 
It remains to show that $\norm{\mb \delta}{}\leq \Delta$. By Lemma~\ref{lem:lip-h-mu}, we know that $\norm{\mb q(\mb w') - \mb q\paren{\mb w}}{} \le 2 \sqrt{n} \norm{\mb w' - \mb w}{} \le 3\Delta/\pi$. Hence,  
\begin{align*}
\norm{\exp_{\mb q}\paren{\mb \delta} - \mb q}{}^2 = \norm{\mb  q\paren{1-\cos\norm{\mb \delta}{}} + \frac{\mb \delta}{\norm{\mb \delta}{}} \sin\norm{\mb \delta}{}}{}^2 = 2 - 2\cos \norm{\mb \delta}{} = 4\sin^2 \frac{\norm{\mb \delta}{}}{2} \le \frac{9\Delta^2}{\pi^2}, 
\end{align*}
which means that $\sin\paren{\norm{\mb \delta}{}/2} \le 3\Delta/\paren{2\pi}$. Because $\sin x \ge \tfrac{3}{\pi} x$ over $x \in \brac{0, \pi/6}$, it implies that $\norm{\mb \delta}{} \le \Delta$. Since $g(\mb w) = f(\mb q(\mb w))$, by summarizing all the results, we conclude that there exists a $\mb \delta$ with $\norm{\mb \delta}{}\leq \Delta$, such that
\begin{align*} 
f(\exp_{\mb q}(\mb \delta)) \le f(\mb q ) - \min \set{ \frac{\beta_g^2}{2 L_g}, \frac{3\beta_g \Delta}{4\pi\sqrt{n}} },
\end{align*}
as claimed. 
\end{proof}

\begin{proof}[of Lemma~\ref{lem:alg_neg_cuv_func}] \label{proof:lem_alg_neg_cuv_func}
Let $\sigma = \mathrm{sign}\paren{ \mb w^* \nabla g(\mb w) }$. For any $t \in \brac{0, \frac{\Delta}{2\sqrt{n}}}$, by integral form of Taylor's theorem in Lemma \ref{lem:Taylor-integral-form}, we have  
\begin{align*}
& g\left( \mb w - t \sigma \frac{\mb w}{\norm{\mb w}{}} \right) \\
=\; & g(\mb w) - t \sigma \frac{\mb w^* \nabla g(\mb w) }{\norm{\mb w}{}} + t^2 \int_0^1 \paren{1-s} \frac{\mb w^* \nabla^2 g\paren{\mb w - st\sigma \frac{\mb w}{\norm{\mb w}{}}}\mb w}{\norm{\mb w}{}^2}\;ds \\
\le\; & g(\mb w) + \frac{t^2}{2} \frac{\mb w^* \nabla^2 g(\mb w) \mb w }{\norm{\mb w}{}^2} + t^2 \int_0^1 \brac{\paren{1-s} \frac{\mb w^* \nabla^2 g\paren{\mb w - st\sigma \frac{\mb w}{\norm{\mb w}{}}}\mb w}{\norm{\mb w}{}^2} -  \paren{1-s} \frac{\mb w^* \nabla^2 g(\mb w) \mb w }{\norm{\mb w}{}^2}} \; ds \\
=\; & g(\mb w) + \frac{t^2}{2} \frac{\mb w^* \nabla^2 g(\mb w) \mb w }{\norm{\mb w}{}^2} \\
& +  t^2\int_0^1 \paren{1-s}\brac{\frac{\paren{\mb w - st\sigma \frac{\mb w}{\norm{\mb w}{}}}^* \nabla^2 g\paren{\mb w - st\sigma \frac{\mb w}{\norm{\mb w}{}}}\paren{\mb w - st\sigma \frac{\mb w}{\norm{\mb w}{}}}}{\norm{\mb w - st\sigma \frac{\mb w}{\norm{\mb w}{}}}{}^2} -\frac{\mb w^* \nabla^2 g(\mb w) \mb w }{\norm{\mb w}{}^2} }\; ds\\
\le\; & g(\mb w) - \frac{t^2}{2} \betaconcave + t^2 \int_0^1\paren{1-s}s\Lconcave t\; ds \;\leq\; g(\mb w) - \frac{t^2}{2} \betaconcave + \frac{ t^3}{6} \Lconcave.
\end{align*}
Minimizing this function over $t \in \left[ 0, \frac{3\Delta}{2\pi \sqrt{n}} \right]$, we obtain 
\begin{align*}
t_\star = \min \set{ \frac{2\betaconcave}{\Lconcave}, \frac{3\Delta}{2\pi \sqrt{n}} },
\end{align*}
and there exists a $\mb w' = \mb w - t_\star \sigma \frac{\mb w}{\norm{\mb w}{}}$ such that
\begin{align*}
g\left( \mb w - t_\star \sigma \frac{\mb w}{\norm{\mb w}{}} \right) \;\le\; g(\mb w) - \min \set{ \frac{2 \betaconcave^3}{3 \Lconcave^2}, \frac{3\Delta^2 \betaconcave}{8\pi^2 n} }. 
\end{align*}
By arguments identical to those used in Lemma \ref{lem:alg_gradient_func}, there exists a tangent vector $\mb \delta\in T_{\mb q}\bb S^{n-1}$ such that $\mb q(\mb w') = \exp_{\mb q}(\mb \delta)$ and $\norm{\mb \delta}{}\leq \Delta$. This completes the proof. 
\end{proof}

\begin{proof}[of Lemma~\ref{lem:alg_strcvx_func}] \label{proof:lem_alg_strcvx_func}
For any $t \in \brac{0, \frac{\Delta}{\norm{\grad f\paren{\mb q^{(k)}}}{}}}$, it holds that $\norm{t\; \grad f\paren{\mb q^{(k)}}}{} \le \Delta$, and the quadratic approximation 
\begin{align*}
\widehat{f}\paren{\mb q^{(k)}, -t\; \grad f\paren{\mb q^{(k)}}}
& \le f\paren{\mb q^{(k)}} - t\norm{\grad f\paren{\mb q^{(k)}}}{}^2 + \frac{M_H}{2}t^2 \norm{\grad f\paren{\mb q^{(k)}}}{}^2 \\
& = f\paren{\mb q^{(k)}} - t\paren{1-\frac{1}{2}M_H t}\norm{\grad f\paren{\mb q^{(k)}}}{}^2. 
\end{align*}
Taking $t_0 = \min\set{\frac{\Delta}{\norm{\grad f\paren{\mb q^{(k)}}}{}}, \frac{1}{M_H}}$, we obtain 
\begin{align} \label{eq:alg_strcvx_key1}
\widehat{f}\paren{\mb q^{(k)}, -t_0\; \grad f\paren{\mb q^{(k)}}} \le f\paren{\mb q^{(k)}} - \frac{1}{2}\min\set{\frac{\Delta}{\norm{\grad f\paren{\mb q^{(k)}}}{}}, \frac{1}{M_H}} \norm{\grad f\paren{\mb q^{(k)}}}{}^2. 
\end{align}
Now let $\mb U$ be an arbitrary orthonormal basis for $T_{\mb q^{(k)}} \bb S^{n-1}$. Since the norm constraint is active, by the optimality condition in \eqref{eqn:ts-optimal-solution-1}, we have 
\begin{align*}
\Delta & \le \norm{\brac{\mb U^* \Hess f\paren{\mb q^{(k)}} \mb U}^{-1} \mb U^* \grad f\paren{\mb q^{(k)}}}{} \\
& \le \norm{\brac{\mb U^* \Hess f\paren{\mb q^{(k)}} \mb U}^{-1}}{} \norm{\mb U^* \grad f\paren{\mb q^{(k)}}}{} \le \frac{\norm{\grad f \paren{\mb q^{(k)}}}{}}{m_H},  
\end{align*} 
which means that $\norm{\grad f \paren{\mb q^{(k)}}}{} \ge m_H\Delta$. Substituting this into~\eqref{eq:alg_strcvx_key1}, we obtain 
\begin{align*}
\widehat{f}\paren{\mb q^{(k)}, -t_0\; \grad f\paren{\mb q^{(k)}}} \le f\paren{\mb q^{(k)}} - \frac{1}{2} \min \Brac{m_H \Delta^2, \frac{m_H^2}{M_H} \Delta^2  }\leq f\paren{\mb q^{(k)}}- \frac{m_H^2 \Delta^2 }{2M_H}.
\end{align*}
By the key comparison result established in proof of Lemma~\ref{lem:alg_approx_bd2}, we have 
\begin{align*}
f\paren{\exp_{\mb q^{(k)}}\paren{-t_0\grad f\paren{\mb q^{(k)}}}} & \le \widehat{f}\paren{\mb q^{(k)}, -t_0\; \grad f\paren{\mb q^{(k)}}} + \frac{1}{6} \eta_f\Delta^3 \\
& \le  f\paren{\mb q^{(k)}} - \frac{m_H^2 \Delta^2}{M_H} + \frac{1}{6} \eta_f\Delta^3. 
\end{align*}
This completes the proof. 
\end{proof}

It takes certain delicate work to prove Lemma~\ref{lem:alg_strcvx_lb}. Basically to use discretization argument, the degree of continuity of the Hessian is needed. The tricky part is that for continuity, we need to compare the Hessian operators at different points, while these Hessian operators are only defined on the respective tangent planes. This is the place where parallel translation comes into play. The next two lemmas compute spectral bounds for the forward and inverse parallel translation operators. 
\begin{lemma} \label{lem:alg_tsp_op}
For $\tau \in [0,1]$ and $\norm{\mb \delta}{} \le 1/ 2$, we have
\begin{eqnarray}
\norm{ \mc P_{\gamma}^{\tau \leftarrow 0} - \mb I }{} &\le& \frac{5}{4}\tau \norm{\mb \delta}{}, \\
\norm{ \mc P_{\gamma}^{0 \leftarrow \tau} - \mb I }{} &\le& \frac{3}{2}\tau \norm{\mb \delta}{}.
\end{eqnarray}
\end{lemma}
\begin{proof}
By~\eqref{eq:alg_tsp_op}, we have 
\begin{align*}
\norm{ \mc P_{\gamma}^{\tau \leftarrow 0} - \mb I }{}
& = \norm{\left( \cos( \tau \norm{\mb \delta}{} ) - 1 \right) \frac{\mb \delta \mb \delta^*}{\norm{\mb \delta }{}^2} - \sin\left( \tau \norm{\mb \delta }{} \right) \frac{\mb q \mb \delta^*}{\norm{\mb \delta}{}}}{} \\
& \le 1-\cos\paren{\tau \norm{\mb \delta}{}} + \sin\paren{\tau \norm{\mb \delta}{}} \\
& \le 2\sin^2\paren{\frac{\tau\norm{\mb \delta}{}}{2}} + \sin\paren{\tau \norm{\mb \delta}{}} \le \frac{1}{4}\tau \norm{\mb \delta}{} +  \tau \norm{\mb \delta}{} \le \frac{5}{4} \tau \norm{\mb \delta}{}, 
\end{align*}
where we have used the fact $\sin\paren{t} \le t$ and $1-\cos x = 2\sin^2\paren{x/2}$. Moreover, $\mc P_{\gamma}^{0 \leftarrow \tau}$ is in the form of $\paren{\mb I + \mb u \mb v^*}^{-1}$ for some vectors $\mb u$ and $\mb v$. By the Sherman-Morrison matrix inverse formula, i.e., $\paren{\mb I + \mb u \mb v^*}^{-1} = \mb I - \mb u \mb v^*/\paren{1+ \mb v^* \mb u}$ (justified as $\norm{\left( \cos( \tau \norm{\mb \delta}{} ) - 1 \right) \frac{\mb \delta \mb \delta^*}{\norm{\mb \delta }{}^2} - \mb q \sin\left( \tau \norm{\mb \delta }{} \right) \frac{\mb \delta^*}{\norm{\mb \delta}{}}}{} \le 5\tau\norm{\mb \delta}{}/4 \le 5/8 < 1$ as shown above), we have 
\begin{align*}
& \norm{ \mc P_{\gamma}^{0 \leftarrow \tau} - \mb I }{} \\
=\; & \norm{\left( \cos( \tau \norm{\mb \delta}{} ) - 1 \right) \frac{\mb \delta \mb \delta^*}{\norm{\mb \delta }{}^2} - \mb q \sin\left( \tau \norm{\mb \delta }{} \right) \frac{\mb \delta^*}{\norm{\mb \delta}{}}}{} \frac{1}{1+\paren{\cos\paren{\tau \norm{\mb \delta}{}}-1}} \quad (\text{as}\; \mb q^* \mb \delta = 0)\\
\le\; & \frac{5}{4}\tau \norm{\mb \delta}{} \frac{1}{\cos\paren{\tau \norm{\mb \delta}{}} } \le \frac{5}{4}\tau \norm{\mb \delta}{} \frac{1}{\cos\paren{1/2}} \le \frac{3}{2} \tau \norm{\mb \delta}{},  
\end{align*}
completing the proof. 
\end{proof}
The next lemma establish the ``local-Lipschitz" property of the Riemannian Hessian. 
\begin{lemma} \label{lem:rie_hess_lip}
Let $\gamma(t) = \exp_{\mb q}\paren{t\mb \delta}$ denotes a geodesic curve on $\bb S^{n-1}$. Whenever $\norm{\mb \delta}{} \le 1/2$ and $\tau \in [0,1]$,
\begin{eqnarray}
\norm{ \mc P_{\gamma}^{0 \leftarrow \tau} \Hess f (\gamma(\tau)) \mc P_{\gamma}^{\tau \leftarrow 0} - \Hess f (\mb q) }{} \le L_H\cdot\tau\norm{\mb \delta}{}, 
\end{eqnarray}
where $L_H = \frac{5}{2\mu^2}n^{3/2} \norm{\mb X}{\infty}^3 + \frac{9}{\mu}n \norm{\mb X}{\infty}^2 + 9\sqrt{n} \norm{\mb X}{\infty}$. 
\end{lemma}
\begin{proof}
First of all, by \eqref{eq:fq_rie_hess} and using the fact that the operator norm of a projection operator is unitary bounded, we have
\begin{align*}
& \norm{\Hess f (\gamma(\tau)) - \Hess f (\mb q)}{} \\
\le\; & \norm{\mc P_{T_{\gamma\paren{\tau}}\bb S^{n-1}} \brac{\nabla^2 f\paren{\gamma\paren{\tau}} - \nabla^2 f\paren{\mb q} -\paren{ \innerprod{\nabla f\paren{\gamma\paren{\tau}}}{\gamma\paren{\tau}} - \innerprod{\nabla f\paren{\mb q}}{\mb q} }\mb I} \mc P_{T_{\gamma\paren{\tau}}\bb S^{n-1} } }{} \\
& \qquad + \left\|\mc P_{T_{\gamma\paren{\tau}}\bb S^{n-1}}  \paren{\nabla^2 f\paren{\mb q} - \innerprod{\nabla f\paren{\mb q}}{\mb q} \mb I} \mc P_{T_{\gamma\paren{\tau}}\bb S^{n-1}}\right. \\
& \qquad \left.  - \mc P_{T_{\mb q}\bb S^{n-1} } \paren{\nabla^2 f\paren{\mb q} - \innerprod{\nabla f\paren{\mb q}}{\mb q} \mb I}  \mc P_{T_{\mb q}\bb S^{n-1}}\right\| \\
\le\; & \norm{\nabla^2 f\paren{\gamma\paren{\tau}} - \nabla^2 f\paren{\mb q}}{} + \abs{\innerprod{\nabla f\paren{\gamma\paren{\tau}} - \nabla f\paren{\mb q}}{\gamma\paren{\tau}}} + \abs{\innerprod{\nabla f\paren{\mb q}}{\gamma\paren{\tau} - \mb q}} \\
& \qquad + \norm{\mc P_{T_{\gamma\paren{\tau}}\bb S^{n-1} } - \mc P_{T_{\mb q}\bb S^{n-1} }}{} \norm{\mc P_{T_{\gamma\paren{\tau}}\bb S^{n-1}} + \mc P_{T_{\mb q}\bb S^{n-1} }}{}\norm{\nabla^2 f\paren{\mb q} - \innerprod{\nabla f\paren{\mb q}}{\mb q} \mb I}{}. 
\end{align*}
By the estimates in Lemma~\ref{lem:mag_lip_fq}, we obtain
\begin{align}
& \norm{\Hess f (\gamma(\tau)) - \Hess f (\mb q)}{} \nonumber  \\
\le\; &  \frac{2}{\mu^2}n^{3/2} \norm{\mb X}{\infty}^3 \norm{\gamma\paren{\tau} - \mb q}{} + \frac{n}{\mu} \norm{\mb X}{\infty}^2  \norm{\gamma\paren{\tau} - \mb q}{} + \sqrt{n} \norm{\mb X}{\infty} \norm{\gamma\paren{\tau} - \mb q}{} \nonumber \\
& \qquad + 2\norm{\gamma\paren{\tau} \gamma^* \paren{\tau} - \mb q \mb q^*}{} \paren{\frac{n}{\mu}\norm{\mb X}{\infty}^2 + \sqrt{n} \norm{\mb X}{\infty}} \nonumber \\
\le & \paren{\frac{5}{2\mu^2} n^{3/2} \norm{\mb X}{\infty}^3 + \frac{25n}{4\mu} \norm{\mb X}{\infty}^2 + \frac{25}{4} \sqrt{n} \norm{\mb X}{\infty}} \tau \norm{\mb \delta}{}, \label{eq:alg_lip_key1}
\end{align}
where at the last line we have used the following estimates: 
\begin{align*}
\norm{\gamma\paren{\tau} - \mb q}{} & = \norm{\mb q \paren{\cos\paren{\tau \norm{\mb \delta}{}} -1} + \frac{\mb \delta}{\norm{\mb \delta}{}} \sin\paren{\tau \norm{\mb \delta}{}}}{} \le \frac{5}{4} \tau \norm{\mb \delta}{}, \; (\text{Proof of Lemma~\ref{lem:alg_tsp_op}}) \\
\norm{\gamma\paren{\tau} \gamma^* \paren{\tau} - \mb q \mb q^*}{} & \le \norm{\paren{\frac{\mb \delta \mb \delta^*}{\norm{\mb \delta}{}^2} - \mb q \mb q^*}\sin^2\paren{\tau \norm{\mb \delta}{}}}{} + 2\sin\paren{\tau \norm{\mb \delta}{}}\cos\paren{\tau \norm{\mb \delta}{}} \\
& \le \sin^2\paren{ \tau \norm{\mb \delta}{}} + \sin\paren{2\tau \norm{\mb \delta}{}} \le \frac{5}{2}\tau \norm{\mb \delta}{}. 
\end{align*}
Therefore, by Lemma \ref{lem:alg_tsp_op}, we obtain 
\begin{align*}
& \norm{ \mc P_{\gamma}^{0 \leftarrow \tau} \Hess f (\gamma(\tau)) \mc P_{\gamma}^{\tau \leftarrow 0} - \Hess f (\mb q) }{}  \\
\le\; & \norm{ \mc P_{\gamma}^{0 \leftarrow \tau} \Hess f (\gamma(\tau)) \mc P_{\gamma}^{\tau \leftarrow 0} -  \Hess f (\gamma(\tau)) \mc P_{\gamma}^{\tau \leftarrow 0}  }{} + \norm{ \Hess f (\gamma(\tau)) \mc P_{\gamma}^{\tau \leftarrow 0}  -  \Hess f (\gamma(\tau)) }{}  \\
& \qquad + \norm{\Hess f (\gamma(\tau)) - \Hess f (\mb q)}{} \nonumber \\
\le\; & \norm{\mc P_{\gamma}^{0 \leftarrow \tau} - \mb I}{} \norm{\Hess f (\gamma(\tau))}{} + \norm{\mc P_{\gamma}^{\tau \leftarrow 0} - \mb I}{} \norm{\Hess f (\gamma(t))}{} + \norm{\Hess f (\gamma(t)) - \Hess f (\mb q)}{} \\
\le\; & \frac{11}{4} \tau \norm{\mb \delta}{} \norm{\nabla^2 f\paren{\gamma\paren{\tau}} - \innerprod{\nabla f\paren{\gamma\paren{\tau}}}{\gamma\paren{t}} \mb I}{} + \norm{\Hess f (\gamma(\tau)) - \Hess f (\mb q)}{}. 
\end{align*}
By Lemma~\ref{lem:mag_lip_fq} and substituting the estimate in~\eqref{eq:alg_lip_key1}, we obtain the claimed result. 
\end{proof}

\begin{proof}[of Lemma~\ref{lem:alg_strcvx_lb}] \label{proof:lem_alg_strcvx_lb}
For any given $\mb q$ with $\norm{\mb w(\mb q)}{} \le \mu/(4\sqrt{2})$, assume $\mb U$ is an orthonormal basis for its tangent space $T_{\mb q}\bb S^{n-1}$. We could compare $\mb U^* \Hess f(\mb q) \mb U$ with $\nabla_{\mb w}^2 g(\mb w)$, and build on the known results for the latter. Instead, we present a direct proof here that yields tighter results as stated in the lemma. Again we first work with the ``canonical'' section in the vicinity of $\mb e_n$ with the ``canonical" reparametrization $\mb q(\mb w) = [\mb w; \sqrt{1-\|\mb w\|^2}]$. 

By definition of the Riemannian Hessian in~\eqref{eq:fq_rie_hess}, expressions of $\nabla^2 f$ and $\nabla f$ in~\eqref{eq:fq_grad} and~\eqref{eq:fq_hess}, and exchange of differential and expectation opeators (justified similarly as in Section~\ref{sec:proof_geo_exp}), we obtain 
\begin{align*}
\mb U^* \Hess \expect{f(\mb q)} \mb U 
& = \expect{\mb U^* \Hess f(\mb q) \mb U} \\
& = \expect{\mb U^* \nabla^2 f(\mb q) \mb U - \innerprod{\mb q}{\nabla f(\mb q)} \mb I_{n-1}} \\
& = \mb U^* \expect{\frac{1}{\mu} \Brac{1-\tanh^2\paren{\frac{\mb q^* \mb x}{\mu}}} \mb x \mb x^*} \mb U - \expect{\tanh\paren{\frac{\mb q^* \mb x}{\mu}} \mb q^* \mb x} \mb I_{n-1}. 
\end{align*} 
We have 
\begin{align*}
\mb U^* \expect{\frac{1}{\mu} \Brac{1-\tanh^2\paren{\frac{\mb q^* \mb x}{\mu}}} \mb x \mb x^*} \mb U \succeq \frac{1-\theta}{\mu} \mb U^* \expect{\Brac{1-\tanh^2\paren{\frac{\mb w^* \ol{\mb x}}{\mu}}}
\begin{bmatrix}
\ol{\mb x}\; \ol{\mb x}^* & \mb 0 \\
\mb 0^*   &  0
\end{bmatrix}} \mb U.
\end{align*}
Now consider any vector $\mb z \in T_{\mb q} \bb S^{n-1}$ such that $\mb z = \mb U \mb v$ for some $\mb v \in \R^{n-1}$ and $\|\mb z\| = 1$. Then 
\begin{align*}
\mb z^* \expect{\Brac{1-\tanh^2\paren{\frac{\mb w^* \ol{\mb x}}{\mu}}}
\begin{bmatrix}
\ol{\mb x} \;\ol{\mb x}^* & \mb 0 \\
\mb 0^*   &  0
\end{bmatrix}} \mb z \ge \frac{\theta}{\sqrt{2\pi}} (2 - 3\sqrt{2}/4) \|\ol{\mb z}\|^2 
\end{align*}
by proof of Proposition~\ref{prop:geometry_asymp_strong_convexity}, where $\ol{\mb z} \in \R^{n-1}$ as above is the first $n-1$ coordinates of $\mb z$. Now we know that $\innerprod{\mb q}{\mb z} = 0$, or 
\begin{align*}
\mb w^* \ol{\mb z} + q_n z_n = 0 \Longrightarrow \frac{\|\ol{\mb z}\|}{\abs{z_n}} = \frac{q_n}{\|\mb w\|} = \frac{\sqrt{1-\|\mb w\|^2}}{\|\mb w\|} \ge 50, 
\end{align*}
where we have used $\|\mb w\| \le \mu/(4\sqrt{2})$ and $\mu \le 1/10$ to obtain the last lower bound. Combining the above with the fact that $\|\mb z\| = 1$, we obtain 
\begin{align}
\mb U^* \expect{\frac{1}{\mu} \Brac{1-\tanh^2\paren{\frac{\mb q^* \mb x}{\mu}}} \mb x \mb x^*} \mb U 
& \succeq \frac{99}{100} \frac{1-\theta}{\mu} \frac{\theta}{\sqrt{2\pi}} (2 - 3\sqrt{2}/4) \mb I_{n-1} \\
&  \succeq \frac{99}{200\sqrt{2\pi}} (2 - 3\sqrt{2}/4) \frac{\theta}{\mu} \mb I_{n-1}, 
\end{align}
where we have simplified the expression using $\theta \le 1/2$. To bound the second term, 
\begin{align*}
& \expect{\tanh\paren{\frac{\mb q^* \mb x_k}{\mu}} \mb q^* \mb x_k} \\
=\; & \bb E_{\mc I} \brac{\bb E_{Z \sim \mc N\paren{0, \norm{\mb q_{\mc I}}{}^2}}\brac{\tanh(Z/\mu) Z}} \\
=\; & \frac{1}{\mu}\bb E_{\mc I} \brac{\|\mb q_{\mc I}\|^2\bb E_{Z \sim \mc N\paren{0, \norm{\mb q_{\mc I}}{}^2}}\brac{1-\tanh^2(Z/\mu)}}   \quad \text{(by Lemma B.1 in~\cite{sun2015complete_a})}  \\
\le\; & \frac{1}{\mu}\bb E_{\mc I} \brac{\bb E_{Z \sim \mc N\paren{0, \norm{\mb q_{\mc I}}{}^2}}\brac{1-\tanh^2(Z/\mu)}}. 
\end{align*}
Now we have the following estimate: 
\begin{align*}
& \bb E_{Z \sim \mc N\paren{0, \norm{\mb w_{\mc J}}{}^2 + q_n^2}}\brac{1-\tanh^2(Z/\mu)} \\
=\; & 2 \bb E_{Z \sim \mc N\paren{0, \norm{\mb w_{\mc J}}{}^2 + q_n^2}}\brac{\paren{1-\tanh^2(Z/\mu)} \indicator{Z > 0}} \\
\le\; & 8 \bb E_{Z \sim \mc N\paren{0, \norm{\mb w_{\mc J}}{}^2 +  q_n^2}}\brac{\exp(-2Z/\mu) \indicator{Z > 0}} \\
=\; & 8 \exp\paren{\frac{2\norm{\mb w_{\mc J}}{}^2 + 2 q_n^2}{\mu^2}} \Phi^c\paren{\frac{2\sqrt{\norm{\mb w_{\mc J}}{}^2 +  q_n^2}}{\mu}}  \quad \text{(by Lemma B.1 in~\cite{sun2015complete_a})} \\
\le\; & \frac{4}{\sqrt{2\pi}} \frac{\mu}{\sqrt{\norm{\mb w_{\mc J}}{}^2 +  q_n^2}}, 
\end{align*}
where at the last inequality we have applied Gaussian tail upper bound of Type II in Lemma~\ref{lem:gaussian_tail_est}.  Since $\norm{\mb w_{\mc J}}{}^2 +  q_n^2 \ge q_n^2 = 1-\norm{\mb w}{}^2 \ge 1 - \mu^2/32 \ge 31/32$ for $\norm{\mb w}{} \le \mu/(4\sqrt{2})$ and $\mu \le 1$, we obtain 
\begin{align} \label{eq:strcvx_qsp_key2}
\bb E_{Z \sim \mc N\paren{0, \norm{\mb w_{\mc J}}{}^2 + q_n^2}}\brac{1-\tanh^2(Z/\mu)} \le \frac{4}{\sqrt{2\pi}} \frac{\mu}{\sqrt{31/32}} \le \frac{4}{\sqrt{2\pi}} \mu. 
\end{align}
Collecting the above estimates, we obtain 
\begin{align} \label{eq:strcvx_qsp_key3}
\mb U^* \Hess \expect{f(\mb q)} \mb U \succeq \frac{99}{200\sqrt{2\pi}} (2 - 3\sqrt{2}/4) \frac{\theta}{\mu} \mb I_{n-1} - \frac{1}{\mu}\frac{4}{\sqrt{2\pi}} \mu \mb I_{n-1} \succeq \frac{1}{4\sqrt{2\pi}} \frac{\theta}{\mu} \mb I_{n-1}, 
\end{align}
where we have used the fact $\mu \le \theta/10$ to obtain the final lower bound.

Next we perform concentration analysis. For any $\mb q$, we can write 
\begin{align*}
\mb U^* \nabla^2 f(\mb q) \mb U = \frac{1}{p} \sum_{k=1}^p \mb W_k, \quad \text{with}\; \mb W_k \doteq \frac{1}{\mu} \brac{1-\tanh^2\paren{\frac{\mb q^* \mb x_k}{\mu}}} \mb U^* \mb x_k \mb x_k^* \mb U. 
\end{align*}
For any integer $m \ge 2$, we have 
\begin{align*}
\mb  0 \preceq \expect{\mb W_k^m} \preceq \frac{1}{\mu^m} \expect{\paren{\mb U^* \mb x_k \mb x_k^* \mb U}^m} \preceq \frac{1}{\mu^m} \expect{\norm{\mb x_k \mb x_k^*}{}^m}\mb I = \frac{1}{\mu^m} \expect{\norm{\mb x_k}{}^{2m}} \mb I \preceq \frac{1}{\mu^m} \bb E_{Z \sim \xi^2\paren{n}} \brac{Z^m} \mb I,  
\end{align*}
where we have used Lemma~\ref{lem:U-moments-bound} to obtain the last inequality. By Lemma~\ref{lem:chi_sq_moment}, we obtain 
\begin{align*}
\mb 0 \preceq \expect{\mb W_k^m} \preceq \frac{1}{\mu^m} \frac{m!}{2} \paren{2n}^m \mb I \preceq \frac{m!}{2} \paren{\frac{2n}{\mu}}^m \mb I. 
\end{align*}
Taking $R_{\mb W} = 2n/\mu$, and $\sigma^2_{\mb W} = 4n^2/\mu^2 \ge \expect{\mb W_k^2}$, by Lemma~\ref{lem:mc_bernstein_matrix}, we obtain 
\begin{align} \label{eq:strcvx_qsp_key4}
\prob{\norm{\frac{1}{p} \sum_{k=1}^p \mb W_k - \frac{1}{p} \sum_{k=1}^p \expect{\mb W_k}}{} > t/2} \le 2n\exp\paren{-\frac{p\mu^2t^2}{32n^2 + 8nt}}
\end{align}
for any $t > 0$. Similarly, we write
\begin{align*}
\innerprod{\nabla f(\mb q)}{\mb q} = \frac{1}{p} \sum_{k=1}^p Z_k, \quad \text{with}\; Z_k \doteq \tanh\paren{\frac{\mb q^* \mb x_k}{\mu}} \mb q^* \mb x_k. 
\end{align*}
For any integer $m \ge 2$, we have 
\begin{align*}
\expect{\abs{Z_k}^m} \le \expect{\abs{\mb q^* \mb x_k}^m} \le \bb E_{Z \sim \mc N\paren{0, 1}}\brac{\abs{Z}^m} \le \frac{m!}{2}, 
\end{align*}
where at the first inequality we used the fact $\abs{\tanh(\cdot)} \le 1$, at the second we invoked Lemma~\ref{lem:U-moments-bound}, and at the third we invoked Lemma~\ref{lem:guassian_moment}. Taking $R_Z = \sigma^2_{Z} = 1$, by Lemma~\ref{lem:mc_bernstein_scalar}, we obtain 
\begin{align} \label{eq:strcvx_qsp_key5}
\prob{\abs{\frac{1}{p}\sum_{k=1}^p Z_k - \frac{1}{p}\sum_{k=1}^p \expect{Z_k}} > t/2} \le 2\exp\paren{-pt^2/16}
\end{align}
for any $t > 0$. Gathering~\eqref{eq:strcvx_qsp_key4} and~\eqref{eq:strcvx_qsp_key5}, we obtain that for any $t > 0$, 
\begin{align} \label{eq:strcvx_qsp_key6}
& \prob{\norm{\mb U^* \Hess \expect{f(\mb q)} \mb U - \mb U^* \Hess f(\mb q) \mb U}{} > t} \nonumber \\
\le\; &  \prob{\norm{\mb U^* \nabla^2 f(\mb q) \mb U - \nabla^2 \expect{f(\mb q)}}{} > t/2} + \prob{\abs{\innerprod{\nabla f(\mb q)}{\mb q} - \innerprod{\nabla \expect{f(\mb q)}}{\mb q} }> t/2} \nonumber \\
\le\; & 2n\exp\paren{-\frac{p\mu^2 t^2}{32n^2 + 8nt}} + 2\exp\paren{-\frac{pt^2}{16}} \le 4n \exp\paren{-\frac{p\mu^2 t^2}{32n^2 + 8nt}}. 
\end{align}

Now we are ready to pull above results together for a discretization argument. For any $\eps \in (0, \mu/(4\sqrt{2}))$, there is an $\eps$-net $N_{\eps}$ of size at most $(3\mu/(4\sqrt{2} \eps))^n$ that covers the region $\set{\mb q: \norm{\mb w(\mb q)}{} \le \mu/(4\sqrt{2})}$. By Lemma~\ref{lem:rie_hess_lip}, the function $\Hess f(\mb q)$ is locally Lipschitz within each normal ball of radius 
\begin{align*}
\norm{\mb q - \exp_{\mb q} (1/2)}{} = \sqrt{2-2\cos(1/2)} \ge 1/\sqrt{5}
\end{align*}
with Lipschitz constant $L_H$ (as defined in Lemma~\ref{lem:rie_hess_lip}). Note that $\eps < \mu/(4\sqrt{2}) < 1/(4\sqrt{2}) < 1/\sqrt{5}$ for $\mu < 1$, so any choice of $\eps \in (0, \mu/(4\sqrt{2}))$ makes the Lipschitz constant $L_H$ valid within each $\eps$-ball centered around one element of the $\eps$-net. Let 
\begin{align*}
\event_{\infty} \doteq \set{1 \le \norm{\mb X_0}{\infty} \le 4\sqrt{\log(np)}}. 
\end{align*}
From Lemma~\ref{lem:X-infinty-tail-bound},  $\prob{\event_{\infty}^c} \le \theta \paren{np}^{-7} + \exp\paren{-0.3\theta np}$. By Lemma~\ref{lem:rie_hess_lip}, with at least the same probability, 
\begin{align*}
L_H \le C_1 \frac{n^{3/2}}{\mu^2} \log^{3/2} (np). 
\end{align*}
Set $\eps = \frac{\theta}{12\sqrt{2\pi}\mu L_H} <\mu/(4\sqrt{2})$, so 
\begin{align*}
\# N_\eps \le \exp\paren{n\log\frac{C_2 n^{3/2} \log^{3/2}(np)}{\theta}}. 
\end{align*}
Let $\event_H$ denote the event that 
\begin{align*}
\event_H \doteq \set{\max_{\mb q \in N_{\eps}} \norm{\mb U^* \Hess \expect{f(\mb q)} \mb U - \mb U^* \Hess f(\mb q) \mb U}{} \le \frac{\theta}{12\sqrt{2\pi}\mu} }. 
\end{align*}
On $\event_{\infty} \cap \event_H$, 
\begin{align*}
\sup_{\mb q: \norm{\mb w(\mb q)}{} \le \mu/(4\sqrt{2})} \norm{\mb U^* \Hess \expect{f(\mb q)} \mb U - \mb U^* \Hess f(\mb q) \mb U}{} \le \frac{\theta}{6\sqrt{2\pi}\mu}. 
\end{align*}
So on $\event_{\infty} \cap \event_H$, we have 
\begin{align}
\mb U^* \Hess f(\mb q) \mb U \succeq c_\sharp \frac{\theta}{\mu}
\end{align}
for any $c_\sharp \le 1/(12\sqrt{2\pi})$. Setting $t = \frac{\theta}{12\sqrt{2\pi}\mu}$ in~\eqref{eq:strcvx_qsp_key6}, we obtain that for any fixed $\mb q$ in this region, 
\begin{align*}
\prob{\norm{\mb U^* \Hess \expect{f(\mb q)} \mb U - \mb U^* \Hess f(\mb q) \mb U}{} > t} \le 4n \exp\paren{-\frac{p\theta^2}{c_3n^2 + c_4n\theta/\mu }}. 
\end{align*}
Taking a union bound, we obtain that 
\begin{align*}
\prob{\event_H^c} \le 4n\exp\paren{-\frac{p\theta^2}{c_3n^2 + c_4n\theta/\mu } + C_5 n\log n + C_6 n\log\log p}. 
\end{align*}
It is enough to make $p \ge C_7 n^3\log (n/(\mu \theta))/(\mu \theta^2)$ to make the failure probability small, completing the proof. 
\end{proof}

\begin{proof}[of Lemma~\ref{lem:alg_gradient_lb}] \label{proof:lem_alg_gradient_lb}
For a given $\mb q$, consider the vector $\mb r \doteq \mb q - \mb e_n/q_n$. It is easy to verify that $\innerprod{\mb q}{\mb r} = 0$, and hence $\mb r \in T_{\mb q} \bb S^{n-1}$. Now, by \eqref{eq:fq_grad} and \eqref{eq:fq_rie_grad}, we have 
\begin{align*}
\innerprod{\grad f\paren{\mb q}}{\mb r} 
& = \innerprod{\paren{\mb I - \mb q \mb q^*} \nabla f\paren{\mb q}}{\mb q - \mb e_n/q_n} \\
& = \innerprod{\paren{\mb I - \mb q \mb q^*} \nabla f\paren{\mb q}}{- \mb e_n/q_n} \\
& = \frac{1}{p}\sum_{k=1}^p \innerprod{\paren{\mb I - \mb q \mb q^*} \tanh\paren{\frac{\mb q^* \mb x_k}{\mu}} \mb x_k}{-\mb e_n/q_n}   \\
& = \frac{1}{p} \sum_{k=1}^p \tanh\paren{\frac{\mb q^* \mb x_k}{\mu}} \paren{-\frac{x_k\paren{n}}{q_n} + \mb q^* \mb x_k} \\
& = \frac{1}{p} \sum_{k=1}^p \tanh\paren{\frac{\mb q^* \mb x_k}{\mu}} \paren{\mb w^*\paren{\mb q} \overline{\mb x}_k - \frac{x_k\paren{n}}{q_n}\norm{\mb w\paren{\mb q}}{}^2} \\
& = \mb w^*\paren{\mb q} \nabla g\paren{\mb w}, 
\end{align*}
where to get the last line we have used~\eqref{eqn:lse-gradient}. Thus, 
\begin{align*}
\frac{\mb w^* \nabla g\paren{\mb w}}{\norm{\mb w}{}} = \frac{\innerprod{\grad f \paren{\mb q}}{\mb r}}{\norm{\mb w}{}} \le \norm{\grad f\paren{\mb q}}{} \frac{\norm{\mb r}{}}{\norm{\mb w}{}}, 
\end{align*}
where 
\begin{align*}
\frac{\norm{\mb r}{}^2}{\norm{\mb w}{}^2} = \frac{\norm{\mb w}{}^2 + \paren{q_n - \frac{1}{q_n}}^2}{\norm{\mb w}{}^2} = \frac{\norm{\mb w}{}^2 + \norm{\mb w}{}^4/q_n^2}{\norm{\mb w}{}^2} = \frac{1}{q_n^2} = \frac{1}{1-\norm{\mb w}{}^2} \le \frac{1}{1-\tfrac{1}{2000}} = \frac{2000}{1999}, 
\end{align*}
where we have invoked our assumption that $\norm{\mb w}{} \le \tfrac{1}{20\sqrt{5}}$. Therefore we obtain 
\begin{align*}
\norm{\grad f\paren{\mb q}}{} \ge \frac{\norm{\mb w}{}}{\norm{\mb r}{}} \frac{\mb w^* \nabla g\paren{\mb w}}{\norm{\mb w}{}} \ge \sqrt{\frac{1999}{2000}} \frac{\mb w^* \nabla g\paren{\mb w}}{\norm{\mb w}{}} \ge \frac{9}{10} \frac{\mb w^* \nabla g\paren{\mb w}}{\norm{\mb w}{}}, 
\end{align*}
completing the proof. 
\end{proof}

Proof of Lemma~\ref{lem:TR-step} combines the local Lipschitz property of $\Hess f(\mb q)$ in Lemma~\ref{lem:rie_hess_lip}, and the Taylor's theorem (manifold version, Lemma 7.4.7 of~\cite{absil2009}).  

\begin{proof}[of Lemma~\ref{lem:TR-step}] \label{proof:lem_TR-step}
Let $\gamma\paren{t}$ be the unique geodesic that satisfies $\gamma\paren{0} = \mb q^{(k)}$, $\gamma\paren{1} = \mb q^{(k+1)}$, and its directional derivative $\dot{\gamma}\paren{0} = \mb \delta_\star$. Since the parallel translation defined by the Riemannian connection is an isometry, then $\norm{\grad f (\mb q^{(k+1)}) }{} = \norm{\mc P_{\gamma}^{0 \leftarrow 1}\grad f(\mb q^{(k+1)}) }{}$. Moreover, since $\norm{\mb \delta_\star}{} \leq \Delta$, the unconstrained optimality condition in \eqref{eqn:ts-optimal-solution-1} implies that $\grad f(\mb q^{(k)}) + \Hess f(\mb q^{(k)}) \mb \delta_\star = \mb 0_{\mb q^{(k)}}$. Thus, by using Taylor's theorem in \cite{absil2009}, we have 
\begin{align*}
\norm{\grad f (\mb q^{(k+1)})}{} 
& = \norm{\mc P_{\gamma}^{0 \leftarrow 1}\grad f \paren{\mb q^{(k+1)}} - \grad f \paren{\mb q^{(k)}} - \Hess f\paren{\mb q^{(k)}} \mb \delta_\star}{} \\
& = \norm{\int_0^1 \brac{\mc P_{\gamma}^{0 \leftarrow t} \Hess f\paren{\gamma\paren{t}}\brac{\dot{\gamma}\paren{t}} -  \Hess f\paren{\mb q^{(k)}} \mb \delta_\star} \; dt}{} \; (\text{Taylor's theorem}) \\
& = \norm{\int_0^1 \paren{\mc P_{\gamma}^{0 \leftarrow t} \Hess f\paren{\gamma\paren{t}}\mc P_{\gamma}^{t \leftarrow 0} \mb \delta_\star -  \Hess f\paren{\mb q^{(k)}} \mb \delta_\star} \; dt}{}\\
& \le \norm{\mb \delta_\star}{} \int_0^1 \norm{\mc P_{\gamma}^{0 \leftarrow t} \Hess f\paren{\gamma\paren{t}}\mc P_{\gamma}^{t \leftarrow 0} -  \Hess f\paren{\mb q^{(k)}}}{} \; dt. 
\end{align*}
From the Lipschitz bound in Lemma~\ref{lem:rie_hess_lip} and the optimality condition in \eqref{eqn:ts-optimal-solution-1}, we obtain 
\begin{align*}
\norm{\grad f \paren{\mb q^{(k+1)}}}{}  
& \le \frac{1}{2}\norm{\mb \delta_\star}{}^2 L_H 
= \frac{L_H}{2m_H^2} \norm{\grad f \paren{\mb q^{(k)}}}{}^2.
\end{align*}
This completes the proof. 
\end{proof}

\begin{proof}[of Lemma~\ref{lem:TR-grad-opt}] \label{proof:lem_TR-grad-opt}
By invoking Taylor's theorem in \cite{absil2009}, we have
\begin{align*}
\mc P_{\gamma}^{0 \leftarrow \tau} \grad f \paren{\gamma\paren{\tau}} = \int_{0}^{\tau} \mc P_{\gamma}^{0 \leftarrow t} \Hess f\paren{\gamma\paren{t}} [\dot{\gamma}\paren{t}]\; dt. 
\end{align*}
Hence, we have 
\begin{align*}
\innerprod{\mc P_{\gamma}^{0 \leftarrow \tau} \grad f \paren{\gamma\paren{\tau}}}{\mb \delta} 
& = \int_{0}^{\tau} \innerprod{ \mc P_{\gamma}^{0 \leftarrow t} \Hess f\paren{\gamma\paren{t}} [\dot{\gamma}\paren{t}]}{\mb \delta}\; dt \\
& = \int_{0}^{\tau} \innerprod{ \mc P_{\gamma}^{0 \leftarrow t} \Hess f\paren{\gamma\paren{t}} [\dot{\gamma}\paren{t}]}{\mc P_{\gamma}^{0 \leftarrow t }\dot{\gamma}\paren{t}}\; dt \\
& = \int_{0}^{\tau} \innerprod{\Hess f\paren{\gamma\paren{t}} [\dot{\gamma}\paren{t}]}{\dot{\gamma}\paren{t}}\; dt\\
& \ge m_H \int_{0}^\tau \norm{\dot{\gamma}\paren{t}}{}^2 \; dt \ge m_H\tau \norm{\mb \delta}{}^2, 
\end{align*}
where we have used the fact that the parallel transport $\mc P_{\gamma}^{0 \leftarrow t}$ defined by the Riemannian connection is an isometry. On the other hand, we have 
\begin{align*}
\innerprod{\mc P_{\gamma}^{0 \leftarrow \tau} \grad f \paren{\gamma\paren{\tau}}}{\mb \delta} & \le \norm{\mc P_{\gamma}^{0 \leftarrow \tau} \grad f \paren{\gamma\paren{\tau}}}{} \norm{\mb \delta}{} = \norm{\grad f \paren{\gamma\paren{\tau}}}{} \norm{\mb \delta}{},
\end{align*}
where again used the isometry property of the operator $\mc P_{\gamma}^{0 \leftarrow \tau}$. Combining the two bounds above, we obtain 
\begin{align*}
\norm{\grad f \paren{\gamma\paren{\tau}}}{} \norm{\mb \delta}{} \ge m_H \tau \norm{\mb \delta}{}^2,  
\end{align*}
which implies the claimed result. 
\end{proof}

%% file: sec/proof_main.tex
\section{Proofs of Technical Results for Section~\ref{sec:main_result}} \label{sec:proof_main}
We need one technical lemma to prove Lemma~\ref{lem:alg_rounding_orth} and the relevant lemma for complete dictionaries. 
\begin{lemma} \label{lem:rounding-0}
There exists a positive constant $C$, such that for all integer $n_1 \in \N$, $\theta \in \paren{0, 1/3}$, and $n_2 \in \N$ with $n_2 \ge Cn_1\log\paren{n_1/\theta}/\theta^2$, any random matrix $\mb M \in \R^{n_1 \times n_2} \sim_{i.i.d.} \mathrm{BG}(\theta)$ obeys the following. For any fixed index set $\mc I \subset [n_2]$ with $\abs{\mc I} \leq \frac{9}{8} \theta n_2$, it holds that 
\begin{align*}
\norm{\mb v^* \mb M_{\mc I^c}}{1} - \norm{\mb v^* \mb M_{\mc I}}{1} \ge \frac{n_2}{6}\sqrt{\frac{2}{\pi}} \theta \norm{\mb v}{}\quad \text{for all}\; \mb v \in \R^{n_1}, 
\end{align*}
with probability at least 
$
1-n_2^{-10} - \theta\paren{n_1n_2}^{-7} - \exp\paren{-0.3\theta n_1 n_2}. 
$
\end{lemma}
\begin{proof}
By homogeneity, it is sufficient to consider all $\mb v \in \bb S^{n_1}$. For any $i \in [n_2]$, let $\mb m_i\in \bb R^{n_1}$ be a column of $\mb M$. For a fixed $\mb v$ such that $\norm{\mb v}{} = 1$, we have
\begin{align*}
T\paren{\mb v} \doteq \norm{\mb v^* \mb M_{\mc I^c}}{1} - \norm{\mb v^* \mb M_{\mc I}}{1} = \sum_{i \in \mc I^c} \abs{\mb v^* \mb m_i} - \sum_{i \in \mc I} \abs{\mb v^* \mb m_i},
\end{align*}
namely as a sum of independent random variables. Since $\abs{\mc I} \le 9n_2\theta/8$, we have 
\begin{align*}
\expect{T\paren{\mb v}} \ge \paren{n_2 - \frac{9}{8} \theta n_2 - \frac{9}{8} \theta n_2} \expect{\abs{\mb v^* \mb m_1}} = \paren{1-\frac{9}{4} \theta } n_2 \expect{\abs{\mb v^* \mb m_1}} \geq \frac{1}{4}n_2 \expect{\abs{\mb v^* \mb m_1}},
\end{align*}
where the expectation $\expect{\abs{\mb v^* \mb m_1}}$ can be lower bounded as 
\begin{align*}
\expect{\abs{\mb v^* \mb m_1}} 
& \;=\; \sum_{k=0}^{n_1} \theta^k \paren{1-\theta}^{n_1 - k} \sum_{\mc J \in \binom{[n_1]}{k}} \bb E_{\mb g \sim \mc N\paren{\mb 0, \mb I}} \brac{\abs{\mb v^*_{\mc J} \mb g}}\\
& \;=\; \sum_{k=0}^{n_1} \theta^k \paren{1-\theta}^{n_1 - k} \sum_{\mc J \in \binom{[n_1]}{k}} \sqrt{\frac{2}{\pi}} \norm{\mb v_{\mc J}}{} 
\geq \sqrt{\frac{2}{\pi}} \norm{\bb E_{\mc J}\brac{\mb v_{\mc J}}}{} = \sqrt{\frac{2}{\pi}} \theta. 
\end{align*}
Moreover, by Lemma~\ref{lem:U-moments-bound} and Lemma \ref{lem:guassian_moment}, for any $i \in [n_2]$ and any integer $m \ge 2$, 
\begin{align*}
\expect{\abs{\mb v^* \mb m_i}^m} \le \bb E_{\mb Z \sim \mc N\paren{0, 1}} \brac{\abs{Z}^m } \le \paren{m-1}!! \le \frac{m!}{2}. 
\end{align*} 
So invoking the moment-control Bernstein's inequality in Lemma~\ref{lem:mc_bernstein_scalar}, we obtain 
\begin{align*}
\prob{T\paren{\mb v} < \frac{n_2}{4}\sqrt{\frac{2}{\pi}} \theta - t} \le \prob{T\paren{\mb v} < \expect{T\paren{\mb v}} - t} \le \exp\paren{-\frac{t^2}{2n_2 + 2t}}. 
\end{align*}
Taking $t = \tfrac{n_2}{20}\sqrt{\tfrac{2}{\pi}} \theta$ and simplifying, we obtain that 
\begin{align}
\prob{T\paren{\mb v} < \frac{n_2}{5}\sqrt{\frac{2}{\pi}} \theta} \le \exp\paren{-c_1 \theta^2 n_2}
\end{align}
for some positive constant $c_1$. Fix $\eps = \sqrt{\frac{2}{\pi}}\frac{\theta}{120} \brac{ n_1 \log \paren{n_1 n_2}}^{-1/2} < 1$. The unit sphere $\bb S^{n_1}$ has an $\eps$-net $N_\eps$ of cardinality at most $\paren{3/\eps}^{n_1}$. Consider the event
\begin{align*}
\event_{bg} \doteq \set{T\paren{\mb v} \ge \frac{n_2}{5}\sqrt{\frac{2}{\pi}} \theta\;\; \forall\; \mb v \in N_\eps}. 
\end{align*} 
A simple union bound implies 
\begin{align}\label{eqn:rounding-union}
\prob{\event_{bg}^c} 
& \le \exp\paren{-c_1\theta^2 n_2 + n_1\log\paren{\frac{3}{\eps}}}
\le \exp\paren{-c_1\theta^2 n_2 + c_2n_1\log \frac{n_1 \log n_2}{\theta}},  
\end{align}
where $c_2 > 0$ is numerical. Conditioned on $\event_{bg}$, we have that any $\mb z \in \bb S^{n_1-1}$ can be written as $\mb z = \mb v + \mb e$ for some $\mb v \in N_\eps$ and $\norm{\mb e}{} \le \eps$. Moreover, 
\begin{align*}
T\paren{\mb z} 
&\; =\; \norm{\paren{\mb v + \mb e}^* \mb M_{\mc I^c}}{1} - \norm{\paren{\mb v + \mb e}^* \mb M_{\mc I}}{1} \ge T\paren{\mb v} - \norm{\mb e^* \mb M_{\mc I^c}}{1} - \norm{\mb e^* \mb M_{\mc I}}{1} \nonumber \\
& \;=\; \frac{n_2}{5}\sqrt{\frac{2}{\pi}} \theta - \norm{\mb e^* \mb M}{1} = \frac{n_2}{5}\sqrt{\frac{2}{\pi}} \theta - \sum_{k=1}^{n_2} \abs{\mb e^* \mb m_k} \nonumber \\
& \;\geq\;  \frac{n_2}{5}\sqrt{\frac{2}{\pi}} \theta - \eps \sum_{k=1}^{n_2} \norm{\mb m_k}{}. 
\end{align*} 
By Lemma~\ref{lem:X-infinty-tail-bound}, with probability at least $1-\theta\paren{n_1n_2}^{-7} - \exp\paren{-0.3\theta n_1 n_2}$, $\norm{\mb M}{\infty} \le 4\sqrt{\log\paren{n_1n_2}}$. Thus, 
\begin{align}
T\paren{\mb z} \ge \frac{n_2}{5}\sqrt{\frac{2}{\pi}} \theta - \sqrt{\frac{2}{\pi}}\frac{\theta}{120}\frac{n_2\sqrt{n_1} 4\sqrt{\log\paren{n_1n_2}}}{\sqrt{n_1} \sqrt{\log\paren{n_1 n_2}}} = \frac{n_2}{6}\sqrt{\frac{2}{\pi}} \theta. \label{eqn:rounding-lower-bound}
\end{align}
Thus, by \eqref{eqn:rounding-union}, it is enough to take $n_2 > Cn_1\log\paren{n_1/\theta}/\theta^2$ for sufficiently large $C > 0$ to make the overall failure probability small enough so that the lower bound \eqref{eqn:rounding-lower-bound} holds. 
\end{proof}

\begin{proof}[Proof of Lemma~\ref{lem:alg_rounding_orth}] \label{proof:lem_alg_rounding_orth}
The proof is similar to that of \cite{qu2014finding}. First, let us assume the dictionary $\mb A_0 = \mb I$. Wlog, suppose that the Riemannian TRM algorithm returns a solution $\widehat{\mb q}$, to which $\mb e_n$ is the nearest signed basis vector. Thus, the rounding LP~\eqref{eqn:LP_rounding} takes the form: 
\begin{align}\label{eqn:rounding-00}
\mini_{\mb q}\; \norm{\mb q^* \mb X_0}{1}, \quad \st \quad \innerprod{\mb r}{\mb q} = 1. 
\end{align}
where the vector $\mb r = \widehat{\mb q}$. Next, We will show whenever $\widehat{\mb q}$ is close enough to $\mb e_n$, w.h.p., the above linear program returns $\mb e_n$.
Let $\mb X_0 = \brac{\ol{\mb X}; \mb x_n^* }$, where $\ol{\mb X} \in \bb R^{(n-1)\times p}$ and $\mb x_n^*$ is the last row of $\mb X_0$. Set $\mb q= \brac{\overline{\mb q}, q_n}$, where $\overline{\mb q}$ denotes the first $n-1$ coordinates of $\mb q$ and $q_n$ is the last coordinate; similarly for $\mb r$. Let us consider a relaxation of the problem~\eqref{eqn:rounding-00},
\begin{align}
	\mini_{\mb q} \norm{\mb q^* \mb X_0}{1},\quad \st \quad q_nr_n+\innerprod{\overline{\mb q}}{\overline{\mb r}} \geq 1, \label{eqn:rounding-1}
\end{align}
It is obvious that the feasible set of \eqref{eqn:rounding-1} contains that of \eqref{eqn:rounding-00}. So if $\mb e_n$ is the unique optimal solution (UOS) of \eqref{eqn:rounding-1}, it is the UOS of \eqref{eqn:rounding-00}. Suppose $\mc I = \supp (\mb x_n)$ and define an event $\event_0 = \Brac{\abs{\mc I}\leq \frac{9}{8} \theta p }$. By Hoeffding's inequality, we know that
$
\prob{\event_0^c }  \le  \exp\paren{- \theta^2 p/2}. 
$
Now conditioned on $\event_0$ and consider a fixed support $\mc I$. \eqref{eqn:rounding-1} can be further relaxed as 
\begin{align} 
	\mini_{\mb q} \norm{\mb x_n}{1} \abs{q_n} - \norm{ \overline{\mb q}^* \ol{\mb X}_{\mc I} }{1} + \norm{\overline{\mb q}^* \ol{\mb X}_{\mc I^c} }{1},\quad \st \quad q_n r_n + \norm{\overline{\mb q}}{} \norm{\overline{\mb r}}{}\geq 1.\label{eqn:rounding-2}
\end{align}
The objective value of \eqref{eqn:rounding-2} lower bounds that of \eqref{eqn:rounding-1}, and are equal when $\mb q = \mb e_n$. So if $\mb q = \mb e_n$ is UOS of \eqref{eqn:rounding-2}, it is UOS of \eqref{eqn:rounding-00}. By Lemma \ref{lem:rounding-0}, we know that
\begin{align*}
	\norm{\overline{\mb q}^*\ol{\mb X}_{\mc I^c} }{1}- \norm{ \overline{\mb q}^* \ol{\mb X}_{\mc I} }{1}  \;\geq\; \frac{p}{6}\sqrt{\frac{2}{\pi}}\theta  \norm{\overline{\mb q}}{}
\end{align*}
holds w.h.p. when $p \ge C_1 (n-1)\log\paren{(n-1)/\theta}/\theta^2$. Let $\zeta = \frac{p}{6}\sqrt{\frac{2}{\pi}}\theta$, thus we can further lower bound the objective value in \eqref{eqn:rounding-2} by
\begin{align}
	\mini_{\mb q} \norm{\mb x_n}{1}\abs{q_n}+ \zeta \norm{\overline{\mb q} }{},\quad \st \quad q_n r_n + \norm{\overline{\mb q}}{} \norm{\overline{\mb r}}{}\geq 1. \label{eqn:rounding-3}
\end{align}
By similar arguments, if $\mb e_n$ is the UOS of \eqref{eqn:rounding-3}, it is also the UOS of \eqref{eqn:rounding-00}. For the optimal solution of \eqref{eqn:rounding-3}, notice that it is necessary to have $\sign\paren{q_n} = \sign\paren{r_n}$ and $q_n r_n + \norm{\overline{\mb q}}{} \norm{\overline{\mb r}}{}=1$. Therefore, the problem \eqref{eqn:rounding-3} is equivalent to
\begin{align}
	\mini_{q_n} \norm{\mb x_n}{1} \abs{q_n} + \zeta \frac{1-\abs{r_n} \abs{q_n} }{\norm{\overline{\mb r}}{}}, \quad \st \quad \abs{q_n}\leq \frac{1}{\abs{r_n}}. \label{eqn:rounding-4}
\end{align}
Notice that the problem \eqref{eqn:rounding-4} is a linear program in $\abs{q_n}$ with a compact feasible set, which indicates that the optimal solution only occurs at the boundary points $\abs{q_n}=0$ and $\abs{q_n} = 1/\abs{r_n}$. Therefore, $\mb q = \mb e_n$ is the UOS of \eqref{eqn:rounding-4} if and only if
\begin{align}
	\frac{1}{\abs{r_n}} \norm{\mb x_n}{1} < \frac{\zeta }{\norm{\overline{\mb r}}{} }. \label{eqn:rounding-5}
\end{align}
Conditioned on $\event_0$, by using the Gaussian concentration bound, we have
\begin{align*}
	\prob{\norm{\mb x_n}{1} \geq  \frac{9}{8}  \sqrt{\frac{2}{\pi}} \theta p + t } \;\leq\; \prob{\norm{\mb x_n}{1} \geq \expect{\norm{\mb x_n}{1}} + t } \;\leq\; \exp\paren{- \frac{t^2}{2p} },
\end{align*}
which means that
\begin{align}
	\prob{\norm{\mb x_n}{1} \geq \frac{5}{4} \sqrt{\frac{2}{\pi}} \theta p  } \;\leq\;\exp\paren{- \frac{\theta^2 p }{64 \pi } }. \label{eqn:rounding-6}
\end{align}
Therefore, by \eqref{eqn:rounding-5} and \eqref{eqn:rounding-6}, for $\mb q = \mb e_n$ to be the UOS of \eqref{eqn:rounding-00} w.h.p., it is sufficient to have
\begin{align}
	\frac{5}{4 \abs{r_n}} \sqrt{\frac{2}{\pi }} \theta p \;<\; \frac{\theta p}{6\sqrt{1-\abs{r_n}^2}} \sqrt{\frac{2}{\pi }},
\end{align}
which is implied by 
\begin{align*}
	\abs{r_n}\;>\; \frac{249}{250}. 
\end{align*}
The failure probability can be estimated via a simple union bound. Since the above argument holds uniformly for any fixed support set $\mc I$, we obtain the desired result. 

When our dictionary $\mb A_0$ is an arbitrary orthogonal matrix, it only rotates the row subspace of $\mb X_0$. Thus, wlog, suppose the TRM algorithm returns a solution $\widehat{\mb q}$, to which $\mb A_0 \mb q_\star$ is the nearest ``target'' with $\mb q_\star$ a signed basis vector. By a change of variable $\tilde{\mb q} = \mb A_0^* \mb q$, the problem \eqref{eqn:rounding-00} is of the form
\begin{align*}
	\mini_{\tilde{\mb q}} \norm{\widetilde{\mb q}^*\mb X_0 }{1},\quad \st \quad \innerprod{\mb A_0^* \mb r}{\tilde{\mb q} }=1,
\end{align*}
obviously our target solution for $\tilde{\mb q}$ is again the standard basis $\mb q_\star$. By a similar argument above, we only need $\innerprod{\mb A_0^* \mb r}{\mb e_n}>249/250$ to exactly recover the target, which is equivalent to 
$
	\innerprod{\mb r}{\widehat{\mb q}_\star} > 249/250.
$ 
This implies that our rounding \eqref{eqn:LP_rounding} is invariant to change of basis, completing the proof. 
\end{proof}

\begin{proof}[of Lemma~\ref{lem:alg_rounding_comp}] \label{proof:lem_alg_rounding_comp}
Define $\wt{\mb q} \doteq (\mb U \mb V^* + \mb \Xi)^* \mb q$. By Lemma~\ref{lem:pert_key_mag}, and in particular~\eqref{eq:pert_upper_bound}, when 
$
p \ge \frac{C}{c_\star^2 \theta} \max\set{\frac{n^4}{\mu^4}, \frac{n^5}{\mu^2}} \kappa^8\paren{\mb A_0} \log^4\paren{\frac{\kappa\paren{\mb A_0} n}{\mu \theta}}
$, 
$\norm{\mb \Xi}{} \le 1/2$ so that $\mb U \mb V^* + \mb \Xi$ is invertible. 
Then the LP rounding can be written as 
\begin{align*} 
\mini_{\wt{\mb q}} \norm{\wt{\mb q}^*\mb X_0}{1},\quad \st \quad \innerprod{(\mb U \mb V^*+ \mb \Xi)^{-1}\mb r}{\wt{\mb q}} = 1. 
\end{align*}
By Lemma~\ref{lem:alg_rounding_orth}, to obtain $\wt{\mb q} = \mb e_n$ from this LP, it is enough to have 
\begin{align*}
\innerprod{(\mb U \mb V^*+ \mb \Xi)^{-1}\mb r}{\mb e_n} \ge 249/250, 
\end{align*}
and $p \ge Cn^2\log(n/\theta)/\theta$ for some large enough $C$. This implies that to obtain $\mb q_\star$ for the original LP, such that $(\mb U \mb V^* + \mb \Xi)^* \mb q_\star = \mb e_n$, it is enough that 
\begin{align*}
\innerprod{(\mb U \mb V^*+ \mb \Xi)^{-1}\mb r}{(\mb U \mb V^* + \mb \Xi)^* \mb q_\star} = \innerprod{\mb r}{\mb q_\star} \ge 249/250, 
\end{align*}
completing the proof. 
\end{proof}

\begin{proof}[of Lemma~\ref{lem:deflation-bound}] \label{proof:lem_deflation-bound}
Note that $[\mb q_\star^1, \dots, \mb q_\star^\ell] = (\mb Q^* + \mb \Xi^*)^{-1} [\mb e_1, \dots, \mb e_\ell]$, we have
\begin{align*}
\mb U^* (\mb Q + \mb \Xi) \mb X_0 
& = \mb U^* (\mb Q^* + \mb \Xi^*)^{-1} (\mb Q + \mb \Xi)^* (\mb Q + \mb \Xi) \mb X_0 \\
& = \mb U^* \left[\mb q_\star^1, \dots, \mb q_\star^\ell \;\vert\; \wh{\mb V}\right] (\mb I + \mb \Delta_1) \mb X_0, 
\end{align*}
where $\wh{\mb V} \doteq (\mb Q^* + \mb \Xi^*)^{-1} [\mb e_{\ell+1}, \dots, \mb e_n]$, and the matrix $\mb \Delta_1 = \mb Q^*\mb \Xi+ \mb \Xi^*\mb Q+\mb \Xi^*\mb \Xi$ so that $\norm{\mb \Delta_1}{}\leq 3\norm{\mb \Xi}{}$. Since $\mb U^* \left[\mb q_\star^1, \dots, \mb q_\star^\ell \;\vert\; \wh{\mb V}\right] = \brac{\mb 0 \; \vert\; \mb U^* \wh{\mb V}}$, we have
\begin{align} \label{eqn:deflation-comp-1}
\mb U^* (\mb Q + \mb \Xi) \mb X_0  = \brac{\mb 0 \; \vert\; \mb U^* \wh{\mb V}} \mb X_0 + \brac{\mb 0 \; \vert\; \mb U^* \wh{\mb V}} \mb \Delta_1 \mb X_0 = \mb U^* \wh{\mb V} \mb X_0^{[n-\ell]} + \mb \Delta_2 \mb X_0, 
\end{align}
where $\mb \Delta_2 = \brac{\mb 0 \; \vert\; \mb U^* \wh{\mb V}} \mb \Delta_1$. Let $\delta = \norm{\mb \Xi}{}$, so that
\begin{align}\label{eqn:delta-2-bound}
	\norm{\mb \Delta_2}{} \leq \frac{\norm{\mb \Delta_1}{}}{\sigma_{\min} \paren{\mb Q+\mb \Xi} } \leq \frac{3\norm{\mb \Xi}{} }{\sigma_{\min} \paren{\mb Q+\mb \Xi} }\leq \frac{3\delta}{1-\delta}.
\end{align}
Since the matrix $\wh{\mb V}$ is near orthogonal, it can be decomposed as $\wh{\mb V} = \mb V + \mb \Delta_3$, where $\mb V$ is orthogonal, and $\mb \Delta_3$ is a small perturbation. Obviously, $\mb V = \mb U \mb R$ for some orthogonal matrix $\mb R$, so that spans the same subspace as that of $\mb U$. Next, we control the spectral norm of $\mb \Delta_3$ so that it is sufficiently small,
\begin{align}
	\norm{\mb \Delta_3}{} = \min_{\mb R \in O_{\ell}}\norm{\mb U \mb R - \wh{\mb V} }{} \le \min_{\mb R \in O_{\ell}} \norm{\mb U \mb R - \mb Q_{[n-\ell]} }{} + \norm{ \mb Q_{[n-\ell]} - \wh{\mb V} }{},
\end{align} 
where $\mb Q_{[n-\ell]}$ collects the last $n -\ell$ columns of $\mb Q$, i.e., $\mb Q=[\mb Q_{[\ell]}, \mb Q_{[n-\ell]}] $. To bound the second term on the right, we have 
\begin{align*}
\norm{\mb Q_{[n-\ell]} - \wh{\mb V} }{} \le \norm{\mb Q^{-1} - (\mb Q + \mb \Xi)^{-1}}{} \le \frac{\norm{\mb Q^{-1}}{} \norm{\mb Q^{-1} \mb \Xi}{}}{1-\norm{\mb Q^{-1} \mb \Xi}{}} \le \frac{\delta}{1-\delta} , 
\end{align*}
where we have used perturbation bound for matrix inverse (see, e.g., Theorem 2.5 of Chapter III in~\cite{stewart1990matrix}). To bound the first term, from Lemma~\ref{lem:sp_angle_norm}, it is enough to upper bound the largest principal angle $\theta_1$ between the subspaces $\text{span}([\mb q_\star^1, \dots, \mb q_\star^\ell])$, and that spanned by $\mb Q [\mb e_1, \dots, \mb e_\ell]$. Write $\mb I_{[\ell]} \doteq [\mb e_1, \dots, \mb e_\ell]$ for short, we bound $\sin \theta_1$ as
\begin{align*}
\sin \theta_1 \le\; & \norm{\mb Q \mb I_{[\ell]} \mb I_{[\ell]}^* \mb Q^* - (\mb Q^* + \mb \Xi^*)^{-1} \mb I_{[\ell]} \paren{\mb I_{[\ell]}^* (\mb Q + \mb \Xi)^{-1} (\mb Q^* + \mb \Xi^*)^{-1} \mb I_{[\ell]} }^{-1} \mb I_{[\ell]}^* (\mb Q + \mb \Xi)^{-1}}{} \\
=\; & \norm{\mb Q \mb I_{[\ell]} \mb I_{[\ell]}^* \mb Q^* - (\mb Q^* + \mb \Xi^*)^{-1} \mb I_{[\ell]} \paren{ \mb I_{[\ell]}^* (\mb I +\mb \Delta_1)^{-1} \mb I_{[\ell]} }^{-1} \mb I_{[\ell]}^* (\mb Q + \mb \Xi)^{-1}}{} \\
\le\; & \norm{\mb Q \mb I_{[\ell]} \mb I_{[\ell]}^* \mb Q^* - (\mb Q^* + \mb \Xi^*)^{-1} \mb I_{[\ell]}\mb I_{[\ell]}^* (\mb Q + \mb \Xi)^{-1}}{} \\
& \qquad + \norm{(\mb Q^* + \mb \Xi^*)^{-1} \mb I_{[\ell]} \brac{\mb I - \paren{ \mb I_{[\ell]}^* (\mb I +\mb \Delta_1)^{-1} \mb I_{[\ell]} }^{-1} }\mb I_{[\ell]}^* (\mb Q + \mb \Xi)^{-1} }{} \\
\le\; & \paren{1+\frac{1}{\sigma_{\min}(\mb Q + \mb \Xi)}}\norm{\mb Q^{-1} - (\mb Q + \mb \Xi)^{-1}}{} + \frac{1}{\sigma^2_{\min}(\mb Q + \mb \Xi)} \norm{\mb I - \paren{ \mb I_{[\ell]}^* (\mb I +\mb \Delta_1)^{-1} \mb I_{[\ell]} }^{-1} }{} \\
\le\; & \paren{1 + \frac{1}{1-\delta}} \frac{\delta}{1-\delta}  + \frac{1}{(1-\delta)^2} \frac{\norm{\mb I_{[\ell]}^*(\mb I+\mb \Delta_1)^{-1} \mb I_{[\ell]}-\mb I  }{}}{1-\norm{\mb I_{[\ell]}^*(\mb I+\mb \Delta_1)^{-1} \mb I_{[\ell]}-\mb I  }{}  } \\
\le\;& \paren{1 + \frac{1}{1-\delta}} \frac{\delta}{1-\delta}  + \frac{1}{(1-\delta)^2} \frac{ \norm{\mb \Delta_1}{} }{1-2\norm{\mb \Delta_1}{} },
\end{align*}
where in the first line we have used the fact that for any full column rank matrix $\mb M$, $\mb M (\mb M^* \mb M)^{-1} \mb M^*$ is the orthogonal projection onto the its column span, and to obtain the fifth and six lines we have invoked the matrix inverse perturbation bound again. Use the facts that $\delta<1/20$ and $\norm{\mb \Delta_1}{} \le 3\delta < 1/2$, we have
\begin{align*}
\sin \theta_1 \leq \frac{(2-\delta)\delta}{(1-\delta)^2}	 + \frac{3\delta}{(1-\delta)^2(1-6\delta)} = \frac{5\delta - 13\delta^2+6\delta^3}{(1-\delta)^2(1-6\delta)}\leq 8\delta.
\end{align*}
For $\delta < 1/20$, the upper bound is nontrivial. By Lemma~\ref{lem:sp_angle_norm}, 
\begin{align*}
\min_{\mb R \in O_\ell}\norm{\mb U \mb R - \mb Q_{[n-\ell]} }{} \le \sqrt{2-2\cos\theta_1} \le \sqrt{2-2\cos^2\theta_1} = \sqrt{2} \sin \theta_1 \le 8\sqrt{2} \delta. 
\end{align*}
Put the estimates above, there exists an orthogonal matrix $\mb R \in O_\ell$ such that $\mb V = \mb U\mb R$ and $\wh{\mb V} = \mb V + \mb \Delta_3$ with 
\begin{align}\label{eqn:delta-3-bound}
	\norm{\mb \Delta_3}{} \leq \delta/(1-\delta) +8\sqrt{2}\delta\leq 12.5\delta.
\end{align}
Therefore, by \eqref{eqn:deflation-comp-1}, we obtain
\begin{align}
	\mb U^*(\mb Q+\mb \Xi)\mb X_0 = \mb U^*\mb V\mb X_0^{[n-\ell]} +\mb \Delta,\quad \text{with}\quad\mb \Delta \doteq \mb U^* \mb \Delta_3 \mb X_0^{[n-\ell]} + \mb \Delta_2 \mb X_0.
\end{align}
By using the results in \eqref{eqn:delta-2-bound} and \eqref{eqn:delta-3-bound}, we get the desired result.
\end{proof}

%% file: sec/appendix.tex
\begin{appendices}
\section{Technical Tools and Basic Facts Used in Proofs}
\input{sec/app_tools}

\section{Auxillary Results for Proofs}
\input{sec/app_aux_geometry}
\end{appendices}

%% file: sec/app_tools.tex
In this section, we summarize some basic calculations that are useful throughout, and also record major technical tools we use in proofs. 
\begin{lemma}[Derivates and Lipschitz Properties of $h_{\mu}\paren{z}$] \label{lem:derivatives_basic_surrogate}
For the sparsity surrogate 
\begin{align*}
h_{\mu}\left(z\right)= \mu \log\paren{\cosh\paren{z/\mu}}, 
\end{align*}
the first two derivatives are
\begin{align}
\dot{ h}_\mu (z) = \tanh\paren{\frac{z}{\mu}},\quad \ddot{h}_\mu (z) = \frac{1}{\mu}\brac{1-\tanh^2\paren{\frac{z}{\mu}}}. 
\end{align}
Also, for any $z>0$, we have
\begin{align}
\frac{1}{2}\paren{1-\exp\paren{-\frac{2z}{\mu}}}\;&\le \;\tanh\paren{\frac{z}{\mu}} \;\le\; 1-\exp\paren{-\frac{2z}{\mu}}, \\
\exp\paren{-\frac{2z}{\mu}}\;&\le \;1-\tanh^2\paren{\frac{z}{\mu}} \;\le\; 4\exp\paren{-\frac{2z}{\mu}}. 
\end{align}
Moreover, for any $z,~z^\prime\in \reals$, we have
\begin{align}
\abs{\dot{h}_{\mu}(z) - \dot{h}_{\mu}(z^\prime) } \le \frac{1}{\mu} \abs{z - z^\prime},\quad \abs{\ddot{h}_{\mu}(z) - \ddot{h}_{\mu}(z^\prime) } \le \frac{2}{\mu^2} \abs{z - z^\prime} 
\end{align}
\end{lemma}

\begin{lemma}[Chebyshev's Association Inequality] \label{lemma:cheby_correlation}
Let $X$ denote a real-valued random variable, and $f, g: \R \mapsto \R$ nondecreasing (nonincreasing) functions of $X$ with $\bb E\left[f\left(X\right)\right] < \infty$ and $\bb E \left[g\left(X\right)\right] < \infty$. Then 
\begin{align}
\bb E\left[f\left(X\right)g\left(X\right)\right] \geq \bb E\left[f\left(X\right)\right] \bb E \left[g\left(X\right)\right]. 
\end{align}
If $f$ is nondecreasing (nonincreasing) and $g$ is nonincreasing (nondecreasing), we have
\begin{align}
\bb E\left[f\left(X\right)g\left(X\right)\right] \leq \bb E\left[f\left(X\right)\right] \bb E \left[g\left(X\right)\right]. 
\end{align}
\end{lemma}
\begin{proof}
Consider $Y$, an independent copy of $X$. Then it is easy to see 
\begin{align*}
\expect{\left(f\left(X\right) - f\left(Y\right)\right) \left(g\left(X\right) - g\left(Y\right)\right)} \geq 0. 
\end{align*}
Expanding the expectation and noticing $\expect{f\left(X\right)g\left(Y\right)} = \expect{f\left(Y\right)g\left(X\right)} = \expect{f\left(X\right)} \expect{g\left(X\right)}$ and also $\expect{f\left(X\right) g\left(X\right)} = \expect{f\left(Y\right)g\left(Y\right)}$ yields the result. Similarly, we can prove the second one. 
\end{proof}
This lemma implies the following lemma.

\begin{lemma}[Harris' Inequality, ~\cite{harris1960lower}, see also Theorem 2.15 of~\cite{boucheron2013concentration}] \label{lemma:harris_ineq}
Let $X_1, \dots, X_n$ be independent, real-valued random variables and $f, g: \R^n \mapsto \R$ be nonincreasing (nondecreasing) w.r.t. any one variable while fixing the others. Define a random vector $\mb X = \paren{X_1,\cdots,X_n}\in \R^n$, then we have 
\begin{align}
\expect{f\left(\mb X\right) g\left(\mb X\right)} \geq \expect{f\left(\mb X\right)} \expect{g\left(\mb X\right)}. 
\end{align}
Similarly, if $f$ is nondecreasing (nonincreasing) and $g$ is nonincreasing (nondecreasing) coordinatewise in the above sense, we have
\begin{align}
\expect{f\left(\mb X\right) g\left(\mb X\right)} \leq \expect{f\left(\mb X\right)} \expect{g\left(\mb X\right)}. 
\end{align}
\end{lemma}
\begin{proof}
Again, it suffices to prove the first equality, which can be shown by induction. For $n = 1$, it reduces to Lemma~\ref{lemma:cheby_correlation}. Suppose the claim is true for any $m < n$. Since both $g$ and $f$ are nondecreasing functions in $X_n$ given $\widehat{\mb X}=\paren{X_1,\cdots,X_{n-1}}$, then 
\begin{align*}
\expect{f\left(\mb X\right) g\left(\mb X\right)} = \bb E\brac{\bb E\brac{f(\mb X)g(\mb X) \mid \widehat{\mb X} } } \geq \bb E\brac{ \bb E\brac{f(\mb X) \mid \widehat{\mb X}}\bb E\brac{g(\mb X) \mid \widehat{\mb X}} }
\end{align*}
Now, it follows by independence that $f^\prime\paren{\widehat{\mb X}} = \bb E\brac{f(\mb X) \mid \widehat{\mb X}}$ and $g^\prime\paren{\widehat {\mb X}} = \bb E\brac{g(\mb X) \mid \widehat{\mb X} }$ are both nondecreasing functions, then by the induction hypothesis, we have
\begin{align*}
\expect{f\left(\mb X\right) g\left(\mb X\right)}\ge \bb E\brac{f^\prime\paren{\widehat{\mb X}}}\bb E\brac{g^\prime\paren{\widehat{\mb X}}} = \bb E\brac{f(\mb X)}\bb E\brac{g(\mb X)}, 
\end{align*}
as desired. 
\end{proof}

\begin{lemma}[Differentiation under the Integral Sign] \label{lemma:exchange_diff_int}
Consider a function $F: \R^n \times \R \mapsto \R$ such that $\frac{\partial F\left(\mb x, s\right)}{\partial s}$ is well defined and measurable over $\mc U \times \left(0, t_0\right)$ for some open subset $\mc U \subset \R^n$ and some $t_0 > 0$. For any probability measure $\mu$ on $\R^n$ and any $t \in \left(0, t_0\right)$ such that $\int_{0}^t \int_{\mc U} \abs{\frac{\partial F\left(\mb x, s\right)}{\partial s}} \; \mu\left(d \mb x\right) ds < \infty$, it holds that 
\begin{align}
\frac{d}{dt} \int_{\mc U} F\left(\mb x, t\right) \mu\left(d\mb x\right) = \int_{\mc U} \frac{\partial F\left(\mb x, t\right)}{\partial t} \mu \left(d\mb x\right), \; \text{or} \; \frac{d}{dt} \bb E_{\mb x} \left[F\left(\mb x, t\right) \indicator{\mc U}\right] = \bb E_{\mb x} \left[\frac{\partial F\left(\mb x, t\right)}{\partial t} \indicator{\mc U}\right]. 
\end{align}
\end{lemma}
\begin{proof}
We have 
\begin{align*}
\int_{\mc U} \frac{\partial F\left(\mb x, t\right)}{\partial t} \mu\left(d\mb x\right) 
& = \frac{d}{dt} \int_{0}^t \int_{\mc U} \frac{\partial F\left(\mb x, s\right)}{\partial s} \mu\left(d\mb x\right) ds \\
& = \frac{d}{dt} \int_{\mc U} \int_0^t \frac{\partial F\left(\mb x, s\right)}{\partial s} \; ds\; \mu\left(d\mb x\right) \\
& = \frac{d}{dt} \int_{\mc U} \left(F\left(\mb x, t\right) - F\left(\mb x, 0\right)\right)\; \mu\left(d\mb x\right) \\
& =  \frac{d}{dt} \int_{\mc U} F\left(\mb x, t\right)\; \mu\left(d\mb x\right),  
\end{align*}
where we have used the fundamental theorem of calculus for the first and third equalities, and measure-theoretic Fubini's theorem (see, e.g., Theorem 2.37 of~\cite{folland1999real}) for the second equality (as justified by our integrability assumption). 
\end{proof}

\begin{lemma}[Gaussian Tail Estimates] \label{lem:gaussian_tail_est}
Let $X \sim \mc N\paren{0, 1}$ and $\Phi\paren{x}$ be CDF of $X$. For any $x \ge 0$, we have the following estimates for $\Phi^c\paren{x} \doteq 1 - \Phi\paren{x}$: 
\begin{align}
\paren{\frac{1}{x} - \frac{1}{x^3}}\frac{\exp\paren{-x^2/2}}{\sqrt{2\pi}} & \le \Phi^c\paren{x} \le \paren{\frac{1}{x} - \frac{1}{x^3} + \frac{3}{x^5}}\frac{\exp\paren{-x^2/2}}{\sqrt{2\pi}}, \quad (\text{Type I}) \\
\frac{x}{x^2 + 1} \frac{\exp\paren{-x^2/2}}{\sqrt{2\pi}}& \le \Phi^c\paren{x} \le \frac{1}{x} \frac{\exp\paren{-x^2/2}}{\sqrt{2\pi}},  \quad (\text{Type II}) \\
\frac{\sqrt{x^2 + 4} - x}{2} \frac{\exp\paren{-x^2/2}}{\sqrt{2\pi}} & \le  \Phi^c\paren{x} \le \paren{\sqrt{2 + x^2} - x}  \frac{\exp\paren{-x^2/2}}{\sqrt{2\pi}} \quad (\text{Type III}). 
\end{align}
\end{lemma}
\begin{proof}
Type I bounds can be obtained by integration by parts with proper truncations. Type II upper bound can again be obtained via integration by parts, and the lower bound can be obtained via considering the function $f\paren{x} \doteq  \Phi^c\paren{x} - \frac{x}{x^2 + 1} \frac{\exp\paren{-x^2/2}}{\sqrt{2\pi}}$ and noticing it is always nonnegative. Type III bounds are mentioned in~\cite{duembgen2010bounding} and reproduced by the systematic approach developed therein (section 2). 
\end{proof}

\begin{lemma}[Moments of the Gaussian Random Variables] \label{lem:guassian_moment}
If $X \sim \mc N\left(0, \sigma^2\right)$, then it holds for all integer $p \geq 1$ that
\begin{align}
\expect{\abs{X}^p} = \sigma^p \paren{p -1}!! \brac{ \sqrt{\frac{2}{\pi}} \indicator{p\; \text{odd}}+\indicator{p\; \text{even}} } \leq \sigma^p \paren{p -1}!!. 
\end{align}
\end{lemma}

\begin{lemma}[Moments of the $\chi^2$ Random Variables] \label{lem:chi_sq_moment}
If $X \sim \mc \chi^2\paren{n}$, then it holds for all integer $p \geq 1$,
\begin{align}
\expect{X^p} = 2^p \frac{\Gamma\paren{p + n/2}}{\Gamma\paren{n/2}} =  \prod_{k=1}^p (n+2k-2) \le \frac{p!}{2}\paren{2n}^p. 
\end{align}
\end{lemma}

\begin{lemma}[Moments of the $\chi$ Random Variables] \label{lem:chi_moment}
If $X \sim \mc \chi\paren{n}$, then it holds for all integer $p \geq 1$,
\begin{align}
\expect{X^p} = 2^{p/2} \frac{\Gamma\paren{p/2 + n/2}}{\Gamma\paren{n/2}} \le p! n^{p/2}. 
\end{align}
\end{lemma}

\begin{lemma}[Moment-Control Bernstein's Inequality for Scalar RVs, Theorem 2.10 of~\cite{foucart2013mathematical}] \label{lem:mc_bernstein_scalar}
Let $X_1, \dots, X_p$ be i.i.d. real-valued random variables. Suppose that there exist some positive number $R$ and $\sigma^2$ such that
\begin{align*}
\expect{\abs{X_k}^m} \leq \frac{m!}{2} \sigma^2 R^{m-2}, \; \; \text{for all integers}\; m \ge 2. 
\end{align*} 
Let $S \doteq \frac{1}{p}\sum_{k=1}^p X_k$, then for all $t > 0$, it holds  that 
\begin{align}
\prob{\abs{S - \expect{S}} \ge t} \leq 2\exp\left(-\frac{pt^2}{2\sigma^2 + 2Rt}\right).   
\end{align}
\end{lemma}

\begin{lemma}[Moment-Control Bernstein's Inequality for Matrix RVs, Theorem 6.2 of~\cite{tropp2012user}] \label{lem:mc_bernstein_matrix}
Let $\mb X_1, \dots, \mb X_p\in \R^{d \times d}$ be i.i.d. random, symmetric matrices. Suppose there exist some positive number $R$ and $\sigma^2$ such that
\begin{align*}
\expect{\mb X_k^m} \preceq \frac{m!}{2} \sigma^2 R^{m-2} \mb I \; \text{and} -\expect{\mb X_k^m} \preceq \frac{m!}{2} \sigma^2 R^{m-2} \mb I\;,  \; \text{for all integers $m \ge 2$}. \label{eqn:bern-moment-conditions}
\end{align*}
Let $\mb S \doteq \frac{1}{p} \sum_{k = 1}^p \mb X_k$, then for all $t > 0$, it holds that 
\begin{align}
\prob{\norm{\mb S - \expect{\mb S}}{} \ge t} \le 2d\exp\paren{-\frac{pt^2}{2\sigma^2 + 2Rt}}.
\end{align}
\end{lemma}

Proving this lemma requires some modification to the original proof of Theorem 6.2 in~\cite{tropp2012user}. We record it here for the sake of completeness. 

\begin{proof}
Let us define $\mb S_p = \sum_{k=1}^p \mb X_k$, by Proposition 3.1 of \cite{tropp2012user}, we have
\begin{align}
\bb P\brac{\lambda_{\max}\paren{\mb S_p - \bb E\brac{\mb S_p}} \geq t }\;\leq\; \inf_{t>0} e^{-\theta t} \bb E\brac{\trace\exp\paren{ \theta \mb S_p - \theta\bb E\brac{\mb S_p}}  }, \label{eqn:bern-matrix-1}
\end{align}
To proceed, notice that
\begin{align*}
&  \bb E\brac{\trace\exp\paren{ \theta \mb S_p - \theta\bb E\brac{\mb S_p}} } \\
=\;&  \bb E_{\mb S_{p-1}}\bb E_{\mb X_p}\brac{\trace\exp \paren{\theta\paren{\mb S_{p-1}- \bb E\brac{\mb S_{p-1}}}+ \theta \mb X_p - \theta\bb E\brac{\mb X_p} } }  \\
\leq \;&\bb E_{\mb S_{p-1}}\brac{\trace\exp \paren{\theta(\mb S_{p-1}- \bb E\brac{\mb S_{p-1}})+ \log \paren{\bb E\brac{e^{\theta \mb X_p}}} - \theta\bb E\brac{\mb X_p} } } \\
\leq \; & \bb E_{\mb S_{p-1}}\brac{\trace\exp \paren{\theta(\mb S_{p-1}- \bb E\brac{\mb S_{p-1}})+ \bb E\brac{e^{\theta \mb X_p}}-\mb I - \theta\bb E\brac{\mb X_p} } } \\
= \;& \bb E_{\mb S_{p-1}}\brac{\trace\exp \paren{\theta(\mb S_{p-1}- \bb E\brac{\mb S_{p-1}})+ \sum_{\ell=2}^\infty \frac{\theta^\ell\bb E\brac{\mb X_k^\ell}}{\ell!} } }
\end{align*}
where at the third line we have used the result of Corollary 3.3 of~\cite{tropp2012user}, i.e., $\bb E\brac{\trace\exp\paren{\mb H+\mb X} }\leq \trace\exp \paren{\mb H + \log\paren{\bb E\brac{e^{\mb X}}}}$ for any fixed $\mb H$ and random, symmetric $\mb X$, at the fourth we have used the fact that $\log\mb X \preceq \mb X- \mb I$ for any $\mb X\succ \mb 0$ (as $\log u \le u - 1$ for any $u > 0$ and transfer rule applies here), and the last line relies on exchange of infinite summation and expectation, justified as $\mb X_p$ has a bounded spectral radius. By repeating the argument backwards for $\mb X_{p-1}, \cdots, \mb X_{1}$, we get
\begin{align}
& \bb E\brac{\trace\exp\paren{ \theta \mb S_p - \theta\bb E\brac{\mb S_p}} } \nonumber \\
\le\; & \trace\exp\paren{p\sum_{\ell=2}^\infty \frac{\theta^\ell \bb E\brac{\mb X_k^\ell}}{\ell!} } \le \trace\exp\paren{p\sum_{\ell=2}^p\frac{\theta^\ell\sigma^2 R^{\ell-2}}{2}\mb I }\nonumber \\
\le\; & d \norm{\exp\paren{p\sum_{\ell=2}^p\frac{\theta^\ell\sigma^2 R^{\ell-2}}{2}\mb I }}{} \le d\exp\paren{\frac{p\theta^2\sigma^2}{2(1- \theta R)}},\label{eqn:bern-matrix-5}
\end{align}
where we used the fact that $\expect{\mb X_i^m} \preceq \frac{m!}{2} \sigma^2 R^{m-2} \mb I$ in \eqref{eqn:bern-moment-conditions} and restrict $\theta<\frac{1}{R}$. Combining the results in \eqref{eqn:bern-matrix-1} and \eqref{eqn:bern-matrix-5}, we have
\begin{align}
\bb P\brac{\lambda_{\max} \paren{\mb S_p - \bb E\brac{\mb S_p} } \geq t } \;\leq\;  d\inf_{\theta<1/R}  \exp\paren{\frac{p\theta^2\sigma^2}{2(1- \theta R)}-\theta t }
\end{align}
by taking $\theta = t/(p\sigma^2+Rt) < 1/R$, we obtain
\begin{align}
\bb P\brac{\lambda_{\max} \paren{\mb S_p - \bb E\brac{\mb S_p} } \geq t } \le d \exp \paren{- \frac{t^2}{2p\sigma^2+2Rt}}. 
\end{align}
Considering $\mb X'_k = -\mb X_k$ and repeating the above argument, we can similarly obtain 
\begin{align}
\bb P\brac{\lambda_{\min} \paren{\mb S_p - \bb E\brac{\mb S_p} } \le -t } \le d \exp \paren{- \frac{t^2}{2p\sigma^2+2Rt}}. 
\end{align}
Putting the above bounds together, we have 
\begin{align}
\bb P\brac{\norm{\mb S_p - \bb E\brac{\mb S_p}}{} \geq t } \;\leq\; 2d \exp \paren{- \frac{t^2}{2p\sigma^2+2Rt}}. 
\end{align}
We obtain the claimed bound by substituting $\mb S_p = p\mb S$ and simplifying the resulting expressions. 
\end{proof}

\begin{corollary}[Moment-Control Bernstein's Inequality for Vector RVs] \label{cor:vector-bernstein} 
Let $\mb x_1, \dots, \mb x_p \in \reals^d$ be i.i.d. random vectors. Suppose there exist some positive number $R$ and $\sigma^2$ such that
\begin{equation*}
\bb E\left[ \norm{\mb x_k }{}^m \right] \;\le\; \frac{m!}{2} \sigma^2R^{m-2}, \quad \text{for all integers $m \ge 2$}. 
\end{equation*}
Let $\mb s = \frac{1}{p}\sum_{k=1}^p \mb x_k$, then for any $t > 0$, it holds that
\begin{align}
\bb P\brac{\norm{\mb s - \bb E\brac{\mb s}}{} \geq t} \; \leq \; 2(d+1)\exp\paren{-\frac{pt^2}{2\sigma^2+2Rt}}.
\end{align}
\end{corollary}

\begin{proof} 
To obtain the result, we apply the matrix Bernstein inequality in Lemma \ref{lem:mc_bernstein_matrix} to a suitable embedding of the random vectors $\Brac{\mb x_k}_{k=1}^p$. For any $k \in [p]$, define the symmetric matrix
\begin{equation*}
\mb X_k = \left[ \begin{array}{cc} 0 & \mb x_k^* \\ \mb x_k & \mb 0 \end{array} \right] \in \reals^{(d+1) \times (d+1)}.
\end{equation*}
Then it holds that 
\begin{align*}
\bm X_k^{2\ell+1} = \norm{\mb x_k}{2}^{2\ell} \brac{\begin{array}{cc}
0 & \mb x_k^*\\
\mb x_k & \mb 0
\end{array} },\; 
\bm X_k^{2\ell+2} = \norm{\mb x_k}{}^{2\ell}\brac{\begin{array}{cc}
\norm{\mb x_k}{}^{2} & \mb 0\\
\mb 0 &  \mb x_k\mb x_k^*
\end{array} },\; \text{for all integers $\ell \ge 0$}.  \label{eqn:vector-bern-1}
\end{align*}
Using the fact that 
\begin{align*}
 \mb x_k\mb x_k^* \preceq \norm{\mb x_k}{}^2 \mb I,\quad \norm{\mb X_k}{} = \sqrt{\norm{\mb X_k^2}{}} = \norm{\mb x_k}{} \Longrightarrow -\norm{\mb x_k}{} \mb I \preceq \mb X_k \preceq  \norm{\mb x_k}{} \mb I, 
\end{align*}
and combining the above expressions for $\mb X_k^{2\ell+1}$ and $\mb X_k^{2\ell+2}$, we obtain 
\begin{align}
\bb E\brac{\mb X_k^{m}}, -\expect{\mb X_k^m} \preceq \bb E\brac{\norm{\mb x_k}{2}^m }\mb I\preceq \frac{m!}{2}\sigma^2R^{m-2}\mb I,\quad \text{for all integers $m \ge 2$}, 
\end{align}
Let $\mb S = \frac{1}{p}\sum_{k=1}^p \mb X_k$, noting that 
\begin{align}
\norm{\mb S - \bb E\brac{\mb S}}{} = \norm{\mb s -\bb E\brac{\mb s}}{}, 
\end{align}
and applying Lemma \ref{lem:mc_bernstein_matrix}, we complete the proof.
\end{proof}

\begin{lemma}[Integral Form of Taylor's Theorem]\label{lem:Taylor-integral-form}
	Let $f(\mb x): \bb R^n \mapsto \bb R$ be a twice continuously differentiable function, then for any direction $\mb y\in \bb R^n$, we have
	\begin{align}
		f(\mb x+t\mb y) &= f(\mb x) + t \int_0^1\innerprod{\nabla f(\mb x+st\mb y)}{\mb y} \; ds, \\
		f(\mb x+t\mb y) &= f(\mb x) + t \innerprod{\nabla f(\mb x)}{\mb y} + t^2 \int_0^1 (1-s)\innerprod{\nabla^2 f(\mb x+st\mb y) \mb y}{\mb y}\; ds.
	\end{align}
\end{lemma}
\begin{proof}
	By the fundamental theorem of calculus, since $f$ is continuous differentiable, it is obvious that
	\begin{align}
		f(\mb x+ t\mb y) = f(\mb x) +\int_0^t \innerprod{\nabla f(\mb x+ \tau \mb y) }{\mb y }d \tau. \label{eqn:taylor-integral-1}
	\end{align}
	If $f$ is twice continuously differentiable, by using integral by parts, we obtain
	\begin{align}
		f(\mb x+t \mb y) &= f(\mb x) + \left.\brac{(\tau-t)  \innerprod{\nabla f(\mb x+ \tau \mb y) }{\mb y } }\right|_{0}^t - \int_{0}^t (\tau-t) \; d \innerprod{\nabla f(\mb x+ \tau \mb y) }{\mb y } \nonumber \\
		& = f(\mb x) + t  \innerprod{\nabla f(\mb x+ \tau \mb y) }{\mb y } +\int_{0}^t (t-\tau ) \innerprod{\nabla^2 f(\mb x+ \tau \mb y) \mb y }{\mb y } d\tau. \label{eqn:taylor-integral-2}
	\end{align}
	By a change of variable $\tau = st~(0\leq s\leq 1)$ for \eqref{eqn:taylor-integral-1} and \eqref{eqn:taylor-integral-2}, we get the desired results.
\end{proof}

%% file: sec/app_aux_geometry.tex
\begin{lemma}\label{lem:aux_asymp_proof_a}
Let $X\sim \mc N(0,\sigma_X^2)$ and $Y\sim \mc N(0,\sigma_Y^2)$ be independent random variables and $\Phi^c\paren{t} = \frac{1}{\sqrt{2\pi}} \int_t^{\infty} \exp\paren{-x^2/2}\; dx$ be the complementary cumulative distribution function of the standard normal. For any $a>0$, we have 
\begin{align}
\expect{X\indicator{X>0}} \;&= \;\frac{\sigma_X}{\sqrt{2\pi}}, \label{eqn:lem-aux-aymp-proof-a-1} \\
\expect{\exp\paren{-aX}X\indicator{X>0}} \;&= \; \frac{\sigma_X}{\sqrt{2\pi}} - a \sigma_X^2\exp\paren{\frac{a^2\sigma_X^2}{2}}\Phi^c\paren{a\sigma_X}, \label{eqn:lem-aux-aymp-proof-a-2} \\
\bb E\brac{\exp\paren{-aX}\indicator{X>0} }\;&=\; \exp\paren{\frac{a^2\sigma_X^2}{2}} \Phi^c\paren{a\sigma_X},  \label{eqn:lem-aux-aymp-proof-a-3}\\
\bb E\brac{\exp\paren{-a(X+Y)}X^2\indicator{X+Y>0}}\; &=\;\sigma_X^2\paren{1+a^2\sigma_X^2}\exp\paren{\frac{a^2\sigma_X^2 + a^2 \sigma_Y^2}{2}}\Phi^c\paren{a\sqrt{\sigma_X^2 + \sigma_Y^2}} \nonumber \\
& \qquad - \frac{a\sigma_X^4}{\sqrt{2\pi}\sqrt{\sigma_X^2 + \sigma_Y^2}}, \label{eqn:lem-aux-aymp-proof-a-4}\\
\bb E\brac{\exp\paren{-a(X+Y)}XY\indicator{X+Y>0} }\; &=\; a^2\sigma_X^2\sigma_Y^2 \exp\paren{\frac{a^2 \sigma_X^2 + a^2 \sigma_Y^2}{2}} \Phi^c\paren{a\sqrt{\sigma_X^2 + \sigma_Y^2}} \nonumber \\
& \qquad - \frac{a\sigma_X^2\sigma_Y^2}{\sqrt{2\pi}\sqrt{\sigma_X^2 + \sigma_Y^2}},  \label{eqn:lem-aux-aymp-proof-a-5} \\
\bb E\brac{\tanh\paren{aX}X} \;&=\; a\sigma_X^2 \bb E\brac{1-\tanh^2\paren{aX}},  \label{eqn:lem-aux-aymp-proof-a-6}\\
\bb E\brac{\tanh\paren{a(X+Y)}X} \;&=\; a\sigma_X^2 \bb E\brac{1-\tanh^2\paren{a(X+Y)}}\label{eqn:lem-aux-aymp-proof-a-7}. 
\end{align}
\end{lemma}

\begin{proof}
Equalities \eqref{eqn:lem-aux-aymp-proof-a-1}, \eqref{eqn:lem-aux-aymp-proof-a-2}, \eqref{eqn:lem-aux-aymp-proof-a-3}, \eqref{eqn:lem-aux-aymp-proof-a-4} and \eqref{eqn:lem-aux-aymp-proof-a-5} can be obtained by direct integrations. Equalities \eqref{eqn:lem-aux-aymp-proof-a-6} and \eqref{eqn:lem-aux-aymp-proof-a-7} can be derived using integration by part. 
\end{proof}


\begin{proof}[of Lemma~\ref{lem:neg_curvature_norm_bound}]
Indeed $\frac{1}{\paren{1+\beta t}^2} = \sum_{k=0}^\infty (-1)^k(k+1)\beta^kt^k$, as 
\begin{align*}
\sum_{k=0}^\infty (-1)^k(k+1)\beta^kt^k = \sum_{k=0}^\infty(- \beta t)^k + \sum_{k=0}^\infty k (-\beta t)^k = \frac{1}{1+\beta t} + \frac{-\beta t}{(1+\beta t)^2} = \frac{1}{(1+\beta t)^2}. 
\end{align*}
The magnitude of the coefficient vector is 
\begin{align*}
\norm{\mb b}{\ell^1} &= \sum_{k=0}^\infty \beta^k (1+k) = \sum_{k=0}^\infty \beta^k + \sum_{k=0}^\infty k \beta^k = \frac{1}{1-\beta} +\frac{\beta}{(1-\beta)^2} = \frac{1}{(1-\beta)^2} = T. 
\end{align*}
Observing that $\frac{1}{\paren{1+\beta t}^2} > \frac{1}{\paren{1+t}^2}$ for $t \in \brac{0, 1}$ when $0 < \beta < 1$, we obtain 
\begin{align}
\norm{p -f}{L^1[0,1]} &= \int_0^1 \abs{p(t) - f(t)}dt = \int_0^1 \brac{\frac{1}{(1+\beta t)^2} -\frac{1}{(1+t)^2}} dt=\frac{1-\beta}{2(1+\beta)}  \le \frac{1}{2\sqrt{T}}. 
\end{align}
Moreover, we have 
\begin{align}
\norm{f-p}{L^\infty[0,1]} = \max_{t\in[0,1]} p(t) - f(t) = \max_{t\in[0,1]} \frac{t(1-\beta)\paren{2+t(1+\beta)}}{(1+t)^2(1+\beta t)^2} \le 1-\beta = \frac{1}{\sqrt{T}}.  
\end{align}
Finally, notice that
\begin{align}
\sum_{k=0}^\infty \frac{b_k}{(1+k)^3} = \sum_{k=0}^\infty \frac{\paren{-\beta}^k}{(1+k)^2} 
& = \sum_{i=0}^\infty\brac{\frac{\beta^{2i}}{(1+2i)^2} - \frac{\beta^{2i+1}}{(2i+2)^2} } \nonumber \\
& = \sum_{i=0}^\infty \beta^{2i} \frac{(2i+2)^2 -\beta(2i+1)^2 }{(2i+2)^2(2i+1)^2} > 0, 
\end{align}
where at the second equality we have grouped consecutive even-odd pair of summands. In addition, we have
\begin{align}
\sum_{k=0}^n \frac{b_k}{(1+k)^3}\le \sum_{k=0}^n \frac{\abs{b_k}}{(1+k)^3} = \sum_{k=0}^n\frac{\beta^k}{(1+k)^2} \le 1+ \sum_{k=1}^n \frac{1}{(1+k)k} = 2-\frac{1}{n+1}, 
\end{align}
which converges to $2$ when $n \to \infty$, completing the proof. 
\end{proof}

\begin{proof}[of Lemma~\ref{lem:U-moments-bound}]
The first inequality is obviously true for $\mb v = \mb 0$. When $\mb v \neq \mb 0$, we have 
\begin{align*}
\expect{\abs{\mb v^* \mb z}^m } 
& =  \sum_{\ell = 0}^n \theta^\ell \paren{1-\theta}^{n - \ell} \sum_{\mc J \in \binom{[n]}{\ell}} \bb E_{Z \sim \mc N\paren{0, \norm{\mb v_{\mc J}}{}^2}}\brac{\abs{Z}^m} \\
& \le \sum_{\ell = 0}^n \theta^\ell \paren{1-\theta}^{n - \ell} \sum_{\mc J \in \binom{[n]}{\ell}} \bb E_{Z \sim \mc N\paren{0, \norm{\mb v}{}^2}}\brac{\abs{Z}^m}\\
& = \bb E_{Z \sim \mc N\paren{0, \norm{\mb v}{}^2}}\brac{\abs{Z}^m} \sum_{\ell = 0}^n \theta^\ell \paren{1-\theta}^{n - \ell} \binom{n}{\ell} \\
& = \bb E_{Z \sim \mc N\paren{0, \norm{\mb v}{}^2}}\brac{\abs{Z}^m}, 
\end{align*}
where the second line relies on the fact $\norm{\mb v_{\mc J}}{} \le \norm{\mb v}{}$ and that for a fixed order, central moment of Gaussian is monotonically increasing w.r.t. its variance. Similarly, to see the second inequality, 
\begin{align*}
\expect{\norm{\mb z}{}^m}
& = \sum_{\ell = 0}^n \theta^\ell \paren{1 -\theta}^{n-\ell} \sum_{\mc J \in \binom{[n]}{\ell}} \expect{\norm{\mb z'_{\mc J}}{}^m} \\
& \le \expect{\norm{\mb z'}{}^m}  \sum_{\ell = 0}^n \theta^\ell \paren{1 -\theta}^{n-\ell} \binom{n}{\ell} = \expect{\norm{\mb z'}{}^m}, 
\end{align*}
as desired. 
\end{proof}

\begin{proof}[of Lemma~\ref{lem:X-infinty-tail-bound}]
Consider one component of $\mb X$, i.e., $X_{ij}=B_{ij}V_{ij}$ for $i \in [n]$ and $j \in [p]$, where $B_{ij}\sim \mathrm{Ber}\paren{\theta}$) and $V_{ij}\sim \mc N(0,1)$. We have 
\begin{align*}
\bb P\brac{\abs{X_{ij}}> 4\sqrt{\log\paren{np}}} \leq \theta \bb P\brac{\abs{V_{ij}}> 4\sqrt{\log(np)}} \leq \theta \exp\paren{-8\log(np)}= \theta (np)^{-8}. 
\end{align*}
And also 
\begin{align*}
\prob{\abs{X_{ij}} < 1} = 1 -\theta + \theta \prob{\abs{V_{ij}} < 1} \le 1-0.3\theta. 
\end{align*}
Applying a union bound as  
\begin{align*}
\prob{\norm{\mb X}{\infty} \le 1 \; \text{or} \; \norm{\mb X}{\infty} \ge 4\sqrt{\log\paren{np}}} \le \paren{1-0.3\theta}^{np} + np\theta \paren{np}^{-8} \le \exp\paren{-0.3\theta n p} + \theta \paren{np}^{-7}, 
\end{align*}
we complete the proof. 
\end{proof}

\begin{lemma} \label{lem:half_inverse_pert}
Suppose $\mb A \succ \mb 0$. Then for any symmetric perturbation matrix $\mb \Delta$ with $\norm{\mb \Delta}{} \le \tfrac{\sigma_{\min}\paren{\mb A}}{2}$, it holds that 
\begin{align}
\norm{\paren{\mb A + \mb \Delta}^{-1/2} - \mb A^{-1/2}}{} \le \frac{2\norm{\mb A}{}^{1/2} \norm{\mb \Delta}{} }{\sigma_{\min}^2\paren{\mb A}}. 
\end{align}
\end{lemma}
\begin{proof}
First note that 
\begin{align*}
\norm{\paren{\mb A + \mb \Delta}^{-1/2} - \mb A^{-1/2}}{} \le \frac{ \norm{\paren{\mb A + \mb \Delta}^{-1} - \mb A^{-1}}{}}{\sigma^{1/2}_{\min}\paren{\mb A^{-1}}}
\end{align*}
as by our assumption $\mb A + \mb \Delta \succ \mb 0$ and the fact (Theorem 6.2 in~\cite{higham2008functions}) that $\norm{\mb X^{1/2} - \mb Y^{1/2}}{} \le \norm{\mb X - \mb Y}{}/\paren{\sigma_{\min}^{1/2}\paren{\mb X} + \sigma_{\min}^{1/2}\paren{\mb Y}}$ for any $\mb X, \mb Y \succ \mb 0$ applies. Moreover, using the fact 
\begin{align*}
\norm{\paren{\mb X + \mb \Delta}^{-1} - \mb X^{-1}}{} \le \frac{\norm{\mb X^{-1}}{} \norm{\mb X^{-1} \mb \Delta}{}}{1-\norm{\mb X^{-1} \mb \Delta}{}} \le \frac{\norm{\mb \Delta}{} \norm{\mb X^{-1}}{}^2}{1-\norm{\mb X^{-1}}{} \norm{\mb \Delta}{}}
\end{align*} 
for nonsingular $\mb X$ and perturbation $\mb \Delta$ with $\norm{\mb X^{-1}}{} \norm{\mb \Delta}{} < 1$ (see, e.g., Theorem 2.5 of Chapter III in~\cite{stewart1990matrix}), we obtain 
\begin{align*}
\tfrac{1}{\sigma^{1/2}_{\min}\paren{\mb A^{-1} }}\norm{\paren{\mb A + \mb \Delta}^{-1} - \mb A^{-1}}{} \le \norm{\mb A}{}^{1/2} \frac{\norm{\mb \Delta}{} \norm{\mb A^{-1}}{}^2}{1-\norm{\mb A^{-1}}{} \norm{\mb \Delta}{}} \le \frac{2\norm{\mb A}{}^{1/2} \norm{\mb \Delta}{} }{\sigma_{\min}^2\paren{\mb A}}, 
\end{align*} 
where we have used the fact $\norm{\mb A^{-1}}{} \norm{\mb \Delta}{} \le 1/2$ to simplify at the last inequality. 
\end{proof}

\begin{lemma} \label{lem:bg_identity_diff}
There exists a positive constant $C$ such that for any $\theta \in \paren{0, 1/2}$ and $n_2 > C n_1^2 \log n_1$, the random matrix $\mb X \in \R^{n_1 \times n_2}$ with $\mb X\sim_{i.i.d.} \mathrm{BG}\paren{\theta}$ obeys 
\begin{align}
\norm{\frac{1}{n_2 \theta} \mb X \mb X^* - \mb I}{} \le 10\sqrt{\frac{\theta n_1 \log n_2}{n_2}}
\end{align}
with probability at least $1 - n_2^{-8}$. 
\end{lemma}
\begin{proof}
Observe that $\expect{\tfrac{1}{\theta}\mb x_k \mb x_k^*} = \mb I$ for any column $\mb x_k$ of $\mb X$ and so $\tfrac{1}{n_2 \theta} \mb X \mb X^*$ can be considered as a normalize sum of independent random matrices. Moreover, for any integer $m \ge 2$, 
\begin{align*}
\expect{\paren{\frac{1}{\theta} \mb x_k \mb x_k^*}^m} = \frac{1}{\theta^m} \expect{\norm{\mb x_k}{}^{2m-2} \mb x_k \mb x_k^*}. 
\end{align*}
Now $\expect{\norm{\mb x_k}{}^{2m-2} \mb x_k \mb x_k^*}$ is a diagonal matrix (as $\expect{\norm{\mb x_k}{}^2 x_k\paren{i} x_k\paren{j}} = - \expect{\norm{\mb x_k}{}^2 x_k\paren{i} x_k\paren{j}}$ for any $i \neq j$ by symmetry of the distribution) in the form $\expect{\norm{\mb x_k}{}^{2m-2} \mb x_k \mb x_k^*} = \expect{\norm{\mb x}{}^{2m-2} x(1)^2}\mb I$ for $\mb x \sim_{i.i.d.} \mathrm{BG}\paren{\theta}$ with $\mb x \in \R^{n_1}$. Let $t^2\paren{\mb x} = \norm{\mb x}{}^2 - x(1)^2$. Then if $m = 2$, 
\begin{align*}
\expect{\norm{\mb x}{}^2 x(1)^2} 
& = \expect{x(1)^4} + \expect{t^2\paren{\mb x}} \expect{x(1)^2} \\
& = \expect{x(1)^4} + \paren{n_1-1} \paren{\expect{x(1)^2}}^2 = 3\theta + \paren{n_1-1} \theta^2 \le 3n_1 \theta, 
\end{align*}
where for the last simplification we use the assumption $\theta \le 1/2$. For $m \ge 3$, 

\begin{align*}
\expect{\norm{\mb x}{}^{2m-2} x(1)^2} 
& = \sum_{k=0}^{m-1} \binom{m-1}{k} \expect{t^{2k}\paren{\mb x} x(1)^{2m-2k}} =  \sum_{k=0}^{m-1} \binom{m-1}{k}  \expect{t^{2k}\paren{\mb x}} \expect{x(1)^{2m-2k}} \\
& \le \sum_{k=0}^{m-1} \binom{m-1}{k} \bb E_{Z \sim \chi^2\paren{n_1 -1}}\brac{Z^k} \theta \bb E_{W \sim \mc N\paren{0, 1}}\brac{W^{2m-2k}} \\
& \le \theta \sum_{k=0}^{m-1} \binom{m-1}{k} \frac{k!}{2} \paren{2n_1 - 2}^k \paren{2m-2k}!! \\
& \le \theta 2^m \frac{m!}{2} \sum_{k=0}^{m-1} \binom{m-1}{k} \paren{n_1-1}^k \\
& \le \frac{m!}{2} n_1^{m-1} 2^{m-1}, 
\end{align*}
where we have used the moment estimates for Gaussian and $\chi^2$ random variables from Lemma~\ref{lem:guassian_moment} and Lemma~\ref{lem:chi_sq_moment}, and also $\theta \le 1/2$. Taking $\sigma^2 = 3n_1 \theta$ and $R = 2n_1$, and invoking the matrix Bernstein in Lemma~\ref{lem:mc_bernstein_matrix}, we obtain 
\begin{align}
\expect{\norm{\frac{1}{p\theta} \sum_{k=1}^p \mb x_k \mb x_k^* - \mb I}{} > t} \le \exp\paren{-\frac{n_2 t^2}{6n_1 \theta + 4n_1 t} + 2\log n_1}
\end{align}
for any $t \ge 0$. Taking $t = 10\sqrt{\theta n_1 \log\paren{n_2}/n_2}$ gives the claimed result. 
\end{proof}

\begin{lemma} \label{lem:sp_angle_norm}
Consider two linear subspaces $\mc U$, $\mc V$ of dimension $k$ in $\R^n$ ($k \in [n]$) spanned by orthonormal bases $\mb U$ and $\mb V$, respectively. Suppose $\pi/2 \ge \theta_1 \ge \theta_2 \dots \ge \theta_k \ge 0$ are the principal angles between $\mc U$ and $\mc V$. Then it holds that \\
i) $\min_{\mb Q \in O_k} \norm{\mb U - \mb V \mb Q}{} \le \sqrt{2-2\cos \theta_1}$; \\
ii) $\sin \theta_1 = \norm{\mb U\mb U^* - \mb V\mb V^*}{}$;\\
iii) Let $\mc U^\perp$ and $\mc V^\perp$ be the orthogonal complement of $\mc U$ and $\mc V$, respectively. Then $\theta_1(\mc U, \mc V) = \theta_1(\mc U^\perp, \mc V^\perp)$. 
\end{lemma}
\begin{proof}
Proof to i) is similar to that of II. Theorem 4.11 in~\cite{stewart1990matrix}. For $2k \le n$, w.l.o.g., we can assume $\mb U$ and $\mb V$ are the canonical bases for $\mc U$ and $\mc V$, respectively. Then 
\begin{align*}
\min_{\mb Q \in O_k} \norm{
\begin{bmatrix}
\mb I - \mb \Gamma \mb Q \\
- \mb \Sigma \mb Q \\
\mb 0
\end{bmatrix}
}{} \le 
\norm{
\begin{bmatrix}
\mb I - \mb \Gamma  \\
- \mb \Sigma \\
\mb 0
\end{bmatrix}
}{} \le \norm{
\begin{bmatrix}
\mb I - \mb \Gamma  \\
- \mb \Sigma
\end{bmatrix}
}{}. 
\end{align*}
Now by definition 
\begin{align*}
\norm{
\begin{bmatrix}
\mb I - \mb \Gamma  \\
- \mb \Sigma
\end{bmatrix}
}{}^2 
& = \max_{\norm{\mb x}{} = 1} \norm{\begin{bmatrix}
\mb I - \mb \Gamma  \\
- \mb \Sigma
\end{bmatrix} \mb x}{}^2 = \max_{\norm{\mb x}{} = 1} \sum_{i=1}^k (1 - \cos  \theta_i)^2 x_i^2 + \sin^2\theta_i x_i^2 \\
& = \max_{\norm{\mb x}{} = 1} \sum_{i=1}^k (2-2\cos \theta_i) x_i^2 \le 2- 2\cos \theta_1. 
\end{align*}
Note that the upper bound is achieved by taking $\mb x = \mb e_1$. When $2k > n$, by the results from CS decomposition (see, e.g., I Theorem 5.2 of~\cite{stewart1990matrix}). 
\begin{align*}
\min_{\mb Q \in O_k} \norm{
\begin{bmatrix}
\mb I & \mb 0 \\
\mb 0 & \mb I \\
\mb 0 & \mb 0
\end{bmatrix}
 - 
\begin{bmatrix}
\mb \Gamma & \mb 0 \\
\mb 0 & \mb I \\
\mb \Sigma & \mb 0
\end{bmatrix} 
}{} \le \norm{
\begin{bmatrix}
\mb I - \mb \Gamma  \\
- \mb \Sigma
\end{bmatrix}
}{}, 
\end{align*}
and the same argument then carries through. To prove ii), note the fact that $\sin \theta_1 = \norm{\mb U \mb U^* - \mb V \mb V^*}{}$ (see, e.g., Theorem 4.5 and Corollary 4.6 of~\cite{stewart1990matrix}). Obviously one also has 
\begin{align*}
\sin \theta_1 = \norm{\mb U \mb U^* - \mb V \mb V^*}{} = \norm{(\mb I - \mb U \mb U^*) - (\mb I - \mb V \mb V^*)}{}, 
\end{align*}
while $\mb I - \mb U \mb U^*$ and $\mb I - \mb V \mb V^*$ are projectors onto $\mc U^\perp$ and $\mc V^\perp$, respectively. This completes the proof. 
\end{proof}